\documentclass[11pt]{article}
\usepackage[utf8]{inputenc}
\usepackage[T1]{fontenc}
\usepackage[margin=1.2in]{geometry}
\usepackage{authblk}
\pdfoutput=1
\usepackage{slashed}
\usepackage{cancel}
\usepackage{comment}
\usepackage[T1]{fontenc}
\usepackage{kpfonts,baskervald}
\usepackage{amsfonts, amssymb}
\usepackage{hyperref}
\usepackage{bbm}
\usepackage[backend=biber,style=alphabetic]{biblatex}
\usepackage{amsmath}
\usepackage{stmaryrd}
\usepackage{physics}
\usepackage{bm}
\usepackage{dsfont}
\usepackage[table,xcdraw]{xcolor}
\usepackage{enumitem}

\usepackage[all]{xy}

\usepackage{adjustbox}
\usepackage{tikz}
\usepackage{tikz-3dplot}
\usetikzlibrary{positioning}
\usetikzlibrary{cd}
\usepackage{pgfplots}
\usepackage[labelfont=bf]{caption}
\usepackage{blindtext}
\pgfplotsset{compat=1.17} \usetikzlibrary{patterns,hobby}

\usepackage{amsthm}
\newtheorem{lemma}{Lemma} [section]

\newtheorem{proposition}{Proposition}[section]
\newtheorem{defn}{Definition} 
\newtheorem{thm}{Theorem}

\usepackage{hyperref}
\numberwithin{equation}{section}

\usepackage{graphicx} % Required for inserting images

\usepackage{xcolor}
\usepackage{subfig}
\usepackage[compat=1.1.0]{tikz-feynman}

\usetikzlibrary{backgrounds}
\usetikzlibrary{shadows}

\usetikzlibrary{automata,positioning}
\tikzset{snake it/.style={decorate, decoration=snake}}

\pgfarrowsdeclaredouble{doubled}{doubled}{stealth}{stealth}
\tikzset{
->-/.style={postaction={decorate,
   decoration={markings,mark=at position .5 with {\arrow{stealth};}}}
   },
->>-/.style={postaction={decorate,
   decoration={markings,mark=at position .57 with {\arrow{doubled};}}}
   },   
}

\pgfdeclarelayer{bottom}
\pgfdeclarelayer{middleb}
\pgfdeclarelayer{middlet}
\pgfdeclarelayer{top}
\pgfsetlayers{bottom,middleb,main,middlet,top}
\usepgfplotslibrary{colorbrewer}
\pgfplotsset{compat=newest}
\pgfplotsset{colormap/violet}

\pgfdeclarelayer{background layer}
\pgfdeclarelayer{foreground layer}
\pgfsetlayers{background layer,main,foreground layer}
\newtheoremstyle{boldremark}
  {6pt}
  {6pt}
  {\normalfont}
  {}
  {\bfseries}
  {.}
  {0.5em}
  {}

\theoremstyle{boldremark}
\newtheorem{remark}{Remark}

\usepackage[nameinlink]{cleveref}
\usetikzlibrary{matrix,calc,positioning,decorations.markings,decorations.pathmorphing,decorations.pathreplacing,cd}
\usetikzlibrary{arrows,decorations.markings}
\usetikzlibrary{shapes.misc}
\usetikzlibrary{shapes.multipart}
\usetikzlibrary{decorations.markings}
\usetikzlibrary{backgrounds}
\usetikzlibrary{shadows}

\usetikzlibrary{automata,positioning}
\tikzset{snake it/.style={decorate, decoration=snake}}

\pgfarrowsdeclaredouble{doubled}{doubled}{stealth}{stealth}
\tikzset{
->-/.style={postaction={decorate,
   decoration={markings,mark=at position .5 with {\arrow{stealth};}}}
   },
   ->>>-/.style={postaction={decorate,
   decoration={markings,mark=at position .75 with {\arrow{stealth};}}}
   },
    ->>>>-/.style={postaction={decorate,
   decoration={markings,mark=at position .55 with {\arrow{stealth};}}}
   },
    ->>>>>-/.style={postaction={decorate,
   decoration={markings,mark=at position .85 with {\arrow{stealth};}}}
   },
->>-/.style={postaction={decorate,
   decoration={markings,mark=at position .50 with {\arrow{doubled};}}}
   },   
   ->>>>>-/.style={postaction={decorate,
   decoration={markings,mark=at position 1 with {\arrow{stealth};}}}
   },   
}

\usetikzlibrary{shapes.geometric} 
\tikzset{cross/.style={cross out, draw=black, minimum size=2*(#1-\pgflinewidth), inner sep=0pt, outer sep=0pt},
cross/.default={4pt}}

\usepackage[bbgreekl]{mathbbol}

\usepackage{bbold}
\DeclareSymbolFontAlphabet{\mathbb}{AMSb}
\DeclareSymbolFontAlphabet{\mathbbl}{bbold}

\def\a{\ensuremath\alpha}
\def\b{\ensuremath\beta}
\def\d{\ensuremath\delta}

\def\o{\ensuremath\omega}

\def\s{\ensuremath\sigma}
\def\e{\ensuremath\epsilon}
\def\ve{\ensuremath\varepsilon}
\def\vphi{\ensuremath\varphi}

\def\de{\ensuremath\partial}

\def\mc{\ensuremath\mathcal}
\def\l{\ensuremath\lambda}
\def\L{\ensuremath\Lambda}
\def\ov{\ensuremath\overline}
\def\wt{\ensuremath\widetilde}

\def\mf{\ensuremath\mathfrak}
\def\wt{\ensuremath\widetilde}
\def\z{\ensuremath\zeta}
\def\mbb{\ensuremath\mathbb}

\newcommand{\virg}[1]{``#1''}

\newcounter{Chapcounter}

\newcommand{\chapter}[1] 
{ {\centering          
  \addtocounter{Chapcounter}{1} \Large \textbf{ \color{black} Chapter \theChapcounter: ~#1}}    
  \addcontentsline{toc}{section}{ \color{blue} Chapter:~\theChapcounter~~ #1}    
}

\hypersetup{colorlinks=true}
\hypersetup{linkcolor=black}

\usepackage[framemethod=tikz]{mdframed}
\definecolor{mycolor1}{rgb}{0.122, 0.435, 0.698}
\newmdenv[innerlinewidth=0.5pt, roundcorner=4pt,linecolor=mycolor1,innerleftmargin=6pt,
innerrightmargin=6pt,innertopmargin=6pt,innerbottommargin=6pt]{mybox1}
\definecolor{cadmiumred}{rgb}{0.89, 0.0, 0.13}
\newmdenv[innerlinewidth=0.5pt, roundcorner=4pt,linecolor=cadmiumred,innerleftmargin=6pt,
innerrightmargin=6pt,innertopmargin=6pt,innerbottommargin=6pt]{mybox2}

\definecolor{brandeisblue}{rgb}{0.0, 0.44, 1.0}
\newmdenv[innerlinewidth=0.5pt, roundcorner=4pt,linecolor=brandeisblue,innerleftmargin=6pt,
innerrightmargin=6pt,innertopmargin=6pt,innerbottommargin=6pt]{mybox3}

\definecolor{orange}{rgb}{1.0, 0.5, 0.0}
\newmdenv[innerlinewidth=0.5pt, roundcorner=4pt,linecolor=orange,innerleftmargin=6pt,
innerrightmargin=6pt,innertopmargin=6pt,innerbottommargin=6pt]{mybox4}

\usepackage{hyperref}
\definecolor{airforceblue}{rgb}{0.36, 0.54, 0.66}
\hypersetup{
    colorlinks,
    citecolor=purple,
    filecolor=black,
    linkcolor=airforceblue,
    urlcolor=red    
}
\title{\vspace{1cm}\bf{Chiral Long-Range Order in three Euclidean Lattice Gross--Neveu Models}\vspace{0.2cm}}
\date{\vspace{0.2cm} \today}
\author[1]{Simone Fabbri\thanks{{sfabbri@sissa.it}}}
\author[1]{Leonardo Goller\thanks{{lgoller@sissa.it}}}
\affil[1]{Mathematics Area, SISSA, Via Bonomea 265, 34136 Trieste, Italy}

\begin{document}

\maketitle
\begin{abstract}
We prove the existence of Long-Range Order in a class of two-dimensional Euclidean lattice Gross–Neveu models with an even number of fermion flavors, covering three standard lattice discretizations, including naive and staggered fermions widely used in numerical studies. By performing a Hubbard–Stratonovich transformation, we map the fermionic systems to bosonic ones and establish Reflection Positivity for the resulting measures. Exploiting this structure, we combine Chessboard Estimates with a Peierls-type contour argument to prove Long-Range Order for the Chirally charged fermion-mass bilinear $\overline{\psi}\psi$ at sufficiently small coupling and sufficiently large flavor number. In particular, we obtain uniform pointwise bounds on the bosonic two-point function, equivalently on the fermionic mass–mass correlator, showing that it is quantitatively controlled by the minimizers of the effective potential. This provides a rigorous, non-perturbative proof of Long-Range Order in lattice Gross–Neveu models.
\end{abstract}

\setcounter{page}{1}
{\hypersetup{linkcolor=black}
\tableofcontents
%\listoffigures
}

\newpage
\section{Introduction}
Relativistic quantum field theories of interacting fermions with four-fermion couplings provide a natural framework for the study of dynamical mass generation and spontaneous Chiral Symmetry Breaking. A prototypical example is the Nambu--Jona-Lasinio model \cite{NambuJonaLasinio1961a,NambuJonaLasinio1961b}, originally introduced as an effective theory of interacting nucleons and later recognized as a precursor of modern mechanisms of Chiral Symmetry Breaking. The model has been extensively investigated using a wide range of analytical and non-perturbative techniques, including large-$N$ expansions \cite{Klevansky1992RMP, HatsudaKunihiro1994PhysRep, Bersini:2025lxs}, functional Renormalization Group methods \cite{BergesTetradisWetterich2002PhysRep, Braun2011FRGFermions}, and lattice regularizations \cite{BitarVranas1993LatticeNJL}. Rigorous results, however, remain comparatively scarce; notable progress includes \cite{salmhofer1991proof}, where Chiral Symmetry Breaking was established using Reflection Positivity techniques.\\

The Gross--Neveu model \cite{GN} can be viewed as a paradigmatic example within this framework, retaining the essential features of Spontaneous Symmetry Breaking while offering a more tractable setting for the analysis of non-perturbative phenomena. It describes $N$ species (flavors) of Dirac fermions coupled to each others via a quartic mass-mass local interaction. In $1+1$ dimensions it is renormalizable, asymptotically free, and admits a formal solution in the large-$N$ limit, which relies upon a saddle-point argument. This yields a detailed analytic description of the dynamically generated mass gap and of the spectrum of bound states \cite{GN,DHN1975}, and establishes the spontaneous breaking of the discrete Chiral symmetry, with the fermion condensate $\langle \overline{\psi} \psi \rangle$ as the natural order parameter.\\

The method of Reflection Positivity originated in the Osterwalder--Schrader axiomatization of Euclidean quantum field theory \cite{OS1,OS2}, as the key positivity condition allowing the reconstruction of a relativistic quantum theory from its Euclidean correlation functions. Its formulation for Euclidean fermionic field theories was subsequently developed by Fröhlich and Osterwalder \cite{Frohlich:1974zs}, while its implementation in lattice field theory, including Lattice Gauge Theories and fermionic models, was further developed by Osterwalder and Seiler \cite{OSTERWALDER1978440} (see also \cite{SeilerBook} for a comprehensive review). Soon afterwards, Fr\"ohlich, Simon and
Spencer recognized that the same positivity property could be exploited
in Classical Statistical Mechanics to prove phase transitions
\cite{FSS}. This point of view was subsequently developed by
Fr\"ohlich, Israel, Lieb and Simon into a general method for classical
and quantum lattice systems, introducing the machinery of successive
reflections and Chessboard Estimates
\cite{cmp/1103904299, Dyson1978-ef, frohlich1980phase, frohlich1978phase}. Since then, Reflection
Positivity has become one of the main non-perturbative tools for proving
infrared bounds, phase transitions, Long-Range Order and Symmetry
Breaking in a wide variety of lattice models, including quantum spin and fermionic systems. A limitation of Reflection Positivity is that it relies on a suitable positivity condition for the Hamiltonian (or action) of the system, which restricts the class of models that can be considered. In parallel, Renormalization Group and Constructive Field Theory methods have provided robust, rigorous approaches to interacting fermionic quantum field theories \cite{gawccdzki1985gross, BG90,BGPS,FelBook,BM02, benfatto2005ward, MasBook}.
These methods allow to construct a large class of quantum field theories and to analyze scaling limits of critical statistical mechanics models; however, their application in proving phase transitions and symmetry breaking is substantially heavier than Reflection Positivity. These developments provide part of the mathematical background for the present work.\\

The Gross--Neveu model has played an important role in the development of constructive and Renormalization Group techniques for fermionic quantum field theories. Its short-distance (or ultraviolet) construction has been carried rigorously in \cite{FMRS85,gawccdzki1985gross, DR00} (see also the recent works \cite{2024AnHP...25.5113D, Duch24, dimock2025correlationfunctionsgrossneveumodel}), where asymptotic freedom was established from a rigorous viewpoint. On the other hand the long-distance (or infrared) analysis has been successfully performed by Kopper, Magnen and Rivasseau \cite{KMR95}, via the combination of Cluster Expansion techniques and large-$N$ analysis for a suitably regularized version of the continuum model.
 See also \cite{K99} where a similar approach was followed for the study of mass generation in the context of Non-Linear Sigma Models.
It is worth stressing, however, that up to now (this work included) there is no mathematical proof of Chiral Symmetry Breaking for \emph{any} version of the Gross--Neveu model in the removed ultraviolet-cutoff limit, which remains a main open problem.\\

In this work, we address the infrared problem by proving the existence of Long-Range Order (LRO) in a class of lattice regularizations of the two-dimensional Euclidean Gross-Neveu model, with a large, even number $N$ of fermion flavors and fixed lattice spacing. We consider three standard lattice discretizations, originally introduced by Cohen, Elitzur and Rabinovici \cite{cohen1981monte, Cohen:1983nr}, including naive fermions and staggered fermions with different interaction structures. These models correspond to formulations widely used in numerical simulations, see e.g. \cite{WU06} and the more recent papers \cite{LPWWW20,SCA22}.
Our result is conceptually comparable with those of \cite{KMR95,K99}, even though we follow a different approach, mainly based on Reflection Positivity and a Peierls' argument. Reflection Positivity methods have the advantage of being more simple and direct compared to Cluster Expansion techniques; however they are less robust with respect to small modifications of the model, as they require a rather specific form of the action (or alternatively the Hamiltonian).
\\

In the same way as \cite{KMR95}, our starting point consists of a Hubbard--Stratonovich transformation,
which introduces an auxiliary scalar field coupled to the Grassmann
fermions. Reflection Positivity is first established for the
Grassmann--boson system arising after the Hubbard--Stratonovich
transformation and is then transferred to the corresponding effective
bosonic measure, providing the starting point for the application of
Chessboard Estimates \cite{cmp/1103904299, frohlich1978phase}, a technique originally developed in the context of classical spin systems and then applied to quantum spin systems and also to fermionic models\footnote{
The Hamiltonian formulation of fermionic Reflection Positivity developed
by Fr\"ohlich, Israel, Lieb and Simon is technically different from the
present Euclidean Grassmann framework. In particular, the proof in
\cite{frohlich1980phase} relies on a Jordan--Wigner-type transformation
and was later revisited by Lieb \cite{Lieb_1994}, Macris and
Nachtergaele \cite{Macris_1996} because of a sign problem for the
standard nearest-neighbour hopping Hamiltonian . This Hamiltonian version
of Reflection Positivity has subsequently found numerous applications in
Mathematical Physics, including
\cite{PhysRevB.51.4777,Koma:2022mua,Goto_2023,Goto_2024,
goto2025antiferromagneticlongrangeorderlattice,
Goller2026-xy,bachmann2026anyonspifluxphasefermionic}. By contrast, the
present work is formulated entirely within the Euclidean Grassmann
framework and is closer in spirit to the constructive field-theoretic
approach developed in \cite{OSTERWALDER1978440,SHARATCHANDRA1981205}.
} \\

The main result of the paper is the derivation of uniform upper and lower bounds on the two-point correlation function of the auxiliary field, implying Long-Range Order for the fermionic bilinear $ \overline{\psi}\psi$. The proof combines Reflection Positivity with a Peierls-type argument \cite{frohlich1978phase}, controlling configurations which interpolate between distinct minima of the effective potential and showing that they are exponentially suppressed in the number of fermion flavors. The bounds are formulated in terms of the minimizers of the effective potential and show that the Long-Range Order parameter converges, up to controlled errors, to its mean-field prediction in the large-$N$ limit. This provides a rigorous justification of the standard large-$N$ picture for Symmetry Breaking in lattice Gross--Neveu models. However, it should be pointed out that our bounds are not uniform in the lattice spacing $\ell$ (see Remark \ref{rmk:main}.\ref{ite:4}), which should be regarded as an ultraviolet cutoff; thus our result provides a description of the behavior of the model only in the presence of a fixed ultraviolet cutoff, or in a combined scaling limit: $\ell\to 0^+$ and $N\to\infty$.
\\

The present result should also be compared with the recent work of Goto and Koma on Chiral Symmetry Breaking in quantum Nambu--Jona-Lasinio models \cite{Goto_2023}. While their analysis is carried out in the operator formalism and applies to spatial dimensions $d\ge 3$ at positive temperature and to $d=2$ at zero temperature, our setting, similarly to \cite{salmhofer1991proof}, is Euclidean and formulated directly in terms of Grassmann fields over a two-dimensional discrete space-time. Moreover, the proof of \cite{Goto_2023} (as well as its Euclidean Lattice counterpart \cite{salmhofer1991proof}) relies on Infrared Bounds, which, however, are not expected to be applicable in our setting, since the integral of the propagator that naturally appears in the lower bound of the $2$-point function is log-divergent. A natural tool in our context is rather given by the Peierls argument. However, it is not clear apriori how such geometric contour method can be implemented directly in a Grassmann formulation. This difficulty is one main motivation for introducing the Hubbard--Stratonovich representation, which reformulates the problem in terms of auxiliary bosonic fields and makes Peierls-type arguments accessible.\\

A key feature of our analysis is its robustness with respect to different lattice discretizations. Although the models differ in their fermionic content and ultraviolet behavior, the proof applies in the same way to all the different variants once a Hubbard--Stratonovich representation is available. This indicates that the mechanism leading to Long-Range Order is not tied to a specific formulation, but rather reflects a structural property of a broader class of fermionic systems admitting an auxiliary-field description.\\

Our contributions can be summarized as follows:
\begin{enumerate}
    \item\label{i1} we establish \textit{Reflection Positivity} for a class of \textit{lattice Gross--Neveu models};
    \item\label{i2} we adapt \textit{Chessboard Estimates} to the \emph{bosonic model} obtained after the Hubbard-Stratonovich transformation; 
    \item\label{i3} we prove \textit{uniform, pointwise upper and lower bounds} on the \textit{bosonic two-point function} (i.e. the mass-mass fermionic correlator), implying \textit{Long-Range Order}.
\end{enumerate}

Specifically, Item \ref{i3} should be regarded as a rigorous, quantitative counterpart of the qualitative prediction arising from a mean-field analysis similar to the one performed in \cite{GN}, placing earlier heuristic and numerical evidence \cite{cohen1981monte,Cohen:1983nr} on a firm mathematical footing.\\

The paper is organized as follows. In Section \ref{sec2} we introduce the lattice models and their bosonic (Hubbard-Stratonovich) formulations and state the main result. In Section \ref{sec3} we establish Reflection Positivity of the effective measures. In Section \ref{sec4} we collect the main properties which lead to the Long-Range Order, and we show how they imply the main theorem. Section \ref{sect:contrib} is devoted to the estimates of some remarkable bosonic expectations, based on Chessboard Estimates and Peierls-type arguments, which constitute a main building block for the proof of Long-Range Order. Appendix \ref{app:Dirac} is devoted to the characterization of the lattice Dirac operator, whose properties are essential for the analysis carried out in Section \ref{sec3}. In Appendix \ref{app:pot} we perform the study of the effective bosonic potential, which plays a central role throughout the paper. In Appendix \ref{app:continuum_sp} we compute the formal continuum limit for our third model under consideration, namely the \emph{Staggered-Plaquette} model, which, differently from the other two, formally reduces to the standard Gross-Neveu model, as the lattice spacing is sent to zero. Finally, in Appendix \ref{applbub} we discuss the bounds for the several contributions to the bosonic two-point function, serving as a main technical ingredient for the proof of Long-Range Order in Section \ref{sec4}.

\paragraph{Outlook.} A natural direction for future work is the construction of infinite-volume Gibbs states associated with the symmetry-broken phases and the characterization of Spontaneous Symmetry Breaking in terms of extremal translation-invariant measures. In particular, it would be desirable to complement the present results on Long-Range Order with a direct construction of distinct infinite-volume states exhibiting a non-vanishing order parameter.\\

Another important problem concerns the identification of a dynamically generated mass gap through the decay properties of correlation functions. Establishing exponential decay of suitable truncated correlators would provide a direct characterization of mass generation within the present framework.\\

Finally, it would be extremely interesting to combine our infrared result with a multiscale analysis able to control the ultraviolet regime of the model, thus allowing to prove Long-Range Order uniformly with respect to the ultraviolet cutoff (namely the inverse of the lattice spacing). This should be regarded as a major open problem, as it would lead to a complete picture of the phenomenon of Spontaneous Symmetry Breaking in the standard Gross-Neveu model, defined as the continuum limit of a suitable lattice regularization thereof.

\paragraph{Acknowledgements.} LG and SF would like to thank Marcello Porta and Alessandro Giuliani for their supervision, insightful discussions, and critical feedback. LG is grateful to J{\"u}rg Fr{\"o}hlich for references on Reflection Positivity and phase transitions, to Matteo Delladio for enlightening conversations on the continuum limit of lattice models, and to Amartya Harsh Singh and Alessandro Piazza for discussions on large-$N$ limits in strongly coupled QFTs. The authors acknowledge financial support from the ERC-StG MaMBoQ, grant agreement n.802901, and the  MUR, PRIN 2022 project MaIQuFi cod. 20223J85K3. 

\paragraph{Statements and Declarations.}
All authors declare that they have no conflicts of interest to disclose.

\paragraph{Data Availability.} Data sharing is not applicable to this article as no datasets were generated or analyzed during the current study.

\section{The Models}\label{sec2}

\paragraph{Conventions.}

We work on the two-dimensional discrete torus which we represent with as its fundamental domain given by the square lattice

\[
\Lambda_L
\doteq
\mathbb Z^2 \cap
\big(-\tfrac L2,\tfrac L2\big]^2,
\]

with lattice spacing set equal to one and $L\in4\mathbb N$. 
For the plaquette formulation introduced below we shall additionally assume
$L\in 8\mathbb N$.
Lattice sites are denoted by
$x=(x_0,x_1)$.\\

Throughout the paper we consider three lattice realizations of the Gross--Neveu model:

\begin{enumerate}
\item The Naive-Fermion model (NN);
\item The Staggered-Fermion model (SN);
\item The Staggered-Fermion Plaquette model (SP).
\end{enumerate}

Despite their different microscopic realizations, the three models share the same overall structure. In each case, we define a lattice Dirac operator $\slashed{D}_{\a}$, introduce a quartic Gross--Neveu-like interaction, and construct the corresponding fermionic Gibbs state. We then analyze the relevant symmetries and spectral properties and derive an equivalent bosonic representation through a Hubbard--Stratonovich transformation. The arguments developed later apply to all three formulations with only minor modifications.

\paragraph{General Set-Up.}

The Grassmann fields
$(\psi,\ov{\psi})$ satisfy anti-periodic boundary conditions in both lattice directions. These boundary conditions ensure that hopping terms crossing the reflection plane enter with the sign required by the Osterwalder--Seiler Reflection Positivity construction.

The fields carry an internal index $J$, whose meaning depends on the formulation:

\[
J=
\begin{cases}
(s,j)
&
\alpha=\mathrm{NN}\\
j
&
\alpha=\mathrm{SN}
\\
(a,j)
&
\alpha=\mathrm{SP}
\end{cases}
\]

For the naive model,
$s=1,2$ is a spin index and
$j=1,\cdots,N$ 
is a flavour index. For the staggered model only the flavour index is present. For the plaquette formulation,
$a\in\mathcal I=\{A,B,C,D\}$ 
labels the four internal degrees of freedom associated with a $2\times2$ unit cell, while $j$ again denotes the flavour index.\\

We shall always assume that $N$ is even. This assumption is used only to guarantee positivity of the effective bosonic measure (see below) obtained after integrating out the fermions. Whether the restriction can be removed would require different techniques and will not be investigated here.\\

For the NN and SN formulations the fermionic fields are indexed by sites of $\Lambda_L$. In the plaquette formulation it is convenient to introduce the coarse lattice

\[
\widetilde\Lambda_L
\doteq
(2\mathbb Z)^2
\cap
\big(-\tfrac L2,\tfrac L2\big]^2,
\]

whose sites label the centers of the unit cells. Accordingly, we define:

\[
\Lambda_L^\alpha =
\begin{cases}
\Lambda_L\,\,\,\,\,\,& \alpha = \mathrm{NN}\\
\Lambda_L\,\,\,\,\,\,& \alpha = \mathrm{SN} \\
\widetilde{\Lambda}_L & \alpha = \mathrm{SP}
\end{cases}
\]

and we denote by $|\L^{\a}_L|$ its cardinality, namely $|\L^{\a}_L|=L^2$ for $\a= $NN, SN and $|\L^{\a}_L|=L^2/4$ for $\a=$SP. For notational simplicity, we use a unified notation and make the dependence on the formulation explicit only when needed. In particular, we write:
\[
(\overline{\psi}\psi)_x \doteq \sum_J \overline{\psi}_{x,J}\psi_{x,J}
\]
 Given a fermionic action $S_{\alpha}(\overline{\psi},\psi)$, the expectation of an even Grassmann polynomial $\mathcal O(\overline{\psi},\psi)$ is defined by: 
\[ \langle\mathcal O\rangle_{\alpha} \doteq \frac{1}{Z^{\alpha}_{\Lambda_L}} \int d\overline{\psi}\,d\psi\; \mathcal O(\overline{\psi},\psi)\,e^{-S_{\alpha}(\overline{\psi},\psi)}, \qquad Z^{\alpha}_{\Lambda_L} \doteq \int d\overline{\psi}\,d\psi\;e^{-S_{\alpha}(\overline{\psi},\psi)}, \] where $d\overline{\psi}\,d\psi$ denotes the product of Grassmann differentials over all lattice sites and internal indices (spin, flavor, and unit-cell degrees of freedom) and $\alpha$ indicates each possible model. 

\medskip

\paragraph{Effective bosonic measure.}

In all three formulations the interaction is local and depends only on the composite field
$(\overline\psi\psi)^2_x$. We can therefore introduce an auxiliary real field $\phi$ through the Hubbard--Stratonovich identity

\begin{equation}
\label{eqn:Hub_Strat}
e^{\frac{\lambda}{2N}(\overline{\psi}\psi)_x^2} =
\sqrt{\frac{N}{2\pi \lambda}}
\int_{-\infty}^{\infty} d\phi_x
e^{-\frac{N}{2\lambda}\phi_x^2 + \phi_x(\overline{\psi}\psi)_x}.
\end{equation}

Applying this transformation independently at every lattice site, integrating out exactly the fermionic fields and absorbing $\phi$-independent normalization constants into the partition function, one obtains an effective bosonic measure\footnote{$d\mu_{\a}$ is actually a probability measure, as $N$ is even and the determinant in the right--hand side of \eqref{eq:boson_measure} is real, as shown in Appendix \ref{app:Dirac}.\\} of the form:

\begin{equation}
\label{eq:boson_measure}
d\mu_{\alpha}(\phi) \doteq \frac{1}{Z^{\alpha}_{\Lambda_L}} \prod_{x \in \L_L^{\a}} d\phi_x \; e^{-\frac{N}{2\lambda} \phi_x^2} \det((\slashed{D}_{\alpha})^{-}_{\Lambda_L}-M_\phi)^N. \end{equation} 

Here $(\slashed{D}_{\alpha})^{-}_{\Lambda_L}$ denotes the appropriate lattice Dirac operator (including anti-periodic boundary conditions) used in each model while $M_\phi$ is the multiplication operator:  

\begin{align*}
(M_{\phi})_{(x,s),(y,s')}&\doteq \d_{s,s'}\d_{x,y}\phi_x, \quad \forall x,y\in\L_L, \; s,s'\in\{1,2\} \quad \, (\a=\mathrm{NN});\\ 
(M_{\phi})_{x,y}&\doteq \d_{x,y}\phi_x, \quad \,\,\,\,\,\,\,\,\,\,\,\,\,\,\forall x,y\in\L_L, \quad \quad \quad \quad \quad \,\,\,\,\,\,\,\,(\a=\mathrm{SN}); \\
(M_{\phi})_{(x,a),(y,b)}&\doteq \d_{a,b}\d_{x,y}\phi_x, \quad \forall x,y\in\wt{\L}_L,\; a,b\in\mc{I} \quad \quad \,\,\,\,\,\, (\a=\mathrm{SP}).
\end{align*}

The determinant is computed by regarding  $(\slashed{D}_{\alpha})^{-}_{\Lambda_L}-M_\phi$ as a square matrix (of size $L^2\times L^2$, if $\a=\mathrm{SN},\mathrm{SP}$, and $2L^2\times2L^2$, if $\a=\mathrm{NN}$).

\paragraph{Order parameter.}

Differentiating the Hubbard--Stratonovich identity with respect to $\phi_x$ and intergrating by parts yields

\begin{equation}\label{eq:HS-2pt-common}
\big\langle(\overline{\psi}\psi)_x
(\overline{\psi}\psi)_y\big\rangle_{\alpha}= 
\Big(\frac{N}{\lambda}\Big)^2
\langle\phi_x\phi_y\rangle_{\mu_{\alpha}}
-\frac{N}{\lambda}\delta_{x,y}.
\end{equation}

Hence it is natural to identify the auxiliary field $\phi$ as the Chiral order parameter. Indeed, proving Long-Range Order for the fermionic bilinear $\overline{\psi}\psi$, namely (see \cite{FSS, cmp/1103904299}) 

\begin{equation*} \liminf_{L \to \infty} \frac{1}{|\Lambda^{\a}_L|^2} \sum_{x,y \in \Lambda^{\a}_L} \left \langle (\overline{\psi}\psi)_x (\overline{\psi} \psi)_y \right \rangle_{\alpha} > C 
\end{equation*}

for some $C>0$ independent of the volume $|\Lambda^{\a}_L|$, is equivalent\footnote{Note that in the thermodynamic limit $|\Lambda_L| \to \infty$, the contact term in the right--hand side of \eqref{eq:HS-2pt-common} is negligible in view of establishing LRO.} to proving the corresponding lower bound for the bosonic two-point function: \begin{equation*} \frac{1}{|\Lambda^{\a}_L|^2} \sum_{x, y \in \Lambda^{\a}_L} \left \langle \phi_x \phi_y \right \rangle_{\mu_{\alpha}} >C', \end{equation*}

Therefore, throughout the paper we shall work primarily with the effective bosonic measure $\mu_\alpha$, establishing Long-Range Order for the auxiliary field $\phi$.

\subsection{Naive Fermions with Naive Interaction (NN)}

The first model we consider is the lattice version of the Gross--Neveu model originally introduced
in \cite{cohen1981monte}, formulated in terms of naive lattice fermions and therefore
subject to the usual fermion--doubling phenomenon. This formulation provides the
simplest discretization of the fermionic action and serves as a natural starting
point for our analysis.\\
\paragraph{Spinors.} At each lattice site $x \in \Lambda_L$, we introduce $2N$ independent Grassmann-valued spinor fields $\psi_{x,j}$ and $\overline{\psi}_{x, j}$ ($j=1, \cdots,N$) with $2$ entries, where $\psi_{x,j}$ is regarded as a column vector and $\overline{\psi}_{x,j}$ as a row vector:

 \begin{equation*}
     \psi_{x,j} = \begin{pmatrix}
          \psi_{x,1,j} \\
            \psi_{x,2,j}\end{pmatrix},\,\,\,\,\,\,\,\,\,\,\,\,\,\,\,\,\,\, \,\,\,\,\,\,\,\,\, \,\,\,\,\,\,\,\,\,  \overline{\psi}_{x,j} = \begin{pmatrix}
          \overline{\psi}_{x,1,j} &
          \overline{ \psi}_{x,2,j}
     \end{pmatrix}.
 \end{equation*}

 \paragraph{Gamma matrices.} We choose two self-adjoint matrices $\{\gamma_\mu\}_{\mu=0,1} \subset M_2(\mathbb{C})$ as a possible irreducible representation of the
$2$-dimensional Euclidean Clifford algebra $\{\gamma_\mu,\gamma_\nu\} = 2\,\delta_{\mu\nu}\,\mathbbl{1}_2$. The Axial matrix $\gamma_A$ is defined as $\gamma_A \doteq i \gamma_0 \gamma_1$ and satisfies:

\begin{equation}
    \gamma_A^2=\mathbbl{1}_2, \qquad \{\gamma_A,\gamma_\mu\}=0, \qquad \gamma_A^{\dagger}=\gamma_A.
     \tag{GA}
\label{eq:GA1}
\end{equation}

In this paper, we adopt the following choice for the $\gamma$-matrices:

\begin{equation*}
    \gamma_0 = \sigma^y \equiv \left(\begin{array}{cc}
        0 & -i \\
        i & 0
    \end{array}\right), \qquad\gamma_1 = \sigma^x\equiv \left(\begin{array}{cc}
        0 & 1 \\
        1 & 0
    \end{array}\right),\qquad\gamma_A = \sigma^z \equiv \left(\begin{array}{cc}
        1 & 0 \\
        0 & -1
    \end{array}\right).
\end{equation*}

\paragraph{Dirac operator.} The Lattice Dirac operator is chosen to correspond with its naive discretization,
acting on $\ell^2(\mathbb{Z}^2, \mathbb{C}^{2})$ according to:

\begin{equation*}
    \slashed{D}_{\mathrm{NN}} = \frac{1}{2} \sum_{\mu=0}^{1} \gamma_{\mu} \otimes (T_{\mu}-T_{\mu}^*),
\end{equation*}

where $T_{\mu}$ is the right translation operator in the $\mu$-th direction, namely:

\begin{equation*}
(T_{\mu}\psi)_x=\psi_{x+\hat{e}_{\mu}}\,\,\,\,\forall \psi\in\ell^2(\mbb{Z}^2).
\end{equation*}

The Fourier representation for the Dirac operator is
\begin{equation}
\label{eq:model_NN}
(\slashed{D}_{\mathrm{NN}})_{x,y}= -i\sum_{\mu=0,1} \iint_{(-\pi,\pi]^2} \frac{d^2p}{(2\pi)^2} e^{-ip\cdot(x-y)}  \gamma_{\mu} \sin(p_{\mu}).     
\end{equation}

 We impose \emph{anti-periodic} boundary conditions on the fermion field in both directions: specifically, these are implemented by introducing the anti-periodic version of the lattice Dirac operator\footnote{In this paper by convention the Grassmann variables $\ov{\psi},\psi$ are instead regarded with \emph{periodic} boundary conditions, namely $\ov{\psi}_{x+nL}=\ov{\psi}_x$ and $\psi_{x+nL}=\psi_x$ for every $n\in\mbb{Z}^2$.}, denoted by $(\slashed{D}_{\mathrm{NN}})_{\Lambda_L}^-$:

\begin{equation*}
\big((\slashed{D}_{\mathrm{NN}})_{\Lambda_L}^-\big)_{x,y} = \sum_{n \in \mathbb{Z}^2} (-1)^{n_0+n_{1}} (\slashed{D}_{\mathrm{NN}})_{x+Ln, y}.
\end{equation*}

\begin{figure}[h]
    \centering
\begin{tikzpicture}
\draw[dotted, ->-] (-0.5,-0.5) -- (3.5,-0.5);
\draw[dotted, ->>-] (3.5,-0.5) -- (3.5,3.5);
\draw[dotted, ->-] (-0.5,3.5) -- (3.5,3.5);
\draw[dotted, ->>-] (-0.5,-0.5) -- (-0.5,3.5);
\draw[] (0,0) -- (3,0);
\draw[line width = 0.05 cm] (0, 0) -- (-0.5,0);
\draw[line width = 0.05 cm] (0,2) -- (-0.5,2);
\draw[line width = 0.05 cm] (3,2) -- (3.5,2);
\draw[line width = 0.05 cm] (3,0) -- (3.5,0);
\draw[] (0, 2) -- (3,2);
\draw[line width = 0.05 cm] (3,1) -- (3.5,1);
\draw[line width = 0.05 cm] (-0.5,1) -- (0,1);
\draw[line width = 0.05 cm] (-0.5,3) -- (0,3);
\draw[line width = 0.05 cm] (3,3) -- (3.5,3);
\draw (0,3) -- (3,3);
\draw (0,1) -- (3,1);
\draw (0,-0.5) -- (0,3.5);
\draw (1,-0.5) -- (1,3.5);
\draw (2,-0.5) -- (2,3.5);
\draw (-0.5,2) -- (0,2);
\draw (3,2) -- (3.5,2);
\draw (-0.5,0) -- (0,0);
\draw (3,0) -- (3.5,0);
\draw (3,-0.5) -- (3,3.5);

\draw[line width = 0.05 cm] (0,-0.5) -- (0,0);
\draw[line width = 0.05 cm] (1,-0.5) -- (1,0);
\draw[line width = 0.05 cm] (2,-0.5) -- (2,0);
\draw[line width = 0.05 cm] (3,-0.5) -- (3,0);

\draw[line width = 0.05 cm] (0,3.5) -- (0,3);
\draw[line width = 0.05 cm] (1,3.5) -- (1,3);
\draw[line width = 0.05 cm] (2,3.5) -- (2,3);
\draw[line width = 0.05 cm] (3,3.5) -- (3,3);
\draw[thin, red] (-0.5,-0.5) -- (-0.5,3.5);
\draw[thin, red] (1.5,-0.5) -- (1.5,3.5);
\node[below, red] (a) at (-0.5,-0.6) {{$\partial_- (\L_L)_+$}};
\node[below, red] (a) at (1.5,-0.6) {{$\partial_+ (\L_L)_+$}};
\node[above] (a) at (0.5,3.6) {{$(\L_L)_+$}};
\node[above] (a) at (2.5,3.6) {{$(\L_L)_-$}};
\draw[->] (-1,-0.5) -- (-1,3.5);
\node[above] (a) at (-1,3.5) {{$x_0$}};
\end{tikzpicture}
\caption{Distribution of Phases of the Anti-symmetrized Forward Dirac Operator $\gamma_{\mu} \otimes T_{\mu}$ in $L=2$ Lattice. Solid Links represent a $-1$ sign. Notice that the set of edges intersecting $\partial_- (\L_L)_+$ are opposite of those intersecting $\partial_+ (\L_L)_+$ because we have chosen an anti-symmetrization convention where the $\mathbb{Z}_2$-twisting is localized on $\partial_- (\L_L)_+$. By means of any $\mathbb{Z}_2$-Gauge Transformation, we can move such line wherever we desire, without changing the physical properties of the system.}
\label{fig1N}
\end{figure}
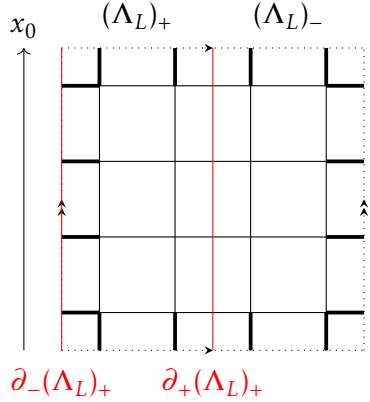

\begin{lemma}[Properties of the NN Dirac operator]\label{PND}
    
The operator $(\slashed{D}_{\mathrm{NN}})_{\Lambda_L}^-$ has the following properties, for every $\phi\in \mathbb{R}^{\Lambda_L}$.    

\begin{enumerate}
        \item \textbf{{Anti--self-adjointness:}} 
        
        \begin{equation}    \big((\slashed{D}_{\mathrm{NN}})_{\Lambda_L}^-\big)^* = - (\slashed{D}_{\mathrm{NN}})_{\Lambda_L}^-.
            \tag{PND1}
            \label{PND1}
        \end{equation}
        \item \textbf{{Chirality under adjoint:}}
\begin{equation}
     \big((\slashed{D}_{\mathrm{NN}})^-_{\Lambda_L}-  M_\phi\big)^* = (\gamma_A \otimes 1)\big((\slashed{D}_{\mathrm{NN}})^-_{\Lambda_L}-  M_{\phi}    \big) (\gamma_A\otimes 1),
      \tag{PND2}
      \label{PND2}
\end{equation}
where $\gamma_A \otimes 1: \mathbb{C}^2\otimes\ell^2(\Lambda_L) \to \mathbb{C}^2 \,\otimes\, \ell^2(\Lambda_L)$ acts as multiplication by $\gamma_A$ on the spinor indices.
  \item \textbf{{Spectrum of the adjoint:}} 
  \begin{equation}
      \mathrm{spec}\Big[\big((\slashed{D}_{\mathrm{NN}})^-_{\Lambda_L}- M_{\phi} \big)^*\Big] = \mathrm{spec}\Big[(\slashed{D}_{\mathrm{NN}})^-_{\Lambda_L}- M_{\phi}\Big].
       \tag{PND3}
            \label{PND3}
  \end{equation}
  
       \item \textbf{Reality of the determinant:}
       \begin{equation}
\det\big((\slashed{D}_{\mathrm{NN}})^-_{\Lambda_L}- M_{\phi}\big) \in \mathbb{R}.
 \tag{PND4}
            \label{PND4}
       \end{equation}
    \end{enumerate}
\end{lemma}

The proof of Lemma \ref{PND} is discussed in Appendix \ref{Ap21}.

\paragraph{Lattice action.} The Dirac fermions are coupled through a local four-fermion interaction of Gross--Neveu type:

\begin{equation}\label{NaiveA}
    S_{\mathrm{NN}}(\overline{\psi}, \psi) = \sum_{x, y \in \Lambda_L} \sum_{j=1}^N \overline{\psi}_{x,j} \big((\slashed{D}_{\mathrm{NN}})_{\Lambda_L}^-\big)_{x,y} \psi_{y,j} - \frac{\lambda}{2N} \sum_{x \in \Lambda_L} \bigg(\sum_{j=1}^N \overline{\psi}_{x,j} \psi_{x,j} \bigg)^2, \qquad \l>0,
\end{equation}

where all spinor products have to be understood as row times column products, such as:

\begin{equation*}
    \overline{\psi}_{x,j} \psi_{x,j}  = \sum_{s=1}^{2}  \overline{\psi}_{x,s,j} \psi_{x,s,j}.
\end{equation*}

\paragraph{Chiral symmetry.} The action is invariant under the discrete $\mathbb{Z}_2$-Chiral symmetry defined by:

\begin{equation}
\label{eq:Chiral_1}
    \begin{cases}
        \psi_{x,j} \to \gamma_A \psi_{x,j},\\
        \overline{\psi}_{x,j} \to -\overline{\psi}_{x,j} \gamma_A,
    \end{cases}
\end{equation}

while the fermion bilinear $\overline{\psi}_{x, j}\psi_{x,j}$ obviously changes sign as one can easily infer from the first two of \eqref{eq:GA1}. 

\begin{remark}
    Alternatively to \eqref{eq:Chiral_1}, one can consider the following transformation:

    \begin{equation*}
    \begin{cases}
        \psi_{x,j} \to e^{i \frac{\pi}{2}\gamma_A} \psi_{x,j} = i \gamma_A \psi_{x,j}\\
        \overline{\psi}_{x,j} \to \overline{\psi}_{x,j} e^{i \frac{\pi}{2}  \gamma_A} = i \overline{\psi}_{x,j} \gamma_A.
    \end{cases}
\end{equation*}
 Such transformation seems apparently a $\mathbb{Z}_4$-symmetry. However, when restricted to even polynomials (which are the actual observables) it reduces to a $\mathbb{Z}_2$-symmetry.    
\end{remark}

\paragraph{Effective bosonic measure.} By applying the general Hubbard--Stratonovich construction \eqref{eqn:Hub_Strat} to the action \eqref{NaiveA}, and integrating out the fermionic fields, we obtain the effective bosonic measure on $\mathbb{R}^{\Lambda_L}$:

\begin{equation*}
    d \mu_{\mathrm{NN}}(\phi) = \frac{1}{Z^{\mathrm{NN}}_{\Lambda_L}} \prod_{x \in \Lambda_L} d \phi_x e^{-\frac{N}{2\lambda}\phi_x^2} \det((\slashed{D}_{\mathrm{NN}})_{\Lambda_L}^- -  M_{\phi})^N. %\doteq \frac{1}{Z^{\mathrm{NN}}_{\Lambda_L}}d \phi e^{-\frac{N}{2\lambda}\phi^2} \det((\slashed{D}_{\mathrm{NN}})_{\Lambda_L}^- - \mathbbl{1}_2 \otimes M_{\phi})^N 
\end{equation*}

\paragraph{Issues, strong coupling and scaling limit.} As anticipated above, the naive lattice formulation of the Gross--Neveu model suffers
from the well-known fermion--doubling problem \cite{PhysRevD.11.395,NIELSEN1981173},
implying that the number of low-energy (long-wavelength) Dirac fermions described by
the free action is four, rather than one as suggested by the continuum theory. It has
been pointed out that such a proliferation of fermionic degrees of freedom in free
lattice theories may lead to significant difficulties once interactions among the
long-wavelength modes are introduced \cite{PhysRevD.16.3031}. These issues were
analyzed in detail in \cite[Sect.~3.1]{Cohen:1983nr}, where it was argued that the
interaction spoils Lorentz invariance at the quantum level. Nevertheless, at least heuristically, this model is still expected to exhibit asymptotic freedom and infrared properties analogous to those of the continuum Gross--Neveu model, up to an overall factor of four arising from the additional poles in the free propagator \cite{GUERIN1980168,Cohen:1983nr}. \\

Motivated by these difficulties, alternative lattice formulations based on staggered
fermions \cite{PhysRevD.11.395} were introduced in \cite{Cohen:1983nr}. While these models share the same
continuum limit and reproduce the Gross--Neveu action, they exhibit
different behaviors in the strong-coupling regime. One of these two formulations, the SP model, will be the third model under study of this paper. We first introduce a
simpler staggered-fermion model, the SN model, which, despite failing to reproduce the correct
continuum limit, still captures the appropriate infrared physics
(see \cite[Sect.~3]{Hands_1993} and \cite{JOLICOEUR1986431}) and has the simplest presentation among all three.

\subsection{Staggered fermions with Naive Interaction (SN)}

The second model we consider is formulated in terms of staggered fermions \cite{PhysRevD.11.395}.
Compared with the naive discretization, spin degrees of freedom are encoded in site-dependent phases, while the interaction remains on-site and of Gross--Neveu type. 
\paragraph{Staggered Fermions.} At each lattice site $x \in \Lambda_L$, we introduce $2N$ independent Grassmann variables ${\psi}_{x,j}$ and $\overline{\psi}_{x,j}$, with $j=1,\dots,N$.
\paragraph{Gamma matrices.} Gamma Matrices are represented by phases:

\begin{equation*}
    \Gamma_{\mu}(x) = \begin{cases}
        1, \,\,\,\,&\mu=0\\
        (-1)^{x_0}, &\mu=1.
    \end{cases}
\end{equation*}

The role of the Axial matrix $\gamma_A$ is played by the multiplication operator:

\begin{equation*}
    \epsilon(x) = (-1)^{x_0+x_1}.
\end{equation*}

\paragraph{Dirac operator.} The Lattice Dirac operator in the Staggered Fermions formalism is an operator acting on $\ell^2(\mathbb{Z}^2, \mathbb{C})$ according to:

\begin{equation*}
    \slashed{D}_{\mathrm{SN}} = \frac{1}{2} \sum_{\mu=0}^1 \Gamma_{\mu} (T_{\mu}- T_{\mu}^*).
\end{equation*}

By labeling the sites of $\L_L$ by $(x,a)$, with $x\in \L_L$ such that $x_0$ is even and $a\in\{1,2\}$, so that every $y\in\L_L$ can be written as $\big(x_0+ \frac{1}{2}(a-1),x_1\big)$, we recover translation invariance of the Dirac operator, which can be therefore expressed by Fourier transform: 
\begin{equation}
\label{eq:model_SN}
(\slashed{D}_{\mathrm{SN}})_{(x,a),(y,b)}
=\iint_{(-\frac{\pi}{2},\frac{\pi}{2}]\times(-\pi,\pi]}\frac{d^2p}{(2\pi)^2}
\,e^{-ip(x-y)}\,
  \begin{pmatrix}
        i\sin p_1 & \frac{1}{2}(e^{2ip_0}-1)\\
        -\frac{1}{2}(e^{-2ip_0}-1) & -i\sin p_1
    \end{pmatrix}_{a,b}.
\end{equation}

Again, we impose anti-periodic boundary conditions by considering an anti-periodic extension of the above Dirac Operator:

\begin{equation*}
    \big((\slashed{D}_{\mathrm{SN}})^-_{\Lambda_L}\big)_{xy} = \sum_{n \in \mathbb{Z}^2} (-1)^{n_0+n_1} (\slashed{D}_{\mathrm{SN}})_{x+ L n,y}.
\end{equation*}

\begin{lemma}[Properties of the SN Dirac operator]\label{AP221}
    The operator $(\slashed{D}_{\mathrm{SN}})^-_{\Lambda_L}$ has the following properties, for every $\phi\in \mathbb{R}^{\Lambda_L}$.
    \begin{enumerate}
        \item \textbf{{Anti-self-adjointness:}} 
        
        \begin{equation}
            \big((\slashed{D}_{\mathrm{SN}})^-_{\Lambda_L}\big)^* = - (\slashed{D}_{\mathrm{SN}})^-_{\Lambda_L}.
            \tag{PSD1}
            \label{PSD1}
        \end{equation}
        \item \textbf{{Chirality under adjoint:}}

\begin{equation}
     \big((\slashed{D}_{\mathrm{SN}})^-_{\Lambda_L}-M_{\phi} \big)^* = \epsilon \big((\slashed{D}_{\mathrm{SN}})^-_{\Lambda_L}-M_{\phi}\big) \epsilon,
      \tag{PSD2}
      \label{PSD2}
\end{equation}
where $\big((\slashed{D}_{\mathrm{SN}})^-_{\Lambda_L} -M_{\phi}\big)_{x,y} = \big((\slashed{D}_{\mathrm{SN}})^-_{\Lambda_L}\big)_{x,y} - \phi_x  \delta_{x,y}$ and $\epsilon$ is the multiplication operator by $\epsilon(x)= (-1)^{x_0+x_1}$.
  \item \textbf{{Spectrum of the adjoint:}} 
  \begin{equation}
      \mathrm{spec} \Big[   \big((\slashed{D}_{\mathrm{SN}})^-_{\Lambda_L}-M_{\phi}\big)^*\Big] = \mathrm{spec}\Big[(\slashed{D}_{\mathrm{SN}})^-_{\Lambda_L}-M_{\phi}\Big].
       \tag{PSD3}
            \label{PSD3}
  \end{equation}
  
       \item \textbf{{Reality of the determinant:}}
       \begin{equation}
\det\big((\slashed{D}_{\mathrm{SN}})^-_{\Lambda_L}-M_{\phi}\big) \in \mathbb{R}.
 \tag{PSD4}
            \label{PSD4}
       \end{equation}
    \end{enumerate}
\end{lemma}

The proof of Lemma \ref{AP221} is discussed in Appendix \ref{AP22}.

\paragraph{Lattice action.} The Dirac fermions are coupled through a local four-fermion interaction of Gross--Neveu type:

\begin{equation}\label{StaggeredA}
    S_{{\mathrm{SN}}}(\overline{\psi},\psi) = \sum_{x,y \in \Lambda_L}\sum_{j=1}^N \overline{\psi}_{x,j} \big((\slashed{D}_{\mathrm{SN}})^-_{\Lambda_L}\big)_{x,y} \psi_{y,j}- \frac{\lambda}{2N} \sum_{x \in \Lambda_L} \bigg(\sum_{j=1}^N \overline{\psi}_{x,j} \psi_{x,j} \bigg)^2, \qquad \lambda>0.
\end{equation}

\paragraph{Chiral symmetry.} The action is invariant under the discrete $\mathbb{Z}_2$-Chiral transformation defined by:

\begin{equation*}
    \begin{cases}
        \psi_{x, j} \to e^{i \frac{\pi}{2} \epsilon(x)}  \psi_{x, j}\\
                \overline{\psi}_{x, j} \to e^{i \frac{\pi}{2} \epsilon(x)}  \overline{\psi}_{x, j}.\\
    \end{cases}
\end{equation*}

Indeed, under this transformation, we have that 

\begin{equation*}
\begin{aligned}
    \overline{\psi}_{x,j} \big((\slashed{D}_{\mathrm{SN}})^-_{\Lambda_L}\big)_{x,y} \psi_{y,j} \to e^{i \frac{\pi}{2} (\epsilon(x)+\epsilon(y))}  \overline{\psi}_{x,j} \big((\slashed{D}_{\mathrm{SN}})^-_{\Lambda_L}\big)_{x,y} \psi_{y,j} = \overline{\psi}_{x,j} \big((\slashed{D}_{\mathrm{SN}})^-_{\Lambda_L}\big)_{x,y} \psi_{y,j}
    \end{aligned}
\end{equation*}

where we used the fact that $\big((\slashed{D}_{\mathrm{SN}})^-_{\Lambda_L}\big)_{x,y}$ is non-zero only if $|x-y|=1$, in which case $\e(x)=-\e(y)$. Conversely, under this transformation the bilinear $(\overline{\psi} \psi)_x$ changes sign: $(\overline{\psi} \psi)_x \to e^{i \pi \epsilon(x)} (\overline{\psi} \psi)_x = - (\overline{\psi} \psi)_x$.

\begin{figure}[h]
    \centering
\begin{tikzpicture}
\draw[dotted, ->-] (-0.5,-0.5) -- (3.5,-0.5);
\draw[dotted, ->>-] (3.5,-0.5) -- (3.5,3.5);
\draw[dotted, ->-] (-0.5,3.5) -- (3.5,3.5);
\draw[dotted, ->>-] (-0.5,-0.5) -- (-0.5,3.5);
\draw[line width = 0.05 cm] (0, 0) -- (3,0);
\draw[line width = 0.05 cm] (0, 2) -- (3,2);
\draw[line width = 0.05 cm] (3,1) -- (3.5,1);
\draw[line width = 0.05 cm] (-0.5,1) -- (0,1);
\draw[line width = 0.05 cm] (-0.5,3) -- (0,3);
\draw[line width = 0.05 cm] (3,3) -- (3.5,3);
\draw (0,3) -- (3,3);
\draw (0,1) -- (3,1);
\draw (0,-0.5) -- (0,3.5);
\draw (1,-0.5) -- (1,3.5);
\draw (2,-0.5) -- (2,3.5);
\draw (-0.5,2) -- (0,2);
\draw (3,2) -- (3.5,2);
\draw (-0.5,0) -- (0,0);
\draw (3,0) -- (3.5,0);
\draw (3,-0.5) -- (3,3.5);

\draw[line width = 0.05 cm] (0,-0.5) -- (0,0);
\draw[line width = 0.05 cm] (1,-0.5) -- (1,0);
\draw[line width = 0.05 cm] (2,-0.5) -- (2,0);
\draw[line width = 0.05 cm] (3,-0.5) -- (3,0);

\draw[line width = 0.05 cm] (0,3.5) -- (0,3);
\draw[line width = 0.05 cm] (1,3.5) -- (1,3);
\draw[line width = 0.05 cm] (2,3.5) -- (2,3);
\draw[line width = 0.05 cm] (3,3.5) -- (3,3);
\draw[thin, red] (-0.5,-0.5) -- (-0.5,3.5);
\draw[thin, red] (1.5,-0.5) -- (1.5,3.5);
\node[below, red] (a) at (-0.5,-0.6) {{$\partial_- (\L_L)_+$}};
\node[below, red] (a) at (1.5,-0.6) {{$\partial_+ (\L_L)_+$}};
\node[above] (a) at (0.5,3.6) {{$(\L_L)_+$}};
\node[above] (a) at (2.5,3.6) {{$(\L_L)_-$}};
\draw[->] (-1,-0.5) -- (-1,3.5);
\node[above] (a) at (-1,3.5) {{$x_0$}};
\end{tikzpicture}
\caption{Distribution of Phases of the Anti-symmetrized Forward Dirac Operator $\Gamma_{\mu}(x) T_{\mu}$ in $L=2$ Lattice. Solid Links represent a $-1$ sign. Notice that the set of edges intersecting $\partial_- (\L_L)_+$ are opposite of those intersecting $\partial_+ (\L_L)_+$ because we have chosen an anti-symmetrization convention where the $\mathbb{Z}_2$-twisting is localized on $\partial_- (\L_L)_+$. By means of any $\mathbb{Z}_2$ Gauge Transformation, we can move such line wherever we desire, without changing the physical properties of the system.}
\label{fig1}
\end{figure}
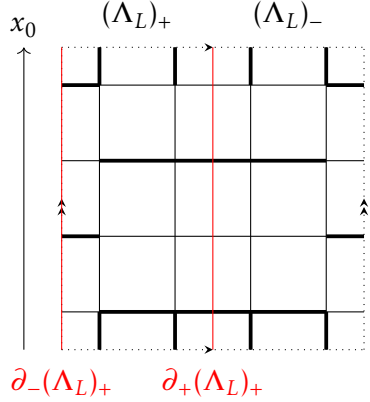

\paragraph{Effective bosonic measure.} By applying the general Hubbard--Stratonovich construction \eqref{eqn:Hub_Strat} to the action \eqref{StaggeredA}, and integrating out the fermionic fields, we obtain the effective bosonic measure on $\mathbb{R}^{\Lambda_L}$: 

\begin{equation*}
    d \mu_{\mathrm{SN}}(\phi) = \frac{1}{Z^{\mathrm{SN}}_{\Lambda_L}} \prod_{x \in \Lambda_L} d \phi_x e^{-\frac{N}{2\lambda}\phi_x^2} \det\big((\slashed{D}_{\mathrm{SN}})^-_{\Lambda_L} - M_\phi\big)^N. %\doteq \frac{1}{Z^{\mathrm{SN}}_{\Lambda_L}}d \phi e^{-\frac{N}{2\lambda}\phi^2} \det\big((\slashed{D}_{\mathrm{SN}})^-_{\Lambda_L} -M_{\phi}\big)^N 
\end{equation*}

\subsection{Staggered Fermions with Plaquette Interaction (SP)}

The third model, originally introduced in \cite{Cohen:1983nr}, is widely considered as the proper lattice discretization of the Gross-Neveu model, as after performing the formal continuum limit one gets exactly the usual Gross-Neveu action \cite[Sect.~3.2]{Cohen:1983nr}, see also Appendix \ref{app:continuum_sp}.\\

The model is most naturally formulated by blocking the original lattice $\Lambda_L$ into four-site unit cells, as illustrated in Fig.~\ref{figstaglat1}. This construction produces a reduced lattice $\widetilde{\Lambda}_L = 2\mathbb{Z}^2 \cap \big(-\frac{L}{2}, \frac{L}{2}\big]^2$, such that $\Lambda_L = \{A,B,C,D\} \times \widetilde{\Lambda}_L$, where $a \in \{A,B,C,D\} \doteq \mathcal{I}$ labels the position within each unit cell. \\

\begin{figure}[h]
\centering
\begin{tikzpicture}
\draw[dotted, ->-] (-0.5,-0.5) -- (3.5,-0.5);
\draw[dotted, ->>-] (3.5,-0.5) -- (3.5,3.5);
\draw[dotted, ->-] (-0.5,3.5) -- (3.5,3.5);
\draw[dotted, ->>-] (-0.5,-0.5) -- (-0.5,3.5);
\draw[line width = 0.05 cm] (0, 0) -- (3,0);
\draw[line width = 0.05 cm] (0, 2) -- (3,2);
\draw[line width = 0.05 cm] (3,1) -- (3.5,1);
\draw[line width = 0.05 cm] (-0.5,1) -- (0,1);
\draw[line width = 0.05 cm] (-0.5,3) -- (0,3);
\draw[line width = 0.05 cm] (3,3) -- (3.5,3);
\draw (0,3) -- (3,3);
\draw (0,1) -- (3,1);
\draw (0,-0.5) -- (0,3.5);
\draw (1,-0.5) -- (1,3.5);
\draw (2,-0.5) -- (2,3.5);
\draw (-0.5,2) -- (0,2);
\draw (3,2) -- (3.5,2);
\draw (-0.5,0) -- (0,0);
\draw (3,0) -- (3.5,0);
\draw (3,-0.5) -- (3,3.5);
\draw[fill, fill opacity =0.1] (-0.25,-0.25) -- (1.25,-0.25) -- (1.25, 1.25) -- (-0.25,1.25) -- (-0.25,-0.25);
\begin{scope} [xshift = 2cm]
\draw[fill, fill opacity =0.1] (-0.25,-0.25) -- (1.25,-0.25) -- (1.25, 1.25) -- (-0.25,1.25) -- (-0.25,-0.25);
\end{scope}
\begin{scope} [yshift = 2cm]
\draw[fill, fill opacity =0.1] (-0.25,-0.25) -- (1.25,-0.25) -- (1.25, 1.25) -- (-0.25,1.25) -- (-0.25,-0.25);
\end{scope}
\begin{scope} [xshift = 2cm, yshift= 2cm]
\draw[fill, fill opacity =0.1] (-0.25,-0.25) -- (1.25,-0.25) -- (1.25, 1.25) -- (-0.25,1.25) -- (-0.25,-0.25);
\end{scope}
\draw[line width = 0.05 cm] (0,-0.5) -- (0,0);
\draw[line width = 0.05 cm] (1,-0.5) -- (1,0);
\draw[line width = 0.05 cm] (2,-0.5) -- (2,0);
\draw[line width = 0.05 cm] (3,-0.5) -- (3,0);
\node[below, left] (a) at (0,1.7) {{$A$}};
\node[below, right] (a) at (1,1.7) {{$B$}};
\node[above, right] (a) at (1,3.3) {{$D$}};
\node[below, left] (a) at (0,3.3) {{$C$}};
\draw[line width = 0.05 cm] (0,3.5) -- (0,3);
\draw[line width = 0.05 cm] (1,3.5) -- (1,3);
\draw[line width = 0.05 cm] (2,3.5) -- (2,3);
\draw[line width = 0.05 cm] (3,3.5) -- (3,3);
\draw[->] (-1,-0.5) -- (-1,3.5);
\node[above] (a) at (-1,3.5) {{$x_0$}};
\draw[->] (0.5,0.5) -- (2.5,0.5);
\draw[->] (0.5,0.5) -- (0.5,2.5);
\end{tikzpicture}
\caption{Blocking of the lattice $\L_L$ ($L=4$ in this example) in terms of the unit cell with points indexed by $\mc{I}$. Lighter (resp. heavier) lines are used for denoting hoppings with a phase $+1$ (resp. $-1$).}
\label{figstaglat1}
\end{figure}
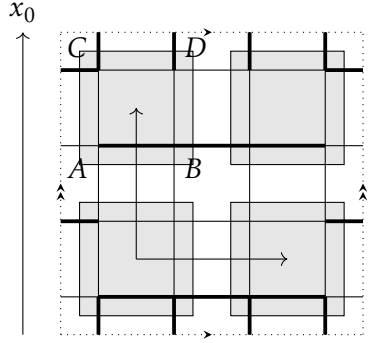

The model is constructed along the same lines as the SN one described above, but with a different interaction structure. We therefore adopt the same notation, adapted to the new labeling.

\paragraph{Staggered fermions.} At each lattice site $(x,a) \in \widetilde{\Lambda}_L \times \mc{I}$, we introduce $2N$ independent Grassmann variables $\psi_{x,a,j}$ and $\overline{\psi}_{x,a,j}$.

\paragraph{Gamma matrices.} The staggered phases are now encoded in
\[
\Gamma_{\mu}(a) \doteq
\begin{cases}
+1 & \mu=0,\\
-1 & \mu=1,\ a=A,B,\\
+1 & \mu=1,\ a=C,D.
\end{cases}
\]
The analogue of the axial matrix is the multiplication operator
\begin{equation}\label{epP}
    \epsilon(a)=
\begin{cases}
+1 & a=A,D,\\
-1 & a=B,C.
\end{cases}
\end{equation}

\paragraph{Dirac operator.} The lattice Dirac operator acts on
$\ell^2(\widetilde{\Lambda}_L \times \mathcal{I},\mathbb{C})$ as
\[
\slashed{D}_{\mathrm{SP}}
=
\frac{1}{2}\sum_{\mu=0}^1
\Gamma_{\mu}(a)(T_\mu - T_\mu^*).
\]
Note that it is translationally invariant on the reduced lattice $\widetilde{\Lambda}_L$:
\begin{equation}
\label{eq:model_SP}
(\slashed{D}_{\mathrm{SP}})_{(x,a),(y,b)}
=\iint_{(-\frac{\pi}{2},\frac{\pi}{2}]^2}\frac{d^2p}{(2\pi)^2}
\,e^{-ip\cdot (x-y)}\,
\frac{1}{2} \begin{pmatrix}
        0 &-1 + e^{2i p_1} &1 - e^{2i p_0} &0\\
        1 - e^{-2i p_1} &0 &0 &1-e^{2ip_0}\\
        -1+e^{-2ip_0} &0 &0 &1 -e^{2i p_1}\\
        0 &-1+e^{-2ip_0} &-1+ e^{-2ip_1} &0
    \end{pmatrix}_{a,b},
\end{equation}

for every $x,y\in\wt{\L}_L$ and $a,b\in\mc{I}$ (the matrix row and columns are indexed in the order $A,B,C,D$).
Anti-periodic boundary conditions are imposed as in the previous models by defining the anti-periodic extension
\begin{equation}
\label{eqn:51}
\big((\slashed{D}_{\mathrm{SP}})^-_{\Lambda_L}\big)_{(x,a),(y,b)}
=
\sum_{n\in\mathbb{Z}^2}
(-1)^{n_0+n_1}
(\slashed{D}_{\mathrm{SP}})_{(x+Ln,a),(y,b)}, \qquad x,y\in\wt{\L}_L, \; a,b\in\mc{I}.
\end{equation}

\begin{lemma}[Properties of the SP Dirac operator]\label{Lemma:Dirac_SP}
The operator $(\slashed{D}_{\mathrm{SP}})^-_{\Lambda_L}$ has the following properties, for every $\phi\in \mathbb{R}^{\wt{\Lambda}_L}$.

    \begin{enumerate}
    \item \textbf{Anti-self-adjointness}:

    \begin{equation*}
      \big((\slashed{D}_{\mathrm{SP}})^-_{\Lambda_L} \big)^* = - (\slashed{D}_{\mathrm{SP}})^-_{\Lambda_L}.
       \tag{PSDP1}
       \label{PSDP1}
    \end{equation*}
        
        \item \textbf{{Chirality under adjoint:}}

\begin{equation}
     \big((\slashed{D}_{\mathrm{SP}})^-_{\Lambda_L}-M_{\phi} \big)^* = \epsilon \big((\slashed{D}_{\mathrm{SP}})^-_{\Lambda_L}-M_{\phi}\big) \epsilon, 
      \tag{PSDP2}
      \label{PSDP2}
\end{equation}
where 
$\epsilon$ is the multiplication operator by $\epsilon(a)$.
\item \textbf{{Spectrum of the Adjoint:}} 
  \begin{equation}
      \mathrm{spec}\Big[\big((\slashed{D}_{\mathrm{SP}})^-_{\Lambda_L}-M_{\phi}\big)^*\Big] = \mathrm{spec}\Big[(\slashed{D}_{\mathrm{SP}})^-_{\Lambda_L}-M_{\phi}\Big].
       \tag{PSDP3}
            \label{PSDP3}
  \end{equation}
  
\item \textbf{Reality of the determinant:}
\begin{equation}
\det\big((\slashed{D}_{\mathrm{SP}})^-_{\Lambda_L}-M_{\phi}\big) \in \mathbb{R}.
 \tag{PSDP4}
\label{PSDP4}
\end{equation}
\end{enumerate}
\end{lemma}

\begin{remark}\label{RPSDP}
   The properties stated in Lemma \ref{Lemma:Dirac_SP} follow from their analogue related to the SN Dirac operator, Lemma \ref{AP221}. This is straightforward since every field configuration $\phi$ associated with the plaquette interaction, can be regarded as a field configuration of the SN model (considered in Lemma \ref{AP221}) which is constant within each unit cell.
\end{remark}

\paragraph{Lattice action.} The third model is similar to the previous one but differs for the interaction, which couples degrees of freedom within the same unit cell:

\begin{equation}\label{StaggeredPA}
    S_{\mathrm{SP}}(\overline{\psi}, \psi) =  \sum_{x,y \in \widetilde{\Lambda}_L} \sum_{a,b\in\mc{I}} \sum_{j=1}^N \overline{\psi}_{x,a, j} \big((\slashed{D}_{\mathrm{SP}})^-_{\Lambda_L}\big)_{(x,a),(y,b)} \psi_{y,b, j}- \frac{\lambda}{2 N} \sum_{x \in \widetilde{\Lambda}_L} \bigg(\sum_{j=1}^N \sum_{a \in \mathcal{I}}\overline{\psi}_{x,a, j} \psi_{x,a, j} \bigg)^2,% - m \sum_{x \in \Lambda_L} \bigg(\sum_{a=1}^N \overline{\psi}_{x,a} \psi_{x,a}\bigg)
\end{equation}

where again $\lambda>0$.  \\

\paragraph{Effective bosonic measure.} Applying the general Hubbard--Stratonovich construction \eqref{eqn:Hub_Strat} to the action \eqref{StaggeredPA}, and integrating out the fermionic fields, we obtain the effective bosonic measure on $\mathbb{R}^{\widetilde{\Lambda}_L}$: 

\begin{equation*}
    d \mu_{\mathrm{SP}}(\phi) = \frac{1}{Z^{\mathrm{SP}}_{\Lambda_L}} \prod_{x \in \widetilde{\Lambda}_L} d \phi_x e^{-\frac{N}{2\lambda}\phi_x^2} \det\big((\slashed{D}_{\mathrm{SP}})^-_{\Lambda_L} - M_\phi\big)^N. %\doteq \frac{1}{Z^{\mathrm{SP}}_{\Lambda_L}}d \phi e^{-\frac{N}{2\lambda}\phi^2} \det(\slashed{D}^-_{\widetilde{\Lambda}_L % \times \{A,B,C,D\}
   % } -M_{\phi})^N 
\end{equation*}

This formulation makes it natural to view $\slashed{D}$ as an operator on the reduced lattice $\widetilde{\Lambda}_L$, endowed with four internal degrees of freedom per site. Equivalently, the SP model may be regarded as a restriction of the SN model to Hubbard--Stratonovich fields that are constant within each unit cell. This restriction, however, has significant consequences for the (formal) continuum limit (see \cite{Cohen:1983nr, Hands_1993}) and Appendix \ref{app:continuum_sp}.\\

\subsection{Main result}

Given the precise definitions of the models, we are ready to state our main theorem.

\begin{thm}[Chiral Long-Range Order]\label{thm1}
There exists $\lambda_0 > 0$ such that for each $\alpha \in \{\mathrm{NN}, \mathrm{SN}, \mathrm{SP}\}$, the following holds. For every $\lambda \in (0,\lambda_0]$ and $\varepsilon \in (0,1)$, there exist $N_0(\lambda,\varepsilon)\ge1$ and $L_0(\lambda,\varepsilon)\ge1$ with the property that, for all $N \in \big[N_0(\lambda,\varepsilon),\infty\big)\cap 2\mbb{N}$ and $L \in \big[L_0(\lambda,\varepsilon),\infty\big)\cap4(1+\d_{\a,\mathrm{SP}})\mbb{N}$, 
\begin{equation}
\label{eqn:33}
(1+\varepsilon)\,\left(\tfrac{N}{\l}\right)^2 \varphi_\alpha^\star(\lambda)^2 
\ge
\frac{1}{|\Lambda_L^\alpha|^2}
\sum_{x,y \in \Lambda_L^\alpha}
\left\langle (\overline{\psi}\psi)_x (\overline{\psi}\psi)_y \right\rangle_{\alpha}
\ge 
(1-\varepsilon)\,\left(\tfrac{N}{\l}\right)^2 \varphi_\alpha^\star(\lambda)^2 > 0.
\end{equation}
Here $\varphi_\alpha^\star(\lambda)$ denotes the (strictly) positive minimizer of the effective potential:
\begin{equation}
\label{eq:mean_field_pot}
v^{\alpha}_{\lambda}(\varphi)
=
\frac{\varphi^2}{2\lambda}
-
C_\alpha
\iint_{[-\pi,\pi]^2} \frac{d^2p}{(2\pi)^2}
\log\bigl(\varphi^2 + \sin^2(p_0) + \sin^2(p_1)\bigr),
\end{equation}
where $C_{\a}=1,\frac{1}{2},2$ for $\a=\mathrm{NN},\mathrm{SN},\mathrm{SP}$ respectively.

\end{thm}

\begin{remark}
\label{rmk:main}
    \leavevmode
    \begin{enumerate}
\item\label{ite:1} \textbf{Interpretation of the Result.} Theorem~\ref{thm1} establishes Long-Range Order for the scalar fermionic
bilinear $(\overline{\psi}\psi)_x$ in the regime of sufficiently weak coupling and
sufficiently large flavor number. Actually, the proof yields a stronger
pointwise statement than the spatially averaged formulation of the theorem:
for every $\varepsilon\in(0,1)$ there exists
$N_0=N_0(\lambda,\varepsilon)$ such that, for all $N\ge N_0$,
\[
(1-\varepsilon)\,\left(\tfrac{N}{\l}\right)^2 \varphi_\alpha^\star(\lambda)^2
\le
\inf_{x,y\in\Lambda_L^\alpha}
\big\langle(\overline{\psi}\psi)_x(\overline{\psi}\psi)_y\big\rangle_{\a}
\le
\sup_{x,y\in\Lambda_L^\alpha}
\big\langle(\overline{\psi}\psi)_x(\overline{\psi}\psi)_y\big\rangle_{\a}
\le
(1+\varepsilon)\,\left(\tfrac{N}{\l}\right)^2 \varphi_\alpha^\star(\lambda)^2.
\]
Thus the mean-field prediction for the fermion condensate is recovered uniformly at the level of individual two-point functions of the fermion bilinear, and not only after spatial averaging. Remarkably, this conclusion is obtained without relying on a large-$N$ expansion or on a saddle-point approximation. Instead, our rigorous upper and lower bounds are formulated directly in terms of the minimizers of the effective potential, namely the same quantities that arise in the heuristic mean-field analysis (similarly to what is done in \cite{GN}). As a consequence, the fermion bilinear two-point function is uniformly trapped, up to controlled errors, between its corresponding mean-field predictions.
\item\label{ite:2} \textbf{Current Limitations and Future Directions.} A notable feature of our analysis is that it is carried out entirely at finite volume, with estimates that are uniform in the lattice size. Consequently, any thermodynamic-limit point of the finite-volume Gibbs states considered in this work would automatically inherit the Long-Range Order established in Theorem~\ref{thm1}. The existence and characterization of such infinite-volume Gibbs states, however, are beyond the scope of the present paper and will be investigated elsewhere.

A natural next step is the rigorous construction of extremal infinite-volume Gibbs states $\langle\cdot\rangle^{\infty,\pm}_{\alpha}$, corresponding to the two symmetry-related minima of the effective potential. One would then like to prove that these states satisfy the clustering property and, more strongly, that the 
\emph{truncated correlations} decay exponentially, namely:
\[\big\langle (\overline{\psi}\psi)_x(\overline{\psi}\psi)_y\big\rangle^{\infty,\pm}_{\alpha} -  \big(\tfrac{N}{\lambda}\vphi_\alpha^\star(\lambda)\big)^2 =\mc{O}\big(e^{-\mf{k}|x-y|} \big), \qquad \mf{k}>0,
\] 

up to the accuracy already controlled by the finite-volume estimates. 
From this perspective, the main remaining challenge is not the identification of the order parameter itself, but rather the rigorous construction and characterization of the corresponding pure infinite-volume phases.

\item\label{ite:3} \textbf{Asymptotic Expansion for $\varphi^{\star}_{\alpha}(\lambda)$.} The positive minimizer of the effective potential $v^{\alpha}_{\l}(\varphi)$ can be implicitly written \cite[Sect.~5]{GUERIN1980168} as the (non-trivial) solution of the Gap Equation $\frac{d}{d\vphi}v^{\a}_{\l}(\vphi^{\star}_{\a})=0$, which in turn can be rewritten as

\begin{equation}
\label{eqn:49}
\frac{\pi}{4 \l} = \frac{C_{\a}}{1+ (\vphi^{\star}_{\a})^2} K\bigg(\sqrt{\frac{1}{1+(\vphi_{\a}^{\star})^2}}\bigg),
\end{equation}

where $K$ is the Complete Elliptic Integral of the first kind\footnote{
Notice that \cite{Abramowitz1964} and \cite{GUERIN1980168} adopt different conventions for the complete elliptic integral of the first kind. Throughout \cite{Abramowitz1964} the argument is the parameter $m$, whereas \cite{GUERIN1980168} uses the modulus $k$. The two conventions are related by $m=k^2$.} \cite[Sect.~17.3]{Abramowitz1964}. Using \cite[Eq. 17.3.26]{Abramowitz1964}, it is possible to obtain the precise asymptotic of $\vphi^{\star}_{\a}(\l)$ as $\l\to 0^+$. In fact: 

\begin{equation*}
\frac{C_{\a}}{1+ (\vphi^{\star}_{\a})^2} K\bigg(\sqrt{\frac{1}{1+(\vphi_{\a}^{\star})^2}}\bigg) = C_{\a} \log(\frac{2 \sqrt{2}}{\vphi_{\a}^{\star}}) \big(1+ o(1)\big),
\end{equation*}

(with $o(1)$ denoting quantities which vanish as $\vphi_{\a}^{\star}\to 0^+$) which leads to

\begin{equation}\label{asymmin}
\vphi_{\a}^{\star}(\l)=  2 \sqrt{2} e^{-\frac{\pi}{4 C_{\a} \lambda}} \big(1+ o(1) \big)
\end{equation}

(here instead $o(1)$ denotes quantities which vanish as $\l\to 0^+$). It is worth observing that this expression is closely related to the dynamically generated mass predicted by the continuum Gross--Neveu gap equation. The only difference lies in the numerical coefficients appearing in the exponent, which reflect the larger number of effective fermionic degrees of freedom arising from lattice fermion doubling.

\item\label{ite:4} \textbf{Extension to Lattices with Arbitrary Mesh Size.} Theorem~\ref{thm1} extends immediately to lattices $\Lambda_{\ell,L}\doteq \ell\mathbb Z^2 \cap \big(-\frac{L}{2},\frac{L}{2}\big]^2$, with arbitrary mesh size
$\ell$ such that 
$L/\ell\in4\mathbb N$. Indeed, the dependence on the lattice spacing can be completely absorbed into a rescaling of the fermionic fields. We illustrate this observation for the NN model on $\Lambda_{\ell,L}$ described by the Gibbs state $\langle\cdot\rangle_{\mathrm{NN}}^{\ell,L}$ given by

\[
\langle\mathcal O\rangle_{\mathrm{NN}}^{\ell,L}
\doteq
\frac{
\int d\overline{\psi}\,d\psi\;
\mathcal O(\overline{\psi},\psi)\,
e^{-S_{\mathrm{NN}}^{\ell,L}(\overline{\psi},\psi)}
}{
\int d\overline{\psi}\,d\psi\;
e^{-S_{\mathrm{NN}}^{\ell,L}(\overline{\psi},\psi)}
},
\]

where\footnote{The extra normalization factor $\ell^2$ is chosen in order for the action to scale as the volume of the system $L^2$. Alternatively, one can think that such factor allows to reconstruct a Riemann Integral in the limit $\ell \to 0^+$.}

\begin{equation}
\label{eqn:45}
S_{\mathrm{NN}}^{\ell,L}(\overline{\psi}, \psi)
=
\ell^2
\sum_{x,y\in\Lambda_{\ell,L}}
\sum_{j=1}^{N}
\overline{\psi}_{x,j}
\bigl((\slashed D_{\mathrm{NN}})^-_{\Lambda_{\ell,L}}\bigr)_{x,y}
\psi_{y,j}
-
\frac{\lambda}{2N}\ell^2
\sum_{x\in\Lambda_{\ell,L}}
\Bigl(
\sum_{j=1}^{N}
\overline{\psi}_{x,j}\psi_{x,j}
\Bigr)^2,
\end{equation}
and
\[
\bigl((\slashed D_{\mathrm{NN}})^-_{\Lambda_{\ell,L}}\bigr)_{x,y}
=
\frac{1}{2\ell}
\sum_{\mu=0}^{1}
\gamma_\mu
\left(
\delta_{x+\ell e_\mu,y}
-
\delta_{x-\ell e_\mu,y}
\right).
\]

Introducing the rescaled spinors
\[
\xi_{x,s,j}
\doteq
\ell^{1/2}\psi_{\ell x, s, j},
\qquad
\overline{\xi}_{x,s,j}
\doteq
\ell^{1/2}\overline{\psi}_{\ell x, s, j},
\qquad
x\in\Lambda_{1,L/\ell},\,\,\,\, s \in \{1,2\},
\]
a direct computation shows that
\[
S_{\mathrm{NN}}^{\ell,L}(\overline{\psi},\psi)
=
S_{\mathrm{NN}}^{1,L/\ell}(\overline\xi,\xi),
\]

and thus\footnote{Since the Jacobian of the change of Grassmann variables cancels between numerator and denominator in the definition of the Gibbs state} we obtain
\[
\big\langle
(\overline{\psi}\psi)_x
(\overline{\psi}\psi)_y
\big\rangle_{\mathrm{NN}}^{\ell,L}
=
\ell^{-2}
\big\langle
(\overline\xi\xi)_{x/\ell}
(\overline\xi\xi)_{y/\ell}
\big\rangle_{\mathrm{NN}}^{1,L/\ell}.
\]

The same rescaling argument applies similarly to the SN and SP models. Since the rescaled theories coincide exactly with the unit-lattice-spacing models studied throughout this paper, Theorem~\ref{thm1} applies verbatim for all three cases under study. 
 
Hence, introducing the \emph{mass order parameter}:
\begin{equation*}
\mf{m}_{\a}^{\ell,L}(\l,N) \doteq  \Bigg(\frac{\l^2}{N^2|\L^{\a}_{\ell,L}|^2} \sum_{x,y\in\L^{\a}_{\ell,L}}
\big\langle
(\overline{\psi}\psi)_x
(\overline{\psi}\psi)_y
\big\rangle_{\a}^{\ell,L}\Bigg)^{\frac{1}{2}},
\end{equation*}
with $\a\in\{\mathrm{NN},\mathrm{SN},\mathrm{SP}\}$, from Theorem \ref{thm1} it follows that, for
$\lambda\le\lambda_0$,
$L/\ell\ge L_0(\lambda,\varepsilon)$
and
$N\ge N_0(\lambda,\varepsilon)$:
\begin{equation*}
\sqrt{1+\varepsilon}
\frac{\varphi_\alpha^\star(\lambda)}{\ell}
\ge \mf{m}_{\a}^{\ell,L}(\l,N)
\ge
\sqrt{1-\varepsilon} \frac{\varphi_\alpha^\star(\lambda)}{\ell},
\end{equation*}
where recall that $|\L^{\a}_{\ell,L}|$ denotes the cardinality of the lattice. It is worth stressing that the sole effect of the mesh size is the dimensional prefactor $\ell^{-1}$ in the mass order parameter.
A noteworthy feature is that, for $N,\lambda$ fixed, $\mf{m}_{\a}^{\ell,L}(\l,N)$ grows like $\ell^{-1}$ as $\ell\to 0^+$. This behavior can be compensated by allowing the bare coupling to run with the cutoff, according to:
\begin{equation}
\label{eqn:bare_coupling}
\lambda(\ell)=
\frac{\lambda(\ell_0)}
{1+\frac{4C_{\alpha}}{\pi}\lambda(\ell_0)\log\!\left(\frac{\ell_0}{\ell}\right)},
\end{equation}
where $\ell_0$ is a fixed reference lattice spacing (say, the \emph{laboratory scale}) and $\lambda(\ell_0)$ is the corresponding physical coupling. By using \eqref{eqn:33} and substituting \eqref{eqn:bare_coupling} into \eqref{asymmin}, one obtains:
\begin{equation}
\label{eqn:55}\begin{split}
\mf{m}_{\a}^{\ell,L}(\l(\ell),N)
&\lessgtr
\sqrt{1\pm\ve} \Big(\frac{\varphi_\alpha^\star(\lambda(\ell))}{\ell}\Big) = \sqrt{1\pm\ve} \Big(\frac{\varphi_\alpha^\star(\lambda(\ell_0))}{\ell_0}\Big) \big(1+ o(1)\big),
\end{split}\end{equation}
where $o(1)$ denotes quantities tending to zero as $\l(\ell_0)\to 0^+$, uniformly w.r.t. $\ell$, showing that the mass order parameter converges to a finite quantity as the ultraviolet cutoff is removed. In this sense, the scaling of the bare coupling compensates exactly for the divergence induced by the shrinking lattice spacing and yields a finite continuum value for the mass parameter.

Nevertheless, our results should not be interpreted as a genuine ultraviolet construction of the model. Indeed, the assumptions of Theorem~\ref{thm1} are imposed directly at the ultraviolet scale and the corresponding estimates are not controlled uniformly in the lattice spacing $\ell$, as the choice of bare coupling in \eqref{eqn:bare_coupling} forces the number of flavors $N$ to diverge as $\ell\to 0^+$. This limitation is not surprising: our proof establishes Symmetry Breaking directly at the cutoff scale, without performing a Renormalization Group analysis. A true continuum limit would require a multiscale argument capable of tracking the relevant estimates uniformly across scales and proving that the symmetry-breaking mechanism persists along the Renormalization Group flow.

\item\label{ite:5} \textbf{Continuum Limit and Relation with the Gross--Neveu beta function.}

The scaling relation \eqref{eqn:bare_coupling} also suggests a natural Renormalization Group interpretation. 
Indeed, after upgrading $\ell$ to a continuum variable, \eqref{eqn:bare_coupling} is the unique solution to
\begin{equation}
\label{eqn:flow}
\ell\frac{d\lambda}{d\ell}
=
\frac{4C_\alpha}{\pi}\lambda^2.
\end{equation}
This turns out to be exactly the one-loop beta-function equation for the standard Gross--Neveu model, with a model-dependent coefficient reflecting the number of effective low-energy fermionic species.
Let us stress, however, that this interpretation is most transparent for the SP model (see Appendix \ref{app:continuum_sp}). In this case the low-energy degrees of freedom are expected to reproduce the Gross--Neveu interaction without extra spurious quartic terms due to the fermion doubling (similar e.g. to the Umklapp--Scattering effects). Such interactions (which are not of the standard Gross--Neveu form) are instead generated at intermediate scales within the NN and SN models, yielding some extra difficulties, as it is not apriori clear whether they are compatible with Lorentz invariance or whether they can be renormalized without introducing additional counterterms 
(even though it is claimed to be the case, at least formally: see \cite[Sect.~3.1]{Cohen:1983nr} for the NN model and \cite[Sect.~3]{Hands_1993}, where the reader is referred to \cite{JOLICOEUR1986431}, for the SN model).

Thus, while the one-loop running coupling constant, extracted from the scaling of the order parameter, has the expected Gross--Neveu form, a complete continuum Renormalization Group analysis for the NN and SN discretizations would require controlling these additional effective interactions and their Reflection Positivity properties. 
\\

Even for the SP model, despite being a natural candidate for reproducing the correct Gross-Neveu model (see Appendix \ref{app:continuum_sp}), a detailed multiscale analysis would be required in order to constructively prove its renormalizability and to establish the rigorous connection with the flow equation \eqref{eqn:flow}. All in all, the above comparison with the continuum Gross--Neveu beta function should be, at least for this moment, regarded as a heuristic indication, rather than as a solid mathematical argument.

\end{enumerate}
\end{remark}

\section{Reflection Positivity}\label{sec3}

\begin{figure}[h]
\begin{tikzpicture}[scale=1.6]
\begin{axis}[axis equal image,
        axis line style={draw=none},
        tick style={draw=none},
        xmax=18,ymax=22,zmax=7,xmin=-15,ymin=-20,zmin=-7,
        ticks=none,
        clip bounding box=upper bound,
        colormap/blackwhite,
        samples=23,
        view={45}{45}]
        \begin{pgfonlayer}{background layer}
            %%plot metà bianca del toro (per rimettere il mesh basta togliere l'opzione faceted color)
            \addplot3[domain=0:360,
                    y domain=90:270,
                    surf,
                    z buffer=sort,
                    color = white!30,
                    mesh/ordering=y varies,
                    shader=flat, 
                    draw=black,
                    line width=0.1pt
                ]
                ({(12 + 3 * cos(x)) * cos(y)} ,
                {(12 + 3 * cos(x)) * sin(y)},
                {3 * sin(x)});
        \end{pgfonlayer} 
        \begin{pgfonlayer}{main}
            %%frecce
            %\foreach \th in {-20,0,30,50} {
            %\addplot3 [samples=70,samples y=1,domain=262:269,-, blue,opacity=.2] 
            %    ({(12 + 3*cos(\th))*cos(x)},
            %    {(12 + 3*cos(\th))*sin(x)},
            %    {3*sin(\th)});

            %\addplot3 [samples=70,samples y=1,domain=91:100,-{Stealth[length=1mm]}, blue,opacity=.2] 
            %    ({(12+ 3*cos(\th+150))*cos(x)},
            %     {(12 + 3*cos(\th+150))*sin(x)},
            %    {3*sin(\th+150)}); };	
            %%piano
            \draw[fill=gray, fill opacity=0.4, draw=black, line width=0.1pt] (0,18,-7) -- (0,18,7) -- (0,-20,7)  -- (0,-20,-7) -- (0,18,-7);
           % \draw[fill=gray, fill opacity=0.4, draw=black, line width=0.1pt] (0,17,-7) -- (0,18,7) -- (0,-20,7)  -- (0,-20,-7) -- (0,18,-7);
        \end{pgfonlayer}

        \begin{pgfonlayer}{foreground layer}
            %%plot metà blu del toro
            \addplot3[domain=0:360,y 
                domain=-90:90,
                surf,
                z buffer=sort,
                color=gray, 
                mesh/ordering=y varies,
                shader=flat, 
                draw=black,
                line width=0.1pt]
                ({(12 + 3 * cos(x)) * cos(y)} ,
                {(12 + 3 * cos(x)) * sin(y)},
                {3 * sin(x)});
            %%cerchio blu sotto
            \addplot3[line width = 0.1pt,domain=0:360,surf, z buffer=sort, black,samples=71, samples y=1]
                (0,
                {(12 + 3 * cos(x)) * sin(-90)},
                {3 * sin(x)});
            %%cerchio blu sopra
            \addplot3[line width=0.1pt,domain=0:360,surf, z buffer=sort, black,samples=71, samples y=1]
                (0,
                {(12 + 3 * cos(x)) * sin(90)},
                {3 * sin(x)});
            %%frecce
            %\foreach \th in {-20,0,30,50} {
            %\addplot3 [samples=70,samples y=1,domain=269:280,-{Stealth[length=1mm]}, blue,opacity=.6] 
            %    ({(12 + 3*cos(\th))*cos(x)},
            %    {(12 + 3*cos(\th))*sin(x)},
            %    {3*sin(\th)});

            %\addplot3 [samples=70,samples y=1,domain=80:91,-, blue,opacity=.6] 
            %    ({(12+ 3*cos(\th+150))*cos(x)},
            %    {(12 + 3*cos(\th+150))*sin(x)},
            %    {3*sin(\th+150)}); };		
            %\node[blue] at (0,-22,5) {\scalebox{0.3}{$\partial_-$}};
            %\node[blue] at (2,10,4.5) {\scalebox{0.3}{$\partial_+$}};
            \node[] at (21,0,1) {\scalebox{1}{$(\L_L)_+$}};
            \node[] at (-12,0,10) {\scalebox{1}{$(\L_L)_-$}};
        \end{pgfonlayer}
    \end{axis}
\begin{scope}[xshift =6cm,yshift=1cm, scale =0.4]
\draw[black]  (1,1) -- (2,1);
\draw[black]  (0,0) -- (0,1);
\draw[black]  (0,0) -- (1,0);
\draw[black]  (1,0) -- (1,1);
\draw[black]  (1,0) -- (2,0);
\draw[black]  (2,0) -- (2,1);
\draw[black]  (2,0) -- (2,1);
\draw[black]  (3,0) -- (4,0);
\draw[black]  (3,0) -- (3,1);
\draw[black]  (4,0) -- (5,0);
\draw[black]  (2,0) -- (3,0); 
\draw[black]  (4,0) -- (4,1);
\draw[black]  (5,0) -- (6,0);
\draw[black]  (5,0) -- (5,1);
\draw[black]  (6,0) -- (7,0);
\draw[black]  (6,0) -- (6,1);
\draw[black]  (7,0) -- (7,1);

\draw[black]  (0,1) -- (1,1);
\draw[black]  (0,1) -- (0,2);
\draw[black]  (1,1) -- (1,2);
\draw[black]  (2,1) -- (2,2);

\draw[black]  (2,1) -- (3,1);
\draw[black]  (3,1) -- (3,2);
\draw[black]  (3,1) -- (4,1);
\draw[black]  (4,1) -- (4,2);

\draw[black]  (4,1) -- (5,1);
\draw[black]  (5,1) -- (5,2);
\draw[black]  (5,1) -- (6,1);
\draw[black]  (6,1) -- (6,2);

\draw[black]  (6,1) -- (7,1);
\draw[black]  (7,1) -- (7,2);
\draw[black]  (1,2) -- (2,2);
\draw[black]  (1,2) -- (1,3);

\draw[black]  (0,2) -- (1,2);
\draw[black]  (0,2) -- (0,3);

\draw[black]  (2,2) -- (3,2);
\draw[black]  (2,2) -- (2,3);
\draw[black]  (3,2) -- (4,2);
\draw[black]  (3,2) -- (3,3);
\draw[black]  (4,2) -- (5,2);
\draw[black]  (4,2) -- (4,3);
\draw[black]  (5,2) -- (6,2);
\draw[black]  (5,2) -- (5,3);
\draw[black]  (5,2) -- (6,2);
\draw[black]  (5,2) -- (5,3);
\draw[black]  (6,2) -- (7,2);
\draw[black]  (6,2) -- (6,3);
\draw[black]  (4,2) -- (5,2);
\draw[black]  (5,3) -- (6,3);
\draw[black]  (7,2) -- (7,3);
\draw[black]  (0,3) -- (1,3);
\draw[black]  (0,3) -- (0,4);
\draw[black]  (1,3) -- (2,3);
\draw[black]  (2,3) -- (3,3);
\draw[black]  (3,3) -- (4,3);
\draw[black]  (4,3) -- (5,3);
\draw[black]  (0,4) -- (0,5);
\draw[black]  (5,3) -- (6,3);
\draw[black]  (6,3) -- (7,3);
\draw[black]  (0,4) -- (0,5);
\draw[black]  (0,4) -- (1,4);
\draw[black]  (1,4) -- (2,4);
\draw[black]  (2,4) -- (3,4);
\draw[black]  (3,4) -- (4,4);
\draw[black]  (4,4) -- (5,4);
\draw[black]  (0,5) -- (0,6);
\draw[black]  (5,4) -- (6,4);
\draw[black]  (6,4) -- (7,4);
\draw[black]  (0,5) -- (1,5);
\draw[black]  (1,5) -- (2,5);
\draw[black]  (2,5) -- (3,5);
\draw[black]  (3,5) -- (4,5);
\draw[black]  (4,5) -- (5,5);
\draw[black]  (0,6) -- (0,7);
\draw[black]  (5,5) -- (6,5);
\draw[black]  (6,5) -- (7,5);
\draw[black]  (0,6) -- (1,6);
\draw[black]  (1,6) -- (2,6);
\draw[black]  (2,6) -- (3,6);
\draw[black]  (3,6) -- (4,6);
\draw[black]  (4,6) -- (5,6);
\draw[black]  (7,6) -- (7,7);
\draw[black]  (5,6) -- (6,6);
\draw[black]  (6,6) -- (7,6);
\draw[black]  (0,7) -- (1,7);
\draw[black]  (1,7) -- (2,7);
\draw[black]  (2,7) -- (3,7);
\draw[black]  (3,7) -- (4,7);
\draw[black]  (4,7) -- (5,7);
\draw[black]  (7,5) -- (7,6);
\draw[black]  (5,7) -- (6,7);
\draw[black]  (6,7) -- (7,7);
\draw[black]  (7,3) -- (7,4);
\draw[black]  (7,4) -- (7,5);
\draw[black]  (1,4) -- (1,5);
\draw[black]  (1,3) -- (1,4);
\draw[black]  (1,5) -- (1,6);
\draw[black]  (1,6) -- (1,7);
\draw[black]  (2,4) -- (2,5);
\draw[black]  (2,3) -- (2,4);
\draw[black]  (2,5) -- (2,6);
\draw[black]  (2,6) -- (2,7);
\draw[black]  (3,4) -- (3,5);
\draw[black]  (3,3) -- (3,4);
\draw[black]  (3,5) -- (3,6);
\draw[black]  (3,6) -- (3,7);
\draw[black]  (4,4) -- (4,5);
\draw[black]  (4,3) -- (4,4);
\draw[black]  (4,5) -- (4,6);
\draw[black]  (4,6) -- (4,7);
\draw[black]  (5,4) -- (5,5);
\draw[black]  (5,3) -- (5,4);
\draw[black]  (5,5) -- (5,6);
\draw[black]  (5,6) -- (5,7);
\draw[black]  (6,4) -- (6,5);
\draw[black]  (6,3) -- (6,4);
\draw[black]  (6,5) -- (6,6);
\draw[black]  (6,6) -- (6,7);
%Boundaries
\draw[black, dotted, ->-] (-0.5,-0.5) -- (7.5,-0.5);
\draw[black, dotted, ->-] (-0.5,7.5) -- (7.5,7.5);
\draw[black, dotted, ->>-] (-0.5,-0.5) -- (-0.5,7.5);
\draw[black, dotted, ->>-] (7.5,-0.5) -- (7.5,7.5);
% Oriented Edges at the boundary
%Left+Down
\draw[black] (0, -0.5) --(0,0);
\draw[black] (-0.5,0) -- (0,0);
\draw[black] (1, -0.5) --(1,0);
\draw[black]  (-0.5,1) -- (0,1);
\draw[black]  (2, -0.5) --(2,0);
\draw[black]  (-0.5,2) -- (0,2);
\draw[black]  (3, -0.5) --(3,0);
\draw[black]  (-0.5,3) -- (0,3);
\draw[black]  (4, -0.5) --(4,0);
\draw[black]  (-0.5,4) -- (0,4);
\draw[black]  (5, -0.5) --(5,0);
\draw[black]  (-0.5,6) -- (0,6);
\draw[black]  (6,-0.5) -- (6,0);
\draw[black]  (-0.5,5) -- (0,5);
\draw[black]  (7, -0.5) --(7,0);
\draw[black]  (-0.5,7) -- (0,7);
%Shaded area
\draw[fill = black,  fill opacity =0.6] (1.5,-0.5) -- (5.5,-0.5) -- (5.5,7.5) -- (1.5,7.5) -- (1.5, -0.5);
%Rigth + Up
\draw[black] (0,7) -- (0,7.5);
\draw[black] (7,7) -- (7.5,7);
\draw[black] (1,7) -- (1,7.5);
\draw[black] (7,6) -- (7.5,6);
\draw[black] (2,7) -- (2,7.5);
\draw[black] (7,5) -- (7.5,5);
\draw[black] (3,7) -- (3,7.5);
\draw[black] (7,4) -- (7.5,4);
\draw[black] (4,7) -- (4,7.5);
\draw[black] (7,0) -- (7.5,0);
\draw[black] (5,7) -- (5,7.5);
\draw[black] (7,3) -- (7.5,3);
\draw[black] (6,7) -- (6,7.5);
\draw[black] (7,2) -- (7.5,2);
\draw[black] (7,7) -- (7,7.5);
\draw[black] (7,1) -- (7.5,1);

\node[below] (a) at (3.5,-0.7) {$(\L_L)_+$};
\node[below] (b) at (0.5,-0.7) {$(\L_L)_-$};
\draw[red, thick] (5.5,-1) -- (5.5,8);
\node[below, red] (c) at (5.5,-1) {$r$};
\end{scope}
\end{tikzpicture}
 \caption{Graphical representation of the cut torus, with the cut plane highlighted in red (credits: Davide Morgante).}
    \label{fig:cut}
\end{figure}
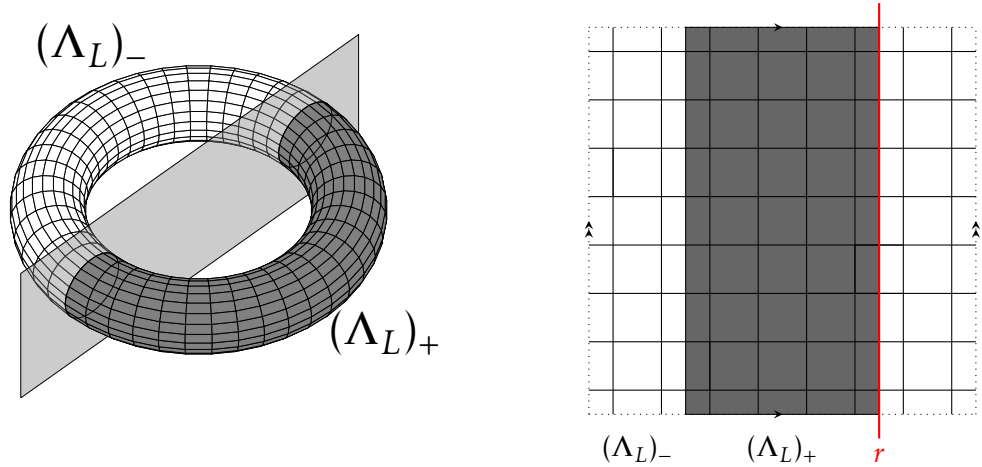

We aim to show that the three bosonic measures $\mu_{\alpha}(\phi)$,
$\alpha\in\{\mathrm{NN},\mathrm{SN},\mathrm{SP}\}$, constructed in the previous section,
are reflection positive \cite{cmp/1103904299} with respect to reflections across bond
planes of the underlying lattice $\Lambda_L^\alpha$, where: 
\[
\Lambda_L^{\mathrm{NN}}=\Lambda_L^{\mathrm{SN}}=\Lambda_L
\qquad
\Lambda_L^{\mathrm{SP}}=\widetilde{\Lambda}_L.
\]

As a very first step towards the notion of Reflection Positivity, we need the precise definition of reflections on $\L^{\a}_L$, regarded as a torus, with respect to some fixed \emph{cut plane}.
\begin{defn}[Cut Plane]
\label{def:cut_plane}
For $\a\in\{\mathrm{NN},\mathrm{SN},\mathrm{SP}\}$, a \virg{cut plane} on $\L^{\a}_L$ is a couple $(\nu,r)$, with $\nu\in\{0,1\}$ and
\begin{itemize}
\item $r\in \big\{-\frac{L-1}{2}, -\frac{L-3}{2},\dots, \frac{L+1}{2} \big\}$, if $\a=\mathrm{NN},\mathrm{SN}$,
\item $r\in \big\{-\frac{L}{2}+1, -\frac{L}{2}+3,\dots, \frac{L}{2}+1 \big\}$, if $\a=\mathrm{SP}$,
\end{itemize}
which, from a geometric viewpoint, identifies the axis $\{x_{\nu}=r\}$. We call \virg{canonical} the cut plane with $r=\frac{1}{2}$, if $\a=\mathrm{NN},\mathrm{SN}$ or $r=1$, if $\a=\mathrm{SP}$. Notice that the canonical cut plane is the only axis that preserves the fundamental domain presentation of each lattice\footnote{In other words it preserves the open square lattice}.
\end{defn}
\begin{defn}[Reflection on the Discrete Torus]
\label{def:reflections}
For $\a\in\{\mathrm{NN},\mathrm{SN},\mathrm{SP}\}$, to any cut plane $(\nu,r)$ on $\L^{\a}_L$, we associate a \virg{full-plane reflection} $\vartheta_{\nu}:\L^{\a}_{\infty}\mapsto \L^{\a}_{\infty}$, defined as:
\begin{equation*}
\vartheta_{\nu}(x)\doteq \begin{cases}
(2r-x_0,x_1), & \nu=0,\\
(x_0,2r-x_1), & \nu=1.
\end{cases}
\end{equation*}
Consequently, we define the \virg{torus reflection}, $\vartheta_{\nu}^L: \L^{\a}_L\mapsto \L^{\a}_L$, as:
\begin{equation*}
\vartheta_{\nu}^L\doteq \pi_{\L^{\a}_L}\circ\vartheta_{\nu},
\end{equation*}
with $\pi_{\L^{\a}_L}:\L^{\a}_{\infty}\mapsto\L^{\a}_L$ the unique map such that $\pi_{\L^{\a}_L}(x)- x \in L\mbb{Z}^2$ for every $x\in\L^{\a}_{\infty}$.
\end{defn}
Any cut plane $(r,\nu)$ induces the splitting $\L^{\a}_L= (\L^{\a}_L)_+\sqcup (\L^{\a}_L)_-$, where (see Fig. \ref{fig:cut}):
\begin{equation*}
(\L^{\a}_L)_+ \doteq \bigg\{x\in \L^{\a}_L: \; (x_{\nu}-r)\, \text{mod}\, L \in \big(-\tfrac{L}{2}, 0\big)  \bigg\}
\end{equation*}
and $(\L^{\a}_L)_-\doteq \L^{\a}_L\setminus (\L^{\a}_L)_+$. It is straightforward to check that $\vartheta_{\nu}^L((\L^{\a}_L)_{\pm})= (\L^{\a}_L)_{\mp}$.
\\

The reflection $\vartheta_{\nu}^L$ induces an action on
field configurations $\phi\in \mathbb{R}^{\Lambda_L^\alpha}$ given by:
\[
(\Theta\phi)_x \doteq \phi_{\vartheta_{\nu}^L(x)}.
\]

In the following, in order to avoid cluttering of the notation, we will denote $\vartheta_{\nu}^L$ simply by $\vartheta_{\nu}$ as reflections on the plane will not play any role in the main text. In Appendix \ref{app:Dirac} we will restore the double notation as we will need both reflections. Besides, we will possibly drop the index $\nu\in\{0,1\}$ for the reflection, whenever its role will not be essential.\\

Let $\mathcal A^\alpha$ be the algebra of complex polynomials on $\mathbb{R}^{\L^{\a}_L}$. We define the Osterwalder--Schrader reflection
$\Theta:\mathcal A^\alpha\to\mathcal A^\alpha$ by
\[
(\Theta F)(\phi) \doteq \overline{F(\Theta\phi)} .
\]

With this definition, $\Theta$ satisfies the following properties discussed in the following proposition.

\begin{proposition}[Properties of $\Theta$]
$\Theta$ is an anti-linear algebra homomorphism of $\mathcal{A}_{\alpha}$:
    \begin{enumerate}
        \item \textbf{\emph{Additivity:}} $\Theta(F+G) = \Theta(F)+\Theta(G)$;
        \item \textbf{\emph{Anti-homogeneity:}} $\Theta(\lambda F) = \overline{\lambda} \Theta(F)$;
        \item \textbf{\emph{Involutivity:}} $\Theta( \Theta(F))= F$;
        \item \textbf{\emph{Multiplicativity:}} $\Theta(FG)=\Theta(F)\Theta(G)$.
    \end{enumerate}
\end{proposition}

Let $\mathcal A^\alpha_\pm\subset\mathcal A^\alpha$ denote the subalgebra of observables
depending only on $\{\phi_x\}_{x\in(\Lambda_L^\alpha)_\pm}$. We define the pairing:
\begin{equation}\label{PS}
\begin{aligned}
\langle \cdot,\cdot\rangle_{\mu_{\a}}:\mathcal A^\alpha_+\times\mathcal A^\alpha_+&\to\mathbb C,\\
(A,B)&\mapsto \int d\mu_\alpha(\phi)\, (\Theta A)(\phi)\, B(\phi).
\end{aligned}
\end{equation}

\begin{thm}[Reflection Positivity of the bosonic measures]\label{rp}
For each $\alpha\in\{\mathrm{NN},\mathrm{SN},\mathrm{SP}\}$, the pairing \eqref{PS}
defines a positive semidefinite sesquilinear form on $\mathcal A^\alpha_+$. In particular,
for all $A,B\in\mathcal A^\alpha_+$,
\begin{enumerate}
\item $\langle A,B\rangle_{\mu_{\a}} = \overline{\langle B,A\rangle_{\mu_{\a}}}$,
\item $\langle A,A\rangle_{\mu_{\a}} \ge 0$.
\end{enumerate}
\end{thm}

\begin{remark}\label{rem:holonomy} 
Without loss of generality, we can always assume the cut plane to be the \emph{canonical one}.\\%, namely $r=\frac{1}{2}$ for $\a=\mathrm{NN},\mathrm{SN}$ and $r=1$ for $\a=\mathrm{SP}$. 

The canonical cut plane has two important properties:
\begin{itemize}
\item $\vartheta_{\nu}(\L_L^{\a})=\L_L^{\a}$;
\item the holonomy line, where the $\mbb{Z}_2$ twisting for anti-periodic boundary conditions is imposed, corresponds to the boundary $\de_-(\L^{\a}_L)_+$ (see Fig. \ref{fig1N}).
\end{itemize} 

 Indeed, if one wishes to implement reflection with respect to an arbitrary reflection plane, one may first perform a $\mathbb{Z}_2$-gauge transformation\footnote{Such a transformation leaves the multiplication operator $M_{\phi}$ 
 invariant, since it acts by conjugation with a diagonal multiplication operator which clearly commutes with $M_{\phi}$ (see \cite{10.1215/S0012-7094-93-07114-1}).} so that the corresponding holonomy defect coincides with the boundary $\partial_-(\Lambda^{\a}_L)_+$. After a suitable relabeling of the lattice coordinates, the problem is reduced to the canonical reflection plane considered above.  
In other words, the change of the reflection plane can always be complemented with a gauge transformation which moves the holonomy line over the boundary corresponding to $\de_-(\L^{\a}_L)_+$ (see Fig. \ref{fig2N}). For this reason, the proof of Theorem \ref{rp} will be carried out for reflections with respect to the canonical cut planes only. 
\end{remark}

  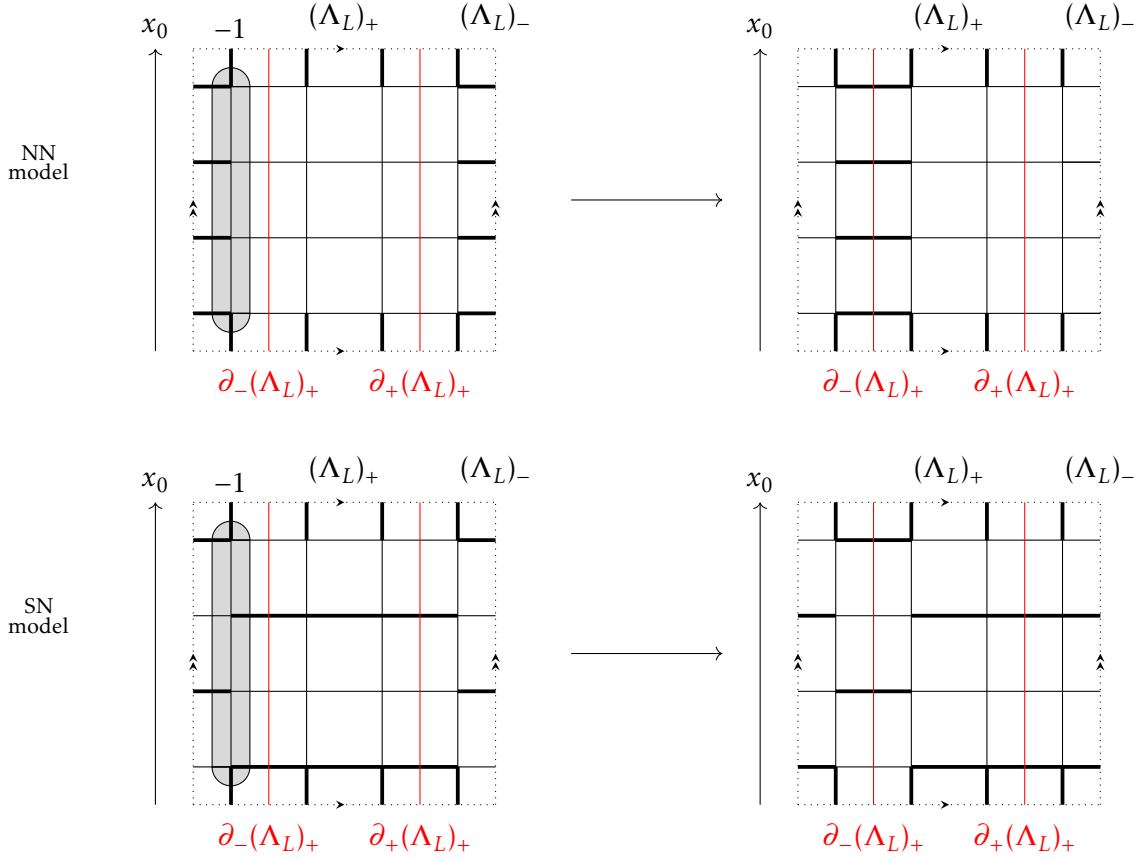
\begin{figure}[h]
    \centering
\begin{tikzpicture}
\node[left] (w) at (-2,2) {{$\substack{\mathrm{NN}\\\mathrm{model}}$}};
\draw[dotted, ->-] (-0.5,-0.5) -- (3.5,-0.5);
\draw[dotted, ->>-] (3.5,-0.5) -- (3.5,3.5);
\draw[dotted, ->-] (-0.5,3.5) -- (3.5,3.5);
\draw[dotted, ->>-] (-0.5,-0.5) -- (-0.5,3.5);
\draw (0, 0) -- (3,0);
\draw (0, 2) -- (3,2);
\draw[line width = 0.05 cm] (3,1) -- (3.5,1);
\draw[line width = 0.05 cm] (-0.5,1) -- (0,1);
\draw[line width = 0.05 cm] (-0.5,3) -- (0,3);
\draw[line width = 0.05 cm] (3,3) -- (3.5,3);
\draw (0,3) -- (3,3);
\draw (0,1) -- (3,1);
\draw (0,-0.5) -- (0,3.5);
\draw (1,-0.5) -- (1,3.5);
\draw (2,-0.5) -- (2,3.5);
\draw (-0.5,2) -- (0,2);
\draw (3,2) -- (3.5,2);
\draw (-0.5,0) -- (0,0);
\draw (3,0) -- (3.5,0);
\draw (3,-0.5) -- (3,3.5);

\draw[line width = 0.05 cm] (0,-0.5) -- (0,0);
\draw[line width = 0.05 cm] (1,-0.5) -- (1,0);
\draw[line width = 0.05 cm] (2,-0.5) -- (2,0);
\draw[line width = 0.05 cm] (3,-0.5) -- (3,0);

\draw[line width = 0.05 cm] (0,3.5) -- (0,3);
\draw[line width = 0.05 cm] (1,3.5) -- (1,3);
\draw[line width = 0.05 cm] (2,3.5) -- (2,3);
\draw[line width = 0.05 cm] (3,3.5) -- (3,3);
\draw[line width = 0.05 cm] (-0.5,0) -- (0,0);
\draw[line width = 0.05 cm] (3.5,0) -- (3,0);
\draw[line width = 0.05 cm] (-0.5,2) -- (0,2);
\draw[line width = 0.05 cm] (3.5,2) -- (3,2);
\draw[thin, red] (0.5,-0.5) -- (0.5,3.5);
\draw[thin, red] (2.5,-0.5) -- (2.5,3.5);
\node[below, red] (a) at (0.5,-0.6) {{$\partial_- (\L_L)_+$}};
\node[below, red] (a) at (2.5,-0.6) {{$\partial_+ (\L_L)_+$}};
\node[above] (a) at (1.5,3.6) {{$(\L_L)_+$}};
\node[above] (a) at (3.5,3.6) {{$(\L_L)_-$}};
\draw[->] (-1,-0.5) -- (-1,3.5);
\node[above] (a) at (-1,3.5) {{$x_0$}};
\draw[->] (4.5,1.5) -- (6.5,1.5);
\draw[fill, fill opacity=0.15] (-0.25,3) -- (-0.25,0) -- (0.25,0) -- (0.25,3) --(-0.25,3);
\draw[fill, fill opacity=0.15] (-0.25,3) -- (0.25,3) arc(0:180:0.25);
\draw[fill, fill opacity=0.15] (0.25,0) -- (-0.25,0) arc(180:360:0.25);
\node[above] (a) at (0,3.5) {{$-1$}};

\begin{scope}[xshift=8cm]
    \draw[dotted, ->-] (-0.5,-0.5) -- (3.5,-0.5);
\draw[dotted, ->>-] (3.5,-0.5) -- (3.5,3.5);
\draw[dotted, ->-] (-0.5,3.5) -- (3.5,3.5);
\draw[dotted, ->>-] (-0.5,-0.5) -- (-0.5,3.5);
\draw[line width = 0.05 cm] (1, 0) -- (0,0);
\draw[] (1,0) -- (3.5,0);
\draw[] (1,2) -- (3.5,2);
\draw[line width = 0.05cm] (1,2) -- (0,2);
\draw[] (0,2) -- (1,2);
\draw[] (0,0) -- (1,0);
\draw[] (3,1) -- (3.5,1);
\draw[] (-0.5,1) -- (0,1);
\draw[line width = 0.05 cm] (0,1) -- (1,1);
\draw[] (-0.5,3) -- (0,3);
\draw[] (3,3) -- (3.5,3);
\draw[line width = 0.05 cm] (0,3) -- (1,3);

\draw (0,3) -- (3,3);
\draw (0,1) -- (3,1);
\draw (0,-0.5) -- (0,3.5);
\draw (1,-0.5) -- (1,3.5);
\draw (2,-0.5) -- (2,3.5);
\draw (-0.5,2) -- (0,2);
\draw (3,2) -- (3.5,2);
\draw (-0.5,0) -- (0,0);
\draw (3,0) -- (3.5,0);
\draw (3,-0.5) -- (3,3.5);

\draw[line width = 0.05 cm] (0,-0.5) -- (0,0);
\draw[line width = 0.05 cm] (1,-0.5) -- (1,0);
\draw[line width = 0.05 cm] (2,-0.5) -- (2,0);
\draw[line width = 0.05 cm] (3,-0.5) -- (3,0);

\draw[line width = 0.05 cm] (0,3.5) -- (0,3);
\draw[line width = 0.05 cm] (1,3.5) -- (1,3);
\draw[line width = 0.05 cm] (2,3.5) -- (2,3);
\draw[line width = 0.05 cm] (3,3.5) -- (3,3);
\draw[thin, red] (0.5,-0.5) -- (0.5,3.5);
\draw[thin, red] (2.5,-0.5) -- (2.5,3.5);
\node[below, red] (a) at (0.5,-0.6) {{$\partial_- (\L_L)_+$}};
\node[below, red] (a) at (2.5,-0.6) {{$\partial_+ (\L_L)_+$}};
\node[above] (a) at (1.5,3.6) {{$(\L_L)_+$}};
\node[above] (a) at (3.5,3.6) {{$(\L_L)_-$}};
\draw[->] (-1,-0.5) -- (-1,3.5);
\node[above] (a) at (-1,3.5) {{$x_0$}};
\end{scope}
\begin{scope}[yshift=-6cm]
\node[left] (w) at (-2,2) {{$\substack{\mathrm{SN}\\\mathrm{model}}$}};
    \draw[dotted, ->-] (-0.5,-0.5) -- (3.5,-0.5);
\draw[dotted, ->>-] (3.5,-0.5) -- (3.5,3.5);
\draw[dotted, ->-] (-0.5,3.5) -- (3.5,3.5);
\draw[dotted, ->>-] (-0.5,-0.5) -- (-0.5,3.5);
\draw[line width = 0.05 cm] (0, 0) -- (3,0);
\draw[line width = 0.05 cm] (0, 2) -- (3,2);
\draw[line width = 0.05 cm] (3,1) -- (3.5,1);
\draw[line width = 0.05 cm] (-0.5,1) -- (0,1);
\draw[line width = 0.05 cm] (-0.5,3) -- (0,3);
\draw[line width = 0.05 cm] (3,3) -- (3.5,3);
\draw (0,3) -- (3,3);
\draw (0,1) -- (3,1);
\draw (0,-0.5) -- (0,3.5);
\draw (1,-0.5) -- (1,3.5);
\draw (2,-0.5) -- (2,3.5);
\draw (-0.5,2) -- (0,2);
\draw (3,2) -- (3.5,2);
\draw (-0.5,0) -- (0,0);
\draw (3,0) -- (3.5,0);
\draw (3,-0.5) -- (3,3.5);

\draw[line width = 0.05 cm] (0,-0.5) -- (0,0);
\draw[line width = 0.05 cm] (1,-0.5) -- (1,0);
\draw[line width = 0.05 cm] (2,-0.5) -- (2,0);
\draw[line width = 0.05 cm] (3,-0.5) -- (3,0);

\draw[line width = 0.05 cm] (0,3.5) -- (0,3);
\draw[line width = 0.05 cm] (1,3.5) -- (1,3);
\draw[line width = 0.05 cm] (2,3.5) -- (2,3);
\draw[line width = 0.05 cm] (3,3.5) -- (3,3);
\draw[thin, red] (0.5,-0.5) -- (0.5,3.5);
\draw[thin, red] (2.5,-0.5) -- (2.5,3.5);
\node[below, red] (a) at (0.5,-0.6) {{$\partial_- (\L_L)_+$}};
\node[below, red] (a) at (2.5,-0.6) {{$\partial_+ (\L_L)_+$}};
\node[above] (a) at (1.5,3.6) {{$(\L_L)_+$}};
\node[above] (a) at (3.5,3.6) {{$(\L_L)_-$}};
\draw[->] (-1,-0.5) -- (-1,3.5);
\node[above] (a) at (-1,3.5) {{$x_0$}};
\draw[->] (4.5,1.5) -- (6.5,1.5);
\draw[fill, fill opacity=0.15] (-0.25,3) -- (-0.25,0) -- (0.25,0) -- (0.25,3) --(-0.25,3);
\draw[fill, fill opacity=0.15] (-0.25,3) -- (0.25,3) arc(0:180:0.25);
\draw[fill, fill opacity=0.15] (0.25,0) -- (-0.25,0) arc(180:360:0.25);
\node[above] (a) at (0,3.5) {{$-1$}};

\begin{scope}[xshift=8cm]
    \draw[dotted, ->-] (-0.5,-0.5) -- (3.5,-0.5);
\draw[dotted, ->>-] (3.5,-0.5) -- (3.5,3.5);
\draw[dotted, ->-] (-0.5,3.5) -- (3.5,3.5);
\draw[dotted, ->>-] (-0.5,-0.5) -- (-0.5,3.5);
\draw[line width = 0.05 cm] (-0.5, 0) -- (0,0);
\draw[line width = 0.05 cm] (1,0) -- (3.5,0);
\draw[line width = 0.05 cm] (1,2) -- (3.5,2);
\draw[line width = 0.05cm] (-0.5,2) -- (0,2);
\draw[] (0,2) -- (1,2);
\draw[] (0,0) -- (1,0);
\draw[] (3,1) -- (3.5,1);
\draw[] (-0.5,1) -- (0,1);
\draw[line width = 0.05 cm] (0,1) -- (1,1);
\draw[] (-0.5,3) -- (0,3);
\draw[] (3,3) -- (3.5,3);
\draw[line width = 0.05 cm] (0,3) -- (1,3);

\draw (0,3) -- (3,3);
\draw (0,1) -- (3,1);
\draw (0,-0.5) -- (0,3.5);
\draw (1,-0.5) -- (1,3.5);
\draw (2,-0.5) -- (2,3.5);
\draw (-0.5,2) -- (0,2);
\draw (3,2) -- (3.5,2);
\draw (-0.5,0) -- (0,0);
\draw (3,0) -- (3.5,0);
\draw (3,-0.5) -- (3,3.5);

\draw[line width = 0.05 cm] (0,-0.5) -- (0,0);
\draw[line width = 0.05 cm] (1,-0.5) -- (1,0);
\draw[line width = 0.05 cm] (2,-0.5) -- (2,0);
\draw[line width = 0.05 cm] (3,-0.5) -- (3,0);

\draw[line width = 0.05 cm] (0,3.5) -- (0,3);
\draw[line width = 0.05 cm] (1,3.5) -- (1,3);
\draw[line width = 0.05 cm] (2,3.5) -- (2,3);
\draw[line width = 0.05 cm] (3,3.5) -- (3,3);
\draw[thin, red] (0.5,-0.5) -- (0.5,3.5);
\draw[thin, red] (2.5,-0.5) -- (2.5,3.5);
\node[below, red] (a) at (0.5,-0.6) {{$\partial_- (\L_L)_+$}};
\node[below, red] (a) at (2.5,-0.6) {{$\partial_+ (\L_L)_+$}};
\node[above] (a) at (1.5,3.6) {{$(\L_L)_+$}};
\node[above] (a) at (3.5,3.6) {{$(\L_L)_-$}};
\draw[->] (-1,-0.5) -- (-1,3.5);
\node[above] (a) at (-1,3.5) {{$x_0$}};
\end{scope}
\end{scope}

\end{tikzpicture}
\caption{By means of a $\mathbb{Z}_2$-Gauge Transformation supported on one side of the holonomy line, we can translate such line to the left making it overlap with the left boundary of $(\L_L)_+$. An example of this construction for the NN and SN models is shown in the upper and lower panels, respectively.}
\label{fig2N}
\end{figure}

\begin{proof}
The proof follows the same general strategy in all three cases 
$\alpha \in \{\mathrm{NN},\mathrm{SN},\mathrm{SP}\}$. Therefore, we first outline the common argument in an abstract setting and then proceed to the specialization to each model.  

    \begin{enumerate}
\item Let $A,B \in \mathcal A_+$. Then:
\begin{equation*}
\begin{aligned}
\langle A,B\rangle_{\mu_{\a}}
&= \int d\mu_\alpha(\phi)\, (\Theta A)(\phi)\, B(\phi) = \int d\mu_\alpha(\phi)\, \overline{A(\Theta\phi)}\, B(\phi).
\end{aligned}
\end{equation*}
Performing the change of variables $\phi \mapsto \Theta\phi$ and using the
reflection invariance of the measure:
\[
\mu_\alpha \circ \Theta^{-1} = \mu_\alpha,
\]
we obtain:
\begin{equation*}
\begin{aligned}
\langle A,B\rangle_{\mu_{\a}}
&= \int d\mu_\alpha(\phi)\, \overline{A(\phi)}\, B(\Theta\phi) = \overline{\int d\mu_\alpha(\phi)\, \overline{B(\Theta\phi)}\, A(\phi)} = \overline{\langle B,A\rangle_{\mu_{\a}}}.
\end{aligned}
\end{equation*}
The reflection invariance of the measure holds in each model as a consequence of the following identity (see Lemmata \ref{RND} and \ref{RSD}): 

\begin{equation}
\det\big((\slashed{D}_{\alpha})^-_{\Lambda_L}- M_{\Theta(\phi)}\big)= \det\big((\slashed{D}_{\alpha})^-_{\Lambda_L}- M_\phi\big).
\tag{H1}
\label{H1}
\end{equation}

%For the naive fermion case this follows from property \eqref{RND4} in Lemma~\ref{prnd}, while for the staggered formulations it follows from property \eqref{RSD4} in Lemma~\ref{rsdo}.
\item Let $A \in \mathcal{A}_+$, then:

   \begin{equation*}
       \begin{aligned}
           \left \langle A, A \right \rangle_{\mu_{\a}} &=  \int d \mu_{\alpha}(\phi) (\Theta A)(\phi) A(\phi).
       \end{aligned}
   \end{equation*} 
   
   The key idea is to represent the determinant appearing in the bosonic measure $\mu_{\alpha}(\phi)$ by reintroducing auxiliary fermionic fields, thereby recovering a local fermionic 
action. One then exploits the reflection covariance properties of the Dirac 
operator to decompose the action into a sum of terms supported on the positive 
and negative halves of the lattice, together with mixed boundary contributions. 
Using the Osterwalder--Schrader reflection and the specific structure of the 
cross terms, one shows that the pairing $\langle A,A\rangle_{\mu_{\a}}$ can be written as a sum of absolute squares and is therefore non-negative.

\paragraph{Step 1: Reintroduction of the fermionic variables.} Using the Grassmann representation of the determinant, we rewrite the bosonic
measure as the marginal of a local measure on $(\phi,\overline{\psi},\psi)$ and thus it factorizes over
$(\Lambda^{\alpha}_L)_\pm$. Writing $\Lambda^{\alpha}_L=(\Lambda^{\alpha}_L)_+\sqcup(\Lambda^{\alpha}_L)_-$ and splitting the
bosonic and on-site fermionic factors accordingly, we obtain:

\begin{equation*}
\begin{aligned}
&\prod_{x \in \Lambda_L^{\alpha} } d \phi_x e^{-\frac{N}{2 \lambda} \phi^2_x}  \det\big((\slashed{D}_{\alpha})^-_{\Lambda_L}- M_\phi\big)^N\\
&= \int d \nu^{\a}_+(\phi, \overline{\psi}, \psi)    d \nu^{\a}_-(\phi, \overline{\psi}, \psi) \exp\Big(-\sum_{J,J'} \sum_{x,y \in \Lambda_L^{\a}} \overline{\psi}_{x, J}\big((\slashed{D}_{\alpha})^-_{\Lambda_L}\big)^{J,J'}_{xy} \psi_{y, J'}\Big)
\end{aligned}
\tag{H2}
\label{H2}
\end{equation*}

   where:

   \begin{equation*}
       d \nu^{\a}_{\pm}(\phi, \overline{\psi}, \psi) = \prod_{x \in (\Lambda_L^{\alpha})_{\pm}} \prod_{J} d \phi_x d \overline{\psi}_{x, J} d\psi_{x,J} e^{-\frac{N}{2\lambda} \phi_x^2} e^{\phi_x \sum_J \overline{\psi}_{x, J} \psi_{x, J}}
   \end{equation*}

and the integral is meant to be taken only on the fermionic part.  Notice that the term $e^{\phi_x \sum_J \overline{\psi}_{x, J} \psi_{x, J}}$ is actually a polynomial in the bosonic field due to the nilpotency of the Grassmann Variables. 

\paragraph{Step 2: RP Decomposition of the fermion kinetic term.} 

Using the notation introduced in Section \ref{sec2}, we write the fermionic kinetic term as

$$K_{\alpha} \doteq -\sum_{J,J'} \sum_{x,y \in \Lambda^{\a}_L} \overline{\psi}_{x, J}\big((\slashed{D}_{\alpha})^-_{\Lambda_L}\big)_{(x,J),(y,J')} \psi_{y, J'}.$$ 

The reflection $\Theta_\nu$, across the canonical axis orthogonal to $e_{\nu}$, acts on fermions depending on the specific formulation:\footnote{Notice that these definitions are the correct ones only for reflections across the \emph{canonical cut plane} (see Definition \ref{def:cut_plane}). For reflections across generic cut planes, \eqref{RNF},\eqref{RSF} and \eqref{RSFP} must be accompanied with the multiplication of the right-hand side by an extra sign:
\[\s_{\nu}^L(x)\doteq (-1)^{\sum_{\mu=0,1}\frac{1}{L}(\vartheta_{\nu}(x)- \vartheta_{\nu}^L(x) )_{\mu}}.\] 
See also Appendix \ref{app:Dirac} for a discussion of reflections across generic cut planes.\\}
    
    \begin{subequations}
    \begin{equation}
        \begin{cases}            \Theta_\nu(\overline{\psi}_{x,s,j})= \sum_{s'}(\gamma_\nu)_{s,s'}\,\psi_{\vartheta_\nu(x),s',j}\\
            \Theta_\nu(\psi_{x,s,J})=\sum_{s'}\overline{\psi}_{\vartheta_{\nu}(x),s',j}\,(\gamma_\nu)_{s',s}
        \end{cases}
        \tag{RNF}
        \label{RNF}
    \end{equation}
    \\
    \begin{equation}
        \begin{cases}            \Theta_\nu(\overline{\psi}_{x,j})= \Gamma_\nu(x)\,\psi_{\vartheta_\nu(x),j}\\            \Theta_\nu(\psi_{x,j})=\overline{\psi}_{\vartheta_\nu(x),j}\,\Gamma_\nu(x)
        \end{cases}
        \tag{RSF}
        \label{RSF}
    \end{equation}
   \\
   \begin{equation}
       \begin{cases}
            \Theta_\nu(\overline{\psi}_{x,a, j})= \Gamma_\nu(a)\,\psi_{\vartheta_\nu(x),\vartheta_{\nu}(a), j}\\           \Theta_\nu(\psi_{x,a, j})=\overline{\psi}_{\vartheta_\nu(x),\vartheta_{\nu}(a),j}\,\Gamma_\nu(a),
        \end{cases}
        \tag{RSFP}
        \label{RSFP}
   \end{equation}
    \end{subequations}

for $\a=\mathrm{NN},\mathrm{SN},\mathrm{SP}$ respectively. In the third case $\vartheta_{\nu}$ also acts on the internal index $J=(j, a)$ because $a$ describes a position within unit cells\footnote{Explicitly: \[\theta_1(a)= \begin{cases}
    A \,\,\,\,\,\,\,\,\,\,&a=B\\
    B &a=A\\
    D &a=C\\
    C &a = D\\
\end{cases},\,\,\,\,\,\,\,\,\,\,\,\, \theta_0(a)= \begin{cases}
    A \,\,\,\,\,\,\,\,\,\,&a=C\\
    C &a=A\\
    D &a=B\\
    B &a= D\\
\end{cases}.\]}. We will show that each such transformation is an involution. In the following, in order to lighten the notation, we will omit the index $\nu$ wherever is not needed (i.e. whenever no explicit computation is performed).  \\

The Reflection operator is then extended to the full algebra of Polynomials in the Grassmann variables $\mathcal{A}^{\alpha}_{\mathrm{Fer}}$ as an anti-linear anti-homomorphism, namely $\Theta(\lambda A) = \overline{\lambda} \Theta(A)$ and $\Theta(AB)= \Theta(B)\Theta(A)$. \\

Let us suppose that each fermionic kinetic term has the form: 
  
  \begin{equation*}
      K_{\alpha} = B_{\alpha} + \Theta(B_{\alpha}) + \sum_{\omega} C^{\omega}_{\a} \Theta(C^{\omega}_{\a}), 
      \tag{H3}
      \label{H3}
  \end{equation*}
  
  with $\o$ an index running over a suitable finite set, $B_{\alpha} \in (\mathcal{A}^{\alpha}_{\mathrm{Fer}})^{\emph{even}}_+$ and $C^{\omega}_{\a} \in (\mathcal{A}^{\alpha}_{\mathrm{Fer}})^{\emph{odd}}_+$. We will show that such condition is sufficient to ensure $\left \langle A, A \right \rangle_{\mu_{\a}} \geq 0$. 

  Since:
  
  \begin{equation*}
      e^{\Theta(B_{\a})} = \sum_{n \geq 0} \frac{\Theta(B_{\a})^n}{n!} = \sum_{n \geq 0} \frac{\Theta(B_{\a}^n)}{n!} = \Theta\bigg(\sum_{n \geq 0} \frac{B_{\a}^n}{n!}\bigg) = \Theta(e^{B_{\a}}),
  \end{equation*}

we are left with proving that:

    \begin{equation*}
        \int d \nu^{\a}_+ d \nu^{\a}_- e^{B_{\a} + \Theta(B_{\a}) + \sum_{\omega} C^{\omega}_{\alpha} \Theta(C^{\omega}_{\alpha}) } \Theta(A) A =   \int d \nu^{\a}_+ e^{B_{\a}}  \int d \nu^{\a}_- \Theta(e^{B_{\alpha}}) e^{\sum_{\omega} C^{\omega}_{\alpha} \Theta(C^{\omega}_{\alpha}) } \Theta(A) A \geq 0.
    \end{equation*}

    Expanding the mixed term, we find that

    \begin{equation*}
        \begin{aligned}
            e^{\sum_{\omega} C^{\omega}_{\a} \Theta(C^{\omega}_{\a})} &= \sum_{n \geq 0} \frac{1}{n!} \sum_{\omega_1 \cdots \omega_n} C^{\omega_1}_{\alpha} \Theta(C^{\omega_1}_{\alpha}) \cdots C^{\omega_n}_{\alpha} \Theta(C^{\omega_n}_{\alpha})\\
            &= \sum_{n \geq 0} \frac{1}{n!} (-1)^{\frac{n(n-1)}{2}}\sum_{\omega_1 \cdots \omega_n} C^{\omega_1}_{\a}  \cdots C^{\omega_n}_{\alpha} \Theta(C^{\omega_1}_{\alpha}) \cdots\Theta(C^{\omega_n}_{\alpha})\\
            &= \sum_{n \geq 0} \frac{1}{n!} (-1)^{\frac{n(n-1)}{2}}\sum_{\omega_1 \cdots \omega_n} C^{\omega_1}_{\alpha}  \cdots C^{\omega_n}_{\alpha} \Theta(C^{\omega_n}_{\alpha} \cdots C^{\omega_1}_{\alpha})\\
            &= \sum_{n \geq 0} \frac{1}{n!} \sum_{\omega_1 \cdots \omega_n} C^{\omega_1}_{\alpha}  \cdots C^{\omega_n}_{\alpha} \Theta(C^{\omega_1}_{\alpha} \cdots C^{\omega_n}_{\alpha}).
        \end{aligned}
    \end{equation*}

    Thus:

\begin{equation*}
\begin{aligned}
\int d \nu^{\a}_+ d \nu^{\a}_- e^{B_{\a} + \Theta(B_{\a}) + \sum_{\omega} C^{\omega}_{\alpha} \Theta(C^{\omega}_{\alpha}) } \Theta(A) A&= \sum_{n \geq 0} \frac{1}{n!} \sum_{\omega_1 \cdots \omega_n} \int d \nu^{\a}_+ e^{B_{\a}}  C^{\omega_1}_{\alpha}  \cdots C^{\omega_n}_{\alpha} A \times \\
&\,\,\,\,\,\,\,\,\,\,\,\,\,\,\,\,\,\,\,\,\,\,\,\,\,\,\,\times \int d \nu^{\a}_- \Theta (e^{B_{\a}} C^{\omega_1}_{\alpha} \cdots C^{\omega_n}_{\alpha}  A).
\end{aligned}
\end{equation*}

   Since\footnote{In order to prove the identity on the reflected measures, it is enough to observe that, for any function \(f\) depending on the Grassmann and real variables in \((\Lambda_L^{\alpha})_{\pm}\), the Grassmann integral of \(f\) with respect to \(d\nu_{\pm}^{\alpha}\) vanishes unless \(f\) contains all Grassmann variables associated with \((\Lambda_L^{\alpha})_{\pm}\). Hence, without loss of generality, we may restrict ourselves to observables of the form $
f = P(\phi)\prod_{x\in (\Lambda_L^{\alpha})_+}\prod_J \overline{\psi}_{x,J}\psi_{x,J}$, where \(P(\phi)\) is an integrable function depending only on \(\phi|_{(\Lambda_L^{\alpha})_+}\). For such observables, $\Theta(f)
=
\Theta(P(\phi))
\prod_{x\in (\Lambda_L^{\alpha})_+}\prod_J
\overline{\psi}_{\theta(x),\theta(J)}\psi_{\theta(x),\theta(J)}$, where, importantly, no sign ambiguity has arisen in this transformation. Since, by definition, $\left\langle f \right\rangle_{\Theta(d\nu_-^{\alpha})}
= \overline{
\left\langle \Theta(f) \right\rangle_{d\nu_-^{\alpha}}}$, the claimed identity immediately follows.} $\Theta(d \nu^{\a}_-) = d \nu^{\a}_+$, we have: 

\begin{equation*}
\int d \nu^{\a}_+ d \nu^{\a}_- e^{B_{\a} + \Theta(B_{\a}) + \sum_{\omega} C^{\omega}_{\a} \Theta(C^{\omega}_{\a}) } \Theta(A) A = \bigg|\int d \nu^{\a}_+ e^{B_{\a}}  C^{\omega_1}_{\alpha}  \cdots C^{\omega_n}_{\alpha} A\bigg|^2 \geq 0.
\end{equation*}

The remaining part of the proof consists of checking \eqref{H1}, \eqref{H2} and \eqref{H3} for each lattice formulation.
\end{enumerate}
\end{proof}

\subsection{Naive Fermions with Naive Interaction}

\begin{lemma}
Properties \eqref{H1}, \eqref{H2} and \eqref{H3} hold for canonical reflections, within the $\mathrm{NN}$ model.
\end{lemma}

\begin{proof}
\leavevmode
    \begin{enumerate}
        \item[\eqref{H1}] See Lemma \ref{RND} in Appendix \ref{Ap21}.        
        \item[\eqref{H2}] By a simple computation.
        \item[\eqref{H3}] Let us first show that the reflection \eqref{RNF} is an involution, as claimed right after its definition.

        \begin{equation*}
        \begin{aligned}
            \Theta_{\nu} (\Theta_{\nu}(\overline{\psi}_{x,s,j})) &= \Theta_{\nu}\big( \sum_{s'=1,2} (\gamma_{\nu})_{s,s'} \psi_{\vartheta_{\nu}(x),s',j}\big)\\
            &= \sum_{s'=1,2} \overline{(\gamma_{\nu})_{s,s'}} \sum_{s''=1,2} \overline{\psi}_{x,s'', j} (\gamma_{\nu})_{s'',s'}\\
            &= \sum_{s',s''=1,2} (\gamma_{\nu})_{s'',s'}(\gamma_{\nu}^{*})_{s',s} \overline{\psi}_{x,s'', j} \\
            &= \sum_{s',s''=1,2} (\gamma_{\nu})_{s'',s'}(\gamma_{\nu})_{s',s} \overline{\psi}_{x,s'', j} \\
            &= \sum_{s''=1,2} (\gamma_{\nu}^2)_{s'',s}\overline{\psi}_{x,s'', j} \\
            &= \overline{\psi}_{x,s, j}.
            \end{aligned}
        \end{equation*}

 We need now to show that the kinetic term: 

\begin{equation*}
  K=  - \sum_{\substack{x, y \in \Lambda_L\\j = 1, \cdots, N}} \overline{\psi}_{x,j} \big((\slashed{D}_{\mathrm{NN}})^-_{\Lambda_L}\big)_{xy} \psi_{y,j}
\end{equation*}

(with $K\equiv K_{\mathrm{NN}}$ in this subsection) fits the decomposition \eqref{H3}. This is not quite true. However, with a simple staggered phase transformation (which preserves the local bilinears and the fermionic measure):

\begin{equation*}
    \begin{cases}
        \overline{\psi}_{x,j} \to  \mathbbl{1}_2 \otimes (-1)^{x_0+x_1}  \overline{\psi}_{x,j}\\
        \psi_{x,j} \to\mathbbl{1}_2 \otimes  (-1)^{x_0+x_1}  \psi_{x,j}. 
    \end{cases}
\end{equation*}

we get:

\begin{equation*}
    K \to K' = \sum_{\substack{x, y \in \Lambda_L\\j=1, \cdots, N}} \overline{\psi}_{x,j} \big((\slashed{D}_{\mathrm{NN}})^-_{\Lambda_L}\big)_{x,y} \psi_{y,j}.
\end{equation*}

Let us split the sum in the following way:

\begin{equation}
\label{eqn:52}
  K'= K'_{++} + K'_{--} + K'_{+-} + K'_{-+}, 
\end{equation}

where:

\begin{equation*}
    \begin{aligned}
        K'_{\e\e'} &= \sum_{x\in \Lambda_L^{\e},y \in \Lambda^{\e'}_L}  \overline{\psi}_{x,j} \big((\slashed{D}_{\mathrm{NN}})^-_{\Lambda_L}\big)_{x,y} \psi_{y,j}, \qquad \e,\e'\in\{\pm\}.        
    \end{aligned}
\end{equation*}

Let us analyze the four terms in the r.h.s. of \eqref{eqn:52}.

\begin{enumerate}
\item Let us begin by discussing $K'_{++}$ and $K'_{--}$. We want to show that $        \Theta_{\nu}(K'_{++}) = K'_{--}$. By a direct computation:

        \begin{equation*}
        \begin{aligned}
            \Theta_{\nu}(K'_{++}) &=  \sum_{\substack{x, y \in (\L_L)_+\\j=1, \cdots, N}}\Theta_{\nu}\big(\overline{\psi}_{x,j} \big((\slashed{D}_{\mathrm{NN}})^-_{\Lambda_L}\big)_{x,y} \psi_{y,j}   \big)\\
            &= \sum_{\substack{x, y \in (\L_L)_+\\s_1,s_2=1,2 \\j=1, \cdots, N}}\Theta_{\nu}\big(\overline{\psi}_{x,s_1,j} \big((\slashed{D}_{\mathrm{NN}})^-_{\Lambda_L}\big)_{(x,s_1),(y,s_2)}\psi_{y,s_2,j}   \big) \\
            &=  \sum_{\substack{x, y \in (\L_L)_+\\s_1,s_2,s_3,s_4=1,2\\j=1, \cdots, N}} \overline{\psi}_{\vartheta_{\nu}(y),s_3,j} (\gamma_{\nu})_{s_3,s_2} \overline{\big((\slashed{D}_{\mathrm{NN}})^-_{\Lambda_L}\big)_{(x,s_1),(y,s_2)}} (\gamma_{\nu})_{s_1,s_4}\psi_{\vartheta(x),s_4,j}\\ 
            &=  \sum_{\substack{x, y \in (\L_L)_+\\s_1,s_2,s_3,s_4=1,2\\j=1, \cdots, N}} \overline{\psi}_{\vartheta_{\nu}(y),s_3,j} (\gamma_{\nu})_{s_3,s_2} {\big((\slashed{D}_{\mathrm{NN}})^-_{\Lambda_L}\big)^*}_{(y,s_2),(x,s_1)} (\gamma_{\nu})_{s_1,s_4}\psi_{\vartheta(x),s_4,j}\\   &=  \sum_{\substack{x, y \in (\L_L)_+\\s_3,s_4=1,2\\j=1, \cdots, N}} \overline{\psi}_{\vartheta(y),s_3,j} \Big(\gamma_{\nu} \big((\slashed{D}_{\mathrm{NN}})^-_{\Lambda_L}\big)^*_{x,y} \gamma_{\nu}\Big)_{s_3,s_4} \psi_{\vartheta(x),s_4,j}\\
            &=  \sum_{\substack{x,y \in (\L_L)_+\\j=1, \cdots, N}} \overline{\psi}_{\vartheta(y),j} \big((\slashed{D}_{\mathrm{NN}})^-_{\Lambda_L}\big)_{\vartheta(y), \vartheta(x)} \psi_{\vartheta(x),j}\\
            &=  \sum_{\substack{x,y \in (\L_L)_-\\j=1, \cdots, N}} \overline{\psi}_{x,j} \big((\slashed{D}_{\mathrm{NN}})^-_{\Lambda_L}\big)_{x,y}\psi_{y,j}\\
            &= K_{--},
            \end{aligned}
        \end{equation*}

        where we have used \eqref{RND1} to pass from the fourth to the third line from the end.
    \item The terms $K'_{+-},K'_{-+}$ are more tricky and show where the anti-symmetric boundary conditions are needed. We will show that:

\begin{equation}
\label{eqn:53a}
K'_{+-} + K'_{-+} = \sum_{j=1}^N \sum_{\epsilon = \pm 1} \sum_{s=1,2} \frac{1}{2} \bigg(\sum_{\substack{x \in (\L_L)_+:\\(x,x+e_{\mu}) \in \partial_+ (\L_L)_+}} \!\!\!\!\!\! \psi^{\epsilon}_{x,s,j} \Theta(\psi^{\epsilon}_{x,s,j})  + \sum_{\substack{x \in (\L_L)_+:\\(x-e_{\nu}, x) \in \partial_- (\L_L)_+}} \!\!\!\!\!\! \psi^{\epsilon}_{x,s,j} \Theta(\psi^{\epsilon}_{x,s,j})\bigg),
\end{equation}

with $\psi^-_{x,s,j} \equiv \psi_{x,s,j}$ and $\psi^+_{x,s,j} \equiv \overline{\psi}_{x,s,j}$. 

Let us consider the subset of cross terms localized around the right boundary $\partial_+ (\L_L)_+$ (or in general on the boundary of the region $(\L_L)_+$ where we have not inserted the $\mathbb{Z}_2$-twisting):

\begin{equation*}
\begin{aligned}
K'_{+-} &= \sum_{\substack{x \in (\L_L)_+:\\
(x,x+e_{\nu}) \in \partial_+ (\L_L)_+\\s_1,s_2=1,2\\
j=1, \cdots, N}} \frac{1}{2}\Big( \overline{\psi}_{x,s_1,j} (\gamma_{\nu})_{s_1,s_2} \psi_{x+e_{\nu},s_2,j} - \overline{\psi}_{x+e_{\nu},s_1j} (\gamma_{\nu})_{s_1,s_2} \psi_{x,s_2,j} \Big) \\
&= \sum_{\substack{x \in (\L_L)_+:\\
(x,x+e_{\nu}) \in \partial_+ (\L_L)_+\\s_1,s_2=1,2\\
j=1, \cdots, N}} \frac{1}{2}\Big( \overline{\psi}_{x,s_1,j} (\gamma_{\nu})_{s_1,s_2} \psi_{x+e_{\nu},s_2,j} + \psi_{x,s_2,j} (\gamma_{\nu})_{s_1,s_2}\overline{\psi}_{x+e_{\nu},s_1,j} \Big) \\
&= \sum_{\substack{x \in (\L_L)_+:\\
(x,x+e_{\nu}) \in \partial_+ (\L_L)_+\\s=1,2\\
j=1, \cdots, N}} \frac{1}{2} \Big( \overline{\psi}_{x,s,j} \Theta(\overline{\psi}_{x,s,j}) + {\psi}_{x,s,j} \Theta({\psi}_{x,s,j}) \Big),
\end{aligned}
\end{equation*}   

where the Grassmann variables $\psi,\ov{\psi}$ are understood with \emph{periodic} boundary conditions.

The terms on the other boundary of $(\L_L)_+$, $\partial_- (\L_L)_+$, come with an extra minus sign which is compensated by the $\mathbb{Z}_2$-holonomy (see Fig. \ref{fig1N}): 

\begin{equation*}
\begin{aligned}
K'_{-+} &= \sum_{\substack{x \in (\L_L)_+:\\
(x-e_{\nu},x) \in \partial_- (\L_L)_+\\s_1,s_2=1,2\\
j=1, \cdots, N}} \frac{1}{2}\Big( -\overline{\psi}_{x-e_{\nu},s_1,j} (\gamma_{\nu})_{s_1,s_2} \, \psi_{x,s_2,j} + \overline{\psi}_{x,s_1,j} (\gamma_{\nu})_{s_1,s_2} \psi_{x-e_{\nu},s_2,j} \Big)\\
&= \sum_{\substack{x \in (\L_L)_+:\\
(x-e_{\nu},x) \in \partial_- (\L_L)_+\\s_1,s_2=1,2\\
j=1, \cdots, N}} \frac{1}{2}\Big( \psi_{x,s_2,j}(\gamma_{\nu})_{s_1,s_2} \overline{\psi}_{x-e_{\nu},s_1,j}  \,  + \overline{\psi}_{x,s_1,j} (\gamma_{\nu})_{s_1,s_2} \psi_{x-e_{\nu},s_2,j} \Big)\\
&=  \sum_{\substack{x \in (\L_L)_+:\\
(x-e_{\nu},x) \in \partial_- (\L_L)_+\\s=1,2\\
j=1, \cdots, N}} \frac{1}{2} \Big( \psi_{x,s,j} \Theta(\psi_{x,s,j}) + \overline{\psi}_{x,s,j} \Theta(\overline{\psi}_{x,s,j} ) \Big).
\end{aligned}
\end{equation*}

\end{enumerate}

\end{enumerate}

\end{proof}

\subsection{Staggered Fermions with Naive Interaction}

\begin{lemma}
    Properties \eqref{H1}, \eqref{H2} and \eqref{H3} hold for canonical reflections, within the $\mathrm{SN}$ model.
\end{lemma}

\begin{proof}
\leavevmode
    \begin{enumerate}
        \item[\eqref{H1}] See Lemma \ref{RSD} in Appendix \ref{AP22}.        
        \item[\eqref{H2}] By a simple computation.
        \item[\eqref{H3}] Let us first show that the reflection \eqref{RSF} is an involution:

        \begin{equation*}
        \begin{aligned}
            \Theta_{\nu} (\Theta_{\nu}(\overline{\psi}_{x,j})) &= \Theta_{\nu}\big( \Gamma_{\nu}(x) \psi_{\vartheta_{\nu}(x),j}\big)\\
            &=  \Gamma_{\nu}(x) \Gamma_{\nu}(\vartheta_{\nu}(x)) \overline{\psi}_{x, j}\\
            &=  \Gamma_{\nu}(x)^2 \overline{\psi}_{x, j} \\
            &= \overline{\psi}_{x, j},
            \end{aligned}
        \end{equation*}

             \begin{equation*}
        \begin{aligned}
            \Theta_{\nu} (\Theta_{\nu}(\psi_{x,j})) &= \Theta_{\nu}\big( \Gamma_{\nu}(x) \overline{\psi}_{\vartheta_{\nu}(x),j}\big)\\
            &=  \Gamma_{\nu}(x) \Gamma_{\nu}(\vartheta_{\nu}(x)) {\psi}_{x, j}\\
            &=  \Gamma_{\nu}(x)^2 {\psi}_{x, j} \\
            &= {\psi}_{x, j},
            \end{aligned}
        \end{equation*}

where we have used that

\begin{equation*}
\Gamma_{\nu}(\vartheta_{\nu}(x)) = \begin{cases}
1 \,\,\,\,&\nu=0\\
(-1)^{\vartheta_{1}(x_0)} & \nu=1
\end{cases}\,\,\,\, =\,\, \begin{cases}
1 \,\,\,\,&\nu=0\\
(-1)^{x_0} & \nu=1
\end{cases}\,\,\,\,  = \,\Gamma_{\nu}(x).
\end{equation*}
        
 We need now to show that the kinetic term:

\begin{equation*}
  K=  - \sum_{\substack{x, y \in \Lambda_L\\j = 1, \cdots, N}} \overline{\psi}_{x,j} \big((\slashed{D}_{\mathrm{SN}})^-_{\Lambda_L}\big)_{xy} \psi_{y,j}
\end{equation*}

(with $K\equiv K_{\mathrm{SN}}$ in this subsection) fits the decomposition \eqref{H3}. Again this is not strictly true; however it can be fixed by a simple staggered phase transformation (which preserves the local bilinears and the fermionic integration):

\begin{equation*}
    \begin{cases}
        \overline{\psi}_{x,j} \to   (-1)^{x_0+x_1}  \overline{\psi}_{x,j}\\
        \psi_{x,j} \to (-1)^{x_0+x_1}  \psi_{x,j}, 
    \end{cases}
\end{equation*}

\begin{equation*}
K \to K' = \sum_{\substack{x, y \in \Lambda_L\\j=1, \cdots, N}} \overline{\psi}_{x,j} \big((\slashed{D}_{\mathrm{SN}})^-_{\Lambda_L}\big)_{x,y}  \psi_{y,j}.
\end{equation*}

Let us split $K'$ in the following way:

\begin{equation}
\label{eqn:52b}
K'= K'_{++} + K'_{--} + K'_{+-} + K'_{-+}, 
\end{equation}

where:

\begin{equation*}
    \begin{aligned}
        K'_{\e\e'} &= \sum_{x\in \Lambda_L^{\e},y \in \Lambda^{\e'}_L}  \overline{\psi}_{x,j} \big((\slashed{D}_{\mathrm{SN}})^-_{\Lambda_L}\big)_{x,y} \psi_{y,j}, \qquad \e,\e'\in\{\pm\}.        
    \end{aligned}
\end{equation*}

Let us analyze the terms in the r.h.s. of \eqref{eqn:52b}.

\begin{enumerate}
\item Let us begin by discussing $K'_{++}$ and $K'_{--}$. We want to show that $        \Theta(K'_{++}) = K'_{--}$. By a direct computation:

\begin{equation*}
\begin{aligned}
\Theta(K'_{++}) &= \sum_{x, y \in (\L_L)_+} \Theta(\psi_y)\ov{\big((\slashed{D}_{\mathrm{SN}})^-_{\Lambda_L}\big)_{x,y}}    \Theta( \overline{\psi}_x )\\
&= \sum_{x, y \in (\L_L)_+} \Gamma_{\nu}(x) \overline{\psi}_{\vartheta(y)} \big((\slashed{D}_{\mathrm{SN}})^-_{\Lambda_L}\big)_{y,x}^*  \Gamma_{\nu}(y) \psi_{\vartheta(x)}\\      
&= \sum_{x,y \in (\L_L)_+} \overline{\psi}_{\vartheta(y)} \big((\slashed{D}_{\mathrm{SN}})^-_{\Lambda_L}\big)_{\vartheta(y), \vartheta(x)} \psi_{\vartheta(x)}\\
&= \sum_{x,y \in (\L_L)_-} \overline{\psi}_{x}\big((\slashed{D}_{\mathrm{SN}})^-_{\Lambda_L}\big)_{x,y} \psi_{y}\\
&= K'_{--},
\end{aligned}
\end{equation*}

where we have used \eqref{RSD1} to pass from the second to the third line.

\item In analogy with the NN model, for the terms $K'_{+-}$ and $K'_{-+}$ we will show a decomposition of the form:
 
\begin{equation*}
K'_{+-} + K'_{-+} = \sum_{j=1}^N \sum_{\epsilon= \pm 1} \frac{1}{2} \bigg( \sum_{\substack{x \in (\L_L)_+:\\(x,x+e_{\nu}) \in \partial_+ (\L_L)_+}} \!\!\psi^{\epsilon}_{x,j} \Theta(\psi^{\epsilon}_{x,j})  + \sum_{\substack{x \in (\L_L)_+:\\(x- e_{\nu}, x) \in \partial_- (\L_L)_+}} \!\!\psi^{\epsilon}_{x,j} \Theta(\psi^{\epsilon}_{x,j})\bigg),
\end{equation*}

with $\psi^-_{x,j}\equiv \psi_{x,j}$ and $\psi^+_{x,j}\equiv \ov{\psi}_{x,j}$.

Let us consider the subset of cross terms localized around the right boundary $\partial_+ (\L_L)_+$ (or in general on the boundary of the region $(\L_L)_+$ where we have not inserted the $\mathbb{Z}_2$-twisting):

\begin{equation*}
\begin{aligned}
K'_{+-} &= \sum_{\substack{x \in (\L_L)_+:\\
(x,x+e_{\nu}) \in \partial_+ (\L_L)_+\\
j=1, \cdots, N}} \frac{1}{2}\Big(\overline{\psi}_{x,j} \Gamma_{\nu}(x)\psi_{x+e_{\nu},j} - \overline{\psi}_{x+e_{\nu},j} \Gamma_{\nu}(x) \psi_{x,j} \Big)\\
&= \sum_{\substack{x \in (\L_L)_+:\\
(x,x+e_{\nu}) \in \partial_+ (\L_L)_+\\
j=1, \cdots, N}} \frac{1}{2} \Big( \overline{\psi}_{x,j} 
\Gamma_{\nu}(x)\psi_{x+e_{\nu},j} + \psi_{x,j} \Gamma_{\nu}(x) \overline{\psi}_{x+e_{\nu},j} \Big) \\
&= \sum_{\substack{x \in (\L_L)_+:\\
(x,x+e_{\nu}) \in \partial_+ (\L_L)_+\\
j=1, \cdots, N}} \frac{1}{2}\Big( \overline{\psi}_{x,j} \Theta(\overline{\psi}_{x,j}) + {\psi}_{x,j} \Theta({\psi}_{x,j}) \Big),
\end{aligned}
\end{equation*}

where again the Grassmann variables $\psi,\ov{\psi}$ are understood with \emph{periodic} boundary conditions.

The terms on the other boundary of $(\L_L)_+$, $\partial_- (\L_L)_+$, come with an extra minus sign which is compensated by the $\mathbb{Z}_2$-holonomy (see Fig. \ref{fig1}):

\begin{equation*}
\begin{aligned}
K'_{-+} &= \sum_{\substack{x \in (\L_L)_+:\\
(x-e_{\nu},x) \in \partial_- (\L_L)_+\\
j=1, \cdots, N}} \frac{1}{2}\Big( -\overline{\psi}_{x-e_{\nu},j} \Gamma_{\nu}(x) \, \psi_{x,j} + \overline{\psi}_{x,j} \Gamma_{\nu}(x)\psi_{x-e_{\nu},j} \Big)\\
&= \sum_{\substack{x \in (\L_L)_+:\\
(x-e_{\nu},x) \in \partial_- (\L_L)_+\\
j=1, \cdots, N}} \frac{1}{2}\Big( \psi_{x,j} \Gamma_{\nu}(x) \overline{\psi}_{x-e_{\nu},j}  \,  + \overline{\psi}_{x,j} \Gamma_{\nu}(x) \psi_{x-e_{\nu},j} \Big)\\
&= \sum_{\substack{x \in (\L_L)_+:\\
(x-e_{\nu},x) \in \partial_- (\L_L)_+\\
j=1, \cdots, N}} \frac{1}{2}\Big( \psi_{x,j} \Theta(\psi_{x,j}) + \overline{\psi}_{x,j} \Theta(\overline{\psi}_{x,j} ) \Big).\\
\end{aligned}
\end{equation*}
        
\end{enumerate} 
        \end{enumerate}
        \end{proof}

\subsection{Staggered Fermions with Plaquette Interaction}

\begin{lemma}
    Properties \eqref{H1}, \eqref{H2} and \eqref{H3} hold for canonical reflections, within the $\mathrm{SP}$ model.
\end{lemma}

\begin{proof}
The Lemma is a corollary of the related lemma for the SN model. Indeed, one can think of the SP model as a particular case of the SN model where the bosonic field is constant within each unit cell $\mathcal{I}$. Since reflections are performed between cells, the structure is preserved under reflections and every consideration made above can be easily translated here.
\end{proof}

\section{Long-Range Order}
\label{sec4}

In this section we establish Long-Range Order for the bosonic field with respect to the effective measures $\mu_\alpha(\phi)$ and, by the Hubbard--Stratonovich identity \eqref{eq:HS-2pt-common}, for the fermionic bilinear $\overline{\psi}\psi$.\\

In each model, the effective bosonic measure can be written in Gibbs form as:

\begin{equation*}
    d \mu_{\alpha}(\phi) = \prod_{x \in \Lambda^{\alpha}_L} d \phi_x e^{-N  V^{\alpha}_{L,\lambda}(\phi)}
 \end{equation*}

where the effective potential is given by

\begin{equation}
\label{EP}\tag{EP}
V^{\a}_{L,\l}(\phi) \doteq \frac{1}{2\l} \sum_{x\in\L_L^{\a}}\phi_x^2 - \log\big|\det\big( (\slashed{D}_{\a})^-_{\L_L} - M_{\phi} \big) \big|.
\end{equation}

\begin{remark}
The determinant appearing in the definition of $V^{\alpha}_{L,\lambda}(\phi)$ may vanish for certain configurations $\phi\in \mbb{R}^{\L^{\a}_L}$. However, since it is a polynomial in the field variables, its zero set has zero Lebesgue measure (and hence also zero Gaussian measure). Throughout the paper, we adopt the convention that $V^{\alpha}_{L,\lambda}(\phi)=+\infty$ on such configurations. In particular, $V^{\alpha}_{L,\lambda}$ takes values in $\mathbb{R}\cup\{+\infty\}$.
\end{remark}

\begin{figure}
    \centering
\begin{tikzpicture}[scale=1.5]
\begin{axis}[
    axis lines = middle,
    xlabel = {$\vphi$},
    xmin = -2.2, xmax = 2.2,
    ymin = -1, ymax = 3,
   samples = 200,
    domain = -1.5:1.5,
    width=9cm,
    height=6cm,
     xtick = \empty,
    ytick = \empty,
]

    \addplot[blue, thick] {x^4 - x^2+0.2};
    \node[blue] at (1.8,2) {\scalebox{1}{$ v^{\alpha}_{\lambda}(\vphi)$}};
    \draw[] (-0.4,-0.5) -- (-0.4, 3);
     \draw[] (0.4,-0.5) -- (0.4, 3);
      \draw[] (-1,-0.5) -- (-1, 3);
     \draw[] (1,-0.5) -- (1, 3);
     \node[below] at (-0.4,-0.5) {\scalebox{0.75}{$-C_{\alpha}$}};
\node[below] at (0.4,-0.5) {\scalebox{0.75}{$C_{\alpha}$}};
\node[below] at (1,-0.5) {\scalebox{0.75}{$K_{\alpha}$}};
\node[below] at (-1,-0.5) {\scalebox{0.75}{$-K_{\alpha}$}};
\addplot[only marks, mark=*, red] coordinates {
    (-0.707,-0.05)
    (0.707,-0.05)
    (0,0.2)
};

\node[below] at (-0.707,-0.05) {\scalebox{0.75}{$-\vphi^{\star}_{\alpha}(\lambda)$}};
\node[below] at (0.707,-0.05) {\scalebox{0.75}{$\vphi^{\star}_{\alpha}(\lambda)$}};
    
\end{axis}

\end{tikzpicture}
\caption{Sketch of the Effective Potential $v^{\alpha}_{\lambda}(\vphi)$}
\label{effpot1}
\end{figure}
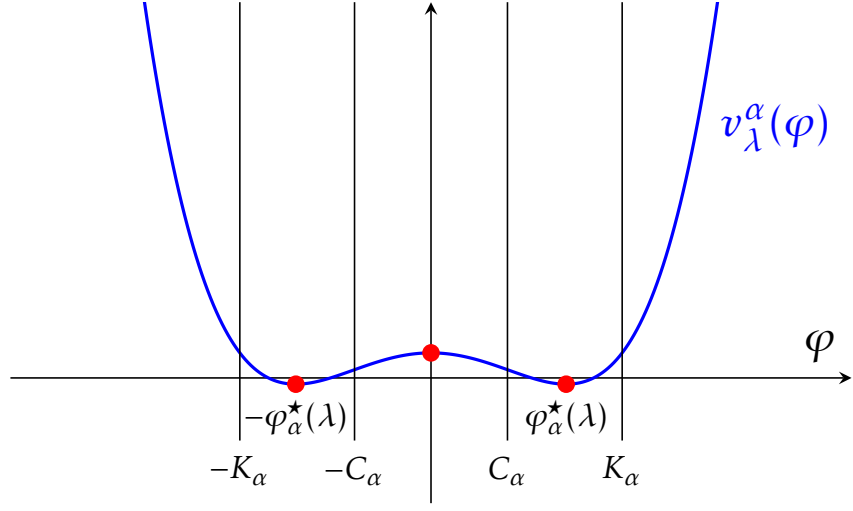

The strategy of the proof is dictated by the structure of the effective potential. Although $ V^{\alpha}_{L,\lambda}$ is highly non-local and does not admit a simple closed expression, Reflection Positivity implies that, among configurations constrained to lie in a given interval, 
its minimizers are spatially constant and can therefore be computed explicitly via Fourier transforms. The corresponding infinite-volume effective potential density:

\begin{equation}
\label{tdpd}
\tag{TDEPD}
    v^{\alpha}_{\lambda}(\vphi) \doteq \lim_{L \to \infty} \frac{1}{|\Lambda^{\a}_L|} V^{\a}_{L, \l}(\phi)\Big|_{\phi= \vphi^{\L^{\a}_L}},
\end{equation}

has a symmetric double-well structure, with minima at $\pm\varphi_\alpha^\star(\lambda) \neq 0$ for any $\lambda >0$ (see Fig.~\ref{effpot1}). Long-Range Order is then proved by decomposing the configuration space into regions around the two minima $\vphi=\pm\vphi^{\star}_{\a}(\l)$, the local maximum $\vphi=0$, and the large-field regions, and by estimating the contribution from each region separately.\\

The following facts about the effective potential will be repeatedly used.

\begin{proposition}[Structural properties of the effective potential]
\label{const_field}
For each $\alpha\in\{\mathrm{NN},\mathrm{SN},\mathrm{SP}\}$ and every interval $I\subseteq\mathbb R$, the following hold.

\begin{enumerate}
\item\label{itt:1}
The constrained infimum of $ V^{\alpha}_{L,\lambda}$ over $I^{\Lambda_L^\alpha}$ is attained among constant configurations:
\begin{equation}
    \label{Minpot}
    \inf_{\phi\in I^{\Lambda_L^\alpha}}  V^{\alpha}_{L,\lambda}(\phi) = \inf_{\varphi\in I}  V^{\alpha}_{L,\lambda}(\phi)\Big|_{\phi= \vphi^{\L^{\a}_L}}
\end{equation}

\item\label{itt:2}
There exists a constant $\mathfrak c_0>0$ such that:
\begin{equation}
    \label{eqn:12a}
    \left|
\tfrac{1}{|\Lambda_L^\alpha|} V^{\alpha}_{L,\lambda}(\varphi^{\L^{\a}_L})
-
 v^{\alpha}_{\lambda}(\varphi)
\right|
\le
\frac{\mathfrak c_0}{L},
\end{equation}

uniformly with respect to $\varphi \in \mathbb{R}$, where the infinite-volume potential density $v^{\alpha}_{\lambda}$ admits the explicit representation:
\begin{equation}
\label{eqn:58}
 v^{\alpha}_{\lambda}(\varphi)
=
\frac{\varphi^2}{2\lambda}
-
C_\alpha
\iint_{(-\pi,\pi]^2}\frac{d^2p}{(2\pi)^2}
\log\big(\varphi^2+\sin^2(p_0)+\sin^2(p_1)\big)
\end{equation}

with $C_{\a}=1,\frac{1}{2},2$ for $\a=\mathrm{NN},\mathrm{SN},\mathrm{SP}$ respectively.

\item\label{itt:3} The function $v^{\alpha}_{\lambda}$ has exactly two non-zero global minima at $\pm\varphi_\alpha^\star(\lambda)$ and a local maximum at $0$. 
$\vphi^{\star}_{\a}$ is a $C^{\infty}$ function of $\l\in(0,\infty)$ and satisfies, for $\l\le 1$,
\begin{equation}
\label{eqn:41}
\mf{a}_1e^{-\frac{\mf{b}_1}{\l}}\le \vphi^{\star}_{\a}(\l) \le   \mf{a}_2e^{-\frac{\mf{b}_2}{\l}}, 
\end{equation}
for suitable constants $\mf{a}_1,\mf{a}_2,\mf{b}_1,\mf{b}_2>0$. Besides, $v^{\a}_{\l}$ is strictly decreasing over $[0,\vphi^{\star}_{\a}(\l)]$ and strictly increasing over $[\vphi^{\star}_{\a}(\l),\infty)$.

\end{enumerate}
\end{proposition}

The proof of Proposition \ref{const_field} is discussed in Appendix \ref{app:pot}.

\subsection{Field-size decomposition}
\label{ssect:field_size}
By Proposition \ref{const_field} we have that the minima of the finite-volume energy density are asymptotically determined by those of $ v^{\alpha}_{\lambda}$. The latter exhibits a symmetric double-well structure (see Fig.~\ref{effpot1}), with absolute minima at $\pm \varphi_\alpha^\star(\lambda)$. This suggests partitioning the configuration space according to the position of the field relative to the minima, the local maximum at the origin, and the large-field regime. Configurations belonging to energetically unfavorable regions, namely large fields or fields near the local maximum, are exponentially suppressed in $N$. This suppression is quantified by comparing the energy of configurations constrained to a given region with that of optimal constant configurations, using Chessboard Estimates \cite{cmp/1103904299, Friedli_Velenik_2017}, which are a remarkable corollary of Reflection Positivity, together with \eqref{Minpot}.

\begin{thm}[Chessboard Estimates]
   Let $\mu$ be a reflection-positive probability measure on $\mbb{R}^{\Lambda_L}$ with respect to reflections between neighboring blocks $\Lambda_B + Bt$ with $t \in \mathbb{Z}^2_{L/B}$ (where both $L, L/B \in 2 \mathbb{N}$). Then for any positive set of functions $f_1, \dots, f_n$ supported in $\Lambda_B$ and for any set of vectors $t_1, \cdots, t_n \in \mathbb{Z}^2_{L/B}$, the following inequality holds:

\begin{equation}\label{CE}\tag{CE}
     \Big\langle \prod_{j=1}^n \Theta_{t_j}(f_j) \Big\rangle_{\mu} \leq \prod_{j=1}^n  \Big\langle \prod_{t \in \mathbb{Z}^2_{L/B}} \Theta_{t}(f_j)\Big\rangle_{\mu}^{|B|/L^2}
\end{equation}

where $\Theta_t$ is the composition of reflections that brings the block $\Lambda_B$ in the block $\Lambda_B + Bt$
\end{thm}

\begin{proof}
    See \cite[Pag. 444-449]{Friedli_Velenik_2017}, \cite[Pag. 43-46]{Biskup_2009} or \cite[Pag. 19-20]{cmp/1103904299}
\end{proof}

The contributions which are not apriori suppressed arise from configurations where both fields lie near the minima. Configurations in which the field at different points lies around different minima can be shown to be again exponentially suppressed in $N$, via a 
Peierls-type argument combined with Reflection Positivity, in the spirit of \cite{frohlich1978phase} (see also \cite{Biskup_2009, Friedli_Velenik_2017}). While this strategy applies to all our three models, the explicit estimates depend on the specific lattice realization of the Dirac operator and will be treated separately in the following subsections.\\

To implement this strategy, we introduce a quantitative partition of the configuration space at the level of single-site variables. Fix parameters $C_\alpha,K_\alpha>0$ such that:
\[
0< C_\alpha < \varphi_\alpha^\star(\lambda) < K_\alpha.
\]
We then split the single-site field values into the following five regions:
\begin{enumerate}
\item\label{it:large+}
\textbf{Large positive fields:}
$\phi_x\in (K_\alpha,\infty)$;

\item\label{it:large-}
\textbf{Large negative fields:} $\phi_x\in (-\infty,-K_\alpha)$;

\item\label{it:max}
\textbf{Neighborhood of the local maximum:} $\phi_x\in (-C_\alpha,C_\alpha)$;

\item\label{it:min+}
\textbf{Neighborhood of the positive minimum:}
$\phi_x\in [C_\alpha,K_\alpha]$;

\item\label{it:min-}
\textbf{Neighborhood of the negative minimum:}
$\phi_x\in [-K_\alpha,-C_\alpha]$.
\end{enumerate}

Accordingly, introduce the following indicator functions:

\begin{align*}
&I_x^{\pm K_\alpha}(\phi)\doteq \mathbbl{1}_{\{\pm \phi_x>K_\alpha\}}, \qquad I_x^{(-C_\alpha,C_\alpha)}(\phi)\doteq\mathbbl{1}_{\{|\phi_x|<C_\alpha\}}, \qquad I_x^{\star_{\pm}}(\phi) \doteq \mathbbl{1}_{\{C_\alpha\le \pm \phi_x\le K_\alpha\}}.
\end{align*}

The lower bound on the two-point function is obtained by isolating the contribution in which both fields lie near the same minimum, and by estimating all the remaining contributions in terms of three errors: \emph{large-field} events, \emph{local-maximum} events, and \emph{opposite-minimum} events. Following the decomposition carried out in Appendix~\ref{applbub}, we obtain:
\begin{equation}
\label{eq:mainbound2}
\begin{aligned}
\inf_{\substack{x,y \in \Lambda^{\alpha}_L}} \langle \phi_x\phi_y\rangle_{\mu_\alpha} 
&\ge C_{\alpha}^2 - \Bigg( 2(K_{\alpha}^2+C_{\alpha}^2) \sup_{x,y\in\L^{\alpha}_L} \left\langle I^{\star_+}_x I^{\star_-}_y \right\rangle_{\mu_{\alpha}} +  4C_{\alpha} (C_{\a}+ K_{\alpha}) \sup_{x\in\L^{\alpha}_L} \left\langle I^{(-C_{\alpha},C_{\alpha})}_x \right\rangle_{\mu_{\alpha}} \\
&+ 4\big(1+ \tfrac{C_{\alpha}^2}{K_{\alpha}^2}\big)\sup_{x\in\L^{\alpha}_L} \left\langle \phi_x^2 I^{+K_{\alpha}}_x \right\rangle_{\mu_{\alpha}} + 4\big(C_{\alpha}+ 2\tfrac{C_{\alpha}^2}{K_{\alpha}}+ 2K_{\alpha}\big) \sup_{x\in\L^{\alpha}_L} \left\langle \phi_x^2  I^{+K_{\alpha}}_x\right\rangle_{\mu_{\alpha}}^{\frac{1}{2}}  \Bigg).
\end{aligned}
\end{equation}

Similarly, see again Appendix \ref{applbub}, for the upper bound we find:
\begin{equation}
\label{upbound}
    \sup_{x,y \in \Lambda^{\alpha}_L}  \langle \phi_x \phi_y \rangle_{\mu_{\alpha}} \leq K_{\alpha}^2 + 2K_{\alpha} \sup_{x \in \Lambda_L^{\alpha}} \left \langle \phi_x^2 I_x^{+K_{\alpha}} \right \rangle_{\mu_{\alpha}}^{\frac{1}{2}} + 4 \sup_{x \in \Lambda_L^{\alpha}} \big \langle \phi_x^2 I_x^{+K_{\alpha}} \big \rangle_{\mu_{\alpha}}.
\end{equation}

Thus, the proof reduces to estimating the following three expectation values:
\[
\big\langle \phi_x^2 I_x^{+K_\alpha}\big\rangle_{\mu_\alpha},
\qquad
\big\langle I_x^{(-C_\alpha,C_\alpha)} \big\rangle_{\mu_\alpha},
\qquad
\big\langle I_x^{\star_+}I_y^{\star_-}\big\rangle_{\mu_\alpha},
\]

whose bounds are discussed in Proposition \ref{prop:contrib} below. The first two correspond to \textit{large-field} and \textit{local-maximum} events and will be controlled directly by Chessboard Estimates together with Proposition~\ref{const_field}. The third one corresponds to configurations in which the field fluctuates near different minima at two distant sites and its control requires a Peierls-type argument combined with Reflection Positivity.

\begin{proposition}
\label{prop:contrib}
There exists $\lambda_0>0$ such that the following holds. For every $0<\lambda\le\lambda_0$ and every choice of parameters
\[
0<C_\alpha<\varphi_\alpha^\star(\lambda)<K_\alpha,
\]
with $\vphi^{\star}_{\a}$ as in Proposition \ref{const_field}, there exist $N_1(K_\alpha,C_\alpha,\lambda)$ and $L_1(K_\alpha,C_\alpha,\lambda)$ such that, for all 
\[
N\ge N_1(K_\alpha,C_\alpha,\lambda), \qquad L\ge L_1(K_\alpha,C_\alpha,\lambda),
\]
we have:
\begin{align}
\sup_{x\in\Lambda_L^\alpha}
\langle \phi_x^2 I_x^{+K_\alpha}\rangle_{\mu_\alpha}
&\le
e^{-\frac N2 c_1^\alpha(\lambda)}
\label{eqn:18}
\\[0.3em]
\sup_{x\in\Lambda_L^\alpha}
\langle I_x^{(-C_\alpha,C_\alpha)}\rangle_{\mu_\alpha}
&\le
e^{-\frac N2 c_2^\alpha(\lambda)}
\label{eqn:18b}
\\[0.3em]
\sup_{x,y\in\Lambda_L^\alpha}
\langle I_x^{\star_+} I_y^{\star_-}\rangle_{\mu_\alpha}
&\le
e^{-\frac N4 c_3^\alpha(\lambda)}
\label{eqn:18c}
\end{align}

where:

\begin{enumerate}
    \item $c_1^\alpha(\lambda)$ is the \virg{large-field energy gap} associated with configurations where $\phi_x > K_\alpha$:

    \[c_1^\alpha(\lambda) \doteq
 v^{\alpha}_{\lambda}(K_\alpha)- v^{\alpha}_{\lambda}(\varphi_\alpha^\star(\lambda))>0;\]

 \item $c_2^\alpha(\lambda)$ is the \virg{local-maximum energy gap} associated with configurations $-C_{\alpha}<\phi_x<C_{\alpha}$:

 \[c_2^\alpha(\lambda)
\doteq
 v^{\alpha}_{\lambda}(C_\alpha)- v^{\alpha}_{\lambda}(\varphi_\alpha^\star(\lambda))>0;\]

 \item $c_3^\alpha(\lambda)>0$ is the \virg{Peierls gap} associated with alternating-sign configurations obtained by reflection of elementary blocks:
\[
c^{\alpha}_3(\lambda)
\doteq
\min\left\{
v^{\alpha}_{\lambda}\!\left(\tfrac{1}{\sqrt{2}}\varphi^{\star}_{\alpha}(\lambda)\right)
-
v^{\alpha}_{\lambda}\!\left(\varphi^{\star}_{\alpha}(\lambda)\right),
\;
-\frac{1}{4} \log \Big(1- \mf{C}_2^{\a}\frac{\varphi^{\star}_{\alpha}(\lambda)^{n_\alpha}}{1+ \varphi^{\star}_{\alpha}(\lambda)^{n_\alpha}} \Big)
\right\},
\]
where $n_{\alpha}=4,2,8$ for $\alpha=\mathrm{NN},\mathrm{SN},\mathrm{SP}$ respectively and $\mf{C}_2^{\a}>0$ are suitable constants.
\end{enumerate}
\end{proposition}

Assuming Proposition \ref{prop:contrib} for the moment, we first derive Theorem \ref{thm1}. The proof of Proposition \ref{prop:contrib} is given in Section \ref{sect:contrib}.

\subsection{Proof of Theorem \ref{thm1}}
\label{ssect:proof_theorem}

We now show how Proposition~\ref{prop:contrib} implies the main theorem. By \eqref{eq:HS-2pt-common} we have that 

\begin{equation*}
\frac{1}{|\Lambda^{\a}_L|^2}\sum_{x,y\in\L^{\a}_L} \left\langle (\overline{\psi}\psi)_x (\overline{\psi} \psi)_y \right\rangle_{\a}= \left(\frac{N}{\l}\right)^2 \frac{1}{|\L^{\a}_L|^2} \sum_{x,y\in\L^{\a}_L} \left\langle \phi_x \phi_y \right\rangle_{\mu_{\a}} - \frac{N}{\l} \frac{1}{|\L^{\a}_L|}.
\end{equation*}

Therefore:

\begin{equation}
\label{eqn:32}
\left(\frac{N}{\l}\right)^2 \Big( \min_{x,y\in\L^{\a}_L} \langle\phi_x\phi_y\rangle_{\mu_{\a}}- \frac{4\l}{N L^2} \Big)\le \frac{1}{|\Lambda^{\a}_L|^2}\sum_{x,y\in\L^{\a}_L} \left\langle (\overline{\psi}\psi)_x (\overline{\psi} \psi)_y \right\rangle_{\a} \le \left(\frac{N}{\l}\right)^2 \max_{x,y\in\L^{\a}_L} \langle\phi_x\phi_y\rangle_{\mu_{\a}}.    
\end{equation}

Given any $\varepsilon\in(0,1)$, we apply Proposition \ref{prop:contrib} with

\[
C_\alpha=\sqrt{1-\frac{\varepsilon}{2}}\,\varphi_\alpha^\star(\lambda),
\qquad
K_\alpha=\sqrt{1+\frac{\varepsilon}{2}}\,\varphi_\alpha^\star(\lambda).
\]

Using the upper and lower bound for the bosonic correlator, \eqref{upbound} and \eqref{eq:mainbound2} respectively, we find that

\begin{align*}
&\sup_{x,y\in\Lambda_L^\alpha}\langle \phi_x\phi_y\rangle_{\mu_\alpha}
\le
\Big(1+\frac{\varepsilon}{2}\Big)\varphi_\alpha^\star(\lambda)^2+\mathcal E_1(N,\l,\ve),\\
&\inf_{x,y\in\Lambda_L^\alpha}\langle \phi_x\phi_y\rangle_{\mu_\alpha}
\ge
\Big(1-\frac{\varepsilon}{2}\Big)\varphi_\alpha^\star(\lambda)^2-\mathcal E_2(N,\l,\ve),
\end{align*}

for every $N\ge N_1(K_{\a},C_{\a},\l)$ and $L\ge L_1(K_{\a},C_{\a},\l)$, where $\mathcal E_1$ and $\mathcal E_2$ are suitable quantities, which are exponentially small in $N$ for any $\l\in(0,\l_0]$ and $\ve\in(0,1)$ fixed. Going back to \eqref{eqn:32}, we have that

\begin{equation*}
\begin{split}
\left(\frac{N}{\l}\right)^2\left[\left(1-\frac{\ve}{2}\right)\vphi^{\star}_{\a}(\l)^2 - \mc{E}_2(N,\l,\ve) - \frac{4\l}{N L^2} \right]&\le \frac{1}{|\Lambda^{\a}_L|^2}\sum_{x,y\in\L^{\a}_L} \left\langle (\overline{\psi}\psi)_x (\overline{\psi} \psi)_y \right\rangle_{\a}\\
&\le \left(\frac{N}{\l}\right)^2 \left[\left(1+\frac{\ve}{2}\right)\vphi^{\star}_{\a}(\l)^2 + \mc{E}_1(N,\l,\ve)\right].  
\end{split}\end{equation*}

Finally we require $N,L$ to be large enough (depending on $\l$ and $\ve$), in order to satisfy: 

\begin{equation*}
\mc{E}_1(N,\l,\ve)\le \frac{\ve}{2} \vphi^{\star}_{\a}(\l)^2, \qquad \mc{E}_2(N,\l,\ve)\le \frac{\ve}{4} \vphi^{\star}_{\a}(\l)^2, \qquad \frac{1}{N L^2}\le \frac{\ve}{8\l} \vphi^{\star}_{\a}(\l)^2,
\end{equation*}

so that \eqref{eqn:33} follows.

\section{Proof of Proposition \ref{prop:contrib}}
\label{sect:contrib}

This section is devoted to the analysis of the three main subdominant contributions to the two-point bosonic correlator $\langle\phi_x\phi_y\rangle_{\mu_{\a}}$, namely the \emph{large-field}, the \emph{local-maximum} and the \emph{Pierles-like} term. These will be shown to satisfy the bounds \eqref{eqn:18}, \eqref{eqn:18b} and \eqref{eqn:18c} respectively, which are the main building blocks for the proof of LRO, as discussed in Subsection \ref{ssect:proof_theorem}. The first two estimates will be proved in Subsections~\ref{ssect:large_fields} and \ref{ssect:loc_max}. The third estimate, treated in Subsection~\ref{ssect:Peierls}, requires a separate analysis of explicit, alternating, periodic configurations, discussed in Subsection \ref{ssect:alternate}. 
%and is the only place where a genuine Peierls' argument enters.

\subsection{Large-field terms}
\label{ssect:large_fields}

We prove the large-field bound \eqref{eqn:18},
%\[\sup_{x\in \Lambda_L^\alpha}\left\langle \phi_x^2 \,I_x^{+K_\alpha}\right\rangle_{\mu_\alpha}\le e^{-\frac N2 c_1^\alpha(\lambda)}\]
namely we show that configurations with $\phi_x > K_\alpha$ are exponentially suppressed in $N$.\\

The proof combines Chessboard Estimates with the fact that, for $K_\alpha>\varphi_\alpha^\star(\lambda)$, the mean-field potential \eqref{eq:mean_field_pot}  is strictly larger than its minimum. Applying the Chessboard Estimate \eqref{CE} to the single-site observable $I^{+K_{\a}}_x(\phi)$, we obtain:
\begin{equation}
\label{eqn:10a}
\left\langle \phi_x^2 I_x^{+K_\alpha}\right\rangle_{\mu_\alpha}
\le
\left(
\frac{
\displaystyle
\int_{[K_\alpha,\infty)^{\Lambda_L^\alpha}}
\Big(\prod_{x \in \L_L^{\a}} d\phi_x\,\phi_x^2\Big)
e^{-N V^{\a}_{L,\lambda}(\phi)}
}{
\displaystyle
\int_{\mathbb R^{\Lambda_L^\alpha}}
\Big(\prod_{x \in \L_L^{\a}} d\phi_x\Big)
e^{-N V^{\a}_{L,\lambda}(\phi)}
}
\right)^{\frac{1}{|\Lambda_L^\alpha|}}.
\end{equation}

We will proceed by looking for an upper and a lower bound for the numerator and denominator of \eqref{eqn:10} respectively.

\paragraph{Upper bound for the numerator of \eqref{eqn:10a}.} We proceed via the following intermediate steps.

\begin{enumerate}
\item We first rewrite the effective potential as:

\[
V^{\alpha}_{L,\lambda}(\phi)
=
\frac{1}{2\lambda N}\sum_{x\in\Lambda_L^\alpha}\phi_x^2
+
V^{\alpha}_{L,\tilde\lambda_N}(\phi),
\qquad
\tilde\lambda_N \doteq \tfrac{\lambda}{1-\frac1N}.
\]

\item As a second step, we extract the infimum of $V^{\alpha}_{L,\tilde\lambda_N}$ over the constrained set $\{\phi_x \geq K_{\a}\, \forall x\in \L^{\a}_L\}$ and we integrate over $\phi$:

\begin{equation*}
\begin{aligned}
&\bigg(\int_{[K_{\alpha},\infty)^{\Lambda^{\alpha}_L}} \big(\prod_{x\in \Lambda_L^{\alpha}}d\phi_x \phi_x^2 \big) \, e^{- N V^{\alpha}_{L,\lambda}(\phi)}\bigg)^{\frac{1}{|\Lambda_L^{\a}|}}  =\\
&\bigg(\int_{[K_{\alpha},\infty)^{\Lambda^{\alpha}_L}} \big(\prod_{x\in \Lambda^{\alpha}_L}d\phi_x \phi_x^2 e^{- \frac{\phi_x^2}{2\lambda}}\big) \, e^{- N V^{\alpha}_{L,\tilde{\lambda}_N}(\phi)} \bigg)^{\frac{1}{|\Lambda_L^{\a}|}}  \leq (\sqrt{2\pi} \lambda^{\frac{3}{2}})\, \exp \Big\{-\tfrac{N}{|\L^{\a}_L|} \inf_{ \phi\in[K_{\alpha},\infty)^{\L^{\alpha}_L}} V^{\alpha}_{L,\tilde\lambda_N}(\phi) \Big\}.
\end{aligned}
\end{equation*}

\item  By Proposition \ref{const_field}.\ref{itt:1}, the infimum of $V^{\alpha}_{L, \lambda}$ is attained for constant field configurations, whose large $L$ asymptotics can be controlled, in force of Proposition \ref{const_field}.\ref{itt:2}, by the infinite-volume potential density $v^{\alpha}_{\l}$:

\begin{equation}
\label{eqn:13s}
\frac{1}{|\L^{\a}_L|}\inf_{\phi\in[K_{\alpha},\infty)^{\L^{\alpha}_L}} V^{\alpha}_{L,\tilde\lambda_N}(\phi) \overset{\eqref{Minpot}}{\geq} \frac{1}{|\L^{\a}_L|}\inf_{\vphi\in[K_{\alpha},\infty)} V^{\alpha}_{L,\tilde\lambda_N}(\vphi) \overset{\eqref{eqn:12a}}{\geq} \inf_{\vphi\in[K_{\alpha},\infty)} v_{\tilde{\l}_N}^{\alpha}(\vphi) -\frac{\mathfrak{c}_0}{L}. 
\end{equation}
\item Since $K_\alpha>\varphi_\alpha^\star(\lambda)$ and the map $\lambda \mapsto  \varphi^{\star}_{\alpha}(\lambda)$ is continuous over $(0,\infty)$ (see Proposition \ref{const_field}.\ref{itt:3}), for $N$ sufficiently large one has $K_\alpha\ge \varphi_\alpha^\star(\tilde\lambda_N)$. Hence, since $v^{\a}_{\tilde{\l}_N}$ is increasing over $[\varphi_{\a}^{\star}(\tilde{\l}_N), \infty)$ (see Proposition \ref{const_field}.\ref{itt:3}), the infimum over $[K_\alpha,\infty)$ is attained at $\varphi=K_\alpha$ and \eqref{eqn:13s} yields:

\begin{equation*}
\frac{1}{|\L^{\a}_L|}\inf_{\phi\in[K_{\alpha},\infty)^{\L^{\alpha}_L}} V^{\alpha}_{L,\tilde\lambda_N}(\phi) \geq v_{\tilde{\l}_N}^{\a}(K_{\alpha}) - \tfrac{\mf{c}_0}{L}.
\end{equation*}
\end{enumerate}

All in all:

\begin{equation}
\label{eqn:13}
\bigg(\int_{[K_{\alpha},\infty)^{\Lambda^{\alpha}_L}} \big(\prod_{x\in \Lambda_L^{\alpha}}d\phi_x \phi_x^2 \big) \, e^{- N V^{\alpha}_{L,\lambda}(\phi)}\bigg)^{\frac{1}{|\Lambda_L^{\a}|}} \le   \left(\sqrt{2\pi} \l^{\frac{3}{2}}\right) \, \exp\Big\{-N \big(v^{\a}_{\tilde{\l}_N}(K_{\a}) - \tfrac{\mf{c}_0}{L}\big)\Big\}.
\end{equation}

\paragraph{Upper bound for the denominator of \eqref{eqn:10a}.} The strategy for bounding from below the denominator of \eqref{eqn:10a} (a.k.a. the partition function) goes through the following steps.
\begin{enumerate}
    \item First, we  restrict the integral over $\phi$ to a neighborhood $Q^{\a}_{\d_N}$  of the homogeneous minimizer $\varphi^{\star}_{\a}(\l)>0$, where
    
    \[Q^{\a}_{\d_N} \doteq \big[\varphi^{\star}_{\alpha}(\lambda)-\tfrac{\d_N}{2},\varphi^{\star}_{\alpha}(\lambda)+\tfrac{\d_N}{2}\big]^{\Lambda^{\a}_L}\]

   with $\d_N>0$ to be specified. Since the integrand is pointwise non-negative for $N$ even, such restriction produces a lower-bound: 

\begin{equation}
\label{eqn:42}
\big(Z^{\alpha}_{\L_L}\big)^{\frac{1}{|\L^{\a}_{L}|}} \geq \Big[ \int_{Q^{\a}_{\d_N}} \prod_{x \in \Lambda^{\alpha}_L} d \phi_x\, \ e^{-\frac{N}{2 \lambda} \phi_x^2}\, \det\big((\slashed{D}_{\alpha})^-_{\Lambda_L}-M_{\phi}\big)^N\Big]^{\frac{1}{|\L^{\a}_{L}|}}.
\end{equation}

    \item In the second place, we perform a change of variables $\phi_x\mapsto \varphi_\alpha^\star(\lambda)+\phi_x$ and we use the following bound (holding for $\phi\in Q^{\a,0}_{\d_N}\doteq [-\frac{\d_N}{2},\frac{\d_N}{2} ]^{\L^{\a}_L}$):

\begin{equation*}
e^{-\frac{N}{2 \lambda} \sum_{x} (\varphi^{\star}_{\a}(\l)+ \phi_x)^2} = e^{-\frac{N}{2 \lambda} \big(\varphi^{\star}_{\a}(\l)^2 |\Lambda^{\alpha}_L| + 2 \varphi^{\star}_{\a}(\l)\sum_{x} \phi_x + \sum_x \phi_x^2 \big)} \geq e^{-\frac{N}{2 \lambda}|\Lambda^{\alpha}_L| \big(\varphi^{\star}_{\a}(\l)^2+ 2 \varphi^{\star}_{\a}(\l) \d_N + \d_N^2\big)}
\end{equation*}

to obtain:
\begin{equation}
\label{eqn:44}
\begin{split}
&\big(Z^{\alpha}_{\L_L}\big)^{\frac{1}{|\L^{\a}_{L}|}}
\ge\\
&
e^{-\frac{N}{|\L^{\a}_L|} V^{\a}_{L,\lambda}(\varphi_\alpha^\star(\lambda))}
e^{-\frac{N}{2\lambda}\big(2\varphi_\alpha^\star(\lambda)\delta_N+\delta_N^2\big)} \bigg[
\int_{Q_{\d_N}^{\a,0}}
\prod_{x \in \L_{L}^{\a}} d\phi_x\,
\det\Big(
\mathbbl{1}-
\big((\slashed D_\alpha)^-_{\Lambda_L}-M_{\varphi_\alpha^\star(\lambda)}\big)^{-1} M_\phi
\Big)^N\bigg]^{\frac{1}{|\L^{\a}_{L}|}}.
\end{split}\end{equation}

\begin{remark}
\label{rmk:inverse_Dirac}
Observe that the Dirac operator is invertible since a non-vanishing, constant background field opens a gap in the dispersion relations (see Appendix \ref{app:pot}). More precisely, introducing the operator norm $\|\,\cdot\,\|_{\mathrm{op}}\doteq \sup_{v\in\mbb{C}^n,\|v\|=1}\|\,\cdot\,v\|_{\ell^2(X)}$, with $X=\L_L\times\{1,2\}, \L_L, \wt{\L}_L\times\{A,B,C,D\}$ for $\a=\mathrm{NN},\mathrm{SN},\mathrm{SP}$ respectively, one has that 

\[\left\| \big((\slashed D_\alpha)^-_{\Lambda_L}- M_{\varphi^{\L^{\a}_L}}\big)^{-1} \right\|_{\mathrm{op}}\le \frac{1}{\vphi},\] 

for any $\vphi\ne0$. Indeed, it is readily checked (see Appendix \ref{app:pot_2}) that the operator $(\slashed D_\alpha)^-_{\Lambda_L}- M_{\varphi^{\L^{\a}_L}}\equiv (\slashed D_\alpha)^-_{\Lambda_L}- \vphi\mathbbl{1}$ is normal and has eigenvalues $\Big\{\pm i \sqrt{\sin^2p_0+ \sin^2p_1} -\vphi \Big\}_{p\in(\L^{\a}_L)^*_{--}}$, where

\begin{equation*}
(\Lambda^{\mathrm{NN}}_L)^*_{--} \doteq \tfrac{2\pi}{L}\big(\mbb{Z}+ \tfrac{1}{2}\big)^2 \cap (-\pi,\pi ]^2, \qquad (\Lambda^{\mathrm{SN}}_L)^*_{--}=(\Lambda^{\mathrm{SP}}_L)^*_{--}\doteq \tfrac{2\pi}{L} \big(\mbb{Z}+\tfrac{1}{2}\big)^2 \cap \Big(\big(-\tfrac{\pi}{2},\tfrac{\pi}{2}\big]\times(-\pi,\pi] \Big).
\end{equation*}

\end{remark}

\item In order to find a lower bound for the determinant in the right--hand side of \eqref{eqn:44}, we use the following standard fact: letting $A$ be a linear operator on a vector space $\mc{V}$ of dimension $n<\infty$, then

\begin{equation}
\label{eqn:43}
|\det(\mathbbl{1}_{\mc{V}} + A)|\ge (1-\|A\|_{\mathrm{op}})^n.
\end{equation}

In our case, setting $A= \big((\slashed D_\alpha)^-_{\Lambda_L}-\varphi_\alpha^\star(\lambda)\big)^{-1}M_{\phi}$, with $\phi\in Q_{\d_N}^{\a,0}$, we see that, in force of Remark \ref{rmk:inverse_Dirac},

\[\left\|
\big((\slashed D_\alpha)^-_{\Lambda_L}-\varphi_\alpha^\star(\lambda)\big)^{-1}M_\phi
\right\|_{\mathrm{op}}
\le \left\| \big((\slashed D_\alpha)^-_{\Lambda_L}-\varphi_\alpha^\star(\lambda)\big)^{-1} \right\|_{\mathrm{op}} \left\| M_\phi\right\|_{\mathrm{op}} \leq
\tfrac{\delta_N}{\varphi_\alpha^\star(\lambda)}. \]

Therefore, for any choice $0<\d_N<\vphi^{\star}_{\a}(\l)$, we have that

\begin{equation*}
\big|\det\Big(
\mathbbl{1}-
\big((\slashed D_\alpha)^-_{\Lambda_L}-M_{\varphi_\alpha^\star(\lambda)}\big)^{-1} M_\phi
\Big)\big|\ge \Big( 1- \tfrac{\d_N}{\vphi^{\star}_{\a}(\l)} \Big)^{4|\L_L^{\a}|}.
\end{equation*}

Plugging the above bound in the right--hand side of \eqref{eqn:44} and performing the integral over $Q^{\a,0}_{\d_N}$, we find the following lower bound for the partition function:

\begin{equation}
\label{eqn:16}
(Z^{\alpha}_{\L_L})^{\frac{1}{|\Lambda_L^{\a}|}}
\ge e^{N\big[- v^{\alpha}_{\lambda}(\varphi^{\star}_{\alpha}(\lambda))- \tfrac{\d_N}{\l}(2\varphi^{\star}_{\alpha}(\lambda)+\d_N)+ \tfrac{1}{N}\log (\d_N)+ 4\log(1- \tfrac{\d_N}{\varphi^{\star}_{\alpha}(\lambda)} ) -\tfrac{\mf{c}_0}{L}\big]}.
\end{equation}

\end{enumerate}

\paragraph{Proof of the bound \eqref{eqn:18}.}

By plugging \eqref{eqn:13} and \eqref{eqn:16} into the right--hand side of \eqref{eqn:10a}, we find:

\begin{equation}
   \label{eqn:17}
   \left\langle\phi_x^2 I_x^{+K_{\alpha}} \right\rangle_{\mu_{\a}} \le e^{-N \big[v^{\a}_{\l}(K_{\alpha})-v^{\a}_{\lambda}(\vphi^{\star}_{\a}(\lambda)) - \mathcal{R}^{(1)}_{N,L}(\lambda)\big]} \doteq e^{-N \big[ c_1^{\a}(\l) - \mathcal{R}^{(1)}_{N,L}(\l)\big]},
\end{equation}

where

\begin{equation*}
\begin{split}
\mc{R}^{(1)}_{N,L}(\l)\doteq &\big( v^{\alpha}_{\lambda}(K_{\alpha})- v^{\a}_{\tilde{\l}_N}(K_{\alpha})\big)+ \tfrac{\d_N}{\l}\big(2\varphi^{\star}_{\alpha}(\lambda)+\d_N\big)- \tfrac{1}{N}\log (\d_N)+\\
&-4 \log\big(1- \tfrac{\d_N}{\varphi^{\star}_{\alpha}(\lambda)} \big)+ \tfrac{1}{2N}\big(\log(2\pi)+ 3\log\l\big) + \tfrac{2\mf{c}_0}{{L}}.
\end{split}\end{equation*}

The main observation is that

\[v^{\alpha}_{\lambda}(K_{\alpha}) -v^{\alpha}_{\lambda}(\varphi^{\star}_{\alpha}(\lambda)) \doteq c_1^{\a}(\l)>0,\]

for every choice of $K_{\alpha}>\varphi^{\star}_{\alpha}(\lambda)$, since the potential $v^{\a}_{\l}$ is strictly-increasing in the interval $[\varphi^{\star}_{\a}(\l), \infty)$. Furthermore, due to the smoothness of the map $\l\mapsto\vphi^{\star}_{\a}(\l)$, for every $\vphi>0$ fixed, we have that $v^{\alpha}_{\lambda}(K_{\alpha})- v^{\a}_{\tilde{\l}_N}(K_{\alpha})\to0$ as $N\to\infty$. In this way, by also fixing $\d_N\equiv \d_N(\l)\doteq \frac{\varphi^{\star}_{\alpha}(\lambda)}{N}$, we readily see that 

\[ \mathcal{R}^{(1)}_{N,L}(\lambda) \overset{N,L \to \infty}{\longrightarrow} 0,\]

for every $\l>0$ and $K_{\alpha}>\varphi^{\star}_{\alpha}(\lambda)$ fixed. Therefore, for $N,L$ sufficiently large (depending on $\l$ and $K_{\alpha}$), we can assume that $|\mc{R}_{N,L}^{(1)}(\l)| \leq \frac{1}{2}c_1^{\a}(\l)$, yielding the bound \eqref{eqn:18}.

\subsection{Local-maximum terms}
\label{ssect:loc_max}

In this subsection we prove the second inequality \eqref{eqn:18b}, namely
\[
\sup_{x \in \Lambda^{\alpha}_L}
\left \langle I_x^{(-C_{\alpha}, C_{\alpha})} \right \rangle_{\mu_{\alpha}}
\leq e^{- \frac{N}{2} c_2^{\alpha}(\lambda)}.
\]
This estimate controls configurations in which the field lies in a neighborhood of a local maximum of the potential. Such configurations are energetically unfavorable and hence exponentially suppressed in $N$. The proof relies on Chessboard Estimates, combined with the fact that, as established by Proposition \ref{const_field}, the potential $v^{\alpha}_{\lambda}$ is strictly larger than its minimum whenever $\phi_x \in (-C_{\alpha},C_{\alpha})\subsetneq (-\vphi^{\star}_{\a}, \vphi^{\star}_{\a})$. Applying the Chessboard Estimate \eqref{CE} to the single-site observable, we obtain:

\begin{equation}
\label{eqn:10}
\left \langle I_x^{(-C_{\alpha}, C_{\alpha})}\right \rangle_{\mu_{\alpha}} \leq \left(\frac{\displaystyle \int_{(-C_{\alpha},C_{\alpha})^{\Lambda^{\alpha}_L}} \big(\prod_{x \in \Lambda_L^{\alpha}}d\phi_x \big) e^{- N V^{\alpha}_{L,\lambda}(\phi)}}{ \displaystyle\int_{\mbb{R}^{\L^{\alpha}_L}} \big(\prod_{x \in \Lambda_L^{\alpha}}d\phi_x \big) e^{- N V^{\alpha}_{L,\lambda}(\phi)}}\right)^{\frac{1}{|\Lambda^{\alpha}_L|}}.
\end{equation}

\paragraph{Upper bound for the numerator  of \eqref{eqn:10}.}

The bound of the numerator is a straightforward application of Proposition \ref{const_field}.

\begin{enumerate}
    \item First, we extract the infimum of $V^{\alpha}_{L,\l}$ over the constrained set $-C_{\a} \leq \phi_x \leq C_{\a}$ and perform the integral over $\phi$:

\begin{equation*}
\begin{aligned}
\bigg(\int_{(-C_{\alpha},C_{\a})^{\Lambda^{\alpha}_L}} \big(\prod_{x\in \Lambda_L^{\alpha}}d\phi_x \big) \, e^{- N V^{\alpha}_{L,\lambda}(\phi)}\bigg)^{\frac{1}{|\Lambda_L^{\a}|}}  &\leq (2C_{\a})\, \exp\Big\{-\tfrac{N}{|\L^{\a}_L|} \inf_{ \phi\in (-C_{\alpha},C_{\a})^{\L^{\alpha}_L}} V^{\alpha}_{L,\lambda}(\phi) \Big\}.
\end{aligned}
\end{equation*}

\item By \eqref{Minpot}, the infimum of $V^{\alpha}_{L, \lambda}$ is attained over constant field configurations, whose large-$L$ asymptotics can be controlled, due to Proposition \ref{const_field}.\ref{itt:2}, by the infinite-volume potential density $v^{\alpha}_{\l}$:

\begin{equation}
\label{eqn:13w}
\frac{1}{|\L^{\a}_L|} \inf_{\phi\in(-C_{\alpha},C_{\a})^{\L^{\alpha}_L}} V^{\alpha}_{L,\lambda}(\phi) \overset{\eqref{Minpot}}{\geq} \frac{1}{|\L^{\a}_L|} \inf_{\vphi\in(-C_{\alpha},C_{\a})} V^{\alpha}_{L,\lambda}(\vphi) \overset{\eqref{eqn:12a}}{\geq} \inf_{\vphi\in(-C_{\alpha},C_{\a})} v_{{\l}}^{\alpha}(\vphi) -\frac{\mathfrak{c}_0}{L}. 
\end{equation}

\item  Since $v^{\a}_{\l}$ is strictly decreasing on $[0, \vphi^{\star}_{\a}(\l)]$ 
(see Proposition \ref{const_field}.\ref{itt:3}) and that $C_{\a} < \vphi^{\star}_{\a}(\l)$, the right--hand side of \eqref{eqn:13w} is larger than $v^{\a}_{\l}(C_{\a}) - \frac{\mf{c}_0}{L}$.
\end{enumerate}

Therefore, we find the following upper bound for the numerator of \eqref{eqn:10}:

\begin{equation}
\label{eqn:16u}
\begin{aligned}
\bigg(\int_{(-C_{\alpha},C_{\a})^{\Lambda^{\alpha}_L}} \big(\prod_{x\in \Lambda_L^{\alpha}}d\phi_x \big) \, e^{- N V^{\alpha}_{L,\lambda}(\phi)}\bigg)^{\frac{1}{|\Lambda_L^{\a}|}} \leq  (2C_{\alpha}) e^{- N  \big[  v^{\alpha}_{\lambda}(C_{\alpha})- \tfrac{\mf{c}_0}{{L}}\big]}.
\end{aligned}
\end{equation}

\paragraph{Proof of the bound \eqref{eqn:18b}.}

Plugging equations \eqref{eqn:13} and \eqref{eqn:16u} into the right--hand side of \eqref{eqn:10}, we find:
\[
\begin{split}
\left\langle I_x^{(-C_{\alpha},C_{\alpha})} \right\rangle_{\mu_{\alpha}}\le e^{-N\big(  v^{\alpha}_{\lambda}(C_{\alpha}) - v^{\alpha}_{\lambda}(\varphi^{\star}_{\alpha}(\lambda))- \mc{R}^{(2)}_{N,L}(\l) \big)},
\end{split}
\]
where

\[\mc{R}^{(2)}_{N,L}(\l)\doteq \tfrac{\d_N}{\l}\big(2\varphi^{\star}_{\alpha}(\lambda)+\d_N\big)+ \tfrac{1}{N}\big(-\log (\d_N)+ \log(2C_{\alpha}) \big)- 4\log\big(1- \tfrac{\d_N}{\varphi^{\star}_{\alpha}(\lambda)} \big) + \tfrac{2\mf{c}_0}{L}. \]

We observe that

\[v^{\alpha}_{\lambda}(C_{\alpha}) -v^{\alpha}_{\lambda}(\varphi^{\star}_{\alpha}(\lambda)) \equiv c_2^{\a}(\l)>0,\]

for every choice of $C_{\alpha}<\varphi^{\star}_{\alpha}(\lambda)$, since the potential $v^{\a}_{\l}$ is strictly decreasing in the interval $[0, \varphi^{\star}_{\a}(\l))$ (see Proposition \ref{const_field}.\ref{itt:3}). By fixing $\d_N=\frac{\varphi^{\star}_{\alpha}(\lambda)}{N}$, we see that for every $\l>0$ and $0<C_{\alpha}<\varphi^{\star}_{\alpha}(\lambda)$,

\[\mc{R}_{N,L}^{(2)}(\l)\overset{N,L\to\infty}{\longrightarrow} 0.\]

Therefore, for $N,L$ sufficiently large (depending on $\l$ and $C_{\alpha}$), we can assume that $|\mc{R}_{N,L}^{(2)}(\l)| \leq \frac{1}{2}c_2^{\a}(\l)$, yielding the bound \eqref{eqn:18b}.

\subsection{Peierls-like terms}
\label{ssect:Peierls}

In this subsection we prove the bound \eqref{eqn:18c}, namely
\[
\sup_{x,y \in \Lambda^{\alpha}_L}
\left \langle I_x^{\star_+} I_y^{\star_-} \right \rangle_{\mu_{\alpha}}
\leq e^{-\frac{N}{4} c_3^{\alpha}(\lambda)}. 
\]

This estimate controls configurations in which the field takes values near opposite minimizers of the potential at two distinct sites. Such configurations necessarily contain an interface separating regions of opposite sign, and are therefore expected to be exponentially suppressed in $N$. \\

In contrast to the previous cases, this contribution cannot be handled directly by means of the Chessboard Estimates alone. Its analysis
requires a more detailed control of the energy cost associated with configurations interpolating between the two minima. This is achieved via a Peierls-type argument (see e.g. \cite[Sect.~3.7.3]{Friedli_Velenik_2017}), combined with Reflection
Positivity, which allows us to control the weight of contours separating regions with opposite sign of $\phi$ by relating them to the energetic penalty of suitably constructed periodic configurations generated through reflections.

\paragraph{General Peierls' argument.}

\begin{enumerate}
\item For any $x, y \in \Lambda_L^{\alpha}$, we have that

\begin{equation*}
    \left \langle I_x^{\star_+} I_y^{\star_-} \right \rangle_{\mu_{\alpha}} \leq \left \langle I_x^+ I_y^-\right \rangle_{\mu_{\alpha}},
\end{equation*}

where $I_x^{\pm}(\phi) \doteq \mathbbl{1}_{\{\pm\phi_x > 0\}}$.

\item Since $I_x^+ + I_x^- = 1$ almost surely, we have that

\begin{equation}
\label{eqn:19}
\left \langle I_x^+ I_y^-\right \rangle_{\mu_{\alpha}} = \Big\langle I_x^+ I_y^- \prod_{z \neq x,y} (I_z^+ + I_z^-)\Big\rangle_{\mu_{\alpha}} = \sum_{\substack{\zeta \in \{\pm \}^{\Lambda_L^{\alpha}}:\\ \zeta_x=+, \zeta_y=-}} \Big\langle \prod_{z\in \Lambda_L^{\alpha}} I_z^{\zeta_z} \Big\rangle_{\mu_{\alpha}}.
\end{equation}
\item To each sign configuration $\zeta\in \{\pm\}^{\L^{\a}_L}$, we associate its set of Peierls' contours, namely the closed curves in the dual lattice $(\Lambda_L^\alpha)^{\mathrm{dual}}\doteq \L^{\alpha}_L+ (\frac{1}{2},\frac{1}{2})$ crossing precisely those edges $\langle x',y'\rangle\subset \L^{\a}_L$ for which $\zeta_{x'}\neq \zeta_{y'}$. As usual, crossings are resolved according to the deformation rule (see \cite[pag. 111]{Friedli_Velenik_2017}) shown in Fig.~\ref{fig:deformation}, so that the contour set is a collection of disjoint loops including the information of the values of the
spins on both “sides”.\\
\item A standard Peierls argument (see \cite[pag. 453-457]{Friedli_Velenik_2017}) then gives:
\begin{equation}
\label{eqn:21}
\left\langle I_x^{\star_+} I_y^{\star_-}\right\rangle_{\mu_\alpha}
\le
2\sum_{\gamma}
\Big\langle
\prod_{\langle x',y'\rangle\in \mathrm{E}(\gamma)} I_{x'}^+ I_{y'}^-
\Big\rangle_{\mu_\alpha},
\end{equation}
where the sum runs over all dual contours $\gamma$ separating $x$ and $y$, and $\mathrm{E}(\gamma)$ denotes the set of crossed primal edges, ordered so that $x'$ lies on the \virg{left hand side} of $\gamma$ and $y'$ on the \virg{right hand side} (with respect to some fixed orientation of the contour $\gamma$).

\vspace{0.1in}

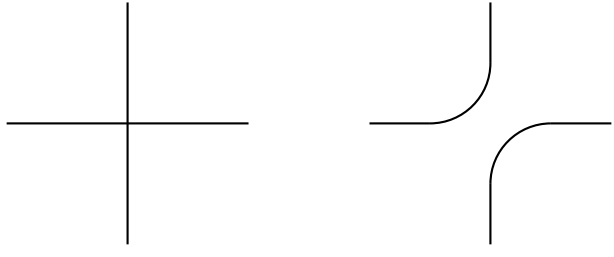
\begin{figure}
    \centering
    \begin{tikzpicture}[scale=0.8]
        \draw[thick] (0,0) -- (0,4);
        \draw[thick] (-2,2) -- (2,2);

        \draw[thick] (6,0) -- (6,1);
        \draw[thick] (6,3) -- (6,4);
        \draw[thick] (4,2) -- (5,2);
        \draw[thick] (7,2) -- (8,2);
       \draw[thick] (6,1) arc (180:90:1cm);
       \draw[thick] (5,2) arc (270:360:1cm);
    \end{tikzpicture}
    \caption{deformation rule adopted for associating a set of Peierls' contours to every sign configuration $\z\in \{\pm\}^{\L^{\a}_L}$.}
    \label{fig:deformation}
\end{figure}

    \end{enumerate}
\paragraph{Reduction to contour estimates.} The main bound for controlling the Peirls'-like contributions is provided by the following lemma.

\begin{lemma}\label{lem:Peierls_bound}
For every contour $\gamma$, the following bound holds true:
    \begin{equation}
\label{eq:Peierls_bound}
\Big\langle
\prod_{\langle x',y'\rangle\in \mathrm{E}(\gamma)} I_{x'}^+ I_{y'}^-
\Big\rangle_{\mu_\alpha}
\le
e^{-\frac14 N c_3^\alpha(\lambda)\,|\gamma|},
\end{equation}

with $c_3^\a(\l)>0$ defined in Proposition \ref{prop:contrib}.
\end{lemma}

In this way the estimate can be reduced to a simple contour-counting argument. It is readily checked that the sum over contours $\gamma$ as in the right--hand side of \eqref{eqn:21} can be bounded as 

\begin{equation*}
\sum_{\gamma \text{ around }x} + \sum_{\gamma \text{ around }y}+ \sum_{\gamma:\; |\gamma|\ge L},
\end{equation*}

where \virg{$\gamma$ around $x$} means that the contour $\gamma$ is contained in a square patch of the torus $\L^{\a}_L$ with side $L$ and centered around $x$. Therefore, assuming the bound \eqref{eq:Peierls_bound}, we readily obtain that
\begin{equation}
\label{eq:Peierls-sum}
\left\langle I_x^{\star_+} I_y^{\star_-}\right\rangle_{\mu_\alpha}
\le
2 \sum_{\gamma \text{ around } x}
e^{-\frac{1}{4} N c_3^\alpha(\lambda)|\gamma|} + 2 \sum_{\gamma \text{ around } y}
e^{-\frac{1}{4} N c_3^\alpha(\lambda)|\gamma|}
+
2 \sum_{\gamma:\ |\gamma|\ge L}
e^{-\frac{1}{4} N c_3^\alpha(\lambda)|\gamma|}.
\end{equation}

Using the standard contour counting bounds:
\[
\#\{\gamma\text{ around }x:\ |\gamma|=k\}\le \frac k2\,3^k,
\qquad
\#\{\gamma:\ |\gamma|=k\}\le L^2 3^k,
\]
we easily find, for $N$ and $L$ sufficiently large:
\[
\sum_{\gamma\mathrm{ around }x}
e^{-\frac14 N c_3^\alpha(\lambda)|\gamma|}
\le \frac18 e^{-\frac14 N c_3^\alpha(\lambda)},
\qquad
\sum_{\gamma:\ |\gamma|\ge L}
e^{-\frac14 N c_3^\alpha(\lambda)|\gamma|}
\le \frac14 e^{-\frac14 N c_3^\alpha(\lambda)}.
\]

Substituting into the right--hand side \eqref{eq:Peierls-sum}, we obtain the desired bound \eqref{eqn:18c}.

\paragraph{Proof of Lemma \ref{lem:Peierls_bound}.} 

We now prove \eqref{eq:Peierls_bound} by combining Chessboard Estimates with an energetic analysis of alternating-sign periodic configurations. 

\begin{enumerate}
\item For any contour $\gamma$, its crossed edges can be partitioned into four classes:
\[
\mathrm{E}(\gamma)=
\mathrm{E}_{h,e}(\gamma)\sqcup \mathrm{E}_{h,o}(\gamma)\sqcup \mathrm{E}_{v,e}(\gamma)\sqcup \mathrm{E}_{v,o}(\gamma),
\]

(here \virg{h}/\virg{v} stands for \emph{horizontal}/\emph{vertical} and \virg{e}/\virg{o} stands for \emph{even}/\emph{odd}) corresponding to the four disjoint tessellations $(\mathrm{v},\mathrm{e})$, $(\mathrm{v},\mathrm{o})$,
$(\mathrm{h},\mathrm{e})$, $(\mathrm{h},\mathrm{o})$ of the torus made of $2\times 1$
rectangles as in Fig.~\ref{tass}. \\

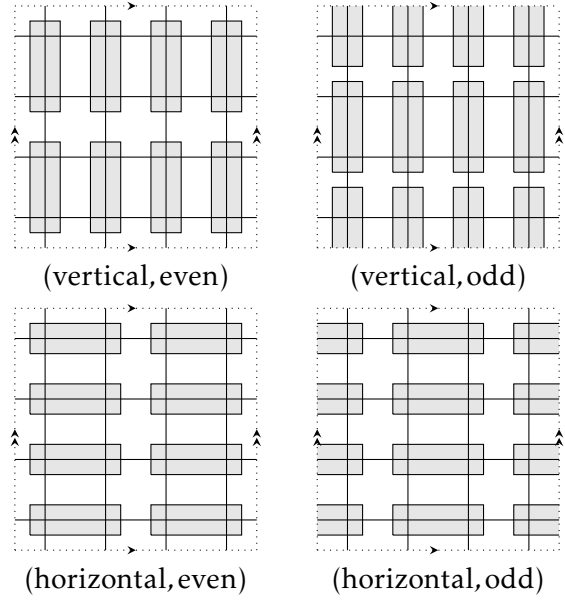
\begin{figure}
        \centering
        \begin{tikzpicture}[scale=0.8]
      %Caso Alto DX
        \begin{scope}[xshift=0cm]
\draw[dotted, ->-] (-0.5,-0.5) -- (3.5,-0.5);
\draw[dotted, ->>-] (3.5,-0.5) -- (3.5,3.5);
\draw[dotted, ->-] (-0.5,3.5) -- (3.5,3.5);
\draw[dotted, ->>-] (-0.5,-0.5) -- (-0.5,3.5);
\draw[] (0, 0) -- (3,0);
\draw[] (0, 2) -- (3,2);
\draw[] (3,1) -- (3.5,1);
\draw[] (-0.5,1) -- (0,1);
\draw[] (-0.5,3) -- (0,3);
\draw[] (3,3) -- (3.5,3);
\draw (0,3) -- (3,3);
\draw (0,1) -- (3,1);
\draw (0,0) -- (0,3);
\draw (1,0) -- (1,3);
\draw (2,0) -- (2,3);
\draw (-0.5,2) -- (0,2);
\draw (3,2) -- (3.5,2);
\draw (-0.5,0) -- (0,0);
\draw (3,0) -- (3.5,0);
\draw (3,0) -- (3,3);
\draw[] (3,3) -- (3,3.5);
\draw[] (2,3) -- (2,3.5);
\draw[] (1,3) -- (1,3.5);
\draw[] (0,3) -- (0,3.5);
\draw[] (3,0) -- (3,-0.5);
\draw[] (2,0) -- (2,-0.5);
\draw[] (1,0) -- (1,-0.5);
\draw[] (0,0) -- (0,-.5);
    \draw[fill, fill opacity =0.1] (-0.25,-0.25) -- (0.25,-0.25) -- (0.25, 1.25) -- (-0.25,1.25) -- (-0.25,-0.25);  
    \begin{scope}[xshift=1cm]
        \draw[fill, fill opacity =0.1] (-0.25,-0.25) -- (0.25,-0.25) -- (0.25, 1.25) -- (-0.25,1.25) -- (-0.25,-0.25); 
    \end{scope}
      \begin{scope}[xshift=2cm]
        \draw[fill, fill opacity =0.1] (-0.25,-0.25) -- (0.25,-0.25) -- (0.25, 1.25) -- (-0.25,1.25) -- (-0.25,-0.25); 
    \end{scope}
      \begin{scope}[xshift=3cm]
        \draw[fill, fill opacity =0.1] (-0.25,-0.25) -- (0.25,-0.25) -- (0.25, 1.25) -- (-0.25,1.25) -- (-0.25,-0.25); 
    \end{scope}
      \begin{scope}[xshift=1cm, yshift=2cm]
        \draw[fill, fill opacity =0.1] (-0.25,-0.25) -- (0.25,-0.25) -- (0.25, 1.25) -- (-0.25,1.25) -- (-0.25,-0.25); 
    \end{scope}
     \begin{scope}[xshift=2cm, yshift=2cm]
        \draw[fill, fill opacity =0.1] (-0.25,-0.25) -- (0.25,-0.25) -- (0.25, 1.25) -- (-0.25,1.25) -- (-0.25,-0.25); 
    \end{scope}
     \begin{scope}[xshift=3cm, yshift=2cm]
        \draw[fill, fill opacity =0.1] (-0.25,-0.25) -- (0.25,-0.25) -- (0.25, 1.25) -- (-0.25,1.25) -- (-0.25,-0.25); 
    \end{scope}
     \begin{scope}[xshift=0cm, yshift=2cm]
        \draw[fill, fill opacity =0.1] (-0.25,-0.25) -- (0.25,-0.25) -- (0.25, 1.25) -- (-0.25,1.25) -- (-0.25,-0.25); 
    \end{scope}
        \node[] at (1.5,-1) {\scalebox{1}{$(\mathrm{vertical, even})$}};
        \end{scope}
            %%Top Right
 \begin{scope}[xshift=5cm]
            \draw[dotted, ->-] (-0.5,-0.5) -- (3.5,-0.5);
\draw[dotted, ->>-] (3.5,-0.5) -- (3.5,3.5);
\draw[dotted, ->-] (-0.5,3.5) -- (3.5,3.5);
\draw[dotted, ->>-] (-0.5,-0.5) -- (-0.5,3.5);
\draw[] (0, 0) -- (3,0);
\draw[] (0, 2) -- (3,2);
\draw[] (3,1) -- (3.5,1);
\draw[] (-0.5,1) -- (0,1);
\draw[] (-0.5,3) -- (0,3);
\draw[] (3,3) -- (3.5,3);
\draw (0,3) -- (3,3);
\draw (0,1) -- (3,1);
\draw (0,0) -- (0,3);
\draw (1,0) -- (1,3);
\draw (2,0) -- (2,3);
\draw (-0.5,2) -- (0,2);
\draw (3,2) -- (3.5,2);
\draw (-0.5,0) -- (0,0);
\draw (3,0) -- (3.5,0);
\draw (3,0) -- (3,3);
\draw[] (3,3) -- (3,3.5);
\draw[] (2,3) -- (2,3.5);
\draw[] (1,3) -- (1,3.5);
\draw[] (0,3) -- (0,3.5);
\draw[] (3,0) -- (3,-0.5);
\draw[] (2,0) -- (2,-0.5);
\draw[] (1,0) -- (1,-0.5);
\draw[] (0,0) -- (0,-.5);
\draw[fill, fill opacity =0.1] (-0.25,0.75) -- (0.25,0.75) -- (0.25, 2.25) -- (-0.25,2.25) -- (-0.25,0.75); 
\begin{scope}[xshift=1cm]
    \draw[fill, fill opacity =0.1] (-0.25,0.75) -- (0.25,0.75) -- (0.25, 2.25) -- (-0.25,2.25) -- (-0.25,0.75); 
\end{scope}
\begin{scope}[xshift=2cm]
    \draw[fill, fill opacity =0.1] (-0.25,0.75) -- (0.25,0.75) -- (0.25, 2.25) -- (-0.25,2.25) -- (-0.25,0.75); 
\end{scope}
\begin{scope}[xshift=3cm]
    \draw[fill, fill opacity =0.1] (-0.25,0.75) -- (0.25,0.75) -- (0.25, 2.25) -- (-0.25,2.25) -- (-0.25,0.75); 
\end{scope}
\fill[fill opacity = 0.1] (-0.25, -0.5) -- (0.25, -0.5) -- (0.25, 0.5) -- (-0.25,0.5) -- (-0.25, -0.5);
 \draw[] (0.25, -0.5) -- (0.25, 0.5) -- (-0.25,0.5) -- (-0.25, -0.5);
 \begin{scope}[xshift=1cm]
   \fill[fill opacity = 0.1]   (-0.25, -0.5) -- (0.25, -0.5) -- (0.25, 0.5) -- (-0.25,0.5) -- (-0.25, -0.5);
 \draw[] (0.25, -0.5) -- (0.25, 0.5) -- (-0.25,0.5) -- (-0.25, -0.5);
 \end{scope}
  \begin{scope}[xshift=2cm]
  \fill[fill opacity = 0.1]    (-0.25, -0.5) -- (0.25, -0.5) -- (0.25, 0.5) -- (-0.25,0.5) -- (-0.25, -0.5);
 \draw[] (0.25, -0.5) -- (0.25, 0.5) -- (-0.25,0.5) -- (-0.25, -0.5);
 \end{scope}
  \begin{scope}[xshift=3cm]
  \fill[fill opacity = 0.1]    (-0.25, -0.5) -- (0.25, -0.5) -- (0.25, 0.5) -- (-0.25,0.5) -- (-0.25, -0.5);
 \draw[] (0.25, -0.5) -- (0.25, 0.5) -- (-0.25,0.5) -- (-0.25, -0.5);
 \end{scope}
 \fill[fill opacity = 0.1] (-0.25,2.5) -- (0.25,2.5) -- (0.25,3.5) -- (-0.25,3.5) -- (-0.25,2.5);
 \draw[] (-0.25, 3.5) -- (-0.25, 2.5) -- (0.25,2.5) -- (0.25, 3.5);
 \begin{scope}[xshift=1cm]
     \fill[fill opacity = 0.1] (-0.25,2.5) -- (0.25,2.5) -- (0.25,3.5) -- (-0.25,3.5) -- (-0.25,2.5);
 \draw[] (-0.25, 3.5) -- (-0.25, 2.5) -- (0.25,2.5) -- (0.25, 3.5);
 \end{scope}
 \begin{scope}[xshift=2cm]
     \fill[fill opacity = 0.1] (-0.25,2.5) -- (0.25,2.5) -- (0.25,3.5) -- (-0.25,3.5) -- (-0.25,2.5);
 \draw[] (-0.25, 3.5) -- (-0.25, 2.5) -- (0.25,2.5) -- (0.25, 3.5);
 \end{scope}
 \begin{scope}[xshift=3cm]
     \fill[fill opacity = 0.1] (-0.25,2.5) -- (0.25,2.5) -- (0.25,3.5) -- (-0.25,3.5) -- (-0.25,2.5);
 \draw[] (-0.25, 3.5) -- (-0.25, 2.5) -- (0.25,2.5) -- (0.25, 3.5);
 \end{scope}
 \node[] at (1.5,-1) {\scalebox{1}{$(\mathrm{vertical}, \mathrm{odd})$}};
 \end{scope}

 %%Down Left
 \begin{scope}[yshift=-5cm]
           \draw[dotted, ->-] (-0.5,-0.5) -- (3.5,-0.5);
\draw[dotted, ->>-] (3.5,-0.5) -- (3.5,3.5);
\draw[dotted, ->-] (-0.5,3.5) -- (3.5,3.5);
\draw[dotted, ->>-] (-0.5,-0.5) -- (-0.5,3.5);
\draw[] (0, 0) -- (3,0);
\draw[] (0, 2) -- (3,2);
\draw[] (3,1) -- (3.5,1);
\draw[] (-0.5,1) -- (0,1);
\draw[] (-0.5,3) -- (0,3);
\draw[] (3,3) -- (3.5,3);
\draw (0,3) -- (3,3);
\draw (0,1) -- (3,1);
\draw (0,0) -- (0,3);
\draw (1,0) -- (1,3);
\draw (2,0) -- (2,3);
\draw (-0.5,2) -- (0,2);
\draw (3,2) -- (3.5,2);
\draw (-0.5,0) -- (0,0);
\draw (3,0) -- (3.5,0);
\draw (3,0) -- (3,3);
\draw[] (3,3) -- (3,3.5);
\draw[] (2,3) -- (2,3.5);
\draw[] (1,3) -- (1,3.5);
\draw[] (0,3) -- (0,3.5);
\draw[] (3,0) -- (3,-0.5);
\draw[] (2,0) -- (2,-0.5);
\draw[] (1,0) -- (1,-0.5);
\draw[] (0,0) -- (0,-.5);
\draw[fill, fill opacity = 0.1] (-0.25, -0.25) -- (1.25, -0.25) -- (1.25, 0.25) -- (-0.25,0.25) -- (-0.25, -0.25);
\begin{scope}[xshift=2cm]
    \draw[fill, fill opacity = 0.1] (-0.25, -0.25) -- (1.25, -0.25) -- (1.25, 0.25) -- (-0.25,0.25) -- (-0.25, -0.25);
\end{scope}
\begin{scope}[xshift=2cm, yshift= 1cm]
    \draw[fill, fill opacity = 0.1] (-0.25, -0.25) -- (1.25, -0.25) -- (1.25, 0.25) -- (-0.25,0.25) -- (-0.25, -0.25);
\end{scope}
\begin{scope}[xshift=2cm, yshift=3cm]
    \draw[fill, fill opacity = 0.1] (-0.25, -0.25) -- (1.25, -0.25) -- (1.25, 0.25) -- (-0.25,0.25) -- (-0.25, -0.25);
\end{scope}
\begin{scope}[xshift=2cm, yshift=2cm]
    \draw[fill, fill opacity = 0.1] (-0.25, -0.25) -- (1.25, -0.25) -- (1.25, 0.25) -- (-0.25,0.25) -- (-0.25, -0.25);
\end{scope}
\begin{scope}[xshift=0cm, yshift= 1cm]
    \draw[fill, fill opacity = 0.1] (-0.25, -0.25) -- (1.25, -0.25) -- (1.25, 0.25) -- (-0.25,0.25) -- (-0.25, -0.25);
\end{scope}
\begin{scope}[xshift=0cm, yshift=3cm]
    \draw[fill, fill opacity = 0.1] (-0.25, -0.25) -- (1.25, -0.25) -- (1.25, 0.25) -- (-0.25,0.25) -- (-0.25, -0.25);
\end{scope}
\begin{scope}[xshift=0cm, yshift=2cm]
    \draw[fill, fill opacity = 0.1] (-0.25, -0.25) -- (1.25, -0.25) -- (1.25, 0.25) -- (-0.25,0.25) -- (-0.25, -0.25);
\end{scope}
 \node[] at (1.5,-1) {\scalebox{1}{$(\mathrm{horizontal, even})$}};
 \end{scope}

 %%Down Left
 \begin{scope}[xshift=5cm, yshift= -5cm]
     \draw[dotted, ->-] (-0.5,-0.5) -- (3.5,-0.5);
\draw[dotted, ->>-] (3.5,-0.5) -- (3.5,3.5);
\draw[dotted, ->-] (-0.5,3.5) -- (3.5,3.5);
\draw[dotted, ->>-] (-0.5,-0.5) -- (-0.5,3.5);
\draw[] (0, 0) -- (3,0);
\draw[] (0, 2) -- (3,2);
\draw[] (3,1) -- (3.5,1);
\draw[] (-0.5,1) -- (0,1);
\draw[] (-0.5,3) -- (0,3);
\draw[] (3,3) -- (3.5,3);
\draw (0,3) -- (3,3);
\draw (0,1) -- (3,1);
\draw (0,0) -- (0,3);
\draw (1,0) -- (1,3);
\draw (2,0) -- (2,3);
\draw (-0.5,2) -- (0,2);
\draw (3,2) -- (3.5,2);
\draw (-0.5,0) -- (0,0);
\draw (3,0) -- (3.5,0);
\draw (3,0) -- (3,3);
\draw[] (3,3) -- (3,3.5);
\draw[] (2,3) -- (2,3.5);
\draw[] (1,3) -- (1,3.5);
\draw[] (0,3) -- (0,3.5);
\draw[] (3,0) -- (3,-0.5);
\draw[] (2,0) -- (2,-0.5);
\draw[] (1,0) -- (1,-0.5);
\draw[] (0,0) -- (0,-.5);
 \draw[fill, fill opacity = 0.1] (0.75, -0.25) -- (2.25, -0.25) -- (2.25, 0.25) -- (0.75,0.25) -- (0.75, -0.25);
 \begin{scope}[yshift=1cm]
  \draw[fill, fill opacity = 0.1] (0.75, -0.25) -- (2.25, -0.25) -- (2.25, 0.25) -- (0.75,0.25) -- (0.75, -0.25);   
 \end{scope}

  \begin{scope}[yshift=2cm]
  \draw[fill, fill opacity = 0.1] (0.75, -0.25) -- (2.25, -0.25) -- (2.25, 0.25) -- (0.75,0.25) -- (0.75, -0.25);   
 \end{scope}

  \begin{scope}[yshift=3cm]
  \draw[fill, fill opacity = 0.1] (0.75, -0.25) -- (2.25, -0.25) -- (2.25, 0.25) -- (0.75,0.25) -- (0.75, -0.25);   
 \end{scope}

 \fill[fill opacity = 0.1] (-0.5,-0.25) -- (0.25,-0.25) -- (0.25, 0.25) -- (-0.5,0.25) -- (-0.5,-0.25); 
 \draw[] (-0.5,-0.25) -- (0.25,-0.25) -- (0.25, 0.25) -- (-0.5,0.25);
 \begin{scope}[xshift= 3.25cm]   
  \fill[fill opacity = 0.1] (-0.5,-0.25) -- (0.25,-0.25) -- (0.25, 0.25) -- (-0.5,0.25) -- (-0.5,-0.25); 
  \draw[] (0.25, -0.25) -- (-0.5,-0.25) -- (-0.5,0.25) -- (0.25, 0.25);
 \end{scope}
 \begin{scope}[xshift= 3.25cm, yshift=1cm]   
  \fill[fill opacity = 0.1] (-0.5,-0.25) -- (0.25,-0.25) -- (0.25, 0.25) -- (-0.5,0.25) -- (-0.5,-0.25); 
  \draw[] (0.25, -0.25) -- (-0.5,-0.25) -- (-0.5,0.25) -- (0.25, 0.25);
 \end{scope}
 \begin{scope}[xshift= 3.25cm, yshift=2cm]   
  \fill[fill opacity = 0.1] (-0.5,-0.25) -- (0.25,-0.25) -- (0.25, 0.25) -- (-0.5,0.25) -- (-0.5,-0.25); 
  \draw[] (0.25, -0.25) -- (-0.5,-0.25) -- (-0.5,0.25) -- (0.25, 0.25);
 \end{scope}
 \begin{scope}[xshift= 3.25cm, yshift=3cm]   
  \fill[fill opacity = 0.1] (-0.5,-0.25) -- (0.25,-0.25) -- (0.25, 0.25) -- (-0.5,0.25) -- (-0.5,-0.25); 
  \draw[] (0.25, -0.25) -- (-0.5,-0.25) -- (-0.5,0.25) -- (0.25, 0.25);
 \end{scope}

 \begin{scope}[yshift=1cm]
     \fill[fill opacity = 0.1] (-0.5,-0.25) -- (0.25,-0.25) -- (0.25, 0.25) -- (-0.5,0.25) -- (-0.5,-0.25); 
 \draw[] (-0.5,-0.25) -- (0.25,-0.25) -- (0.25, 0.25) -- (-0.5,0.25);
 \end{scope}

 \begin{scope}[yshift=2cm]
     \fill[fill opacity = 0.1] (-0.5,-0.25) -- (0.25,-0.25) -- (0.25, 0.25) -- (-0.5,0.25) -- (-0.5,-0.25); 
 \draw[] (-0.5,-0.25) -- (0.25,-0.25) -- (0.25, 0.25) -- (-0.5,0.25);
 \end{scope}

 \begin{scope}[yshift=3cm]
     \fill[fill opacity = 0.1] (-0.5,-0.25) -- (0.25,-0.25) -- (0.25, 0.25) -- (-0.5,0.25) -- (-0.5,-0.25); 
 \draw[] (-0.5,-0.25) -- (0.25,-0.25) -- (0.25, 0.25) -- (-0.5,0.25);
 \end{scope}
 
 \node[] at (1.5,-1) {\scalebox{1}{$(\mathrm{horizontal, odd})$}};
 
 \end{scope}
 
        \end{tikzpicture}
        \caption{$4$ Tassellation Patterns Used to Cover Lattices $\Lambda^{\a}_L$: notice that using all such covering one is able to recover all edges crossed by the contour $\gamma$. (In the Lattice $\Lambda^{\mathrm{SP}}_L$ all spacings are doubled)}
        \label{tass}
    \end{figure}
 
By Cauchy--Schwartz's inequality:
\begin{equation}
\label{eqn:26}
\Big\langle
\prod_{\langle x,y\rangle\in \mathrm{E}(\gamma)} I_x^+ I_y^-
\Big\rangle_{\mu_\alpha} = 
\Big\langle \prod_{\substack{a=h,v\\ b=e,o}}
\prod_{\langle x,y\rangle\in \mathrm{E}_{a,b}(\gamma)} I_x^+ I_y^-
\Big\rangle_{\mu_\alpha}
\le
\prod_{\substack{a=h,v\\ b=e,o}}
\Big\langle
\prod_{\langle x,y\rangle\in \mathrm{E}_{a,b}(\gamma)} I_x^+ I_y^-
\Big\rangle_{\mu_\alpha}^{1/4}.
\end{equation}

\item 

Applying the chessboard estimate \eqref{CE} to each factor produces, up to translations, one of the four periodic patterns shown in Fig.~\ref{pat5}. Therefore:
\begin{equation}
\label{eq:Peierls-patterns}
\Big\langle
\prod_{\langle x,y\rangle\in \mathrm{E}(\gamma)} I_x^+ I_y^-
\Big\rangle_{\mu_\alpha}
\le
\left\langle I^{(-++-)}\right\rangle_{\mu_\alpha}^{\frac{|\mathrm{E}_{h,e}(\gamma)|}{2|\Lambda_L^\alpha|}}
\left\langle I^{(--++)}\right\rangle_{\mu_\alpha}^{\frac{|\mathrm{E}_{h,o}(\gamma)|}{2|\Lambda_L^\alpha|}}
\left\langle I^{\left(\substack{-\\+\\+\\-}\right)}\right\rangle_{\mu_\alpha}^{\frac{|\mathrm{E}_{v,e}(\gamma)|}{2|\Lambda_L^\alpha|}}
\left\langle I^{\left(\substack{-\\-\\+\\+}\right)}\right\rangle_{\mu_\alpha}^{\frac{|\mathrm{E}_{v,o}(\gamma)|}{2|\Lambda_L^\alpha|}}
\end{equation}

where (see also Fig.~\ref{pat5}):

\begin{equation*}
\begin{aligned}
&I^{\left(-++-\right)}
\doteq
\prod_{x\in \L^{\a,h}_L}
I^-_{x} I^+_{x+ e_0^{\a}} I^+_{x+ 2 e_0^{\a}} I^-_{x+ 3e_0^{\a}}, \qquad
I^{\left(--++\right)} \doteq
\prod_{x\in \L^{\a,h}_L}
I^-_{x} I^-_{x+ e_0^{\a}} I^+_{x+ 2 e_0^{\a}} I^+_{x+ 3e_0^{\a}},
\\
&I^{\left(\substack{-\\+\\+\\-}\right)} \doteq
\prod_{x\in \L^{\a,v}_L}
I^-_{x} I^+_{x+ e_1^{\a}} I^+_{x+ 2e_1^{\a}} I^-_{x+ 3e_1^{\a}}, \qquad I^{\left(\substack{-\\-\\+\\+}\right)} \doteq
\prod_{x\in \L^{\a,v}_L}
I^-_{x} I^-_{x+ e_1^{\a}} I^+_{x+ 2 e_1^{\a}} I^+_{x+ 3e_1^{\a}},
\end{aligned}
\end{equation*}

with $e_0^{\a}= 2^{\d_{\a,\mathrm{SP}}}(1,0)$, $e_1^{\a}= 2^{\d_{\a,\mathrm{SP}}} (0,1)$, and 

\begin{align*}
\L^{\a,h}_L&= \L^{\a}_L \cap \left\{\begin{array}{cc}
(4\mbb{Z})\times \mbb{Z},     & \a\in\{\mathrm{NN},\mathrm{SN}\} \\
(8\mbb{Z})\times(2\mbb{Z}),     & \a=\mathrm{SP} 
\end{array}\right. \\
\L^{\a,v}_L&= \L^{\a}_L \cap \left\{\begin{array}{cc}
\mbb{Z}\times (4\mbb{Z}),     & \a\in\{\mathrm{NN},\mathrm{SN}\} \\
(2\mbb{Z})\times(8\mbb{Z}),     & \a=\mathrm{SP} 
\end{array}. \right.
\end{align*}

    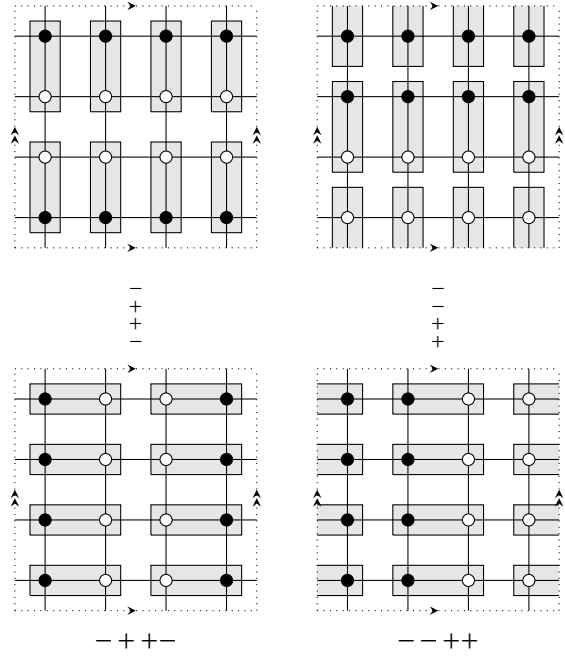
\begin{figure}
        \centering
        \begin{tikzpicture}[scale=0.8]
      %Caso Alto DX
        \begin{scope}[xshift=0cm]
\draw[dotted, ->-] (-0.5,-0.5) -- (3.5,-0.5);
\draw[dotted, ->>-] (3.5,-0.5) -- (3.5,3.5);
\draw[dotted, ->-] (-0.5,3.5) -- (3.5,3.5);
\draw[dotted, ->>-] (-0.5,-0.5) -- (-0.5,3.5);
\draw[] (0, 0) -- (3,0);
\draw[] (0, 2) -- (3,2);
\draw[] (3,1) -- (3.5,1);
\draw[] (-0.5,1) -- (0,1);
\draw[] (-0.5,3) -- (0,3);
\draw[] (3,3) -- (3.5,3);
\draw (0,3) -- (3,3);
\draw (0,1) -- (3,1);
\draw (0,0) -- (0,3);
\draw (1,0) -- (1,3);
\draw (2,0) -- (2,3);
\draw (-0.5,2) -- (0,2);
\draw (3,2) -- (3.5,2);
\draw (-0.5,0) -- (0,0);
\draw (3,0) -- (3.5,0);
\draw (3,0) -- (3,3);
\draw[] (3,3) -- (3,3.5);
\draw[] (2,3) -- (2,3.5);
\draw[] (1,3) -- (1,3.5);
\draw[] (0,3) -- (0,3.5);
\draw[] (3,0) -- (3,-0.5);
\draw[] (2,0) -- (2,-0.5);
\draw[] (1,0) -- (1,-0.5);
\draw[] (0,0) -- (0,-.5);
    \draw[fill, fill opacity =0.1] (-0.25,-0.25) -- (0.25,-0.25) -- (0.25, 1.25) -- (-0.25,1.25) -- (-0.25,-0.25);  
    \begin{scope}[xshift=1cm]
        \draw[fill, fill opacity =0.1] (-0.25,-0.25) -- (0.25,-0.25) -- (0.25, 1.25) -- (-0.25,1.25) -- (-0.25,-0.25); 
    \end{scope}
      \begin{scope}[xshift=2cm]
        \draw[fill, fill opacity =0.1] (-0.25,-0.25) -- (0.25,-0.25) -- (0.25, 1.25) -- (-0.25,1.25) -- (-0.25,-0.25); 
    \end{scope}
      \begin{scope}[xshift=3cm]
        \draw[fill, fill opacity =0.1] (-0.25,-0.25) -- (0.25,-0.25) -- (0.25, 1.25) -- (-0.25,1.25) -- (-0.25,-0.25); 
    \end{scope}
      \begin{scope}[xshift=1cm, yshift=2cm]
        \draw[fill, fill opacity =0.1] (-0.25,-0.25) -- (0.25,-0.25) -- (0.25, 1.25) -- (-0.25,1.25) -- (-0.25,-0.25); 
    \end{scope}
     \begin{scope}[xshift=2cm, yshift=2cm]
        \draw[fill, fill opacity =0.1] (-0.25,-0.25) -- (0.25,-0.25) -- (0.25, 1.25) -- (-0.25,1.25) -- (-0.25,-0.25); 
    \end{scope}
     \begin{scope}[xshift=3cm, yshift=2cm]
        \draw[fill, fill opacity =0.1] (-0.25,-0.25) -- (0.25,-0.25) -- (0.25, 1.25) -- (-0.25,1.25) -- (-0.25,-0.25); 
        
    \end{scope}
     \begin{scope}[xshift=0cm, yshift=2cm]
        \draw[fill, fill opacity =0.1] (-0.25,-0.25) -- (0.25,-0.25) -- (0.25, 1.25) -- (-0.25,1.25) -- (-0.25,-0.25); 
        
    \end{scope}
        \node[below] at (1.5,-1) {\scalebox{1}{$\substack{-\\+\\+\\-}$}};
         \draw[black,fill=black] (0,0) circle (.1 cm);
         \draw[black,fill=black] (1,0) circle (.1 cm);
             \draw[black,fill=black] (2,0) circle (.1 cm);
              \draw[black,fill=black] (3,0) circle (.1 cm);
               \draw[black,fill=black] (0,3) circle (.1 cm);
                \draw[black,fill=black] (1,3) circle (.1 cm);
                 \draw[black,fill=black] (2,3) circle (.1 cm);
                  \draw[black,fill=black] (3,3) circle (.1 cm);
                   \draw[black,fill=white] (0,1) circle (.1 cm);
                   \draw[black,fill=white] (0,2) circle (.1 cm);
                   \draw[black,fill=white] (1,1) circle (.1 cm);
                   \draw[black,fill=white] (1,2) circle (.1 cm);
                   \draw[black,fill=white] (2,1) circle (.1 cm);
                   \draw[black,fill=white] (2,2) circle (.1 cm);
                   \draw[black,fill=white] (3,1) circle (.1 cm);
                   \draw[black,fill=white] (3,2) circle (.1 cm); 
        \end{scope}
            %%Top Right
 \begin{scope}[xshift=5cm]
           \draw[dotted, ->-] (-0.5,-0.5) -- (3.5,-0.5);
\draw[dotted, ->>-] (3.5,-0.5) -- (3.5,3.5);
\draw[dotted, ->-] (-0.5,3.5) -- (3.5,3.5);
\draw[dotted, ->>-] (-0.5,-0.5) -- (-0.5,3.5);
\draw[] (0, 0) -- (3,0);
\draw[] (0, 2) -- (3,2);
\draw[] (3,1) -- (3.5,1);
\draw[] (-0.5,1) -- (0,1);
\draw[] (-0.5,3) -- (0,3);
\draw[] (3,3) -- (3.5,3);
\draw (0,3) -- (3,3);
\draw (0,1) -- (3,1);
\draw (0,0) -- (0,3);
\draw (1,0) -- (1,3);
\draw (2,0) -- (2,3);
\draw (-0.5,2) -- (0,2);
\draw (3,2) -- (3.5,2);
\draw (-0.5,0) -- (0,0);
\draw (3,0) -- (3.5,0);
\draw (3,0) -- (3,3);
\draw[] (3,3) -- (3,3.5);
\draw[] (2,3) -- (2,3.5);
\draw[] (1,3) -- (1,3.5);
\draw[] (0,3) -- (0,3.5);
\draw[] (3,0) -- (3,-0.5);
\draw[] (2,0) -- (2,-0.5);
\draw[] (1,0) -- (1,-0.5);
\draw[] (0,0) -- (0,-.5);
\draw[fill, fill opacity =0.1] (-0.25,0.75) -- (0.25,0.75) -- (0.25, 2.25) -- (-0.25,2.25) -- (-0.25,0.75); 
\begin{scope}[xshift=1cm]
    \draw[fill, fill opacity =0.1] (-0.25,0.75) -- (0.25,0.75) -- (0.25, 2.25) -- (-0.25,2.25) -- (-0.25,0.75); 
\end{scope}
\begin{scope}[xshift=2cm]
    \draw[fill, fill opacity =0.1] (-0.25,0.75) -- (0.25,0.75) -- (0.25, 2.25) -- (-0.25,2.25) -- (-0.25,0.75); 
\end{scope}
\begin{scope}[xshift=3cm]
    \draw[fill, fill opacity =0.1] (-0.25,0.75) -- (0.25,0.75) -- (0.25, 2.25) -- (-0.25,2.25) -- (-0.25,0.75); 
\end{scope}
\fill[fill opacity = 0.1] (-0.25, -0.5) -- (0.25, -0.5) -- (0.25, 0.5) -- (-0.25,0.5) -- (-0.25, -0.5);
 \draw[] (0.25, -0.5) -- (0.25, 0.5) -- (-0.25,0.5) -- (-0.25, -0.5);
 \begin{scope}[xshift=1cm]
   \fill[fill opacity = 0.1]   (-0.25, -0.5) -- (0.25, -0.5) -- (0.25, 0.5) -- (-0.25,0.5) -- (-0.25, -0.5);
 \draw[] (0.25, -0.5) -- (0.25, 0.5) -- (-0.25,0.5) -- (-0.25, -0.5);
 \end{scope}
  \begin{scope}[xshift=2cm]
  \fill[fill opacity = 0.1]    (-0.25, -0.5) -- (0.25, -0.5) -- (0.25, 0.5) -- (-0.25,0.5) -- (-0.25, -0.5);
 \draw[] (0.25, -0.5) -- (0.25, 0.5) -- (-0.25,0.5) -- (-0.25, -0.5);
 \end{scope}
  \begin{scope}[xshift=3cm]
  \fill[fill opacity = 0.1]    (-0.25, -0.5) -- (0.25, -0.5) -- (0.25, 0.5) -- (-0.25,0.5) -- (-0.25, -0.5);
 \draw[] (0.25, -0.5) -- (0.25, 0.5) -- (-0.25,0.5) -- (-0.25, -0.5);
 \end{scope}
 \fill[fill opacity = 0.1] (-0.25,2.5) -- (0.25,2.5) -- (0.25,3.5) -- (-0.25,3.5) -- (-0.25,2.5);
 \draw[] (-0.25, 3.5) -- (-0.25, 2.5) -- (0.25,2.5) -- (0.25, 3.5);
 \begin{scope}[xshift=1cm]
     \fill[fill opacity = 0.1] (-0.25,2.5) -- (0.25,2.5) -- (0.25,3.5) -- (-0.25,3.5) -- (-0.25,2.5);
 \draw[] (-0.25, 3.5) -- (-0.25, 2.5) -- (0.25,2.5) -- (0.25, 3.5);
 \end{scope}
 \begin{scope}[xshift=2cm]
     \fill[fill opacity = 0.1] (-0.25,2.5) -- (0.25,2.5) -- (0.25,3.5) -- (-0.25,3.5) -- (-0.25,2.5);
 \draw[] (-0.25, 3.5) -- (-0.25, 2.5) -- (0.25,2.5) -- (0.25, 3.5);
 \end{scope}
 \begin{scope}[xshift=3cm]
     \fill[fill opacity = 0.1] (-0.25,2.5) -- (0.25,2.5) -- (0.25,3.5) -- (-0.25,3.5) -- (-0.25,2.5);
 \draw[] (-0.25, 3.5) -- (-0.25, 2.5) -- (0.25,2.5) -- (0.25, 3.5);
 \end{scope}
 \node[below] at (1.5,-1) {\scalebox{1}{$\substack{-\\-\\+\\+}$}};
 \draw[black,fill=white] (0,0) circle (.1 cm);
         \draw[black,fill=white] (1,0) circle (.1 cm);
             \draw[black,fill=white] (2,0) circle (.1 cm);
              \draw[black,fill=white] (3,0) circle (.1 cm);
               \draw[black,fill=black] (0,3) circle (.1 cm);
                \draw[black,fill=black] (1,3) circle (.1 cm);
                 \draw[black,fill=black] (2,3) circle (.1 cm);
                  \draw[black,fill=black] (3,3) circle (.1 cm);
                   \draw[black,fill=white] (0,1) circle (.1 cm);
                   \draw[black,fill=black] (0,2) circle (.1 cm);
                   \draw[black,fill=white] (1,1) circle (.1 cm);
                   \draw[black,fill=black] (1,2) circle (.1 cm);
                   \draw[black,fill=white] (2,1) circle (.1 cm);
                   \draw[black,fill=black] (2,2) circle (.1 cm);
                   \draw[black,fill=white] (3,1) circle (.1 cm);
                   \draw[black,fill=black] (3,2) circle (.1 cm);

 \end{scope}

 %%Down Left
 \begin{scope}[yshift=-6cm]
           \draw[dotted, ->-] (-0.5,-0.5) -- (3.5,-0.5);
\draw[dotted, ->>-] (3.5,-0.5) -- (3.5,3.5);
\draw[dotted, ->-] (-0.5,3.5) -- (3.5,3.5);
\draw[dotted, ->>-] (-0.5,-0.5) -- (-0.5,3.5);
\draw[] (0, 0) -- (3,0);
\draw[] (0, 2) -- (3,2);
\draw[] (3,1) -- (3.5,1);
\draw[] (-0.5,1) -- (0,1);
\draw[] (-0.5,3) -- (0,3);
\draw[] (3,3) -- (3.5,3);
\draw (0,3) -- (3,3);
\draw (0,1) -- (3,1);
\draw (0,0) -- (0,3);
\draw (1,0) -- (1,3);
\draw (2,0) -- (2,3);
\draw (-0.5,2) -- (0,2);
\draw (3,2) -- (3.5,2);
\draw (-0.5,0) -- (0,0);
\draw (3,0) -- (3.5,0);
\draw (3,0) -- (3,3);
\draw[] (3,3) -- (3,3.5);
\draw[] (2,3) -- (2,3.5);
\draw[] (1,3) -- (1,3.5);
\draw[] (0,3) -- (0,3.5);
\draw[] (3,0) -- (3,-0.5);
\draw[] (2,0) -- (2,-0.5);
\draw[] (1,0) -- (1,-0.5);
\draw[] (0,0) -- (0,-.5);
\draw[fill, fill opacity = 0.1] (-0.25, -0.25) -- (1.25, -0.25) -- (1.25, 0.25) -- (-0.25,0.25) -- (-0.25, -0.25);
\begin{scope}[xshift=2cm]
    \draw[fill, fill opacity = 0.1] (-0.25, -0.25) -- (1.25, -0.25) -- (1.25, 0.25) -- (-0.25,0.25) -- (-0.25, -0.25);
\end{scope}
\begin{scope}[xshift=2cm, yshift= 1cm]
    \draw[fill, fill opacity = 0.1] (-0.25, -0.25) -- (1.25, -0.25) -- (1.25, 0.25) -- (-0.25,0.25) -- (-0.25, -0.25);
\end{scope}
\begin{scope}[xshift=2cm, yshift=3cm]
    \draw[fill, fill opacity = 0.1] (-0.25, -0.25) -- (1.25, -0.25) -- (1.25, 0.25) -- (-0.25,0.25) -- (-0.25, -0.25);
\end{scope}
\begin{scope}[xshift=2cm, yshift=2cm]
    \draw[fill, fill opacity = 0.1] (-0.25, -0.25) -- (1.25, -0.25) -- (1.25, 0.25) -- (-0.25,0.25) -- (-0.25, -0.25);
\end{scope}
\begin{scope}[xshift=0cm, yshift= 1cm]
    \draw[fill, fill opacity = 0.1] (-0.25, -0.25) -- (1.25, -0.25) -- (1.25, 0.25) -- (-0.25,0.25) -- (-0.25, -0.25);
\end{scope}
\begin{scope}[xshift=0cm, yshift=3cm]
    \draw[fill, fill opacity = 0.1] (-0.25, -0.25) -- (1.25, -0.25) -- (1.25, 0.25) -- (-0.25,0.25) -- (-0.25, -0.25);
\end{scope}
\begin{scope}[xshift=0cm, yshift=2cm]
    \draw[fill, fill opacity = 0.1] (-0.25, -0.25) -- (1.25, -0.25) -- (1.25, 0.25) -- (-0.25,0.25) -- (-0.25, -0.25);
\end{scope}
 \node[] at (1.5,-1) {\scalebox{1}{$-++-$}};
  \draw[black,fill=black] (0,0) circle (.1 cm);
         \draw[black,fill=white] (1,0) circle (.1 cm);
             \draw[black,fill=white] (2,0) circle (.1 cm);
              \draw[black,fill=black] (3,0) circle (.1 cm);
               \draw[black,fill=black] (0,3) circle (.1 cm);
                \draw[black,fill=white] (1,3) circle (.1 cm);
                 \draw[black,fill=white] (2,3) circle (.1 cm);
                  \draw[black,fill=black] (3,3) circle (.1 cm);
                   \draw[black,fill=black] (0,1) circle (.1 cm);
                   \draw[black,fill=black] (0,2) circle (.1 cm);
                   \draw[black,fill=white] (1,1) circle (.1 cm);
                   \draw[black,fill=white] (1,2) circle (.1 cm);
                   \draw[black,fill=white] (2,1) circle (.1 cm);
                   \draw[black,fill=white] (2,2) circle (.1 cm);
                   \draw[black,fill=black] (3,1) circle (.1 cm);
                   \draw[black,fill=black] (3,2) circle (.1 cm); 
 \end{scope}

 %%Down Left
 \begin{scope}[xshift=5cm, yshift= -6cm]
         \draw[dotted, ->-] (-0.5,-0.5) -- (3.5,-0.5);
\draw[dotted, ->>-] (3.5,-0.5) -- (3.5,3.5);
\draw[dotted, ->-] (-0.5,3.5) -- (3.5,3.5);
\draw[dotted, ->>-] (-0.5,-0.5) -- (-0.5,3.5);
\draw[] (0, 0) -- (3,0);
\draw[] (0, 2) -- (3,2);
\draw[] (3,1) -- (3.5,1);
\draw[] (-0.5,1) -- (0,1);
\draw[] (-0.5,3) -- (0,3);
\draw[] (3,3) -- (3.5,3);
\draw (0,3) -- (3,3);
\draw (0,1) -- (3,1);
\draw (0,0) -- (0,3);
\draw (1,0) -- (1,3);
\draw (2,0) -- (2,3);
\draw (-0.5,2) -- (0,2);
\draw (3,2) -- (3.5,2);
\draw (-0.5,0) -- (0,0);
\draw (3,0) -- (3.5,0);
\draw (3,0) -- (3,3);
\draw[] (3,3) -- (3,3.5);
\draw[] (2,3) -- (2,3.5);
\draw[] (1,3) -- (1,3.5);
\draw[] (0,3) -- (0,3.5);
\draw[] (3,0) -- (3,-0.5);
\draw[] (2,0) -- (2,-0.5);
\draw[] (1,0) -- (1,-0.5);
\draw[] (0,0) -- (0,-.5);
 \draw[fill, fill opacity = 0.1] (0.75, -0.25) -- (2.25, -0.25) -- (2.25, 0.25) -- (0.75,0.25) -- (0.75, -0.25);
 \begin{scope}[yshift=1cm]
  \draw[fill, fill opacity = 0.1] (0.75, -0.25) -- (2.25, -0.25) -- (2.25, 0.25) -- (0.75,0.25) -- (0.75, -0.25);   
 \end{scope}

  \begin{scope}[yshift=2cm]
  \draw[fill, fill opacity = 0.1] (0.75, -0.25) -- (2.25, -0.25) -- (2.25, 0.25) -- (0.75,0.25) -- (0.75, -0.25);   
 \end{scope}

  \begin{scope}[yshift=3cm]
  \draw[fill, fill opacity = 0.1] (0.75, -0.25) -- (2.25, -0.25) -- (2.25, 0.25) -- (0.75,0.25) -- (0.75, -0.25);   
 \end{scope}

 \fill[fill opacity = 0.1] (-0.5,-0.25) -- (0.25,-0.25) -- (0.25, 0.25) -- (-0.5,0.25) -- (-0.5,-0.25); 
 \draw[] (-0.5,-0.25) -- (0.25,-0.25) -- (0.25, 0.25) -- (-0.5,0.25);
 \begin{scope}[xshift= 3.25cm]   
  \fill[fill opacity = 0.1] (-0.5,-0.25) -- (0.25,-0.25) -- (0.25, 0.25) -- (-0.5,0.25) -- (-0.5,-0.25); 
  \draw[] (0.25, -0.25) -- (-0.5,-0.25) -- (-0.5,0.25) -- (0.25, 0.25);
 \end{scope}
 \begin{scope}[xshift= 3.25cm, yshift=1cm]   
  \fill[fill opacity = 0.1] (-0.5,-0.25) -- (0.25,-0.25) -- (0.25, 0.25) -- (-0.5,0.25) -- (-0.5,-0.25); 
  \draw[] (0.25, -0.25) -- (-0.5,-0.25) -- (-0.5,0.25) -- (0.25, 0.25);
 \end{scope}
 \begin{scope}[xshift= 3.25cm, yshift=2cm]   
  \fill[fill opacity = 0.1] (-0.5,-0.25) -- (0.25,-0.25) -- (0.25, 0.25) -- (-0.5,0.25) -- (-0.5,-0.25); 
  \draw[] (0.25, -0.25) -- (-0.5,-0.25) -- (-0.5,0.25) -- (0.25, 0.25);
 \end{scope}
 \begin{scope}[xshift= 3.25cm, yshift=3cm]   
  \fill[fill opacity = 0.1] (-0.5,-0.25) -- (0.25,-0.25) -- (0.25, 0.25) -- (-0.5,0.25) -- (-0.5,-0.25); 
  \draw[] (0.25, -0.25) -- (-0.5,-0.25) -- (-0.5,0.25) -- (0.25, 0.25);
 \end{scope}

 \begin{scope}[yshift=1cm]
     \fill[fill opacity = 0.1] (-0.5,-0.25) -- (0.25,-0.25) -- (0.25, 0.25) -- (-0.5,0.25) -- (-0.5,-0.25); 
 \draw[] (-0.5,-0.25) -- (0.25,-0.25) -- (0.25, 0.25) -- (-0.5,0.25);
 \end{scope}

 \begin{scope}[yshift=2cm]
     \fill[fill opacity = 0.1] (-0.5,-0.25) -- (0.25,-0.25) -- (0.25, 0.25) -- (-0.5,0.25) -- (-0.5,-0.25); 
 \draw[] (-0.5,-0.25) -- (0.25,-0.25) -- (0.25, 0.25) -- (-0.5,0.25);
 \end{scope}

 \begin{scope}[yshift=3cm]
     \fill[fill opacity = 0.1] (-0.5,-0.25) -- (0.25,-0.25) -- (0.25, 0.25) -- (-0.5,0.25) -- (-0.5,-0.25); 
 \draw[] (-0.5,-0.25) -- (0.25,-0.25) -- (0.25, 0.25) -- (-0.5,0.25);
 \end{scope}
 
 \node[] at (1.5,-1) {\scalebox{1}{$--++$}};
 \draw[black,fill=black] (0,0) circle (.1 cm);
         \draw[black,fill=black] (1,0) circle (.1 cm);
             \draw[black,fill=white] (2,0) circle (.1 cm);
              \draw[black,fill=white] (3,0) circle (.1 cm);
               \draw[black,fill=black] (0,3) circle (.1 cm);
                \draw[black,fill=black] (1,3) circle (.1 cm);
                 \draw[black,fill=white] (2,3) circle (.1 cm);
                  \draw[black,fill=white] (3,3) circle (.1 cm);
                   \draw[black,fill=black] (0,1) circle (.1 cm);
                   \draw[black,fill=black] (0,2) circle (.1 cm);
                   \draw[black,fill=black] (1,1) circle (.1 cm);
                   \draw[black,fill=black] (1,2) circle (.1 cm);
                   \draw[black,fill=white] (2,1) circle (.1 cm);
                   \draw[black,fill=white] (2,2) circle (.1 cm);
                   \draw[black,fill=white] (3,1) circle (.1 cm);
                   \draw[black,fill=white] (3,2) circle (.1 cm);

 \end{scope}
 
        \end{tikzpicture}
        \caption{Periodic patterns generated by chessboard reflections of the gray boxes representing the covering configurations (h,e), (h,o), (v,e), and (v,o).}
        \label{pat5}
    \end{figure}

\item By translation and rotation invariance of $\mu_\alpha$, the four expectations coincide\footnote{\label{footnote:symmetries} One can verify that the same conclusion holds for staggered fermions by observing the following facts. Any geometric transformation of the field $\phi$ is implemented by conjugation with a suitable unitary operator $U$. Since we are only interested in spectral properties, one may equivalently regard this transformation as acting on the Dirac operator, because
\[
\det\!\big((\slashed D_\alpha)^- - U M_\phi U^\ast - \lambda \mathbf 1\big)
=
\det\!\big(U(\slashed D_\alpha)^-U^\ast - M_\phi - \lambda \mathbf 1\big).
\]
Now, the $\mathbb Z_2$-holonomy background is invariant under translations and rotations, since every plaquette carries flux $-1$. Moreover, the spectrum of a magnetic Laplacian, of which the staggered Dirac operator is an example, depends only on the holonomies of the background Gauge field (see \cite[Lemma 2.1]{10.1215/S0012-7094-93-07114-1}). In the proof of that result, the unitary conjugation relating two magnetic Laplacians with the same fluxes is constructed explicitly and shown to be a gauge transformation, namely multiplication by site-dependent phases. Such transformations commute with the multiplication operator $M_\phi$, and therefore do not affect it.} and thus, using also that $|\mathrm{E}_{h,e}(\gamma)|+|\mathrm{E}_{h,o}(\gamma)|+|\mathrm{E}_{v,e}(\gamma)|+|\mathrm{E}_{v,o}(\gamma)|=|\gamma|$,

\begin{equation}
\label{eq:Peierls-single-pattern}
\Big\langle
\prod_{\langle x,y\rangle\in \mathrm{E}(\gamma)} I_x^+ I_y^-
\Big\rangle_{\mu_\alpha}
\le
\left\langle I^{(-++-)}\right\rangle_{\mu_\alpha}^{\frac{|\gamma|}{2|\Lambda_L^\alpha|}}.
\end{equation}

\item Following the strategy used to bound the numerator in the right--hand side of \eqref{eqn:10a}, we obtain:
\begin{equation}
\label{eq:Peierls-I-bound}
\left\langle I^{(-++-)}\right\rangle_{\mu_\alpha}^{\frac{|\gamma|}{2|\Lambda^{\a}_L|}}
\le(2 \pi \lambda)^{\frac{|\gamma|}{4}} \bigg[ \frac{1}{Z^{\alpha}_{\L_L}} \exp\big\{-N \inf_{\phi \in \operatorname{supp} I^{\left(-++-\right)}} V^{\alpha}_{L,\tilde{\lambda}_N}(\phi)\big\} \bigg]^{\frac{|\gamma|}{2 |\L^{\a}_L|}},
\end{equation}

with $\tilde{\l}_N= \l \big(1-\frac{1}{N}\big)^{-1}$.

\item A further RP argument, closely analogous to Proposition~\ref{const_field}, shows that the infimum over $\mathrm{supp}\,I^{(-++-)}$ is attained by periodic configurations  of the form:
\[
(\varphi_-,\varphi_+,\varphi_+,\varphi_-),
\qquad
\varphi_+\ge 0,\ \varphi_-\le 0,
\]
depending only on \emph{two real
parameters} $(\vphi_+,\vphi_-)$. Accordingly:

\begin{equation}
\label{eqn:47}
\inf_{\phi\in\mathrm{supp}\,I^{(-++-)}}
V^{\alpha}_{L,\tilde\lambda_N}(\phi)
=
\inf_{\varphi_+\ge0,\ \varphi_-\le0}
V^{\alpha}_{L,\tilde\lambda_N}(\varphi_+,\varphi_-)
\end{equation}

where, by abuse of notation, $V^{\alpha}_{L,\tilde\lambda_N}(\varphi_+,\varphi_-)$ denotes the potential $V^{\alpha}_{L,\tilde\lambda_N}(\varphi_+,\varphi_-)$ evaluated at the periodic configuration corresponding to the repetition of the pattern $(\varphi_-,\varphi_+,\varphi_+,\varphi_-)$.

Moreover, for any $\lambda>0$ and $\vphi_+,\vphi_-\in\mbb{R}$, the finite-volume potential density $|\L^{\a}_L|^{-1}V^{\a}_{L,\l}(\vphi_+, \vphi_-)$ is well approximated by its infinite-volume limit, to be denoted by $w^{\alpha}_{\lambda}(\varphi_+,\varphi_-)$ (in full analogy with \eqref{eqn:12a}):
\begin{equation}
\label{eq:Peierls-thermo}
\left|
\tfrac{1}{|\Lambda_L^\alpha|}V^{\alpha}_{L,\lambda}(\varphi_+,\varphi_-)
-
w^{\alpha}_{\lambda}(\varphi_+,\varphi_-)
\right|
\le
\frac{\mathfrak c_1}{L}.
\end{equation}
It is possible to write down an explicit expression for $w^{\a}_{\l}$:
\begin{equation}
\label{eq:Peierls-w}
w^{\alpha}_{\lambda}(\varphi_+,\varphi_-)
=
\frac{\varphi_+^2+\varphi_-^2}{4\lambda}
-
N_\alpha
\iint_{(-\pi,\pi]^2}\frac{d^2p}{(2\pi)^2}
\log \Delta^\alpha(p;\varphi_+,\varphi_-),
\end{equation}

where $N_{\alpha}$ is a multiplicity factor ($N_{\a}=\frac{1}{4},\frac{1}{4},\frac{1}{2}$ for $\a=\mathrm{NN},\mathrm{SN},\mathrm{SP}$ respectively) and  $\Delta^\alpha(p;\varphi_+,\varphi_-)$ is an explicit function (see \eqref{eqn:46a}, \eqref{eqn:46b} and \eqref{eq:mySP1}), related to the determinant of the Dirac operator in alternating background. 
The proof of \eqref{eq:Peierls-thermo} can be performed following the same strategy adopted for the convergence of the uniform-field potential \eqref{eqn:12a}, described in Appendix \ref{app:pot_2}. In fact, a key point is that the functions $\Delta^{\a}$ behave in the same qualitative way, as $p\to0$, as the integrand in the r.h.s. of \eqref{eqn:58} (see
\eqref{eqn:6}, \eqref{eqn:6b} and \eqref{eqn:6c} below), allowing one to replicate the argument exposed in Appendix \ref{app:pot_2}.

Summarizing, so far we have shown that

\begin{equation}
\label{eqn:56}
\Big\langle I^{(-++-)}\Big\rangle_{\mu_\alpha}^{\frac{|\gamma|}{2|\Lambda^{\a}_L|}}
\le(2 \pi \lambda)^{\frac{|\gamma|}{4}} \left[ \exp\Big\{-N |\L^{\a}_L| \Big(\inf_{\varphi_+\ge0,\ \varphi_-\le0} w^{\a}_{\tilde{\l}_N}(\vphi_+,\vphi_-) -\frac{\mf{c}_1}{L} \Big)\Big\}\right]^{\frac{|\gamma|}{2 |\L^{\a}_L|}}.
\end{equation}

\item Plugging the lower-bound on the partition function \eqref{eqn:16} in \eqref{eqn:56}, we get:

\begin{equation*}
\Big\langle I^{\left(-++- \right)} \Big\rangle_{\mu_{\a}}^{\frac{|\gamma|}{2|\L^{\a}_L|}} \leq  e^{-N  \frac{|\gamma|}{2}\big( \inf_{\vphi_+,\vphi_-} w^{\alpha}_{\tilde{\lambda}_N}(\varphi_+,\varphi_-)-
v^{\alpha}_{\tilde{\lambda}_N}(\varphi_\alpha^\star(\tilde{\lambda}_N)) -  \mc{R}^{(3)}_{N,L}(\l) \big) }, 
\end{equation*}

where

\begin{equation*}
\begin{split}
\mc{R}^{(3)}_{N,L}(\l)\doteq &\big( v^{\alpha}_{\lambda}(\varphi^{\star}_{\alpha}(\lambda))- v^{\a}_{\tilde\l_N}(\vphi^{\star}_{\a}(\tilde\l_N)) \big) + \tfrac{\log(2\pi\l)}{2N}+ \tfrac{\d_N}{\l}\big(2\varphi^{\star}_{\alpha}(\lambda)+\d_N\big)\\
&-\tfrac{1}{N}\log (\d_N)- 4\log\big(1- \tfrac{\d_N}{\varphi^{\star}_{\alpha}(\lambda)} \big) + \tfrac{\mf{c}_0+\mf{c}_1}{L}.
\end{split}
\end{equation*}

\item Since $\varphi^{\star}_{\alpha}$ is a smooth function of $\l\in(0,\infty)$ (see Lemma \ref{const_field}) and $ v^{\alpha}_{\lambda}(\vphi)$ is smooth with respect to $\l$ and $\vphi\in (0,\infty)$, we easily get that

\begin{equation*}
\lim_{N\to\infty} v^{\a}_{\tilde\l_N}(\vphi^{\star}_{\a}(\tilde\l_N))=  v^{\alpha}_{\lambda}(\varphi^{\star}_{\alpha}(\lambda)).
\end{equation*}

Using this information, we can follow a similar reasoning as before: by fixing $\d_N=\frac{\varphi^{\star}_{\alpha}(\lambda)}{N}$, we see that for every $\l>0$:

\begin{equation*}
    \mc{R}^{(3)}_{N,L}(\l) \overset{N,L \to \infty }{\longrightarrow} 0.
\end{equation*}
\end{enumerate}

The crucial input, about the energetic penalty associated with the Peirels-like terms, is provided by the following lemma.

\begin{lemma}%[Alternating configurations are penalized]
\label{lem:Peierls-gap}
There exists $\lambda_0>0$ such that, for every $0<\lambda\le\lambda_0$,
\[
\inf_{\varphi_+\ge0,\ \varphi_-\le0}
w^{\alpha}_{\lambda}(\varphi_+,\varphi_-)
-
v^{\alpha}_{\lambda}(\varphi_\alpha^\star(\lambda))
\geq
c_3^\alpha(\lambda)>0,
\]

where $c^{\alpha}_3(\l)$ as in Proposition \ref{prop:contrib}.   
\end{lemma}

\begin{remark}\label{rem2}
As clear from the incoming analysis, the assumption that the coupling constant $\lambda$ is sufficiently small ensures that the effective potential $w^{\a}_{\l}(\vphi_+, \vphi_-)$ does not develop stationary points in the interior of the region $\{\vphi_{+}\ge0, \vphi_-\le0\}$. In fact, a direct inspection of the effective potential suggests that, for sufficiently large values of $\lambda$, this condition is violated and nontrivial stationary points appear (whose energy gap is not easily controlled via our method). This behavior is consistent with the expected phase structure of the model, where no spontaneous breaking of Chiral symmetry occurs at strong coupling due to the approximate product structure of the measure \cite{cohen1981monte, Cohen:1983nr}.
\end{remark}

\begin{enumerate}[resume]
    \item If $\lambda \leq \lambda_0$, for $N$ large enough, $\tilde{\l}_N \leq \l_0$ as well. Therefore, we can apply Lemma \ref{lem:Peierls-gap} to obtain:

    \begin{equation*}
         \left\langle I^{\left(-++- \right)} \right\rangle_{\mu_{\a}}^{\frac{|\gamma|}{2|\L^{\a}_L|}} \leq  e^{-N  \frac{|\gamma|}{2}\big( c_3^{\a}(\tilde\l_N) -  \mc{R}^{(3)}_{N,L}(\l) \big) }. 
\end{equation*}

Also, for $N, L$ large enough (depending on $\l$), we can assume that $|\mc{R}^{(3)}_{N,L}(\l)| \leq \frac{1}{4} c_3^{\a}(\l)$, leading to

\begin{equation*}
         \left\langle I^{\left(-++- \right)} \right\rangle_{\mu_{\a}}^{\frac{|\gamma|}{2|\L^{\a}_L|}} \leq  e^{-N  \frac{|\gamma|}{2} \big(c_3^{\a}(\tilde\l_N) - \frac{1}{4} c_3^{\a}(\l)\big) }. 
\end{equation*}

\item   Furthermore, by the very definition (cf. Proposition \ref{prop:contrib}), $c^{\alpha}_3$ is a continuous function of $\l>0$, so that $\lim_{N\to\infty} c_3^{\a}(\tilde\l_N)= c_3^{\a}(\l)$. Therefore, noting also that $c_3^{\a}(\l)>0$, for $N$ large enough we can assume that

\[
\big|c^{\a}_3(\l)-c^{\a}_3(\tilde{\l}_N)\big| \leq \tfrac{1}{4}c^{\a}_3(\l),
\]

which yields the bound \eqref{eqn:18c}.
\end{enumerate}

\subsection{Proof of Lemma \ref{lem:Peierls-gap}}
\label{ssect:alternate}

We discuss the proof for the three models, $\a\in\{\mathrm{NN},\mathrm{SN},\mathrm{SP}\}$, separately. We prefer to start by discussing the SN model, which involves slightly simpler computations than the other two models.

\paragraph{Proof for the SN model.}

By using the periodicity of the alternate field configurations $\phi$ (consisting of only two independent variables $\vphi_+,\vphi_-$), see Fig. \ref{per}, one can easily compute the expression for the potential $V^{\mathrm{SN}}_{L,\l}(\vphi_+,\vphi_-)$ introduced after \eqref{eqn:47}:

\begin{equation*}
V^{\mathrm{SN}}_{L,\l}(\vphi_+,\vphi_-)= \frac{|\L_L|}{2} \frac{\vphi_+^2+ \vphi_-^2}{2\l}- \sum_{p\in \Xi^*_L}\log \Delta^{\mathrm{SN}}(p;\vphi_+,\vphi_-),
\end{equation*}

where $\Xi^*_L= \Big(\frac{2\pi}{L} \big( \mbb{Z}+\frac{1}{2}\big)^2\Big) \cap \Big((-\frac{\pi}{4},\frac{\pi}{4}]\times(-\pi,\pi]\Big)$ and

\begin{equation*}
\Delta^{\mathrm{SN}}(p;\vphi_+,\vphi_-)= \det\begin{pmatrix}
i \sin(p_1) - \vphi_- &-\tfrac{1}{2} &0 &\tfrac{1}{2}e^{4i p_0}\\
\tfrac{1}{2} &- i \sin(p_1) - \vphi_+ &-\tfrac{1}{2} &0\\
0 &\tfrac{1}{2} &i \sin(p_1) - \vphi_+ &-\tfrac{1}{2}\\
-\frac{1}{2}e^{-4 i p_0} &0 &\tfrac{1}{2} &-i \sin(p_1) - \vphi_-
\end{pmatrix}.
\end{equation*}

It follows that

\begin{equation}
\label{eqn:24}
\begin{split}
&w^{\mathrm{SN}}_{\l}(\vphi_+,\vphi_-)= \lim_{L\to\infty} \tfrac{1}{|\L_L|} V^L_{\l}(\vphi_+,\vphi_-)= \frac{\vphi_+^2+ \vphi_-^2}{4\l}- \iint_{(-\frac{\pi}{4},\frac{\pi}{4}]\times(-\pi,\pi]} \frac{d^2p}{(2\pi)^2} \log \Delta^{\mathrm{SN}}(p;\vphi_+,\vphi_-);
\end{split}
\end{equation}

moreover, an explicit computation leads to 

\begin{equation}
\label{eqn:46a}
\Delta^{\mathrm{SN}}(p;\vphi_+,\vphi_-)=  \frac{1}{4}\sin^2 (2p_0)+ \frac{1}{4}(\vphi_++\vphi_-)^2 + (\vphi_+\vphi_-)^2 + \sin^2 (p_1) \big(1+ \vphi_+^2+ \vphi_-^2 +\sin^2 (p_1) \big).
\end{equation}

Note that due to the periodicity of $\Delta^{\mathrm{SN}}(\,\cdot\, ;\vphi_+,\vphi_-)$ under translations of $\frac{\pi}{2}\mbb{Z}\times \pi\mbb{Z}$, we can either rewrite $w_{\l}^{\mathrm{SN}}$ in the form of \eqref{eq:Peierls-w}, or alternatively

\begin{equation}
\label{eqn:24bis}
\begin{split}
&w^{\mathrm{SN}}_{\l}(\vphi_+,\vphi_-)= \frac{\vphi_+^2+ \vphi_-^2}{4\l}- 2\iint_{(-\frac{\pi}{4},\frac{\pi}{4}]\times(-\frac{\pi}{2},\frac{\pi}{2}]} \frac{d^2p}{(2\pi)^2} \log \Delta^{\mathrm{SN}}(p;\vphi_+,\vphi_-),
\end{split}
\end{equation}

so that within the latter domain of integration, if $\vphi_+=\vphi_-=0$, $\Delta^{\mathrm{SN}}(\,\cdot\, ;\vphi_+,\vphi_-)$ is singular only at $p=0$.

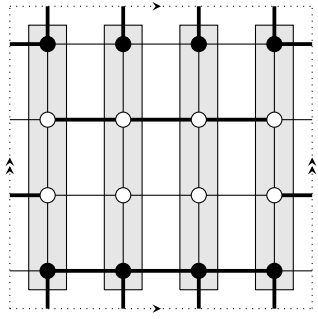
\begin{figure}
\centering
\begin{tikzpicture}
\draw[dotted, ->-] (-0.5,-0.5) -- (3.5,-0.5);
\draw[dotted, ->>-] (3.5,-0.5) -- (3.5,3.5);
\draw[dotted, ->-] (-0.5,3.5) -- (3.5,3.5);
\draw[dotted, ->>-] (-0.5,-0.5) -- (-0.5,3.5);
\draw[line width = 0.05 cm] (0, 0) -- (3,0);
\draw[line width = 0.05 cm] (0, 2) -- (3,2);
\draw[line width = 0.05 cm] (3,1) -- (3.5,1);
\draw[line width = 0.05 cm] (-0.5,1) -- (0,1);
\draw[line width = 0.05 cm] (-0.5,3) -- (0,3);
\draw[line width = 0.05 cm] (3,3) -- (3.5,3);
\draw (0,3) -- (3,3);
\draw (0,1) -- (3,1);
\draw (0,0) -- (0,3);
\draw (1,0) -- (1,3);
\draw (2,0) -- (2,3);
\draw (-0.5,2) -- (0,2);
\draw (3,2) -- (3.5,2);
\draw (-0.5,0) -- (0,0);
\draw (3,0) -- (3.5,0);
\draw (3,0) -- (3,3);
\draw[line width = 0.05 cm] (3,3) -- (3,3.5);
\draw[line width = 0.05 cm] (2,3) -- (2,3.5);
\draw[line width = 0.05 cm] (1,3) -- (1,3.5);
\draw[line width = 0.05 cm] (0,3) -- (0,3.5);
\draw[line width = 0.05 cm] (3,0) -- (3,-0.5);
\draw[line width = 0.05 cm] (2,0) -- (2,-0.5);
\draw[line width = 0.05 cm] (1,0) -- (1,-0.5);
\draw[line width = 0.05 cm] (0,0) -- (0,-.5);
      \begin{scope}[xshift=1cm]
         \draw[fill, fill opacity =0.1] (-0.25,-0.25) -- (0.25,-0.25) -- (0.25, 3.25) -- (-0.25,3.25) -- (-0.25,-0.25); 
    \end{scope}
     \begin{scope}[xshift=2cm]
         \draw[fill, fill opacity =0.1] (-0.25,-0.25) -- (0.25,-0.25) -- (0.25, 3.25) -- (-0.25,3.25) -- (-0.25,-0.25);  
    \end{scope}
     \begin{scope}[xshift=3cm]
        \draw[fill, fill opacity =0.1] (-0.25,-0.25) -- (0.25,-0.25) -- (0.25, 3.25) -- (-0.25,3.25) -- (-0.25,-0.25); 
    \end{scope}
     \begin{scope}
        \draw[fill, fill opacity =0.1] (-0.25,-0.25) -- (0.25,-0.25) -- (0.25, 3.25) -- (-0.25,3.25) -- (-0.25,-0.25); 
    \end{scope}
        \draw[black,fill=black] (0,0) circle (.1 cm);
         \draw[black,fill=black] (1,0) circle (.1 cm);
             \draw[black,fill=black] (2,0) circle (.1 cm);
              \draw[black,fill=black] (3,0) circle (.1 cm);
               \draw[black,fill=black] (0,3) circle (.1 cm);
                \draw[black,fill=black] (1,3) circle (.1 cm);
                 \draw[black,fill=black] (2,3) circle (.1 cm);
                  \draw[black,fill=black] (3,3) circle (.1 cm);
                   \draw[black,fill=white] (0,1) circle (.1 cm);
                   \draw[black,fill=white] (0,2) circle (.1 cm);
                   \draw[black,fill=white] (1,1) circle (.1 cm);
                   \draw[black,fill=white] (1,2) circle (.1 cm);
                   \draw[black,fill=white] (2,1) circle (.1 cm);
                   \draw[black,fill=white] (2,2) circle (.1 cm);
                   \draw[black,fill=white] (3,1) circle (.1 cm);
                   \draw[black,fill=white] (3,2) circle (.1 cm);                   
           \end{tikzpicture}
           \caption{Graphical Representation of the Dirac Operator with periodic pattern produced by Chessboard's Estimates. In Grey we have drawn the unit cells used to perform Fourier Transforms}
           \label{per}
       \end{figure}

We are interested in finding the minimizers of $w^{\mathrm{SN}}_{\l}$ and understanding their properties. By computing derivatives, we get:

\begin{equation*}
\begin{aligned}
\partial_{\vphi_+} w_{\lambda}^{\mathrm{SN}}(\vphi_+,\vphi_-) &= \frac{\vphi_+}{2\lambda}  - 2\iint_{(-\frac{\pi}{4},\frac{\pi}{4}]\times(-\frac{\pi}{2},\frac{\pi}{2}]} \frac{d^2p}{(2 \pi)^2} \frac{ \frac{1}{2}(\vphi_+ +\vphi_-) + 2 \vphi_+ \vphi_-^2 +2 \sin^2(p_1) \vphi_+}{\Delta^{\mathrm{SN}}(p;\vphi_+,\vphi_-)}, \\
\partial_{\vphi_-} w^{\mathrm{SN}}_{\lambda}(\vphi_+,\vphi_-) &= \frac{\vphi_-}{2\lambda} - 2\iint_{(-\frac{\pi}{4},\frac{\pi}{4}]\times(-\frac{\pi}{2},\frac{\pi}{2}]} \frac{d^2p}{(2 \pi)^2} \frac{ \frac{1}{2}(\vphi_+ +\vphi_-) + 2 \vphi_- \vphi_+^2 +2 \sin^2(p_1) \vphi_-}{\Delta^{\mathrm{SN}}(p;\vphi_+,\vphi_-)}.
\end{aligned}
\end{equation*}

Since we are looking for the minimizer over th constrained region $\{\vphi_+\ge0, \vphi_-\le0\}$, looking at the zeros of $\nabla w_{\l}$ is not enough. However, we observe the following:

\begin{equation}
\label{eqn:5}
\begin{split}
& (1,-1)\cdot \nabla w_{\l}^{\mathrm{SN}}(\vphi_+,\vphi_-)\equiv \de_{\vphi_+}w_{\l}^{\mathrm{SN}}(\vphi_+,\vphi_-)- \de_{\vphi_-}w_{\l}^{\mathrm{SN}}(\vphi_+,\vphi_-)=\\
& \frac{1}{2\lambda} (\vphi_+-\vphi_-) - 2\iint_{(-\frac{\pi}{4},\frac{\pi}{4}]\times(-\frac{\pi}{2},\frac{\pi}{2}]} \frac{d^2p}{(2 \pi)^2} \frac{-2(\vphi_+-\vphi_-) \vphi_+ \vphi_- +2(\vphi_+-\vphi_-) \sin^2(p_1) }{\Delta^{\mathrm{SN}}(p;\vphi_+,\vphi_-)}=\\
& (\vphi_+-\vphi_-) \left\{ \frac{1}{2\l} - 2\iint_{(-\frac{\pi}{4},\frac{\pi}{4}]\times(-\frac{\pi}{2},\frac{\pi}{2}]} \frac{d^2p}{(2 \pi)^2} \frac{-2\vphi_+ \vphi_- +2\sin^2(p_1) }{\Delta^{\mathrm{SN}}(p;\vphi_+,\vphi_-)} \right\}.
\end{split}\end{equation}

The expression in brace can be shown to be positive for $\l$ small enough, say $\l\le \l_0$ with $\l_0$ a suitable constant (recall also the discussion in Remark \ref{rem2}). Indeed, since

\begin{equation}
\label{eqn:6}
\Delta(p;\vphi_+,\vphi_-)\gtrsim m_{\mathrm{SN}}(\vphi_+,\vphi_-)^2 + |p|^2 \qquad \forall p\in \big(-\tfrac{\pi}{4},\tfrac{\pi}{4}\big]\times\big(-\tfrac{\pi}{2},\tfrac{\pi}{2}\big],\end{equation}

with $m_{\mathrm{SN}}(\vphi_+,\vphi_-)= |\vphi_+\vphi_-|$ (hereafter by writing $a\gtrsim b$ we mean that $a\ge \mf{C}b$, for some universal constant $\mf{C}>0$), it follows that

\begin{equation}
\begin{split}
\iint_{(-\frac{\pi}{4},\frac{\pi}{4}]\times(-\frac{\pi}{2},\frac{\pi}{2}]} \frac{d^2p}{(2 \pi)^2} \left|\frac{-2\vphi_+ \vphi_- +2\sin^2(p_1) }{\Delta^{\mathrm{SN}}(p;\phi_+,\phi_-)} \right|&\lesssim  \iint_{(-\frac{\pi}{4},\frac{\pi}{4}]\times(-\frac{\pi}{2},\frac{\pi}{2}]} \frac{d^2p}{(2 \pi)^2} \frac{|p|+m_{\mathrm{SN}}(\vphi_+,\vphi_-)}{|p|^2+m_{\mathrm{SN}}(\vphi_+,\vphi_-)^2}\\
& \lesssim \iint_{(-\frac{\pi}{4},\frac{\pi}{4}]\times(-\frac{\pi}{2},\frac{\pi}{2}]} \frac{d^2p}{(2 \pi)^2} \frac{1}{\sqrt{|p|^2+m_{\mathrm{SN}}(\vphi_+,\vphi_-)^2}}\\
&\lesssim 1.
\end{split}
\end{equation}

As a consequence, for $\l$ small enough the right--hand side of \eqref{eqn:5} is always non-negative on the domain of our interest, namely $\{\vphi_+\ge0,\vphi_-\le0\}$ and so is the quantity $(1,-1)\cdot \nabla w_{\l}(\vphi_+,\vphi_-)$. In other words, $w_{\l}(\vphi_+,\vphi_-)$ is minimized on the boundary of the region $\{\vphi_+\ge0,\vphi_-\le0\}$. Besides, since by inspection $w_{\l}(\vphi_+,\vphi_-)= w_{\l}(-\vphi_-,-\vphi_+)$, it suffices to look for the minima on the set $\{\vphi_+\ge0,\vphi_-=0\}$. In order to proceed, we note the following.

\begin{enumerate}
\item By construction, $w_{\l}^{\mathrm{SN}}(\vphi,\vphi)= \lim_{L\to\infty} \frac{1}{|\L_L|}  V^{\mathrm{SN}}_{L,\lambda}(\phi)\Big|_{\phi= \vphi^{\L^{\a}_L}} \equiv v_{\l}^{\mathrm{SN}}(\vphi)$, with $v^{\mathrm{SN}}_{\l}$ minimized by $\vphi=\vphi^{\star}_{\mathrm{SN}}(\l)$ (see Proposition \ref{const_field}).
\item Explicitly, we have that

\begin{align}
& \Delta^{\mathrm{SN}}(p;\vphi,\vphi)= \tfrac{1}{4}\sin^2 (2p_0)+ \vphi^2 + \vphi^4 + \sin^2 (p_1) \big(1+ 2\vphi^2 +\sin^2 (p_1) \big),\\
&\Delta^{\mathrm{SN}}(p;\vphi,0)= \tfrac{1}{4}\sin^2 (2p_0)+ \tfrac{1}{4}\vphi^2  + \sin^2 (p_1) \big(1+ \vphi^2 +\sin^2 (p_1) \big),
\end{align}

from which we see that $\Delta^{\mathrm{SN}}(p;\vphi,0) = \Delta^{\mathrm{SN}}\big(p;\tfrac{1}{\sqrt{2}}\vphi, \tfrac{1}{\sqrt{2}}\vphi\big)- \tfrac{1}{4}\vphi^2- \tfrac{1}{4}\vphi^4$.
\end{enumerate}

Therefore we have that 

\[\begin{split}
w_{\l}^{\mathrm{SN}}(\vphi,0)&=   \frac{\vphi^2}{4\l}- 2\iint_{(-\frac{\pi}{4},\frac{\pi}{4}]\times(-\frac{\pi}{2},\frac{\pi}{2}]} \frac{d^2p}{(2\pi)^2} \log \Delta^{\mathrm{SN}}(p;\vphi,0)\\
&= \frac{\vphi^2}{4\l}- 2\iint_{(-\frac{\pi}{4},\frac{\pi}{4}]\times(-\frac{\pi}{2},\frac{\pi}{2}]} \frac{d^2p}{(2\pi)^2} \log\Big( \Delta^{\mathrm{SN}}\big(p;\tfrac{1}{\sqrt{2}}\vphi,\tfrac{1}{\sqrt{2}}\vphi\big) - \tfrac{1}{4}\vphi^2(1+ \vphi^2)\Big)\\
&= \frac{\vphi^2}{4\l}- 2\iint_{(-\frac{\pi}{4},\frac{\pi}{4}]\times(-\frac{\pi}{2},\frac{\pi}{2}]} \frac{d^2p}{(2\pi)^2} \log \Delta^{\mathrm{SN}}\big(p;\tfrac{1}{\sqrt{2}}\vphi,\tfrac{1}{\sqrt{2}}\vphi\big) \\
&\,\,\,\,\,\,\,\,\,\,\,\,\,\,\,\,\,\,-  2\iint_{(-\frac{\pi}{4},\frac{\pi}{4}]\times(-\frac{\pi}{2},\frac{\pi}{2}]} \frac{d^2p}{(2\pi)^2} \log\left( 1 - \frac{\tfrac{1}{4}\vphi^2(1+ \vphi^2)}{\Delta^{\mathrm{SN}}\big(p;\tfrac{1}{\sqrt{2}}\vphi,\tfrac{1}{\sqrt{2}}\vphi\big)} \right)\\
&= v_{\l}^{\mathrm{SN}}\big( \tfrac{1}{\sqrt{2}}\vphi\big)- 2\iint_{(-\frac{\pi}{4},\frac{\pi}{4}]\times(-\frac{\pi}{2},\frac{\pi}{2}]} \frac{d^2p}{(2\pi)^2} \log\left( 1 - \frac{\tfrac{1}{4}\vphi^2(1+\vphi^2)}{\Delta^{\mathrm{SN}}\big(p;\tfrac{1}{\sqrt{2}}\vphi,\tfrac{1}{\sqrt{2}}\vphi\big)} \right)\\
&\doteq v_{\l}^{\mathrm{SN}}\big( \tfrac{1}{\sqrt{2}}\vphi\big) + \d w^{\mathrm{SN}}(\vphi).
\end{split}\]

Observe that $\Delta^{\mathrm{SN}}(p;\vphi,\vphi)\lesssim 1+ \vphi^2+\vphi^4\le (1+\vphi^2)^2$, therefore:

\begin{equation}
\label{eqn:30}
\begin{split}
\d w^{\mathrm{SN}}(\vphi)& \equiv -2 \iint_{(-\frac{\pi}{4},\frac{\pi}{4}]\times(-\frac{\pi}{2},\frac{\pi}{2}]} \frac{d^2p}{(2\pi)^2} \log\left( 1 - \frac{\tfrac{1}{4}\vphi^2(1+\vphi^2)}{\Delta^{\mathrm{SN}}\big(p;\tfrac{1}{\sqrt{2}}\vphi,\tfrac{1}{\sqrt{2}}\vphi\big)} \right) \\
&\ge - \frac{1}{4}\log\left( 1 -  \mf{C}_2\frac{\vphi^2(1+\vphi^2)}{\big( 1+ \vphi^2\big)^2} \right) = - \frac{1}{4}\log \left(1- \mf{C}_2\frac{\vphi^2}{1+\vphi^2} \right), 
\end{split}\end{equation}

for a suitable constant $\mf{C}_2^{\mathrm{SN}}>0$. Note that the right--hand side of \eqref{eqn:30} is positive and monotone increasing for $\vphi>0$. We deduce the following:

\begin{itemize}
\item if $\vphi\in\big[0,\varphi^{\star}_{\mathrm{SN}}(\lambda)\big]$, we have that  $w_{\l}^{\mathrm{SN}}(\vphi,0)\ge  v_{\l}^{\mathrm{SN}}\big(\frac{1}{\sqrt{2}}\vphi\big)\ge v_{\l}^{\mathrm{SN}}\big(\tfrac{1}{\sqrt{2}}\vphi_{\mathrm{SN}}^{\star}(\l)\big)$ (recall from Proposition \ref{const_field} the form of $v_{\l}^{\mathrm{SN}}$);
\item if $\vphi\in\big[\varphi^{\star}_{\mathrm{SN}}(\lambda),\infty\big)$, then 

$$w_{\l}^{\mathrm{SN}}(\vphi,0) \ge - \frac{1}{4} \log \left(1- \mf{C}_2^{\mathrm{SN}} \tfrac{\vphi^2}{1+ \vphi^2} \right) +  v_{\l}^{\mathrm{SN}}\big(\tfrac{1}{\sqrt{2}}\vphi\big)\ge - \frac{1}{4} \log \left(1- \mf{C}_2^{\mathrm{SN}} \tfrac{\big(\varphi^{\star}_{\mathrm{SN}}(\lambda)\big)^2}{1+ \big(\varphi^{\star}_{\mathrm{SN}}(\lambda)\big)^2} \right) +  v_{\l}^{\mathrm{SN}}\big(\vphi_{\mathrm{SN}}^{\star}(\l)\big).$$

\end{itemize}

We conclude that

\[\begin{split} &\min_{\vphi_+\ge0,\vphi_-\le0} w_{\l}^{\mathrm{SN}}(\vphi_+,\vphi_-) -  v_{\l}^{\mathrm{SN}}\big(\vphi_{\mathrm{SN}}^{\star}(\l)\big) =\min_{\vphi\ge0} w_{\l}^{\mathrm{SN}}(\vphi,0) -  v_{\l}^{\mathrm{SN}}\big(\vphi_{\mathrm{SN}}^{\star}(\l)\big) \\
&\ge \min\left\{  v_{\l}^{\mathrm{SN}}\big(\tfrac{1}{\sqrt{2}}\vphi_{\mathrm{SN}}^{\star}(\l)\big)- v_{\l}^{\mathrm{SN}}\big(\vphi_{\mathrm{SN}}^{\star}(\l)\big), \, - \frac{1}{4} \log \left(1- \mf{C}_2^{\mathrm{SN}}\tfrac{\big(\varphi^{\star}_{\mathrm{SN}}(\lambda)\big)^2}{1+ \big(\varphi^{\star}_{\mathrm{SN}}(\lambda)\big)^2} \right)  \right\}\equiv c_3^{\mathrm{SN}}(\l). \end{split}\]

\paragraph{Proof for the NN model.}

Proceeding in analogy with the case of the SN model (see Fig. \ref{diracnaivestag} for a graphical understanding of the field configuration), we find that

\begin{equation}
\label{eqn:46b}
w_{\l}^{\mathrm{NN}}(\vphi_+,\vphi_-)= \frac{\vphi_+^2+\vphi_-^2}{4\l}- \iint_{(-\frac{\pi}{4},\frac{\pi}{4}]\times(-\pi,\pi]} \frac{d^2p}{(2\pi)^2} \log \Delta^{\mathrm{NN}}(p;\vphi_+,\vphi_-),
\end{equation}

where:

\begin{equation}
\label{eq:myNN1}
\begin{split}
&\Delta^{\mathrm{NN}}(p;\vphi_+,\vphi_-)=\\
&=\det \left(\!\!\begin{array}{cccccccc}
    -\vphi_- & -i\sin p_1&0 &-\frac{i}{2} &0 &0 &0 &\frac{i}{2}e^{4ip_0} \\
    -i\sin p_1 & -\vphi_- &\frac{i}{2}&0&0&0&-\frac{i}{2} e^{4i p_0}&0\\
    0&\frac{i}{2}&-\vphi_+ &-i\sin p_1&0&-\frac{i}{2}&0&0\\
    -\frac{i}{2}&0&-i\sin p_1&-\vphi_+&\frac{i}{2}&0&0&0\\
    0&0&0&\frac{i}{2}&-\vphi_+&-i\sin p_1&0&-\frac{i}{2}\\
    0&0&-\frac{i}{2}&0&-i\sin p_1&-\vphi_+&\frac{i}{2}&0\\
    0&-\frac{i}{2} e^{-4i p_0}&0&0&0&\frac{i}{2}&-\vphi_-&-i\sin p_1\\
    \frac{i}{2} e^{-4i p_0}&0&0&0&-\frac{i}{2}&0&-i\sin p_1&-\vphi_-
\end{array} \!\!\right)\\
&=\frac{1}{64} \big[8+ 4\vphi_+\vphi_- + 6\vphi_+^2+ 6\vphi_-^2 + 8\vphi_+^2\vphi_-^2 - \cos(4p_0) -4\big(2+\vphi_+^2+ \vphi_-^2 \big) \cos(2p_1) + \cos(4p_1)\big]^2.
\end{split}
\end{equation}

\begin{figure}
\centering
\begin{tikzpicture}
\draw[dotted, ->-] (-0.5,-0.5) -- (3.5,-0.5);
\draw[dotted, ->>-] (3.5,-0.5) -- (3.5,3.5);
\draw[dotted, ->-] (-0.5,3.5) -- (3.5,3.5);
\draw[dotted, ->>-] (-0.5,-0.5) -- (-0.5,3.5);

\draw (0, 0) -- (3,0);
\draw (0, 2) -- (3,2);
\draw (3,1) -- (3.5,1);
\draw (-0.5,1) -- (0,1);
\draw (-0.5,3) -- (0,3);
\draw (3,3) -- (3.5,3);
\draw (0,3) -- (3,3);
\draw (0,1) -- (3,1);
\draw (0,0) -- (0,3);
\draw (1,0) -- (1,3);
\draw (2,0) -- (2,3);
\draw (-0.5,2) -- (0,2);
\draw (3,2) -- (3.5,2);
\draw (-0.5,0) -- (0,0);
\draw (3,0) -- (3.5,0);
\draw (3,0) -- (3,3);

\draw[line width = 0.05 cm] (3,3) -- (3,3.5);
\draw[line width = 0.05 cm] (2,3) -- (2,3.5);
\draw[line width = 0.05 cm] (1,3) -- (1,3.5);
\draw[line width = 0.05 cm] (0,3) -- (0,3.5);
\draw[line width = 0.05 cm] (3,0) -- (3,-0.5);
\draw[line width = 0.05 cm] (2,0) -- (2,-0.5);
\draw[line width = 0.05 cm] (1,0) -- (1,-0.5);
\draw[line width = 0.05 cm] (0,0) -- (-0.5,0);
\draw[line width = 0.05 cm] (0,0) -- (0,-0.5);
\draw[line width = 0.05 cm] (0,1) -- (-0.5,1);
\draw[line width = 0.05 cm] (0,2) -- (-0.5,2);
\draw[line width = 0.05 cm] (0,3) -- (-0.5,3);
\draw[line width = 0.05 cm] (3,0) -- (3.5,0);
\draw[line width = 0.05 cm] (3,1) -- (3.5,1);
\draw[line width = 0.05 cm] (3,2) -- (3.5,2);
\draw[line width = 0.05 cm] (3,3) -- (3.5,3); 
      \begin{scope}[xshift=1cm]
         \draw[fill, fill opacity =0.1] (-0.25,-0.25) -- (0.25,-0.25) -- (0.25, 3.25) -- (-0.25,3.25) -- (-0.25,-0.25); 
    \end{scope}
     \begin{scope}[xshift=2cm]
         \draw[fill, fill opacity =0.1] (-0.25,-0.25) -- (0.25,-0.25) -- (0.25, 3.25) -- (-0.25,3.25) -- (-0.25,-0.25);  
    \end{scope}
     \begin{scope}[xshift=3cm]
        \draw[fill, fill opacity =0.1] (-0.25,-0.25) -- (0.25,-0.25) -- (0.25, 3.25) -- (-0.25,3.25) -- (-0.25,-0.25); 
    \end{scope}
     \begin{scope}
        \draw[fill, fill opacity =0.1] (-0.25,-0.25) -- (0.25,-0.25) -- (0.25, 3.25) -- (-0.25,3.25) -- (-0.25,-0.25); 
    \end{scope}
        \draw[black,fill=black] (0,0) circle (.1 cm);
         \draw[black,fill=black] (1,0) circle (.1 cm);
             \draw[black,fill=black] (2,0) circle (.1 cm);
              \draw[black,fill=black] (3,0) circle (.1 cm);
               \draw[black,fill=black] (0,3) circle (.1 cm);
                \draw[black,fill=black] (1,3) circle (.1 cm);
                 \draw[black,fill=black] (2,3) circle (.1 cm);
                  \draw[black,fill=black] (3,3) circle (.1 cm);
                   \draw[black,fill=white] (0,1) circle (.1 cm);
                   \draw[black,fill=white] (0,2) circle (.1 cm);
                   \draw[black,fill=white] (1,1) circle (.1 cm);
                   \draw[black,fill=white] (1,2) circle (.1 cm);
                   \draw[black,fill=white] (2,1) circle (.1 cm);
                   \draw[black,fill=white] (2,2) circle (.1 cm);
                   \draw[black,fill=white] (3,1) circle (.1 cm);
                   \draw[black,fill=white] (3,2) circle (.1 cm);                   
           \end{tikzpicture}
           \caption{Graphical Representation of the Naive Dirac Operator with periodic pattern produced by Chessboard Estimates. In Grey we have drawn the unit cells used to perform Fourier Transforms}
           \label{diracnaivestag}
       \end{figure}

Note that due to the periodicity of $\Delta^{\mathrm{NN}}(\,\cdot\,;\vphi_+,\vphi_-)$ under translations of $\frac{\pi}{2}\mbb{Z}\times \pi\mbb{Z}$, we can either rewrite $w_{\l}^{\mathrm{NN}}$ in the form of \eqref{eq:Peierls-w}, or alternatively, by also writing $\Delta^{\mathrm{NN}}(p;\vphi_+,\vphi_-)= \big(\sqrt{\Delta^{\mathrm{NN}}}(p;\vphi_+,\vphi_-)\big)^2$,

\begin{equation}
w_{\l}^{\mathrm{NN}}(\vphi_+,\vphi_-)= \frac{\vphi_+^2+\vphi_-^2}{4\l}- 4\iint_{(-\frac{\pi}{4},\frac{\pi}{4}]\times(-\frac{\pi}{2},\frac{\pi}{2}]} \frac{d^2p}{(2\pi)^2} \log \sqrt{\Delta^{\mathrm{NN}}}(p;\vphi_+,\vphi_-),
\end{equation}

so that, on the latter domain, the integrand has only one singularity at $p=0$, if $\vphi_+=\vphi_-=0$. Some basic but cumbersome algebraic manipulations lead to the following rewriting:

\begin{equation}
\label{eq:myNN4}
\begin{split}
&\sqrt{\Delta^{\mathrm{NN}}}(p;\vphi_+,\vphi_-)= \frac{1}{4}(\vphi_+ + \vphi_-)^2 + \vphi_+^2\vphi_-^2 + \frac{1}{4}\sin^2(2p_0) + \big(1+\vphi_+^2+ \vphi_-^2 + \sin^2 (p_1)\big) \sin^2(p_1).
\end{split}
\end{equation}

We are going to adapt the same arguments used in the proof for the SN model. As a first step, we shall show that $\de_{\vphi_+}w_{\l}^{\mathrm{NN}}(\vphi_+,\vphi_-)- \de_{\vphi_-}w_{\l}^{\mathrm{NN}}(\vphi_+,\vphi_-)\ge0$, implying that the minimum of $w_{\l}^{\mathrm{NN}}$ is attained on the set $\{\vphi_+=0\}\cup\{\vphi_-=0\}$. Since:

\begin{equation*}
\begin{split}
&\de_{\vphi_+}w_{\l}^{\mathrm{NN}}(\vphi_+,\vphi_-)- \de_{\vphi_-}w_{\l}^{\mathrm{NN}}(\vphi_+,\vphi_-)\\
&= \frac{\vphi_+-\vphi_-}{2\l}- 4 \iint_{(-\frac{\pi}{4},\frac{\pi}{4}]\times(-\frac{\pi}{2},\frac{\pi}{2}]} \frac{d^2p}{(2\pi)^2} \frac{\de_{\vphi_+}\sqrt{\Delta^{\mathrm{NN}}}(p;\vphi_+,\vphi_-)- \de_{\vphi_-}\sqrt{\Delta^{\mathrm{NN}}}(p;\vphi_+,\vphi_-)}{\sqrt{\Delta^{\mathrm{NN}}}(p;\vphi_+,\vphi_-)},
\end{split}\end{equation*}

and, as an explicit computation shows,

\begin{equation}
\label{eq:myNN2}
\begin{split}
\big(\de_{\vphi_+}-\de_{\vphi_-}\big)\sqrt{\Delta^{\mathrm{NN}}}(p;\vphi_+,\vphi_-) = 2(\vphi_+-\vphi_-) \big( -\vphi_+\vphi_- + \sin^2(p_1)\big),
\end{split}
\end{equation}

we can write:

\begin{equation}
\label{eq:myNN3}
\begin{split}
&\de_{\vphi_+}w_{\l}^{\mathrm{NN}}(\vphi_+,\vphi_-)- \de_{\vphi_-}w_{\l}^{\mathrm{NN}}(\vphi_+,\vphi_-)= (\vphi_+-\vphi_-)\left\{\frac{1}{2\l}- 8 \iint_{(-\frac{\pi}{4},\frac{\pi}{4}]\times(-\frac{\pi}{2},\frac{\pi}{2}]} \frac{d^2p}{(2\pi)^2} \frac{-\vphi_+\vphi_- + \sin^2p_1}{\sqrt{\Delta^{\mathrm{NN}}}(p;\vphi_+,\vphi_-)} \right\}.
\end{split}\end{equation}

Notice that for $p\in (-\frac{\pi}{4},\frac{\pi}{4}]\times (-\frac{\pi}{2},\frac{\pi}{2}]$, we have that

\begin{equation}
\label{eqn:6b}
\sqrt{\Delta^{\mathrm{NN}}}(p;\vphi_+,\vphi_-) \gtrsim |p|^2+ m_{\mathrm{NN}}(\vphi_+,\vphi_-)^2,
\end{equation}

with $m_{\mathrm{NN}}(\vphi_+,\vphi_-)\doteq \sqrt{(\vphi_+ + \vphi_-)^2 + \vphi_+^2\vphi_-^2}$. Notice that it furthermore holds
$|\vphi_+\vphi_-|\le m_{\mathrm{NN}}(\vphi_+,\vphi_-)$. It follows that the integral in the right--hand side of \eqref{eq:myNN3} is uniformly bounded:

\begin{equation*}
\begin{split}
 \iint_{(-\frac{\pi}{4},\frac{\pi}{4}]\times(-\frac{\pi}{2},\frac{\pi}{2}]} \frac{d^2p}{(2\pi)^2} \left|\frac{-\vphi_+\vphi_- + \sin^2p_1}{\sqrt{\Delta^{\mathrm{NN}}}(p;\vphi_+,\vphi_-)}\right|& \lesssim \iint_{(-\frac{\pi}{4},\frac{\pi}{4}]\times(-\frac{\pi}{2},\frac{\pi}{2}]} \frac{d^2p}{(2\pi)^2} \frac{m_{\mathrm{NN}}(\vphi_+,\vphi_-)+ |p|^2}{m_{\mathrm{NN}}(\vphi_+,\vphi_-)^2+ |p|^2}\\
&\lesssim \iint_{(-\frac{\pi}{4},\frac{\pi}{4}]\times(-\frac{\pi}{2},\frac{\pi}{2}]} \frac{d^2p}{(2\pi)^2} \left(\frac{1}{\sqrt{|p|^2 + m_{\mathrm{NN}}(\vphi_+,\vphi_-)^2}}+ 1 \right) \\
&\lesssim 1.
\end{split}
\end{equation*}

As an immediate consequence, from \eqref{eq:myNN3} we have that for $\l$ small enough, $(1,-1)\cdot \nabla w_{\l}^{\mathrm{NN}}(\vphi_+,\vphi_-)\ge0$ on the domain $\{\vphi_+\ge0,\vphi_-\le0\}$. Using also the symmetry $w_{\l}^{\mathrm{NN}}(\vphi_+,\vphi_-)= w_{\l}^{\mathrm{NN}}(-\vphi_-,-\vphi_+)$, it follows that

\begin{equation*}
\min_{\vphi_+\ge0,\vphi_-\le0} w^{\mathrm{NN}}_{\l}(\vphi_+,\vphi_-)= \min_{\vphi\ge0} w_{\l}^{\mathrm{NN}}(\vphi,0).
\end{equation*}

Recalling that, by construction, $w_{\l}^{\mathrm{NN}}(\vphi,\vphi)= v_{\l}^{\mathrm{NN}}(\vphi)$; it is then natural to compare $w_{\l}^{\mathrm{NN}}(\vphi,0)$ with $w_{\l}^{\mathrm{NN}}\big(\frac{1}{\sqrt{2}}\vphi,\frac{1}{\sqrt{2}}\vphi\big)$, in the same way as for the SN model. From \eqref{eq:myNN4}, it is straightforward to check that

\begin{equation}
\label{eq:myNN5}
\begin{split}
\sqrt{\Delta^{\mathrm{NN}}}\big(p;\tfrac{1}{\sqrt{2}}\vphi,\tfrac{1}{\sqrt{2}}\vphi\big)&= \sqrt{\Delta^{\mathrm{NN}}}(p;\vphi,0)+ \tfrac{1}{4} \big(\vphi^2+\vphi^4\big).
\end{split}
\end{equation}

This allows us to rewrite:

\begin{equation*}
\begin{split}
&w_{\l}^{\mathrm{NN}}(\vphi,0)= \frac{\vphi^2}{4\l}- 4\iint_{(-\frac{\pi}{4},\frac{\pi}{4}]\times(-\frac{\pi}{2},\frac{\pi}{2}]} \frac{d^2p}{(2\pi)^2} \log \sqrt{\Delta^{\mathrm{NN}}}(p;\vphi,0)=\\
& w_{\l}^{\mathrm{NN}}\big(\tfrac{1}{\sqrt{2}}\vphi, \tfrac{1}{\sqrt{2}}\vphi\big) - 4\iint \frac{d^2p}{(2\pi)^2} \log \left(1- \frac{\vphi^2+\vphi^4}{4\sqrt{\Delta^{\mathrm{NN}}}\big(p;\tfrac{1}{\sqrt{2}}\vphi,\tfrac{1}{\sqrt{2}}\vphi\big)} \right) \doteq v_{\l}\big(\tfrac{1}{\sqrt{2}}\vphi\big)+ \d w^{\mathrm{NN}}(\vphi).
\end{split}
\end{equation*}

Noting that $\sqrt{\Delta^{\mathrm{NN}}}\big(p;\tfrac{1}{\sqrt{2}}\vphi,\tfrac{1}{\sqrt{2}}\vphi\big) \lesssim  1+ \vphi^4$, it is easy to get the following lower bound for $\d w^{\mathrm{NN}}$:

\begin{equation}
\label{eq:myNN6}
\d w^{\mathrm{NN}}(\vphi)\ge - 4 \iint_{(-\frac{\pi}{4},\frac{\pi}{4}]\times(-\frac{\pi}{2},\frac{\pi}{2}]} \frac{d^2p}{(2\pi)^2} \log\left( 1- \mf{C}_2^{\mathrm{NN}} \frac{\vphi^4}{1+\vphi^4}\right)= - \frac{1}{2} \log\left( 1- \mf{C}_2^{\mathrm{NN}} \frac{\vphi^4}{1+\vphi^4}\right),
\end{equation}

for a suitable constant $\mf{C}_2^{\mathrm{NN}}>0$. Note that the right--hand side of \eqref{eq:myNN6} is positive and monotone increasing for $\vphi>0$. All in all:

\begin{equation*}
w_{\l}^{\mathrm{NN}}(\vphi,0) \ge v_{\l}^{\mathrm{NN}} \big(\tfrac{1}{\sqrt{2}}\vphi \big)- \frac{1}{2} \log\left( 1- \mf{C}_2^{\mathrm{NN}} \frac{\vphi^4}{1+\vphi^4}\right).
\end{equation*}

By replicating the same reasoning presented at the end of the previous paragraph, we find that

\[\begin{split} &\min_{\vphi_+\ge0,\vphi_-\le0} w_{\l}^{\mathrm{NN}}(\vphi_+,\vphi_-) -  v_{\l}^{\mathrm{NN}}\big(\vphi_{\mathrm{NN}}^{\star}(\l)\big) \ge \\
& \min\left\{  v_{\l}^{\mathrm{NN}}\big(\tfrac{1}{\sqrt{2}}\vphi_{\mathrm{NN}}^{\star}(\l)\big)- v_{\l}^{\mathrm{NN}}\big(\vphi_{\mathrm{NN}}^{\star}(\l)\big), \, - \frac{1}{4} \log \left(1- \mf{C}_2^{\mathrm{NN}}\frac{\big(\varphi^{\star}_{\mathrm{NN}}(\lambda)\big)^4}{1+ \big(\varphi^{\star}_{\mathrm{NN}}(\lambda)\big)^4} \right)  \right\}\equiv c_3^{\mathrm{NN}}(\l). \end{split}\]

\paragraph{Proof for the SP model.} 

Again we follow the same ideas used for the SN model (see Fig. \ref{dirastagplaq} for a graphical understanding of the field configuration). An explicit computation leads to

\begin{equation}
w_{\l}^{\mathrm{SP}}(\vphi_+,\vphi_-)= \frac{\vphi_+^2+\vphi_-^2}{4\l}- 4\iint_{(-\frac{\pi}{8},\frac{\pi}{8}]\times(-\pi,\pi]} \frac{d^2p}{(2\pi)^2} \log \Delta^{\mathrm{SP}}(p;\vphi_+,\vphi_-),
\end{equation}

where

\begin{equation}
\label{eq:mySP1}
\begin{split}
&\Delta^{\mathrm{SP}}(p;\vphi_+,\vphi_-)=\\
&\frac{1}{256}\Big\{ 16\big[\vphi_+^2+\vphi_+^4+ 2\vphi_-(\vphi_++2\vphi_+^3)+ 4\vphi_-^3(\vphi_++2\vphi_+^3)+ \vphi_-^2(1+10\vphi_+^2+12\vphi_+^4)+ \\
&+\vphi_-^4(1+12\vphi_+^2+ 16\vphi_+^4)\big] + 4\sin^2(4p_0)+ 64\sin^2p_1 \big(1+ \vphi_-^2+ \vphi_+^2+ \sin^2p_1 \big)\cdot\\
&\cdot \big[  1+ 2\vphi_-\vphi_+ + 3\vphi_+^2 + \vphi_-^2(3+8\vphi_+^2)+ 4\sin^2p_1\big(1+\vphi_-^2+ \vphi_+^2 + \sin^2p_1\big)\big]\Big\}.
\end{split}
\end{equation}

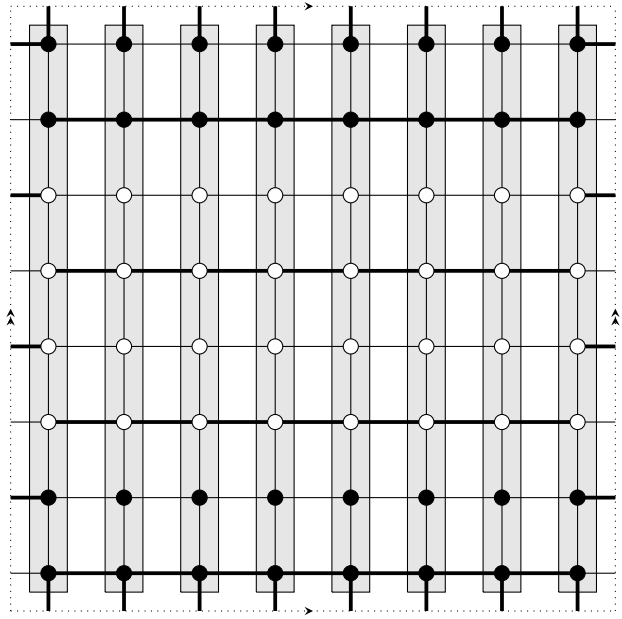
\begin{figure}
\centering
\begin{tikzpicture}
\draw[dotted, ->-] (-0.5,-0.5) -- (7.5,-0.5);
\draw[dotted, ->>-] (7.5,-0.5) -- (7.5,7.5);
\draw[dotted, ->-] (-0.5,7.5) -- (7.5,7.5);
\draw[dotted, ->>-] (-0.5,-0.5) -- (-0.5,7.5);

\draw (-0.5, 0) -- (7.5,0);
\draw (-0.5, 1) -- (7.5,1);
\draw (-0.5, 2) -- (7.5,2);
\draw (-0.5, 3) -- (7.5,3);
\draw (-0.5, 4) -- (7.5,4);
\draw (-0.5, 5) -- (7.5,5);
\draw (-0.5, 6) -- (7.5,6);
\draw (-0.5, 7) -- (7.5,7);

\draw (0,-0.5) -- (0,7.5);
\draw (1,-0.5) -- (1,7.5);
\draw (2,-0.5) -- (2,7.5);
\draw (3,-0.5) -- (3,7.5);
\draw (4,-0.5) -- (4,7.5);
\draw (5,-0.5) -- (5,7.5);
\draw (6,-0.5) -- (6,7.5);
\draw (7,-0.5) -- (7,7.5);
\draw[line width = 0.05 cm] (0,0) -- (7,0);
\draw[line width = 0.05 cm] (0,2) -- (7,2);
\draw[line width = 0.05 cm] (0,4) -- (7,4);
\draw[line width = 0.05 cm] (0,6) -- (7,6);
\draw[line width = 0.05 cm] (-0.5,1) -- (0,1);
\draw[line width = 0.05 cm] (-0.5,3) -- (0,3);
\draw[line width = 0.05 cm] (-0.5,5) -- (0,5);
\draw[line width = 0.05 cm] (-0.5,7) -- (0,7);

\draw[line width = 0.05 cm] (7,1) -- (7.5,1);
\draw[line width = 0.05 cm] (7,3) -- (7.5,3);
\draw[line width = 0.05 cm] (7,5) -- (7.5,5);
\draw[line width = 0.05 cm] (7,7) -- (7.5,7);

\draw[line width = 0.05 cm] (0,-0.5) -- (0,0);
\draw[line width = 0.05 cm] (1,-0.5) -- (1,0);
\draw[line width = 0.05 cm] (2,-0.5) -- (2,0);
\draw[line width = 0.05 cm] (3,-0.5) -- (3,0);
\draw[line width = 0.05 cm] (4,-0.5) -- (4,0);
\draw[line width = 0.05 cm] (5,-0.5) -- (5,0);
\draw[line width = 0.05 cm] (6,-0.5) -- (6,0);
\draw[line width = 0.05 cm] (7,-0.5) -- (7,0);

\draw[line width = 0.05 cm] (0,7.5) -- (0,7);
\draw[line width = 0.05 cm] (1,7.5) -- (1,7);
\draw[line width = 0.05 cm] (2,7.5) -- (2,7);
\draw[line width = 0.05 cm] (3,7.5) -- (3,7);
\draw[line width = 0.05 cm] (4,7.5) -- (4,7);
\draw[line width = 0.05 cm] (5,7.5) -- (5,7);
\draw[line width = 0.05 cm] (6,7.5) -- (6,7);
\draw[line width = 0.05 cm] (7,7.5) -- (7,7);

      \begin{scope}[xshift=1cm]
         \draw[fill, fill opacity =0.1] (-0.25,-0.25) -- (0.25,-0.25) -- (0.25, 7.25) -- (-0.25,7.25) -- (-0.25,-0.25); 
    \end{scope}
     \begin{scope}[xshift=2cm]
            \draw[fill, fill opacity =0.1] (-0.25,-0.25) -- (0.25,-0.25) -- (0.25, 7.25) -- (-0.25,7.25) -- (-0.25,-0.25); 
    \end{scope}
     \begin{scope}[xshift=3cm]
           \draw[fill, fill opacity =0.1] (-0.25,-0.25) -- (0.25,-0.25) -- (0.25, 7.25) -- (-0.25,7.25) -- (-0.25,-0.25); 
    \end{scope}
     \begin{scope}[xshift=4cm]
           \draw[fill, fill opacity =0.1] (-0.25,-0.25) -- (0.25,-0.25) -- (0.25, 7.25) -- (-0.25,7.25) -- (-0.25,-0.25); 
    \end{scope}
     \begin{scope}[xshift=5cm]
           \draw[fill, fill opacity =0.1] (-0.25,-0.25) -- (0.25,-0.25) -- (0.25, 7.25) -- (-0.25,7.25) -- (-0.25,-0.25); 
    \end{scope}
     \begin{scope}[xshift=6cm]
           \draw[fill, fill opacity =0.1] (-0.25,-0.25) -- (0.25,-0.25) -- (0.25, 7.25) -- (-0.25,7.25) -- (-0.25,-0.25); 
    \end{scope}
     \begin{scope}[xshift=7cm]
           \draw[fill, fill opacity =0.1] (-0.25,-0.25) -- (0.25,-0.25) -- (0.25, 7.25) -- (-0.25,7.25) -- (-0.25,-0.25); 
    \end{scope}
     \begin{scope}
         \draw[fill, fill opacity =0.1] (-0.25,-0.25) -- (0.25,-0.25) -- (0.25, 7.25) -- (-0.25,7.25) -- (-0.25,-0.25); 
    \end{scope}
        \draw[black,fill=black] (0,0) circle (.1 cm);
         \draw[black,fill=black] (1,0) circle (.1 cm);
             \draw[black,fill=black] (1,1) circle (.1 cm);
              \draw[black,fill=black] (0,1) circle (.1 cm);
              \begin{scope}[xshift=2cm]
                  \draw[black,fill=black] (0,0) circle (.1 cm);
         \draw[black,fill=black] (1,0) circle (.1 cm);
             \draw[black,fill=black] (1,1) circle (.1 cm);
              \draw[black,fill=black] (0,1) circle (.1 cm);
              \end{scope}
               \begin{scope}[xshift=4cm]
                  \draw[black,fill=black] (0,0) circle (.1 cm);
         \draw[black,fill=black] (1,0) circle (.1 cm);
             \draw[black,fill=black] (1,1) circle (.1 cm);
              \draw[black,fill=black] (0,1) circle (.1 cm);
              \end{scope}
                   \begin{scope}[xshift=6cm]
                  \draw[black,fill=black] (0,0) circle (.1 cm);
         \draw[black,fill=black] (1,0) circle (.1 cm);
             \draw[black,fill=black] (1,1) circle (.1 cm);
              \draw[black,fill=black] (0,1) circle (.1 cm);
              \end{scope}

\begin{scope}[yshift=6cm]
       \draw[black,fill=black] (0,0) circle (.1 cm);
         \draw[black,fill=black] (1,0) circle (.1 cm);
             \draw[black,fill=black] (1,1) circle (.1 cm);
              \draw[black,fill=black] (0,1) circle (.1 cm);
              \begin{scope}[xshift=2cm]
                  \draw[black,fill=black] (0,0) circle (.1 cm);
         \draw[black,fill=black] (1,0) circle (.1 cm);
             \draw[black,fill=black] (1,1) circle (.1 cm);
              \draw[black,fill=black] (0,1) circle (.1 cm);
              \end{scope}
               \begin{scope}[xshift=4cm]
                  \draw[black,fill=black] (0,0) circle (.1 cm);
         \draw[black,fill=black] (1,0) circle (.1 cm);
             \draw[black,fill=black] (1,1) circle (.1 cm);
              \draw[black,fill=black] (0,1) circle (.1 cm);
              \end{scope}
                   \begin{scope}[xshift=6cm]
                  \draw[black,fill=black] (0,0) circle (.1 cm);
         \draw[black,fill=black] (1,0) circle (.1 cm);
             \draw[black,fill=black] (1,1) circle (.1 cm);
              \draw[black,fill=black] (0,1) circle (.1 cm);
              \end{scope}
\end{scope}

\begin{scope}[yshift=4cm]
       \draw[black,fill=white] (0,0) circle (.1 cm);
         \draw[black,fill=white] (1,0) circle (.1 cm);
             \draw[black,fill=white] (1,1) circle (.1 cm);
              \draw[black,fill=white] (0,1) circle (.1 cm);
              \begin{scope}[xshift=2cm]
                  \draw[black,fill=white] (0,0) circle (.1 cm);
         \draw[black,fill=white] (1,0) circle (.1 cm);
             \draw[black,fill=white] (1,1) circle (.1 cm);
              \draw[black,fill=white] (0,1) circle (.1 cm);
              \end{scope}
               \begin{scope}[xshift=4cm]
                  \draw[black,fill=white] (0,0) circle (.1 cm);
         \draw[black,fill=white] (1,0) circle (.1 cm);
             \draw[black,fill=white] (1,1) circle (.1 cm);
              \draw[black,fill=white] (0,1) circle (.1 cm);
              \end{scope}
                   \begin{scope}[xshift=6cm]
                  \draw[black,fill=white] (0,0) circle (.1 cm);
         \draw[black,fill=white] (1,0) circle (.1 cm);
             \draw[black,fill=white] (1,1) circle (.1 cm);
              \draw[black,fill=white] (0,1) circle (.1 cm);
              \end{scope}
\end{scope}

\begin{scope}[yshift=2cm]
       \draw[black,fill=white] (0,0) circle (.1 cm);
         \draw[black,fill=white] (1,0) circle (.1 cm);
             \draw[black,fill=white] (1,1) circle (.1 cm);
              \draw[black,fill=white] (0,1) circle (.1 cm);
              \begin{scope}[xshift=2cm]
                  \draw[black,fill=white] (0,0) circle (.1 cm);
         \draw[black,fill=white] (1,0) circle (.1 cm);
             \draw[black,fill=white] (1,1) circle (.1 cm);
              \draw[black,fill=white] (0,1) circle (.1 cm);
              \end{scope}
               \begin{scope}[xshift=4cm]
                  \draw[black,fill=white] (0,0) circle (.1 cm);
         \draw[black,fill=white] (1,0) circle (.1 cm);
             \draw[black,fill=white] (1,1) circle (.1 cm);
              \draw[black,fill=white] (0,1) circle (.1 cm);
              \end{scope}
                   \begin{scope}[xshift=6cm]
                  \draw[black,fill=white] (0,0) circle (.1 cm);
         \draw[black,fill=white] (1,0) circle (.1 cm);
             \draw[black,fill=white] (1,1) circle (.1 cm);
              \draw[black,fill=white] (0,1) circle (.1 cm);
              \end{scope}
\end{scope}

           \end{tikzpicture}
           \caption{Graphical Representation of the Staggered Dirac Operator with Plaquette Interaction with periodic pattern produced by Chessboard Estimates. In Grey we have drawn the unit cells used to perform Fourier Transforms.}
           \label{dirastagplaq}
       \end{figure}

Note that due to the periodicity of $\Delta^{\mathrm{SP}}(\,\cdot\,;\vphi_+,\vphi_-)$ under translations of $\frac{\pi}{4}\mbb{Z}\times \pi\mbb{Z}$, we can either rewrite $w_{\l}^{\mathrm{SP}}$ in the form of \eqref{eq:Peierls-w}, or alternatively

\begin{equation}
w_{\l}^{\mathrm{SP}}(\vphi_+,\vphi_-)= \frac{\vphi_+^2+\vphi_-^2}{4\l}- 8\iint_{(-\frac{\pi}{8},\frac{\pi}{8}]\times(-\frac{\pi}{2},\frac{\pi}{2}]} \frac{d^2p}{(2\pi)^2} \log \Delta^{\mathrm{SP}}(p;\vphi_+,\vphi_-).
\end{equation}

In this way, if $\vphi_+=\vphi_-=0$, the integrand has only one singularity at $p=0$. Trivial though cumbersome algebraic manipulations lead to

\begin{equation}
\label{eq:mySP4}
\begin{split}
\Delta^{\mathrm{SP}}(p;\vphi_+,\vphi_-)&= \frac{1}{16} \big[ 16\vphi_+^4\vphi_-^4 + 8\vphi_-^2\vphi_+^2(\vphi_-^2+\vphi_+^2)+ (\vphi_+^2+ \vphi_-^2)^2 + (1+2\vphi_-\vphi_+)^2(\vphi_-+\vphi_+)^2 \big]+\\
&+ \frac{1}{64}\sin^2(4p_0)+ \frac{1}{4}\sin^2p_1 \big(1+ \vphi_-^2+ \vphi_+^2 + \sin^2p_1\big)\cdot\\
&\cdot \big[  1+ 2\vphi_-\vphi_+ + 3\vphi_+^2 + \vphi_-^2(3+8\vphi_+^2)+ 4\sin^2p_1\big(1+\vphi_-^2+ \vphi_+^2 + \sin^2p_1\big)\big].
\end{split}
\end{equation}

Notice that since $2\vphi_-\vphi_+ + 3\vphi_+^2 + \vphi_-^2(3+8\vphi_+^2)\ge 2\vphi_-\vphi_+ + \vphi_+^2 + \vphi_-^2\ge0$, we have that

\begin{equation}
\label{eqn:6c}
\Delta^{\mathrm{SP}}(p;\vphi_+,\vphi_-)\gtrsim m_{\mathrm{SP}}(\vphi_+,\vphi_-)^2+ |p|^2 \qquad \forall p\in \big(-\tfrac{\pi}{8},\tfrac{\pi}{8}\big]\times\big(-\tfrac{\pi}{2},\tfrac{\pi}{2}\big],
\end{equation}

where we have introduced

\begin{equation*}
m_{\mathrm{SP}}(\vphi_+,\vphi_-) \doteq \sqrt{ \vphi_+^4\vphi_-^4 + \vphi_-^2\vphi_+^2(\vphi_-^2+\vphi_+^2)+ (\vphi_+^2+ \vphi_-^2)^2 }.
\end{equation*}

Once again our goal is to replicate the reasoning from the analysis of the SN model. An explicit computation shows that

\begin{equation}
\label{eq:mySP3}
\begin{split}
&(1,-1)\cdot\nabla w_{\l}^{\mathrm{SP}}(\vphi_+,\vphi_-) = (\vphi_+-\vphi_-)\left\{\frac{1}{2\l}- 8 \iint_{(-\frac{\pi}{8},\frac{\pi}{8}]\times(-\frac{\pi}{2},\frac{\pi}{2}]} \frac{d^2p}{(2\pi)^2} \frac{\Gamma^{\mathrm{SP}}(p;\vphi_+,\vphi_-)}{\Delta^{\mathrm{SP}}(p;\vphi_+,\vphi_-)} \right\},
\end{split}\end{equation}

where 

\begin{equation*}
\begin{split}
\Gamma^{\mathrm{SP}}(p;\vphi_+,\vphi_-) &= \frac{1}{2} \Big\{ -\vphi_+\vphi_-\big[1+3 \vphi_+^2+ \vphi_-^2(3+8\vphi_+^2) \big] +\\
&+\sin^2(p_1)\big[3-8\vphi_-^3\vphi_+ + 5\vphi_+^2 + \vphi_-^2(5+8\vphi_+^2) - \vphi_+\vphi_-(6+ 8\vphi_+^2) \big]\\
&+ 2\sin^2(p_1) \big[5+ 4(\vphi_-^2+\vphi_+^2- \vphi_-\vphi_+) + 4\sin^2(p_1) \big]\Big\}.
\end{split}
\end{equation*}

Observe that $|\Gamma^{\mathrm{SP}}(p;\vphi_+,\vphi_-)|\lesssim m_{\mathrm{SP}}(\vphi_+,\vphi_-)+ m_{\mathrm{SP}}(\vphi_+,\vphi_-)^2+ |p|^2$; therefore the integral in the right--hand side of \eqref{eq:mySP3} is uniformly bounded:

\begin{equation*}
\begin{split}
&\iint_{(-\frac{\pi}{8},\frac{\pi}{8}]\times(-\frac{\pi}{2},\frac{\pi}{2}]} \frac{d^2p}{(2\pi)^2} \left|\frac{\Gamma^{\mathrm{SP}}(p;\vphi_+,\vphi_-)}{\Delta^{\mathrm{SP}}(p;\vphi_+,\vphi_-)} \right| \lesssim \\
&\iint_{(-\frac{\pi}{8},\frac{\pi}{8}]\times(-\frac{\pi}{2},\frac{\pi}{2}]} \frac{d^2p}{(2\pi)^2} \frac{m_{\mathrm{SP}}(\vphi_+,\vphi_-)+ m_{\mathrm{SP}}(\vphi_+,\vphi_-)^2 + |p|^2}{|p|^2 + m_{\mathrm{SP}}(\vphi_+,\vphi_-)^2} \lesssim 1.
\end{split}
\end{equation*}

As an immediate consequence, from \eqref{eq:mySP3} we have that for $\l$ small enough, $
(1,-1)\cdot \nabla w_{\l}^{\mathrm{SP}}(\vphi_+,\vphi_-)\ge0$ on $\{\vphi_+\ge0, \vphi_-\le0\}$.
Using also the symmetry $w_{\l}(\vphi_+,\vphi_-)= w_{\l}(-\vphi_-,-\vphi_+)$, it follows that

\begin{equation}
\min_{\vphi_+\ge0,\vphi_-\le0} w_{\l}^{\mathrm{SP}}(\vphi_+,\vphi_-)= \min_{\vphi\ge0} w_{\l}^{\mathrm{SP}}(\vphi,0).
\end{equation}

Again we aim to compare $w_{\l}^{\mathrm{SP}}(\vphi,0)$ with $w_{\l}^{\mathrm{SP}}\big(\frac{1}{\sqrt{2}}\vphi,\frac{1}{\sqrt{2}}\vphi\big)$, keeping in mind that $w_{\l}^{\mathrm{SP}}(\vphi,\vphi)= v_{\l}^{\mathrm{SP}}(\vphi)$. From \eqref{eq:mySP4} we have that

\begin{equation*}
\begin{split}
&\Delta^{\mathrm{SP}}(p;\vphi,0)=\\
&\frac{1}{16} \big[\vphi^2+ \vphi^4\big]+ \frac{1}{64}\sin^2(4p_0)+ \frac{1}{4}\sin^2p_1 \big(1+ \vphi^2+ \sin^2p_1\big) \big[  1+ 3\vphi^2 + 4\sin^2p_1\big(1+ \vphi^2 + \sin^2p_1\big)\big],
\end{split}
\end{equation*}

while

\begin{equation*}
\begin{split}
\Delta^{\mathrm{SP}}\big(p;\tfrac{1}{\sqrt{2}}\vphi,\tfrac{1}{\sqrt{2}}\vphi\big)&= \frac{1}{16} \big[2\vphi^2+ 5\vphi^4 + 4\vphi^6 +  \vphi^8 \big]+ \frac{1}{64}\sin^2(4p_0)+\\
&+ \frac{1}{4}\sin^2p_1 \big(1+ \vphi^2 + \sin^2p_1\big)\cdot \big[  1+ 4\vphi^2 + 2\vphi^4+ 4\sin^2p_1\big(1+\vphi^2 + \sin^2p_1\big)\big].
\end{split}
\end{equation*}

Therefore we see that $\Delta^{\mathrm{SP}}\big(p;\tfrac{1}{\sqrt{2}}\vphi,\tfrac{1}{\sqrt{2}}\vphi\big)= \Delta^{\mathrm{SP}}(p;\vphi,0)+ \tilde\Delta^{\mathrm{SP}}(p;\vphi)$, where

\begin{equation}
\label{eq:mySP5}
\begin{split}
\tilde\Delta^{\mathrm{SP}}\big(p;\vphi\big) =  \tfrac{1}{16}\big(\vphi^2+ 4\vphi^4+ 4\vphi^6+ \vphi^8\big)+  \tfrac{1}{4}\sin^2p_1 \big(1+ \vphi^2 + \sin^2p_1\big)\big(\vphi^2+ 2\vphi^4\big).
\end{split}
\end{equation}

This allows us to rewrite

\begin{equation*}
\begin{split}
&w_{\l}^{\mathrm{SP}}(\vphi,0)=\\
&v_{\l}^{\mathrm{SP}}\big(\tfrac{1}{\sqrt{2}}\vphi\big)- 8\iint_{(-\frac{\pi}{8},\frac{\pi}{8}]\times(-\frac{\pi}{2},\frac{\pi}{2}]} \frac{d^2p}{(2\pi)^2} \log \left(1- \frac{\tilde\Delta^{\mathrm{SP}}(p;\vphi)}{\Delta^{\mathrm{SP}}\big(p;\tfrac{1}{\sqrt{2}}\vphi,\tfrac{1}{\sqrt{2}}\vphi\big)} \right) \doteq v_{\l}^{\mathrm{SP}}\big(\tfrac{1}{\sqrt{2}}\vphi\big)+ \d w^{\mathrm{SP}}(\vphi).
\end{split}
\end{equation*}

In order to get a lower bound for $\d w^{\mathrm{SP}}(\vphi)$, we observe that $\tilde\Delta^{\mathrm{SP}}(p;\vphi)\gtrsim \frac{1}{16}\vphi^8$ and $\Delta^{\mathrm{SP}}\big(p;\tfrac{1}{\sqrt{2}}\vphi,\tfrac{1}{\sqrt{2}}\vphi\big) \lesssim 1+\vphi^8$. Therefore

\begin{equation}
\label{eq:mySP6}
\d w^{\mathrm{SP}}(\vphi)\ge - 8 \iint_{(-\frac{\pi}{8},\frac{\pi}{8}]\times(-\frac{\pi}{2},\frac{\pi}{2}]} \frac{d^2p}{(2\pi)^2} \log\left( 1- \mf{C}_2^{\mathrm{SP}} \frac{\vphi^8}{1+\vphi^8}\right)= - \frac{1}{2} \log\left( 1- \mf{C}_2^{\mathrm{SP}} \frac{\vphi^8}{1+\vphi^8}\right),
\end{equation}

for a suitable $\mf{C}_2^{\mathrm{SP}}>0$. Note that the right--hand side is positive and monotone increasing for $\vphi>0$. All in all:

\begin{equation*}
w_{\l}^{\mathrm{SP}}(\vphi,0) \ge v_{\l}^{\mathrm{SP}} \big(\tfrac{1}{\sqrt{2}}\vphi \big)- \frac{1}{2} \log\left( 1- \mf{C}_2^{\mathrm{SP}} \frac{\vphi^8}{1+\vphi^8}\right).
\end{equation*}

Finally, the same argument applied to the SN and NN models yields

\[\begin{split} &\min_{\vphi_+\ge0,\vphi_-\le0} w_{\l}^{\mathrm{SP}}(\vphi_+,\vphi_-) -  v_{\l}^{\mathrm{SP}}\big(\vphi_{\mathrm{SP}}^{\star}(\l)\big) \ge \\
& \min\left\{  v_{\l}^{\mathrm{SP}}\big(\tfrac{1}{\sqrt{2}}\vphi_{\mathrm{SP}}^{\star}(\l)\big)- v_{\l}^{\mathrm{SP}}\big(\vphi_{\mathrm{SP}}^{\star}(\l)\big), \, - \frac{1}{4} \log \left(1- \mf{C}_2^{\mathrm{SP}}\frac{\big(\varphi^{\star}_{\mathrm{SP}}(\lambda)\big)^8}{1+ \big(\varphi^{\star}_{\mathrm{SP}}(\lambda)\big)^8} \right)  \right\}\equiv c_3^{\mathrm{SP}}(\l). \end{split}\]

\appendix

\section{Properties of the Dirac Operator and of its anti-symmetrization}
\label{app:Dirac}

In this appendix we collect and prove several technical results concerning the Dirac operators associated with the three lattice fermion models used in the main body of the paper.  
These results are repeatedly employed in the proof of Reflection Positivity for the corresponding measures $d\mu_{\alpha}(\phi)$, but are of a purely operator-theoretic nature and can therefore be isolated from the main exposition.
\\

For each model, we analyze the algebraic and spectral properties of the associated Dirac operators, with particular emphasis on adjointness relations, symmetry under reflections, and the reality and invariance of the fermionic determinant.  
\\

In order to evoid cluttering of notation, in Subsections \ref{Ap21} and \ref{AP22} we will use the symbol $\slashed{D}$ to denote $\slashed{D}_{\mathrm{NN}}$ and $\slashed{D}_{\mathrm{SN}}$ respectively. On the other hand, we will need to distinguish between reflections on the full plane and reflections on the torus (cf. Definition \ref{def:reflections}), which will be denoted by $\vartheta_{\nu}$ and $\vartheta_{\nu}^L$ respectively.

\subsection{Naive Dirac Operator}\label{Ap21}

\begin{proof}[Proof of Lemma \ref{PND}.]
    \leavevmode
    \begin{enumerate}
        \item[\eqref{PND1}]: \textbf{Anti-self-adjointness.}  Let us consider first the Infinite Plane Naive Dirac Operator $\slashed{D}$:

\begin{equation*}
    \slashed{D} = \frac{1}{2} \sum_{\mu=0}^{1} \gamma_{\mu} \otimes (T_{\mu}- T_{\mu}^*), \qquad \slashed{D}^* = - \frac{1}{2} \sum_{\mu=0}^{1} \gamma_{\mu}^{\dagger} \otimes (T_{\mu}- T_{\mu}^*)  = -\slashed{D};
\end{equation*}

thus:

\begin{equation}
\label{eqn:48}
\begin{split}
\overline{\big(\slashed{D}^-_{\L_L}\big)_{(y,s'),(x,s)}} &=  \sum_{n \in \mathbb{Z}^2} (-1)^{n_0+ n_{1}} \overline{(\slashed{D})_{(y+nL,s'),(x,s)}}\\
&= \sum_{n \in \mathbb{Z}^2} (-1)^{n_0+ n_{1}} \overline{(\slashed{D})_{(y,s'),(x-nL,s)}}\\
&= -  \sum_{n \in \mathbb{Z}^2} (-1)^{n_0+ n_{1}} (\slashed{D})_{(x-nL,s),(y,s')}= - \big(\slashed{D}^-_{\L_L}\big)_{(x,s),(y,s')}.
\end{split}
\end{equation}

\item[\eqref{PND2}:] \textbf{Chirality under adjoint.} Since $\gamma_A \gamma_{\mu} \gamma_A = - \gamma_{\mu}$ and $\gamma_A^2= \mathbbl{1}_2$, we have that

\begin{equation*}
\begin{aligned}
(\gamma_A \otimes 1) \slashed{D}^-_{\Lambda_L} (\gamma_A  \otimes 1)= -\slashed{D}^-_{\Lambda_L},\qquad
(\gamma_A \otimes 1) M_{\phi}  (\gamma_A \otimes 1)= M_{\phi},
\end{aligned}
\end{equation*}

so that

\begin{equation*}
(\slashed{D}^-_{\Lambda_L} - M_{\phi})^* = (\gamma_A \otimes 1) (\slashed{D}^-_{\Lambda_L} - M_{\phi}) (\gamma_A \otimes 1).
\end{equation*}
\item[\eqref{PND3}:] \textbf{Spectrum of the adjoint.} Since $\gamma_A^2 = \mathbbl{1}_2$ and $\gamma_A^{\dagger} =  \gamma_A$, $\gamma_A$ is a unitary matrix. Thus, \eqref{PND2} implies that $(\slashed{D}^-_{\Lambda_L} -  M_{\phi})^*$ and $\slashed{D}^-_{\Lambda_L} -  M_{\phi}$ are unitarily conjugated through $\gamma_A \otimes 1$ and thus have the same spectrum. 
\item[\eqref{PND4}:] \textbf{Reality of the determinant.} By applying the determinant to \eqref{PND2}, we get:
\begin{equation*}   \overline{\det(\slashed{D}^-_{\Lambda_L} -  M_{\phi})} =\det((\slashed{D}^-_{\Lambda_L} -  M_{\phi})^*) \overset{\eqref{PND2}}{=} \det(\slashed{D}^-_{\Lambda_L} -  M_{\phi}).
\end{equation*}
\end{enumerate}
\end{proof}

We now analyze the properties of the Dirac operator under reflections. To this purpose we will need the notion of the \emph{Reflection Operator}.

\begin{defn}
[Reflection Operator]
Given a cut plane $(\nu,r)$ on $\L_L$, and letting $\vartheta_{\nu}$ and $\vartheta_{\nu}^L$ be its related reflections on the full plane and on the torus respectively (cf. Definitions \ref{def:cut_plane}-\ref{def:reflections}), we define the linear operators $\hat{\vartheta}_{\nu}: \ell^2(\mbb{Z}^2)\mapsto\ell^2(\mbb{Z}^2)$ and $\hat{\vartheta}_{\nu}^L: \ell^2(\L_L)\mapsto\ell^2(\L_L)$ as:

\begin{align*}
&(\hat{\vartheta}_{\nu}f)_x\doteq f_{\vartheta_{\nu}(x)}, \qquad \qquad \qquad \qquad \,\,\,\,\,\,\,\,\,\,\,\,\,\,\,\,\,\,\,\,\,\,\,\,\,\,\,\,\,\,\,\,\forall f\in \ell^2(\mbb{Z}^2),\\
&(\hat{\vartheta}_{\nu}^Lf)_x\doteq \underbrace{(-1)^{\sum_{\mu=0,1}\frac{1}{L}(\vartheta_{\nu}(x)- \vartheta_{\nu}^L(x) )_{\mu}}}_{\doteq\s_{\nu}^L(x)}f_{\vartheta_{\nu}^L(x)}, \qquad \forall f\in \ell^2(\L_L),
\end{align*}

where note that $\vartheta_{\nu}(x)- \vartheta_{\nu}^L(x)\in L\mbb{Z}^2$.
\end{defn}

\begin{lemma}[Properties under the Reflection Operator]\label{RND}
The following identities hold, for every reflection $\vartheta_{\nu}$ on $\L_L$.

    \begin{enumerate}
        \item  \textbf{\emph{Reflection Covariance of the Dirac Operator:}}
        \begin{equation}
      (\mathbbl{1}_2 \otimes \hat{\vartheta}_{\nu}^L) (\slashed{D}^-_{\Lambda_L} ) (\mathbbl{1}_2 \otimes \hat{\vartheta}_{\nu}^L) =  \big((\gamma_{\nu} \otimes 1)(\slashed{D}^-_{\Lambda_L})(\gamma_{\nu} \otimes 1) \big)^*. 
        \tag{RND1}
        \label{RND1}
    \end{equation}  
    \item \textbf{\emph{Reflection Covariance :}}
    \begin{equation}
       \slashed{D}^-_{\Lambda_L}- M_{\Theta_{\nu}(\phi)} = (\mathbbl{1}_2 \otimes \hat{\vartheta}_{\nu}^L) \big((\gamma_{\nu} \otimes 1)(\slashed{D}^-_{\Lambda_L} -  M_{\phi})(\gamma_{\nu} \otimes 1) \big)^* (\mathbbl{1}_2 \otimes \hat{\vartheta}_{\nu}^L).
        \tag{RND2}
        \label{RND2}
    \end{equation}  
        \item \textbf{\emph{Determinant Invariance:}}
        \begin{equation}
        \det(\slashed{D}^-_{\Lambda_L}-   M_{\Theta_{\nu}(\phi)}) = \overline{\det(\slashed{D}^-_{\Lambda_L} - M_{\phi})}= \det(\slashed{D}^-_{\Lambda_L} -  M_{\phi}). 
        \tag{RND3}
        \label{RND3}
    \end{equation}  

    \end{enumerate}
\end{lemma}

\begin{proof}
\leavevmode
\begin{enumerate}
\item[\eqref{RND1}.] Let us first analyze the action of the reflection on the infinite-plane Dirac operator.

\begin{equation*}
\begin{aligned}
\big((\mathbbl{1}_2 \otimes \hat{\vartheta}_{\nu}) \slashed{D} ( \mathbbl{1}_2 \otimes \hat{\vartheta}_{\nu})\big)_{(x,s),(y,s')} &= \frac{1}{2} \sum_{\mu=0}^1 (\gamma_{\mu})_{s,s'} \big(T_{\mu}-T_{\mu}^*\big)_{\vartheta_{\nu}(x),\vartheta_{\nu}(y)}\\
&= \frac{1}{2} \sum_{\mu=0}^1 (\gamma_{\mu})_{s,s'} \big(\delta_{\vartheta_{\nu}(x), \vartheta_{\nu}(y)-e_{\mu}} - \delta_{\vartheta_{\nu}(x), \vartheta_{\nu}(y)+e_{\mu}}  \big)\\
&= \frac{1}{2} \sum_{\mu=0}^1 (\gamma_{\mu})_{s,s'} \big(\delta_{x, y-(-1)^{\delta_{\mu \nu}}e_{\mu}} - \delta_{x, y+ (-1)^{\delta_{\mu \nu}}e_{\mu}}  \big)\\
&=\frac{1}{2} \sum_{\mu=0}^1 (-1)^{\delta_{\mu \nu}}(\gamma_{\mu})_{s,s'} (T_{\mu}-T_{\mu}^*)_{x,y}.
\end{aligned}
\end{equation*}

With a simple manipulation with $\gamma$-matrices, we get:

\begin{equation*}
\begin{aligned}
\frac{1}{2} \sum_{\mu=0}^1 (-1)^{\delta_{\mu \nu}}(\gamma_{\mu})_{s,s'} (T_{\mu}-T_{\mu}^*)_{x,y}
&= -\frac{1}{2} \sum_{\mu=0}^1 (\gamma_{\nu}\gamma_{\mu} \gamma_{\nu})_{s,s'} (T_{\mu}-T_{\mu}^*)_{x,y}\\
&=  \big((\gamma_{\nu} \otimes 1 )(- \slashed{D})(\gamma_{\nu} \otimes 1) \big)_{(x,s),(y,s')}.
\end{aligned}
\end{equation*}

Therefore, we get that $(\mathbbl{1}_2 \otimes \hat{\vartheta}_{\nu}) \slashed{D} ( \mathbbl{1}_2 \otimes \hat{\vartheta}_{\nu}) = -(\gamma_{\nu} \otimes 1 )\slashed{D} (\gamma_{\nu} \otimes 1)$. From this identity, first note that

\begin{equation*}
\begin{aligned}
\big((\mathbbl{1}_2 \otimes \hat{\vartheta}_{\nu}^L) (\slashed{D}^-_{\Lambda_L})( \mathbbl{1}_2 \otimes \hat{\vartheta}_{\nu}^L)\big)_{(x,s),(y,s')}&=
 \s_{\nu}^L(x) \s_{\nu}^L(y) (\slashed{D}^-_{\Lambda_L})_{(\vartheta_{\nu}^L(x),s), (\vartheta_{\nu}^L(y),s')}\\
 &= (\slashed{D}^-_{\Lambda_L})_{(\vartheta_{\nu}(x),s), (\vartheta_{\nu}(y),s')},
\end{aligned}
\end{equation*}

essentially by the definition of the reflection operator and the anti-periodicity of $\slashed{D}^-_{\L_L}$. Furthermore, notice that

\begin{equation*}
\begin{split}
(\slashed{D}^-_{\Lambda_L})_{(\vartheta_{\nu}(x),s), (\vartheta_{\nu}(y),s')} &= \sum_{n \in \mathbb{Z}^2} (-1)^{n_0+n_1} \slashed{D}_{(\vartheta_{\nu}(x)+nL,s),(\vartheta_{\nu}(y),s')}= \sum_{n \in \mathbb{Z}^2} (-1)^{n_0+n_1} \slashed{D}_{(\vartheta_{\nu}(x+nL),s),(\vartheta_{\nu}(y),s')}.
\end{split}
\end{equation*}

Therefore:

\begin{equation*}
\begin{aligned}
\big((\mathbbl{1}_2 \otimes \hat{\vartheta}_{\nu}^L) (\slashed{D}^-_{\Lambda_L})( \mathbbl{1}_2 \otimes \hat{\vartheta}_{\nu}^L)\big)_{(x,s),(y,s')} &= \sum_{n \in \mathbb{Z}^2} (-1)^{n_0+n_1} \big((\mathbbl{1}_2 \otimes \vartheta_{\nu}) (\slashed{D})( \mathbbl{1}_2 \otimes \vartheta_{\nu})\big)_{(x+nL,s),(y,s')}\\
&= -\sum_{n \in \mathbb{Z}^2} (-1)^{n_0+n_1} \big((\gamma_{\nu}\otimes 1) \slashed{D}( \gamma_{\nu} \otimes 1)\big)_{(x+nL,s),(y,s')}\\
&= -\big((\gamma_{\nu} \otimes 1 )( \slashed{D}^-_{\Lambda_L})(\gamma_{\nu} \otimes 1) \big)_{(x,s),(y,s')}\\
&= \big((\gamma_{\nu} \otimes 1 )( \slashed{D}^-_{\Lambda_L})(\gamma_{\nu} \otimes 1) \big)^*_{(x,s),(y,s')},
\end{aligned}
\end{equation*}

where we have also used that, by \eqref{PND1}, $(\slashed{D}^-_{\L_L})^*=- \slashed{D}^-_{\L_L}$. By the arbitrariness of $x, y \in \Lambda_L$ and $s,s'\in\{1,2\}$, it follows that

\begin{equation*}
(\mathbbl{1}_2 \otimes \hat{\vartheta}_{\nu}^L) (\slashed{D}^-_{\Lambda_L})( \mathbbl{1}_2 \otimes \hat{\vartheta}_{\nu}^L) = \big((\gamma_{\nu} \otimes 1 )( \slashed{D}^-_{\Lambda_L})(\gamma_{\nu} \otimes 1) \big)^*.
\end{equation*}

\item[\eqref{RND2}.]
Since $M_{\Theta_{\nu}(\phi)}= (\mathbbl{1}_2\otimes \hat{\vartheta}_{\nu}^L) M_{\phi} (\mathbbl{1}_2\otimes \hat{\vartheta}_{\nu}^L)$ and, as one can easily check, $(\hat{\vartheta}_{\nu}^L)^2= 1$, we have:

\begin{equation}\label{eqPR}
\slashed{D}^-_{\Lambda_L} - M_{\Theta_{\nu}(\phi)} = (\mathbbl{1}_2 \otimes \hat{\vartheta}_{\nu}^L) \big( (\mathbbl{1}_2 \otimes \hat{\vartheta}_{\nu}^L) (\slashed{D}^-_{\Lambda_L}) (\mathbbl{1}_2 \otimes \hat{\vartheta}_{\nu}^L) - M_{\phi} \big) (\mathbbl{1}_2 \otimes \hat{\vartheta}_{\nu}^L),
\end{equation}

which, after \eqref{RND1}, implies \eqref{RND2}. 
\item[\eqref{RND3}] readily follows by taking the determinant at both sides of \eqref{RND2}, using the fact that $\det(\gamma_{\nu} \otimes 1)^2= \det(\mathbbl{1}_2 \otimes \hat{\vartheta}_{\nu}^L)^2=1$ and exploiting the reality of the determinant, established in \eqref{PND4}.
\end{enumerate}
     
\end{proof}

\subsection{Staggered Dirac Operator}\label{AP22}

\begin{proof}[Proof of Lemma \ref{AP221}.]
\leavevmode
\begin{enumerate}
\item[\eqref{PSD1}:] \textbf{Anti-self-adjointness.} The infinite-volume operator $\slashed{D} \equiv \frac{1}{2} \Gamma_{\mu} (T_{\mu}- T_{\mu}^*)$ is anti-self-adjoint by direct inspection. About $\slashed{D}^-_{\Lambda_L}$, we have that

\begin{equation*}
\begin{aligned}
\big( (\slashed{D}^-_{\Lambda_L})^*\big)_{x,y} &= \overline{\big(\slashed{D}^-_{\Lambda_L}\big)_{y,x}} = \sum_{n \in \mathbb{Z}^2} (-1)^{n_0+n_1} \overline{\slashed{D}_{y+Ln, x} }\\
&= \frac{1}{2}\sum_{n \in \mathbb{Z}^2} (-1)^{n_0+n_1} \sum_{\mu=0}^1\Gamma_{\mu}(y+L n) (\delta_{y+Ln+e_{\mu},x} - \delta_{y+Ln-e_{\mu},x} )\\
&= \frac{1}{2} \sum_{n \in \mathbb{Z}^2} \sum_{\mu=0}^1(-1)^{n_0+n_1} \Gamma_{\mu}(x) (\delta_{y,x-Ln-e_{\mu}} - \delta_{y,x-Ln+e_{\mu}} ),
\end{aligned}   
\end{equation*} 

where we used the fact that $\Gamma_{\mu}(x-Ln \pm e_{\mu})= \Gamma_{\mu}(x \pm e_{\mu})= \Gamma_{\mu}(x)$. From the previous formula we recognize that

\begin{equation*}
\begin{aligned}
\big( (\slashed{D}^-_{\Lambda_L})^*\big)_{x,y} &= \frac{1}{2}\sum_{n \in \mathbb{Z}^2} \sum_{\mu=0}^1(-1)^{n_0+n_1} \Gamma_{\mu}(x) (\delta_{x,y+Ln+e_{\mu}} - \delta_{x,y+Ln- e_{\mu}} )\\
&= - \frac{1}{2}\sum_{n \in \mathbb{Z}^2} \sum_{\mu=0}^1(-1)^{n_0 +n_1} \Gamma_{\mu}(x-nL)(T_{\mu}-T_{\mu}^*)_{x-nL,y}= -  (\slashed{D}^-_{\Lambda_L})_{x,y}.
\end{aligned}
\end{equation*}

\item[\eqref{PSD2}:] \textbf{Chirality under adjoint.} It is convenient to first check the properties for the Infinite Plane Dirac Operator $\slashed{D}$, whose only non-vanishing entries have different parity, because they are related by a single translation in the $\mu$-th direction:

\begin{equation*}
\begin{aligned}            
\big(\epsilon\slashed{D} \epsilon\big)_{x,y} &= (-1)^{x_0+x_1} \frac{1}{2} \sum_{\mu=0,1} \Gamma_{\mu}(x) (\delta_{x+e_{\mu},y} - \delta_{x-e_{\mu},y}) (-1)^{y_1+y_2} \\
&= - \frac{1}{2}\sum_{\mu=0,1} \Gamma_{\mu}(x) (\delta_{x+e_{\mu},y} - \delta_{x-e_{\mu},y}) = - \slashed{D}_{x,y}. 
\end{aligned}
\end{equation*}

Thus, for the anti-periodic operator:

\begin{equation*}
\big(\epsilon \slashed{D}^-_{\Lambda_L} \epsilon\big)_{x,y} = \sum_{n \in \mathbb{Z}^2} (-1)^{n_0+n_1} \big(\epsilon \slashed{D} \epsilon \big)_{x+nL,y} =  -\sum_{n \in \mathbb{Z}^2} (-1)^{n_0+n_1} \slashed{D}_{x+nL,y}= -\big( \slashed{D}^-_{\Lambda_L} \big)_{x,y}.
\end{equation*}

This shows, after \eqref{PSD1}, that $\epsilon \slashed{D}^-_{\Lambda_L} \epsilon = - \slashed{D}^-_{\Lambda_L}= (\slashed{D}^-_{\Lambda_L})^*$. Besides, the fact that $\e M_{\phi}\e= M_{\phi}^*= M_{\phi}$ is obvious since $\e^2=1$ and $\e M_{\phi}= M_{\phi}\e$.

\item[\eqref{PSD3}:] \textbf{Spectrum of the adjoint.}  Since $\epsilon^2=1$ and $\epsilon^* = \epsilon$, \eqref{PSD2} shows that $(\slashed{D}^-_{\Lambda_L}-M_{\phi})^*$ and  $\slashed{D}^-_{\Lambda_L}-M_{\phi}$ are unitarily equivalent, so, in particular, they share the same spectrum. 

\item[\eqref{PSD4}:] \textbf{Reality of the determinant.} Taking the determinant at both sides of \eqref{PSD2}, we get:

\begin{equation*}
\begin{aligned}
\overline{\det(\slashed{D}^-_{\Lambda_L}-M_{\phi})} = \det\big((\slashed{D}^-_{\Lambda_L}-M_{\phi})^*\big) = \det\big(\epsilon(\slashed{D}^-_{\Lambda_L}-M_{\phi})\epsilon\big) = \cancel{\det(\epsilon)^2}\det(\slashed{D}^-_{\Lambda_L}-M_{\phi}).   
\end{aligned}
\end{equation*}   
\end{enumerate}
\end{proof}

\begin{lemma}[Properties under the Reflection Operator]\label{RSD}
The following identities hold, for every reflection $\vartheta_{\nu}$ on $\L_L$.

\begin{enumerate}
\item  \textbf{\emph{Reflection Covariance of the Dirac Operator:}}
\begin{equation}
\hat{\vartheta}_{\nu}^L \slashed{D}^-_{\Lambda_L}   \hat{\vartheta}_{\nu}^L =  \big(\Gamma_{\nu} \slashed{D}^-_{\Lambda_L}\Gamma_{\nu}  \big)^*, 
\tag{RSD1}
\label{RSD1}
\end{equation}  

where $\Gamma_{\nu}$ denotes the multiplication operator by $\Gamma_{\nu}(x)$.

\item \textbf{\emph{Reflection Covariance :}}
    \begin{equation}
       \slashed{D}^-_{\Lambda_L}- M_{\Theta_{\nu}(\phi)} = \hat{\vartheta}_{\nu}^L \big(\Gamma_{\nu} (\slashed{D}^-_{\Lambda_L} - M_{\phi})\Gamma_{\nu} \big)^* \hat{\vartheta}_{\nu}^L.
        \tag{RSD2}
        \label{RSD2}
    \end{equation}  
        \item \textbf{\emph{Determinant Invariance:}}
        \begin{equation}
        \det(\slashed{D}^-_{\Lambda_L}- M_{\Theta_{\nu}(\phi)}) = \overline{\det(\slashed{D}^-_{\Lambda_L} - M_{\phi}) }= \det(\slashed{D}^-_{\Lambda_L} - M_\phi).
        \tag{RSD3}
        \label{RSD3}
    \end{equation}  
    \end{enumerate}
\end{lemma}

\begin{proof}
\leavevmode
\begin{enumerate}
\item[\eqref{RSD1}.] Let us first analyze the action of the reflection on the Infinite Plane Dirac operator.

\begin{equation*}
\begin{aligned}
\big( \hat{\vartheta}_{\nu} \slashed{D} \hat{\vartheta}_{\nu} \big)_{x,y} &=\slashed{D}_{\vartheta_{\nu}(x),\vartheta_{\nu}(y)} =  \frac{1}{2}\sum_{\mu=0,1}\Gamma_{\mu}(\vartheta_{\nu}(x)) \big(\delta_{\vartheta_{\nu}(x)+e_{\mu},\vartheta_{\nu}(y)} - \delta_{\vartheta_{\nu}(x)-e_{\mu},\vartheta_{\nu}(y)} \big)\\     
&=  \frac{1}{2}\sum_{\mu=0,1}\Gamma_{\mu}(\vartheta_{\nu}(x)) \big(\delta_{x+(-1)^{\delta_{\mu \nu}}e_{\mu},y} - \delta_{x-(-1)^{\delta_{\mu \nu}}e_{\mu},y} \big).
\end{aligned}
\end{equation*}

Using the fact that $\Gamma_{\mu}(\vartheta_{\nu}(x))$ equals $(-1)^{\d_{\mu,\nu}+1}\Gamma_{\mu}(x)$, if $\nu=0$, and $\Gamma_{\mu}(x)$, if $\nu=1$, we find that

\begin{equation*}
\begin{aligned}
\big( \hat{\vartheta}_{\nu} \slashed{D} \hat{\vartheta}_{\nu} \big)_{x,y} &= \frac{1}{2}\sum_{\mu}   \Gamma_{\mu}(x) (\delta_{x+e_{\mu},y}-\delta_{x-e_{\mu},y}) \cdot
\begin{cases}
(-1),&(\nu=0)\\
(-1)^{\d_{\mu,\nu}}, &(\nu=1)
\end{cases}
\; = -\big(\Gamma_{\nu}\slashed{D} \Gamma_{\nu} \big)_{x,y},
\end{aligned}
\end{equation*}

where we also used the fact that, for $y=x\pm e_{\mu}$, $\Gamma_{0}(x)\Gamma_{0}(y)= 1$ and $\Gamma_1(x)\Gamma_1(y)= (-1)^{\d_{\mu,0}}$. Now, for the anti-periodic operator:

\begin{equation}\label{refllap}
\begin{aligned}
\big(\hat{\vartheta}_{\nu}^L \slashed{D}^-_{\Lambda_L} \hat{\vartheta}_{\nu}^L \big)_{x,y} &=  \s_{\nu}^L(x) \s_{\nu}^L(y) (\slashed{D}^-_{\Lambda_L})_{\vartheta_{\nu}^L(x),\vartheta_{\nu}^L(y)}= (\slashed{D}^-_{\Lambda_L})_{\vartheta_{\nu}(x),\vartheta_{\nu}(y)}\\
&=\sum_{n\in\mbb{Z}^2} (-1)^{n_0+n_1} \slashed{D}_{\vartheta_{\nu}(x)+nL, \vartheta_{\nu}(y)}= \sum_{n\in\mbb{Z}^2} (-1)^{n_0+n_1} \slashed{D}_{\vartheta_{\nu}(x+nL), \vartheta_{\nu}(y)}\\
&= -\sum_{n\in\mbb{Z}^2} (-1)^{n_0+n_1} \big(\Gamma_{\nu}\slashed{D} \Gamma_{\nu}\big)_{x+nL,y}= -\Gamma_{\nu}(x) (\slashed{D}^-_{\L_L})_{x,y} \Gamma_{\nu}(y)\\
&= - \big(\Gamma_{\nu}(\slashed{D}^-_{\L_L})\Gamma_{\nu}\big).
\end{aligned}
\end{equation}

Since $(\slashed{D}^-_{\Lambda_L})^* = - \slashed{D}^-_{\Lambda_L}$ and the $\Gamma_{\nu}$ operators are self-adjoint, it follows that
   
\begin{equation*}
\hat{\vartheta}_{\nu}^L \slashed{D}^-_{\Lambda_L} \hat{\vartheta}_{\nu}^L =  \big( \Gamma_{\nu} \slashed{D}^{-}_{\Lambda_L} \Gamma_{\nu}\big)^*.
\end{equation*}

\item[\eqref{RSD2}.] It is clear that 
$$\big(\hat{\vartheta}_{\nu}^L M_{\phi} \hat{\vartheta}_{\nu}^L\big)_{x,y}= \phi_{\vartheta_{\nu}^L(x)} \d_{\vartheta_{\nu}^L(x),\vartheta_{\nu}^L(y)}= (M_{\Theta_{\nu}(\phi)})_{x,y}$$

and also that $M_{\phi} = \Gamma_{\nu} M_{\phi} \Gamma_{\nu}$. Combining these two identities with the self-adjointness of $M_{\phi}$, the fact that $(\hat{\vartheta}_{\nu}^L)^2 = 1$ and with \eqref{RSD1}, we get:

\begin{equation*}
\slashed{D}^-_{\Lambda_L} - M_{\Theta_{\nu}(\phi)}= \hat{\vartheta}_{\nu}^L \big(\hat{\vartheta}_{\nu}^L \slashed{D}^-_{\Lambda_L} \hat{\vartheta}_{\nu}^L - M_{\phi}\big) \hat{\vartheta}_{\nu}^L = \hat{\vartheta}_{\nu}^L \big(\Gamma_{\nu} (\slashed{D}^-_{\Lambda_L} - \phi) \Gamma_{\nu}\big)^* \hat{\vartheta}_{\nu}^L.
\end{equation*}

\item[\eqref{RSD3}.] Taking the determinant at both sides of \eqref{RSD2}, and using that $\det (\hat{\vartheta}_{\nu}^L)^2= \det(\Gamma_{\nu}^2)=1$, we get:

\begin{equation*}
\begin{aligned}
\det\big(\slashed{D}^-_{\Lambda_L} -M_{\Theta_{\nu}(\phi)}\big) &=   \det\big( \hat{\vartheta}_{\nu}^L \big(\Gamma_{\nu} (\slashed{D}^-_{\Lambda_L} - M_\phi) \Gamma_{\nu}\big)^* \hat{\vartheta}_{\nu}^L\big) = \det\big(\slashed{D}^-_{\Lambda_L} - M_\phi\big)^* =\overline{\det(\slashed{D}^-_{\Lambda_L}-M_{\phi})}.
\end{aligned}
\end{equation*}

Combining \eqref{RSD3} with \eqref{PSD4}, \eqref{RSD3} follows.
   \end{enumerate}
\end{proof}

\section{Analysis of the Bosonic Potential: Proof of Proposition \ref{const_field}}
\label{app:pot}

This appendix is dedicated to the proof of the main properties, stated in Proposition \ref{const_field}, of the bosonic potential:

\begin{equation*}
V^{\a}_{L,\l}(\phi) = \frac{1}{2\l} \sum_{x\in\L_L^{\a}}\phi_x^2 - \log\big|\det\big( (\slashed{D}_{\a})^-_{\L_L} - M_{\phi} \big) \big|
\end{equation*}

and its infinite volume limit: $v^{\a}_{\l}(\vphi)\doteq \lim_{L\to\infty} \frac{1}{|\L^{\a}_L|} V^{\a}_{\l}(\phi)\big|_{\phi=\vphi^{\L^{\a}_L}}$, with $\vphi\in\mbb{R}$.

\subsection{Proof of Item \ref{itt:1}}
\label{app:pot_1}

Here we shall prove that the potential $V^{\a}_{L,\l}$ is minimized by constant field configurations, namely, given any interval $I\subseteq \mbb{R}$,

\begin{equation}
\label{Minpot_21}
\inf_{\phi\in I^{\Lambda_L^\alpha}}  V^{\alpha}_{L,\lambda}(\phi) = \inf_{\varphi\in I}  V^{\alpha}_{L,\lambda}(\phi)\Big|_{\phi= \vphi^{\L^{\a}_L}}.
\end{equation}

Note that if $I$ contains one only point, the claim is trivially true. Hence let us suppose that $I$ contains an open interval. The key tool is again Reflection Positivity, which implies the following lemma.

\begin{lemma}
\label{lemma:reflection}
Given any cut plane of $\L^{\a}_L$, with $\vartheta:(\L^{\a}_L)^{\pm}\mapsto (\L^{\a}_L)^{\mp}$ being the related reflection operator (see Definitions \ref{def:cut_plane} and \ref{def:reflections}), and given any $\phi\in\mbb{R}^{\L^{\a}_L}$, it holds true that

\begin{equation}
\label{eqn:36}
\big(e^{-NV^{\a}_{L,\l}(\phi)} \big)^2 \leq e^{-NV^{\a}_{L,\l}(\phi^+)} e^{-NV^{\a}_{L,\l}(\phi^-)},  
\end{equation}

where $\phi^{\pm}$ is the configuration obtained from $\phi$ by reflecting, respectively, the left/right configuration (with respect to the cut plane) to the right/left:

\begin{equation*}
\phi^{\pm}_x \doteq \begin{cases}
\phi_x \,\,\,\,&x \in (\Lambda^{\alpha}_L)_{\pm}\\
\phi_{\vartheta(x)} &x \in (\Lambda^{\alpha}_L)_{\mp}.
\end{cases}
\end{equation*}
    
\end{lemma}

Before proving Lemma \ref{lemma:reflection}, let us show how it readily implies the validity of \eqref{Minpot_21}. The idea is to prove that, for any given field configuration $\phi\in I^{\L^{\a}_L}$, there exists a uniform configuration, $\phi=\vphi^{\L^{\a}_L}$ for some $\vphi\in I$, with a smaller energy. First of all observe that, since $e^{-N V^{\alpha}_{L,\lambda}(\phi)} \geq 0$ and $e^{-N V^{\alpha}_{L,\lambda}(\phi)}  \to 0^+$ as $\|\phi\|_{\ell^2(\L^{\a}_L)} \to \infty$, we have that

\begin{equation}
\label{eqn:57}
\inf_{\phi\in I^{\Lambda_L^\alpha}}  V^{\alpha}_{L,\lambda}(\phi)= - \frac{1}{N}\log \sup _{\phi\in I^{\Lambda_{L}^\alpha}} e^{-N V^{\a}_{L,\l}(\phi)}= - \frac{1}{N}\log \max _{\phi\in \ov{I}^{\Lambda_L^\alpha}} e^{-N V^{\a}_{L,\l}(\phi)},
\end{equation}

where $\ov{I}$ denotes the closure of the interval. If we denote by $\tilde\phi$ the maximizer of $e^{-N V^{\a}_{L,\l}}$ over $\ov{I}^{\L^{\a}_L}$, then by Lemma \ref{lemma:reflection} we have that

\begin{equation}
\label{eqn:54}
   \big(e^{-NV^{\a}_{L,\l}(\tilde\phi)}\big)^2 \leq e^{-NV^{\a}_{L,\l}(\tilde\phi^+)} e^{-NV^{\a}_{L,\l}(\tilde\phi^-)} \leq  e^{-NV^{\a}_{L,\l}(\tilde\phi)} e^{-N V^{\a}_{L,\l}(\tilde\phi^+)}.  
\end{equation}

Note that, assuming that $I$ is not a single point, we must have that $V^{\a}_{L,\l}(\tilde\phi)<+\infty$, since the function $e^{-NV^{\a}_{L,\l}}$ cannot be identically zero over $\ov{I}$. Hence, by \eqref{eqn:54}, we have that $e^{-NV^{\a}_{L,\l}(\tilde\phi)} = e^{-N V^{\a}_{L,\l}(\tilde\phi^+)}$. In other words, we have shown that a configuration obtained by reflections from the minimizing one has the same energy as the latter. We can repeat this argument first for all the horizontal and then for all the vertical cutting planes, so that we end up with a constant field configuration $\tilde\varphi^{\L^{\a}_L}$, with $\tilde\vphi\in \ov{I}$, which has the same energy as $\tilde\phi$. All in all, in view also of \eqref{eqn:57}, we have that

\begin{equation}
\label{eqn:59}
\inf_{\phi\in I^{\Lambda_L^\alpha}}  V^{\alpha}_{L,\lambda}(\phi)= V^{\a}_{L,\l}(\phi)\big|_{\phi=\tilde{\vphi}^{\L^{\a}_L}}, 
\end{equation}

with $\tilde\vphi$ as above. Finally, using \eqref{eqn:59} and the continuity of the map $\mbb{R}\ni\vphi \mapsto V^{\a}_{L,\l}(\phi)\big|_{\phi=\vphi^{\L^{\a}_L}}$ \footnote{For $L$ fixed the continuity is obvious from the direct expression of the potential (see Subsection \ref{app:pot_2}) since there are no exact zero-modes for the Dirac Operators $(\slashed{D}_{\a})^-_{\L_L}$ with antiperiodic boundary conditions; however, it is easy to show that the continuity holds also for the infinite-volume \emph{potential density} $v^{\a}_{\l}$, see \eqref{eqn:58}, as the integrand therein exhibits at most an integrable singularity, at $\vphi=0$.} , we have that
\begin{equation*}
\inf_{\phi\in I^{\Lambda_L^\alpha}}  V^{\alpha}_{L,\lambda}(\phi)= V^{\a}_{L,\l}(\phi)\big|_{\phi=\tilde{\vphi}^{\L^{\a}_L}}\ge \min_{\vphi\in \ov{I}} V^{\a}_{L,\l}(\phi)\big|_{\phi=\vphi^{\L^{\a}_L}}= \inf_{\vphi\in I} V^{\a}_{L,\l}(\phi)\big|_{\phi=\vphi^{\L^{\a}_L}}\ge \inf_{\phi\in I^{\Lambda_L^\alpha}}  V^{\alpha}_{L,\lambda}(\phi),
\end{equation*}
thus proving \eqref{Minpot_21}.

\begin{proof}[Proof of Lemma \ref{lemma:reflection}.]
We must employ Reflection Positivity for the fermionic integration. Recall that

\begin{equation*}
e^{-N V^{\alpha}_{L,\lambda}(\phi)} =  \bigg(\prod_{x \in \Lambda^{\alpha}_L} e^{- \frac{N}{2 \lambda} \phi_x^2} \bigg)  \det\big((\slashed{D}_{\alpha})^{-}_{\Lambda_L} - M_\phi\big)^N = \bigg(\prod_{x \in \Lambda^{\alpha}_L} e^{- \frac{N}{2 \lambda} \phi_x^2} \bigg)  \bigg(\int d \overline{\psi} d \psi\, e^{\sum_{x,y} \overline{\psi}_x\big((\slashed{D}_{\alpha})^-_{\Lambda_L}) - M_\phi  \big)_{xy} \psi_y}\bigg)^N.
\end{equation*}

The idea is to think of the determinant $\det\big((\slashed{D}_{\alpha})^{-}_{\Lambda_L} - M_\phi\big)$ as a bilinear product $\left \langle e^A, e^B \right \rangle_{\rho_\alpha}$, where, given two functions $F,G$ over the Grassmann algebra generated by $\overline{\psi},\psi$,

\begin{equation*}
    \left \langle F, G \right \rangle_{\rho_\alpha} \doteq   \int d\rho_{\a}(\overline{\psi},\psi) \big(F\, \Theta(G)\big)(\overline{\psi},\psi),
\end{equation*}

with the reflection operator $\Theta$ as in \eqref{RNF}, \eqref{RSF}, \eqref{RSFP} and:

\begin{align}
&\label{eqn:34}
A(\overline{\psi},\psi)\doteq \sum_{x \in (\L^{\a}_L)_+} \phi_x \overline{\psi}_x \psi_x, \qquad B(\overline{\psi},\psi) \doteq \sum_{x \in (\L^{\a}_L)_+} \phi_{\vartheta(x)} \overline{\psi}_x \psi_x,\\
&\label{eqn:34b} d \rho_{\alpha}(\overline{\psi}, \psi) \doteq   d\overline{\psi} d \psi \exp\Big\{\sum_{x,y\in\L^{\a}_L} \overline{\psi}_x\big((\slashed{D}_{\alpha})^-_{\Lambda_L}\big)_{x,y} \psi_y \Big\}.
\end{align}

Notice that here $\phi$ has to be understood as a fixed static, external field (versus the dynamical field in the full measure $d \mu_{\alpha}$ as in \eqref{eq:boson_measure}). The Grassmann integration $d \rho_{\alpha}(\overline{\psi}, \psi)$ has already been (implicitly) shown to be reflection positive, in proving Theorem \ref{rp}, so that the sesquilinear form $\left \langle \cdot, \cdot \right \rangle_{\rho_\alpha}$ is actually semi-positive definite. As a consequence, by the Cauchy-Schwartz inequality:

\begin{equation}
\label{eqn:35}
\det\big((\slashed{D}_{\alpha})^-_{\Lambda_L}- M_\phi\big)^2 = \left \langle e^A, e^B \right \rangle_{\rho_\alpha}^2 \leq  \left \langle e^A, e^A \right \rangle_{\rho_\alpha}  \left \langle e^B, e^B \right \rangle_{\rho_\alpha}= \det\big((\slashed{D}_{\alpha})^-_{\Lambda_L}- M_{\phi^+}\big) \det\big((\slashed{D}_{\alpha})^-_{\Lambda_L}- M_{\phi^-}\big),
\end{equation}

with $A$ and $B$ as in \eqref{eqn:34}. Observe that each factor in the right--hand side is non-negative, being of the form $\|e^{A}\|_{\rho_\alpha}^2 \ge 0$.
Therefore, if either of the two determinants on the right--hand side vanishes, the left--hand side. must vanish as well. Furthermore, since the left--hand side. is manifestly non-negative, we also have, for any $N \in \mathbb{N}$,

\begin{equation*}     \det\big((\slashed{D}_{\alpha})^-_{\Lambda_L}- M_\phi\big)^{2N} \leq \det\big((\slashed{D}_{\alpha})^-_{\Lambda_L}- M_{\phi^+}\big)^N \det\big((\slashed{D}_{\alpha})^-_{\Lambda_L}- M_{\phi^-}\big)^N.
\end{equation*}

Finally, by multiplying both sides by $\prod_{x \in \Lambda^{\alpha}_L} e^{-\frac{N}{2 \lambda} \phi^2_x}$ and noting that

\begin{equation*}
\bigg( \prod_{x \in \Lambda^{\alpha}_L} e^{-\frac{N}{2 \lambda} \phi^2_x}\bigg)^2 = \bigg(\prod_{x \in \Lambda^{\alpha}_L} e^{-\frac{N}{2 \lambda} (\phi^+_x)^2}\bigg)\bigg( \prod_{x \in \Lambda^{\alpha}_L} e^{-\frac{N}{2 \lambda} (\phi^-_x)^2} \bigg),
\end{equation*}

we find the validity of \eqref{eqn:36}, hence proving Lemma \ref{lemma:reflection}.
\end{proof}

\begin{remark}
It is interesting to observe that inequality \eqref{eqn:35} implies that $\det\big((\slashed{D}_{\alpha})^-_{\Lambda_L} - M_{\phi}\big)$ is positive whenever the configuration $\phi$ is symmetric under reflection around some cut plane, as it can be written as:

\begin{equation*}
\det\big((\slashed{D}_{\alpha})^-_{\Lambda_L} - M_\phi\big) = \left \langle e^A, e^A \right \rangle_{\rho_{\alpha}} \geq 0,
\end{equation*}

with $A$ as in \eqref{eqn:34} and $\L_L^+$ the right half portion of $\L_L$ determined by the cut plane.
\end{remark}

\subsection{Proof of Item \ref{itt:2}}
\label{app:pot_2}

\paragraph{Computation of the uniform-field potential.} Since $\phi$ is constant, namely $\phi= \vphi^{\L^{\a}_L}$, the spectrum of the operator $(\slashed{D}_{\alpha})^-_{\Lambda_L}- M_{\phi}$ can be explicitly obtained using Fourier transforms. Such operation actually depends on the specific model under consideration, thus we will work out the three cases $\a=\mathrm{NN},\mathrm{SN},\mathrm{SP}$ separately.

\begin{itemize}
\item \textbf{NN model.} The model has a full translational invariance by $\mathbb{Z}^2$, thus the space of allowed momenta with anti-periodic boundary conditions is

\begin{equation*}
(\Lambda^{\mathrm{NN}}_L)^*_{--} \doteq \tfrac{2\pi}{L}\big(\mbb{Z}+ \tfrac{1}{2}\big)^2 \cap (-\pi,\pi ]^2. 
\end{equation*}

The determinant of the finite-volume Dirac operator is computed starting from \eqref{eq:model_NN}:

    \begin{equation*} 
    \begin{aligned}    \det\big((\slashed{D}_{\mathrm{NN}})^{-}_{\Lambda_L}-  M_{\phi}\big)\Big|_{\phi=\vphi^{\L^{\a}_L}} &= \prod_{p \in (\Lambda^{\mathrm{NN}}_L)^*_{--}} \det\big(-i \gamma_0 \sin(p_0) - i \gamma_1 \sin(p_1) - \mathbbl{1}_2 \varphi\big)\\
        &=  \prod_{p \in (\Lambda^{\mathrm{NN}}_L)^*_{--}}  \big(\varphi^2 +\sin^2(p_0) + \sin^2(p_1)\big),
        \end{aligned}
    \end{equation*}

from which:

\begin{equation*}        
\frac{1}{|\L^{\mathrm{NN}}_L|}\log\big|\det\big((\slashed{D}_{\mathrm{NN}})^{-}_{\Lambda_L}-  M_\phi\big)\Big|_{\phi=\vphi^{\L^{\a}_L}}\big| = \frac{1}{L^2}\sum_{p \in (\Lambda^{\mathrm{NN}}_L)^*_{--}} \log(\varphi^2 +\sin^2(p_0) + \sin^2(p_1)).
\end{equation*}

\item \textbf{SN model.} The model has a reduced translational invariance by $(2\mbb{Z})\times\mbb{Z}$. It is therefore convenient to considered an enlarged $2\times1$ fundamental cell. The resulting lattice is drawn in fig. \ref{fig2}.

\begin{figure}[h]
    \centering
\begin{tikzpicture}
\draw[dotted, ->-] (-0.5,-0.5) -- (3.5,-0.5);
\draw[dotted, ->>-] (3.5,-0.5) -- (3.5,3.5);
\draw[dotted, ->-] (-0.5,3.5) -- (3.5,3.5);
\draw[dotted, ->>-] (-0.5,-0.5) -- (-0.5,3.5);
\draw[line width = 0.05 cm] (0, 0) -- (3,0);
\draw[line width = 0.05 cm] (0, 2) -- (3,2);
\draw[line width = 0.05 cm] (3,1) -- (3.5,1);
\draw[line width = 0.05 cm] (-0.5,1) -- (0,1);
\draw[line width = 0.05 cm] (-0.5,3) -- (0,3);
\draw[line width = 0.05 cm] (3,3) -- (3.5,3);
\draw (0,3) -- (3,3);
\draw (0,1) -- (3,1);
\draw (0,-0.5) -- (0,3.5);
\draw (1,-0.5) -- (1,3.5);
\draw (2,-0.5) -- (2,3.5);
\draw (-0.5,2) -- (0,2);
\draw (3,2) -- (3.5,2);
\draw (-0.5,0) -- (0,0);
\draw (3,0) -- (3.5,0);
\draw (3,-0.5) -- (3,3.5);
\draw[fill, fill opacity =0.1] (-0.25,-0.25) -- (0.25,-0.25) -- (0.25, 1.25) -- (-0.25,1.25) -- (-0.25,-0.25);
\begin{scope} [xshift = 1cm]
\draw[fill, fill opacity =0.1] (-0.25,-0.25) -- (0.25,-0.25) -- (0.25, 1.25) -- (-0.25,1.25) -- (-0.25,-0.25);    
\end{scope}
\begin{scope} [xshift = 2cm]
\draw[fill, fill opacity =0.1] (-0.25,-0.25) -- (0.25,-0.25) -- (0.25, 1.25) -- (-0.25,1.25) -- (-0.25,-0.25);    
\end{scope}
\begin{scope} [xshift = 3cm]
\draw[fill, fill opacity =0.1] (-0.25,-0.25) -- (0.25,-0.25) -- (0.25, 1.25) -- (-0.25,1.25) -- (-0.25,-0.25);    
\end{scope}
\begin{scope} [xshift = 3cm, yshift= 2cm]
\draw[fill, fill opacity =0.1] (-0.25,-0.25) -- (0.25,-0.25) -- (0.25, 1.25) -- (-0.25,1.25) -- (-0.25,-0.25);    
\end{scope}
\begin{scope} [xshift = 2cm, yshift= 2cm]
\draw[fill, fill opacity =0.1] (-0.25,-0.25) -- (0.25,-0.25) -- (0.25, 1.25) -- (-0.25,1.25) -- (-0.25,-0.25);    
\end{scope}
\begin{scope} [xshift = 1cm, yshift= 2cm]
\draw[fill, fill opacity =0.1] (-0.25,-0.25) -- (0.25,-0.25) -- (0.25, 1.25) -- (-0.25,1.25) -- (-0.25,-0.25);    
\end{scope}
\begin{scope} [yshift= 2cm]
\draw[fill, fill opacity =0.1] (-0.25,-0.25) -- (0.25,-0.25) -- (0.25, 1.25) -- (-0.25,1.25) -- (-0.25,-0.25);    
\end{scope}
\draw[line width = 0.05 cm] (0,-0.5) -- (0,0);
\draw[line width = 0.05 cm] (1,-0.5) -- (1,0);
\draw[line width = 0.05 cm] (2,-0.5) -- (2,0);
\draw[line width = 0.05 cm] (3,-0.5) -- (3,0);

\draw[line width = 0.05 cm] (0,3.5) -- (0,3);
\draw[line width = 0.05 cm] (1,3.5) -- (1,3);
\draw[line width = 0.05 cm] (2,3.5) -- (2,3);
\draw[line width = 0.05 cm] (3,3.5) -- (3,3);
\draw[->] (-1,-0.5) -- (-1,3.5);
\node[above] (a) at (-1,3.5) {{$x_0$}};
\draw[->] (0,0.5) -- (1,0.5);
\draw[->] (0,0.5) -- (0,2.5);
\end{tikzpicture}
\caption{Fundamental Cells: In absence of twistings of the boundary conditions one gets a manifestly translation invariant system.}
\label{fig2}
\end{figure}
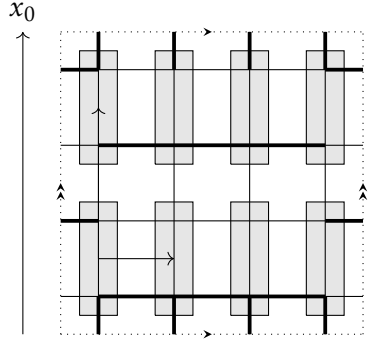

The set of allowed momenta is given by the dual of this reduced lattice:

\begin{equation*}
(\Lambda^{\mathrm{SN}}_L)^*_{--}= \tfrac{2\pi}{L} \big(\mbb{Z}+\tfrac{1}{2}\big)^2 \cap \Big(\big(-\tfrac{\pi}{2},\tfrac{\pi}{2}\big]\times(-\pi,\pi] \Big).
\end{equation*}

Again the determinant can be computed starting from \eqref{eq:model_SN}:

\begin{equation*} 
    \begin{aligned}    \det\big((\slashed{D}_{\mathrm{SN}})^{-}_{\Lambda_L}- M_\phi\big)\Big|_{\phi=\vphi^{\L^{\a}_L}} &= \prod_{p \in (\Lambda^{\mathrm{SN}}_L)^*_{--}} \det \begin{pmatrix}
        i\sin p_1 -\vphi & \frac{1}{2}(e^{2ip_0}-1)\\
        -\frac{1}{2}(e^{-2ip_0}-1) & -i\sin p_1 - \vphi
    \end{pmatrix} \\
        &=  \prod_{p \in (\Lambda^{\mathrm{SN}}_L)^*_{--}}  \big(\varphi^2 +\sin^2(p_0) + \sin^2(p_1)\big).
        \end{aligned}
    \end{equation*}

    Thus:
 \begin{equation*}        
 \frac{1}{|\L^{\mathrm{SN}}_L|} \log \big|\det\big((\slashed{D}_{\mathrm{SN}})^{-}_{\Lambda_L}- M_\phi\big)\Big|_{\phi=\vphi^{\L^{\a}_L}}\big| = \frac{1}{L^2}\sum_{p \in (\Lambda^{\mathrm{SN}}_L)^*_{--}} \log(\varphi^2 +\sin^2(p_0) + \sin^2(p_1)).
\end{equation*}

\item \textbf{SP model.} The analysis is identical to the one for the model SN.  Indeed, when the field $\phi$ is constant, the models SN and SP are actually the same, as well as the set of momenta: $(\L^{\mathrm{SP}}_L)^*_{--}= (\L^{\mathrm{SN}}_L)^*_{--}$. The only difference comes from the fact that $|\L^{\mathrm{SP}}_L|= \frac{1}{4}|\L^{\mathrm{SN}}_L|$, so that

 \begin{equation*}        
 \frac{1}{|\L^{\mathrm{SP}}_L|} \log \big|\det\big((\slashed{D}_{\mathrm{SP}})^{-}_{\Lambda_L}- M_\phi\big)\Big|_{\phi=\vphi^{\L^{\a}_L}}\big| = \frac{4}{L^2}\sum_{p \in (\Lambda^{\mathrm{SP}}_L)^*_{--}} \log(\varphi^2 +\sin^2(p_0) + \sin^2(p_1)).
\end{equation*}
\end{itemize}

Notice that, due to the periodicity of $\hat{f}(p;\vphi) \doteq \log\big(\vphi^2+ \sin^2(p_0)+ \sin^2(p_1)\big)$ with respect to translations of $\pm\pi$ in both directions, it is possible to rewrite:

\begin{equation}
\label{eqn:37}
\begin{split}
\frac{1}{|\L^{\a}_L|}V^{\a}_{L,\l}(\vphi^{\L^{\a}_L})&= \frac{\vphi^2}{2\l}-
\frac{1}{|\L^{\a}_L|}\sum_{p\in(\L_L^{\a})^*} \hat{f}(p;\vphi)= \frac{\vphi^2}{2\l}- \frac{C_{\a}}{L^2}\sum_{p\in (\L^{\mathrm{NN}}_L)^*} \hat{f}(p;\vphi)\\
&= \frac{\vphi^2}{2\l}- \frac{4C_{\a}}{L^2} \sum_{p\in \frac{2\pi}{L}(\mbb{Z}+\frac{1}{2})^2\cap (-\frac{\pi}{2},\frac{\pi}{2})^2} \hat{f}(p;\vphi),
\end{split}\end{equation}

where $C_{\a}=1,\frac{1}{2},2$ for $\a=\mathrm{NN},\mathrm{SN},\mathrm{SP}$ respectively. 

\paragraph{Convergence to the infinite-volume limit.} By standard arguments (e.g. the Dominated Convergence Theorem), we have the convergence, in the limit $L\to\infty$, of the Riemann sum  in \eqref{eqn:37} to the respective integral $C_{\a}\iint_{(-\pi,\pi]^2} \frac{d^2p}{(2\pi)^2} \hat{f}(p;\vphi)$, so that 

\begin{equation}
\label{eqn:38}
\begin{split}
v^{\a}_{\l}(\vphi) \equiv \lim_{L\to\infty}\frac{1}{|\L^{\a}_L|} V^{\a}_{L,\l}(\vphi^{\L^{\a}_L})&= \frac{\vphi^2}{2\l}- C_{\a}\iint_{(-\pi,\pi]^2} \frac{d^2p}{(2\pi)^2} \hat{f}(p;\vphi)\\
&= \frac{\vphi^2}{2\l}- 4C_{\a}\iint_{(-\frac{\pi}{2},\frac{\pi}{2}]^2} \frac{d^2p}{(2\pi)^2} \hat{f}(p;\vphi).
\end{split}
\end{equation}

The representation in the second line of \eqref{eqn:38} is convenient because in the domain $(-\frac{\pi}{2},\frac{\pi}{2}]^2$, the integrand $\hat{f}(p;\vphi)$ has only one singularity at $p=(0,0)$, if $\vphi=0$. Observe also that the Riemann sum in the right--hand side of \eqref{eqn:37} can be rewritten as follows:

\begin{equation}
\frac{1}{L^2}\sum_{p\in \frac{2\pi}{L}(\mbb{Z}+\frac{1}{2})^2\cap (-\frac{\pi}{2},\frac{\pi}{2})^2} \hat{f}(p;\vphi)=\iint_{(-\frac{\pi}{2},\frac{\pi}{2})^2} \frac{d^2p}{(2\pi)^2} \hat{f}\big(\pi_L(p);\vphi\big),
\end{equation}

where $\pi_L: \big(-\tfrac{\pi}{2},\tfrac{\pi}{2}\big)^2\mapsto \big(-\tfrac{\pi}{2},\tfrac{\pi}{2}\big)^2 \cap \tfrac{2\pi}{L}\big(\mbb{Z}+\tfrac{1}{2}\big)^2$ is the projection defined by the following condition (see Fig. \ref{figproj}):

\[\text{if}\, p\in \frac{2\pi}{L}\big([n_0,n_0+1)\times [n_1,n_1+1)\big)\text{, then}\, \pi_L(p)= \frac{2\pi}{L}\big(n_0+\frac{1}{2},n_1+\frac{1}{2}\big)\]

. \\

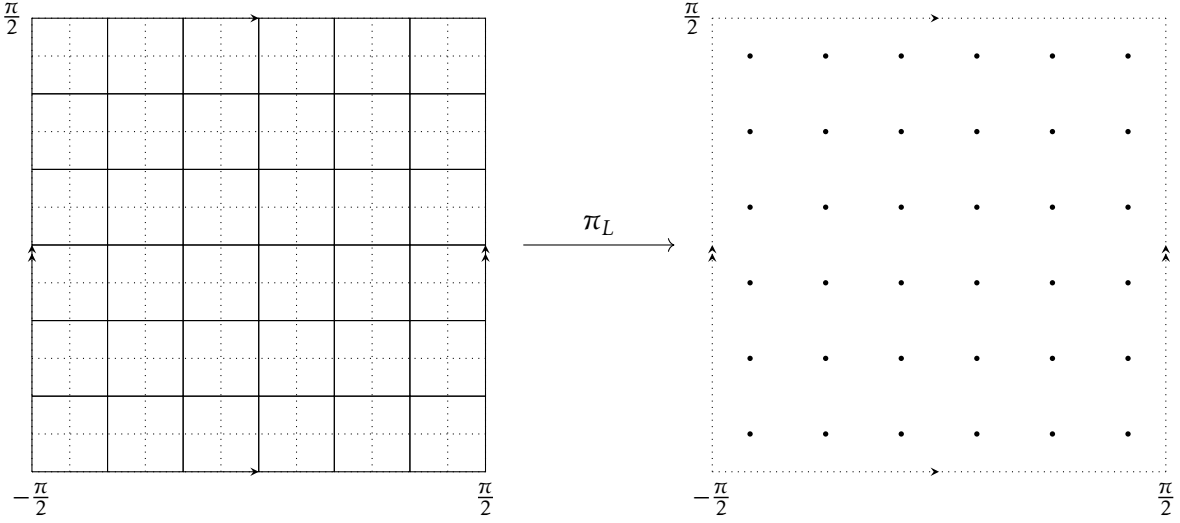
\begin{figure}[h]
    \centering
\begin{tikzpicture}
\draw[dotted, ->-] (-3,-3) -- (3,-3);
\draw[dotted, ->>-] (3,-3) -- (3,3);
\draw[dotted, ->-] (-3,3) -- (3,3);
\draw[dotted, ->>-] (-3,-3) -- (-3,3);
\node[below] at (-3,-3.01) {{$-\frac{\pi}{2}$}};
\node[below] at (3,-3.01) {{$\frac{\pi}{2}$}};
\node[left] at (-2.99,3) {{$\frac{\pi}{2}$}};
\draw[dotted] (-2.5,-3) -- (-2.5,3); 
\draw[dotted] (-1.5,-3) -- (-1.5,3);
\draw[dotted] (-0.5,-3) -- (-0.5,3); 
\draw[dotted] (0.5,-3) -- (0.5,3); 
\draw[dotted] (1.5,-3) -- (1.5,3);
\draw[dotted] (2.5,-3) -- (2.5,3); 
\draw[dotted] (-3,-2.5) -- (3,-2.5); 
\draw[dotted] (-3,-1.5) -- (3,-1.5); 
\draw[dotted] (-3,-0.5) -- (3,-0.5);
\draw[dotted] (-3,0.5) -- (3,0.5); 
\draw[dotted] (-3,1.5) -- (3,1.5);
\draw[dotted] (-3,2.5) -- (3,2.5); 
\begin{scope}[xshift= 0cm]
\draw[] (-3,-3) -- (-2,-3) -- (-2,-2)--(-3,-2) -- (-3,-3);
\end{scope}
\begin{scope}[xshift= 1cm]
\draw[] (-3,-3) -- (-2,-3) -- (-2,-2)--(-3,-2) -- (-3,-3);    
\end{scope}
\begin{scope}[xshift= 2cm]
\draw[] (-3,-3) -- (-2,-3) -- (-2,-2)--(-3,-2) -- (-3,-3);    
\end{scope}
\begin{scope}[xshift= 3cm]
\draw[] (-3,-3) -- (-2,-3) -- (-2,-2)--(-3,-2) -- (-3,-3);    
\end{scope}
\begin{scope}[xshift= 4cm]
\draw[] (-3,-3) -- (-2,-3) -- (-2,-2)--(-3,-2) -- (-3,-3);    
\end{scope}
\begin{scope}[xshift= 5cm]
\draw[] (-3,-3) -- (-2,-3) -- (-2,-2)--(-3,-2) -- (-3,-3);    
\end{scope}
\begin{scope}[xshift= 0cm, yshift=1cm]
\draw[] (-3,-3) -- (-2,-3) -- (-2,-2)--(-3,-2) -- (-3,-3);
\end{scope}
\begin{scope}[xshift= 1cm, yshift=1cm]
\draw[] (-3,-3) -- (-2,-3) -- (-2,-2)--(-3,-2) -- (-3,-3);    
\end{scope}
\begin{scope}[xshift= 2cm, yshift=1cm]
\draw[] (-3,-3) -- (-2,-3) -- (-2,-2)--(-3,-2) -- (-3,-3);    
\end{scope}
\begin{scope}[xshift= 3cm, yshift=1cm]
\draw[] (-3,-3) -- (-2,-3) -- (-2,-2)--(-3,-2) -- (-3,-3);    
\end{scope}
\begin{scope}[xshift= 4cm, yshift=1cm]
\draw[] (-3,-3) -- (-2,-3) -- (-2,-2)--(-3,-2) -- (-3,-3);    
\end{scope}
\begin{scope}[xshift= 5cm, yshift=1cm]
\draw[] (-3,-3) -- (-2,-3) -- (-2,-2)--(-3,-2) -- (-3,-3);    
\end{scope}

\begin{scope}[yshift=1cm]
    \begin{scope}[xshift= 0cm, yshift=1cm]
\draw[] (-3,-3) -- (-2,-3) -- (-2,-2)--(-3,-2) -- (-3,-3);
\end{scope}
\begin{scope}[xshift= 1cm, yshift=1cm]
\draw[] (-3,-3) -- (-2,-3) -- (-2,-2)--(-3,-2) -- (-3,-3);    
\end{scope}
\begin{scope}[xshift= 2cm, yshift=1cm]
\draw[] (-3,-3) -- (-2,-3) -- (-2,-2)--(-3,-2) -- (-3,-3);    
\end{scope}
\begin{scope}[xshift= 3cm, yshift=1cm]
\draw[] (-3,-3) -- (-2,-3) -- (-2,-2)--(-3,-2) -- (-3,-3);    
\end{scope}
\begin{scope}[xshift= 4cm, yshift=1cm]
\draw[] (-3,-3) -- (-2,-3) -- (-2,-2)--(-3,-2) -- (-3,-3);    
\end{scope}
\begin{scope}[xshift= 5cm, yshift=1cm]
\draw[] (-3,-3) -- (-2,-3) -- (-2,-2)--(-3,-2) -- (-3,-3);    
\end{scope}
\end{scope}
\begin{scope}[yshift=2cm]
    \begin{scope}[xshift= 0cm, yshift=1cm]
\draw[] (-3,-3) -- (-2,-3) -- (-2,-2)--(-3,-2) -- (-3,-3);
\end{scope}
\begin{scope}[xshift= 1cm, yshift=1cm]
\draw[] (-3,-3) -- (-2,-3) -- (-2,-2)--(-3,-2) -- (-3,-3);    
\end{scope}
\begin{scope}[xshift= 2cm, yshift=1cm]
\draw[] (-3,-3) -- (-2,-3) -- (-2,-2)--(-3,-2) -- (-3,-3);    
\end{scope}
\begin{scope}[xshift= 3cm, yshift=1cm]
\draw[] (-3,-3) -- (-2,-3) -- (-2,-2)--(-3,-2) -- (-3,-3);    
\end{scope}
\begin{scope}[xshift= 4cm, yshift=1cm]
\draw[] (-3,-3) -- (-2,-3) -- (-2,-2)--(-3,-2) -- (-3,-3);    
\end{scope}
\begin{scope}[xshift= 5cm, yshift=1cm]
\draw[] (-3,-3) -- (-2,-3) -- (-2,-2)--(-3,-2) -- (-3,-3);    
\end{scope}
\end{scope}
\begin{scope}[yshift=3cm]
    \begin{scope}[xshift= 0cm, yshift=1cm]
\draw[] (-3,-3) -- (-2,-3) -- (-2,-2)--(-3,-2) -- (-3,-3);
\end{scope}
\begin{scope}[xshift= 1cm, yshift=1cm]
\draw[] (-3,-3) -- (-2,-3) -- (-2,-2)--(-3,-2) -- (-3,-3);    
\end{scope}
\begin{scope}[xshift= 2cm, yshift=1cm]
\draw[] (-3,-3) -- (-2,-3) -- (-2,-2)--(-3,-2) -- (-3,-3);    
\end{scope}
\begin{scope}[xshift= 3cm, yshift=1cm]
\draw[] (-3,-3) -- (-2,-3) -- (-2,-2)--(-3,-2) -- (-3,-3);    
\end{scope}
\begin{scope}[xshift= 4cm, yshift=1cm]
\draw[] (-3,-3) -- (-2,-3) -- (-2,-2)--(-3,-2) -- (-3,-3);    
\end{scope}
\begin{scope}[xshift= 5cm, yshift=1cm]
\draw[] (-3,-3) -- (-2,-3) -- (-2,-2)--(-3,-2) -- (-3,-3);    
\end{scope}
\end{scope}
\begin{scope}[yshift=4cm]
    \begin{scope}[xshift= 0cm, yshift=1cm]
\draw[] (-3,-3) -- (-2,-3) -- (-2,-2)--(-3,-2) -- (-3,-3);
\end{scope}
\begin{scope}[xshift= 1cm, yshift=1cm]
\draw[] (-3,-3) -- (-2,-3) -- (-2,-2)--(-3,-2) -- (-3,-3);    
\end{scope}
\begin{scope}[xshift= 2cm, yshift=1cm]
\draw[] (-3,-3) -- (-2,-3) -- (-2,-2)--(-3,-2) -- (-3,-3);    
\end{scope}
\begin{scope}[xshift= 3cm, yshift=1cm]
\draw[] (-3,-3) -- (-2,-3) -- (-2,-2)--(-3,-2) -- (-3,-3);    
\end{scope}
\begin{scope}[xshift= 4cm, yshift=1cm]
\draw[] (-3,-3) -- (-2,-3) -- (-2,-2)--(-3,-2) -- (-3,-3);    
\end{scope}
\begin{scope}[xshift= 5cm, yshift=1cm]
\draw[] (-3,-3) -- (-2,-3) -- (-2,-2)--(-3,-2) -- (-3,-3);    
\end{scope}
\end{scope}

\draw[->] (3.5,0) -- (5.5,0);
\node[above] at (4.5,0) {{$\pi_L$}};

\begin{scope}[xshift= 9cm]
    \draw[dotted, ->-] (-3,-3) -- (3,-3);
\draw[dotted, ->>-] (3,-3) -- (3,3);
\draw[dotted, ->-] (-3,3) -- (3,3);
\draw[dotted, ->>-] (-3,-3) -- (-3,3);
\node[below] at (-3,-3.01) {{$-\frac{\pi}{2}$}};
\node[below] at (3,-3.01) {{$\frac{\pi}{2}$}};
\node[left] at (-2.99,3) {{$\frac{\pi}{2}$}};
\fill (-2.5,-2.5) circle (1pt);
\begin{scope}[xshift=1cm]
    \fill (-2.5,-2.5) circle (1pt);
\end{scope}
\begin{scope}[xshift=1cm]
    \fill (-1.5,-2.5) circle (1pt);
\end{scope}
\begin{scope}[xshift=1cm]
    \fill (-0.5,-2.5) circle (1pt);
\end{scope}
\begin{scope}[xshift=1cm]
    \fill (0.5,-2.5) circle (1pt);
\end{scope}
\begin{scope}[xshift=1cm]
    \fill (1.5,-2.5) circle (1pt);
\end{scope}

\begin{scope}[yshift=1cm]
    \fill (-2.5,-2.5) circle (1pt);
\begin{scope}[xshift=1cm]
    \fill (-2.5,-2.5) circle (1pt);
\end{scope}
\begin{scope}[xshift=1cm]
    \fill (-1.5,-2.5) circle (1pt);
\end{scope}
\begin{scope}[xshift=1cm]
    \fill (-0.5,-2.5) circle (1pt);
\end{scope}
\begin{scope}[xshift=1cm]
    \fill (0.5,-2.5) circle (1pt);
\end{scope}
\begin{scope}[xshift=1cm]
    \fill (1.5,-2.5) circle (1pt);
\end{scope}
\end{scope}
\begin{scope}[yshift=2cm]
    \fill (-2.5,-2.5) circle (1pt);
\begin{scope}[xshift=1cm]
    \fill (-2.5,-2.5) circle (1pt);
\end{scope}
\begin{scope}[xshift=1cm]
    \fill (-1.5,-2.5) circle (1pt);
\end{scope}
\begin{scope}[xshift=1cm]
    \fill (-0.5,-2.5) circle (1pt);
\end{scope}
\begin{scope}[xshift=1cm]
    \fill (0.5,-2.5) circle (1pt);
\end{scope}
\begin{scope}[xshift=1cm]
    \fill (1.5,-2.5) circle (1pt);
\end{scope}
\end{scope}
\begin{scope}[yshift=3cm]
    \fill (-2.5,-2.5) circle (1pt);
\begin{scope}[xshift=1cm]
    \fill (-2.5,-2.5) circle (1pt);
\end{scope}
\begin{scope}[xshift=1cm]
    \fill (-1.5,-2.5) circle (1pt);
\end{scope}
\begin{scope}[xshift=1cm]
    \fill (-0.5,-2.5) circle (1pt);
\end{scope}
\begin{scope}[xshift=1cm]
    \fill (0.5,-2.5) circle (1pt);
\end{scope}
\begin{scope}[xshift=1cm]
    \fill (1.5,-2.5) circle (1pt);
\end{scope}
\end{scope}
\begin{scope}[yshift=4cm]
    \fill (-2.5,-2.5) circle (1pt);
\begin{scope}[xshift=1cm]
    \fill (-2.5,-2.5) circle (1pt);
\end{scope}
\begin{scope}[xshift=1cm]
    \fill (-1.5,-2.5) circle (1pt);
\end{scope}
\begin{scope}[xshift=1cm]
    \fill (-0.5,-2.5) circle (1pt);
\end{scope}
\begin{scope}[xshift=1cm]
    \fill (0.5,-2.5) circle (1pt);
\end{scope}
\begin{scope}[xshift=1cm]
    \fill (1.5,-2.5) circle (1pt);
\end{scope}
\end{scope}
\begin{scope}[yshift=5cm]
    \fill (-2.5,-2.5) circle (1pt);
\begin{scope}[xshift=1cm]
    \fill (-2.5,-2.5) circle (1pt);
\end{scope}
\begin{scope}[xshift=1cm]
    \fill (-1.5,-2.5) circle (1pt);
\end{scope}
\begin{scope}[xshift=1cm]
    \fill (-0.5,-2.5) circle (1pt);
\end{scope}
\begin{scope}[xshift=1cm]
    \fill (0.5,-2.5) circle (1pt);
\end{scope}
\begin{scope}[xshift=1cm]
    \fill (1.5,-2.5) circle (1pt);
\end{scope}
\end{scope}

\end{scope}

\end{tikzpicture}
\caption{Picture of the action of the projector $\pi_L$: each cell is sent to its center}
\label{figproj}
\end{figure}

In this way:

\begin{equation}
\label{eqn:31}
\begin{split}
\left|\frac{1}{|\L_L^{\a}|}V^{\a}_{L,\l}(\vphi^{\L^{\a}_L})- v^{\a}_{\l}(\vphi) \right|&\le 4C_{\a}\iint_{(-\frac{\pi}{2},\frac{\pi}{2}]^2} \frac{d^2p}{(2\pi)^2} \Big|\hat{f}\big(\pi_L(p);\vphi\big)- \hat{f}(p;\vphi)\Big|.
\end{split}\end{equation}

In order to proceed, observe that for any $p\in (-\frac{\pi}{2},\frac{\pi}{2})^2$ and $\mu\in\{0,1\}$, $|\pi_L(p)_{\mu}|\ge \frac{1}{2}|p_{\mu}|$. Note also that for any $p\in (-\frac{\pi}{2},\frac{\pi}{2})^2$, we can write:

\begin{equation*}
\begin{split}
\Big|\hat{f}\big(\pi_L(p);\vphi\big)- \hat{f}(p;\vphi)\Big|&\le \sum_{\mu=0,1}\big|\pi_L(p)_{\mu}-p_{\mu}\big| \Bigg|\int_0^1 ds \left[\tfrac{2\sin(q_{\mu})\cos(q_{\mu})}{\vphi^2+ \sin^2(q_0)+ \sin^2(q_1)} \right]_{q= s\pi_L(p)+ (1-s)p} \Bigg|\\
&\le \sum_{\mu=0,1}\big|\pi_L(p)_{\mu}-p_{\mu}\big| \int_0^1 ds \left[\tfrac{2}{\sqrt{\sin^2(q_0)+ \sin^2(q_1)}} \right]_{q= s\pi_L(p)+ (1-s)p}\\
&\le \frac{4\pi}{L}\tfrac{2}{\sqrt{\sin^2(\frac{1}{2}p_0)+ \sin^2(\frac{1}{2}p_1)}} \le \frac{8\pi^2}{L|p|},
\end{split}
\end{equation*}

where we also used that for $|t|\le\frac{\pi}{2}$, $|\sin t|\ge \frac{2}{\pi}|t|$. Therefore from \eqref{eqn:31} we find:

\begin{equation*}
\begin{split}
&\left|\frac{1}{|\L_L|} V^{\alpha}_{L,\lambda}(\vphi^{\L^{\a}_L})- v^{\a}_{\l}(\vphi) \right|\le 8 \iint_{(-\frac{\pi}{2},\frac{\pi}{2}]^2} \frac{d^2p}{(2\pi)^2} \frac{8\pi^2}{L|p|}\le \frac{32\pi^2}{L},
\end{split}\end{equation*}

concluding the proof of Item \ref{itt:2}.

\subsection{Proof of Item \ref{itt:3}}
\label{app:pot_3}

\paragraph{Existence of non-zero minima.} It is clear by inspection that, for any $\l>0$, the mean-field potential $v^{\a}_{\l}$, as a function of $\vphi$, is even and $C^1$ over $\mbb{R}$. In order to find its minima it is convenient to analyze its derivative:

\begin{equation}
\frac{d v^{\a}_{\l}}{d\vphi}(\vphi)= \vphi\left( \frac{1}{\l}- 4 C_{\a} \iint_{(-\frac{\pi}{2},\frac{\pi}{2}]^2} \frac{d^2p}{(2\pi)^2} \frac{1}{\vphi^2+ \sin^2p_0+ \sin^2 p_1}\right) \doteq \vphi g^{\a}_{\l}(\vphi).
\end{equation}

We see that $v^{\a}_{\l}$ has a local maximum at $\vphi=0$. Moreover, as it is easy to check, for $\vphi>0$ the function $g^{\a}_{\l}$ is monotone increasing and $\lim_{\vphi\to+\infty}g^{\a}_{\l}(\vphi)=\frac{1}{\l}$. Besides, $\lim_{\vphi\to 0^+}g^{\a}_{\l}(\vphi)=-\infty$, indeed:

\begin{equation}
\label{eqn:39}
g^{\a}_{\l}(\vphi)\le \frac{1}{\l}- 4 C_{\a} \iint_{(-\frac{\pi}{2},\frac{\pi}{2}]^2} \frac{d^2p}{(2\pi)^2} \frac{1}{\vphi^2+ |p|^2} \le \frac{1}{\l}- 2 \iint_{|p|\le 1} \frac{d^2p}{(2\pi)^2} \frac{1}{\vphi^2+ |p|^2}= \frac{1}{\l}- \frac{1}{2\pi}\log\big(1+ \tfrac{1}{\vphi^2} \big).
\end{equation}

Since $g^{\a}_{\l}$ is also obviously continuous over $(0,\infty)$, it follows that there exists a unique value $\vphi^{\star}_{\a}(\l)\in (0,\infty)$ such that $g^{\a}_{\l}\big(\vphi^{\star}_{\a}(\l)\big)=0$. Such point is of course an absolute minimum for $v^{\a}_{\l}$.

\paragraph{Regularity of the minima.} As we just argued, the minimum $\vphi^{\star}_{\a}(\l)$ can be interpreted as the unique positive solution to the \emph{gap equation}:

\begin{equation}
\label{eqn:40}
g^{\a}_{\l}(\vphi)=0, \qquad \vphi\in(0,\infty).
\end{equation}

The solution of such equation can be characterized by the Implicit Function Theorem. The map $(\l,\vphi)\mapsto g^{\a}_{\l}(\vphi)$ is clearly $C^{\infty}$ over $(0,\infty)^2$. Besides, since the function

\begin{equation*}
\frac{\de g^{\a}_{\l}}{\de\vphi}(\vphi)= 8 C_{\a}\vphi \iint_{(-\frac{\pi}{2},\frac{\pi}{2}]^2} \frac{d^2p}{(2\pi)^2} \frac{1}{\big(\vphi^2+ \sin^2p_0+ \sin^2 p_1\big)^2}
\end{equation*}

is never zero, we can apply the (global version of the) Implicit Function Theorem, which implies that the (unique) solution $\vphi^{\star}_{\a}(\l)$ to \eqref{eqn:40} is actually $C^{\infty}$ with respect to $\l\in(0,\infty)$.

\paragraph{Exponential smallness of the minima.} Finally we prove \eqref{eqn:41}. Using that $g^{\a}_{\l}\big(\vphi^{\star}_{\a}(\l)\big)=0$, from \eqref{eqn:39} we have that $0\le \frac{1}{\l}- \frac{1}{2\pi}\log\big(1+ \vphi^{\star}_{\a}(\l)^{-2} \big)$, which implies that $\vphi^{\star}_{\a}(\l) \ge e^{-\frac{\pi}{\l}}$. In order to provide an upper bound, we observe that

\begin{equation*}
\begin{split}
g^{\a}_{\l}(\vphi)&= \frac{1}{\l}- 4 C_{\a} \iint_{(-\frac{\pi}{2},\frac{\pi}{2}]^2} \frac{d^2p}{(2\pi)^2} \frac{1}{\vphi^2+ \sin^2p_0+ \sin^2 p_1}\\
&\ge \frac{1}{\l}- 8 \iint_{|p|\le \pi} \frac{d^2p}{(2\pi)^2} \frac{1}{\vphi^2+ \frac{4}{\pi^2}|p|^2}= \frac{1}{\l}- \frac{\pi}{2}\log\big(1+ \tfrac{4}{\vphi^2}\big),
\end{split}
\end{equation*}

where again we used the fact that $|\sin t|\ge \frac{2}{\pi}|t|$ for $|t|\le\frac{\pi}{2}$. The condition $g^{\a}_{\l}\big(\vphi^{\star}_{\a}(\l) \big)=0$ then yields, for $\l\le 1$,

\begin{equation*}
\frac{4}{\vphi^{\star}_{\a}(\l)^2}\ge e^{\frac{2}{\pi\l}} -1 \ge e^{\frac{2}{\pi\l}}\big(1- e^{-\frac{2}{\pi}} \big) \Rightarrow \vphi^{\star}_{\a}(\l)\le \frac{2}{\sqrt{1- e^{-\frac{2}{\pi}}}} e^{-\frac{1}{\pi\l}}.
\end{equation*}

\section{Formal Continuum Limit of the Staggered Plaquette Model}
\label{app:continuum_sp}

In this appendix we derive the formal continuum limit of the Staggered Plaquette model and identify the corresponding continuum Gross--Neveu theory. We will show that:

\begin{enumerate}
\item after a suitable change of fermionic variables, the non-interacting lattice theory converges to one describing $N_{\mathrm{eff}}\equiv 2N$ non-interacting Dirac fermions;
\item the plaquette interaction becomes the standard Gross--Neveu quartic interaction.
\end{enumerate}

\paragraph{Lattice Action.}

We consider a two-dimensional lattice of linear size $L$ and mesh $\ell>0$,
$\Lambda_{\ell,L}=\ell\mbb{Z}^2\cap \big(-\tfrac{L}{2},\tfrac{L}{2} \big]^2$, with
$L/\ell \in 4\mathbb N$. We partition $\Lambda_{\ell,L}$ into $2\times2$ plaquettes and introduce the coarse lattice: $\widetilde\Lambda_{\ell,L}
=
2\ell \mbb{Z}^2\cap\big(-\tfrac{L}{2},\tfrac{L}{2}\big]^2$ and we consider the set of internal plaquette sites $\mathcal I=\{A,B,C,D\}$, so that we can identify $\L_{\ell,L}$ with $\widetilde\Lambda_{\ell,L}
\times
\mathcal I$ (see Fig. \ref{figstaglat1}). The action of the Staggered Plaquette model reads $S_{\mathrm{SP}}^{\ell,L}
=
S_0^{\ell,L}
+
S_{\mathrm{int}}^{\ell,L}$, where

\begin{align}
&S_0^{\ell,L}\doteq \ell^2 \sum_{x,y \in \wt{\L}_{\ell,L}}\sum_{j=1}^N\sum_{a,b\in \mc{I}} \ov{\psi}_{x,a,j}\big((\slashed{D}_{\mathrm{SP}})^-_{\L_{\ell,L}}\big)_{(x,a),(y,b)} \psi_{y,b,j},\\
&\label{eq:sp_interaction}
S_{\mathrm{int}}^{\ell,L}
\doteq
-\frac{\lambda(2\ell)^2}{8N}
\sum_{x\in\widetilde\Lambda_{\ell,L}}
\Bigg(
\sum_{a\in\mathcal I}
\sum_{j=1}^N
\ov{\psi}_{x,a,j}\psi_{x,a,j}
\Bigg)^2
\end{align}

and $\big((\slashed{D}_{\mathrm{SP}})^-_{\L_{\ell,L}}\big)_{x,y} \doteq \ell^{-1} \big((\slashed{D}_{\mathrm{SP}})^-_{\L_{L/\ell}}\big)_{x/\ell,y/\ell}$, with the latter defined in \eqref{eqn:51}. 

\paragraph{Fourier Representation.}

We define, for $p\in (\widetilde\Lambda_{\ell,L})^*_{--}
\doteq
\frac{2\pi}{L}
\Big(\mathbb Z+\tfrac12\Big)^2
\cap
\Big(
-\frac{\pi}{2\ell},
\frac{\pi}{2\ell}
\Big)^2$,

\[
\hat{\psi}^-_{p,a,j} \doteq (2\ell)^2\sum_{x\in\wt{\L}_{\ell,L}} e^{ip\cdot x} \psi_{x,a,j}, \qquad \hat{\psi}^+_{p,a,j} \doteq (2\ell)^2\sum_{x\in\wt{\L}_{\ell,L}} e^{-ip\cdot x} \overline\psi_{x,a,j},
\]

so that

\[
\psi_{x,a,j}
=
\frac1{L^2}
\sum_{p\in(\widetilde\Lambda_{\ell,L})^*_{--}}
e^{-ipx}
\hat{\psi}^-_{p,a,j}, \qquad
\overline{\psi}_{x,a,j}
=
\frac1{L^2}
\sum_{p\in(\widetilde\Lambda_{\ell,L})^*_{--}}
e^{ipx}
\hat{\psi}^+_{p,a,j}.
\]

By introducing the four-component spinor:

\[
\hat{\Psi}^-_{p,j}
=
\begin{pmatrix}
\hat\psi^-_{p,A,j}\\
\hat\psi^-_{p,B,j}\\
\hat\psi^-_{p,C,j}\\
\hat\psi^-_{p,D,j}
\end{pmatrix},
\qquad
\hat{\Psi}^+_{p,j}
=
\begin{pmatrix}
\hat{\psi}^+_{p,A,j}&
\hat{\psi}^+_{p,B,j}&
\hat{\psi}^+_{p,C,j}&
\hat{\psi}^+_{p,D,j}
\end{pmatrix},
\]

the free action takes the form

\[
S_0^{\ell,L}
=
\frac1{4L^2}
\sum_{j=1}^{N}
\sum_{p}
\hat{\Psi}^+_{p,j}
\hat{D}(p)
\hat{\Psi}^-_{p,j}, 
\]

where, letting $\sigma_\mu(p)
\doteq
\frac{1-e^{-2i\ell p_\mu}}{2i\ell}$,

\[\hat{D}(p)\doteq \begin{pmatrix}
0
&
-i\ov{\sigma}_1(p)
&
i\ov{\sigma}_0(p)
&
0
\\[6pt]
-i\sigma_1(p)
&
0
&
0
&
i\ov{\sigma}_0(p)
\\[6pt]
i\sigma_0(p)
&
0
&
0
&
i\ov{\sigma}_1(p)
\\[6pt]
0
&
i\sigma_0(p)
&
i\sigma_1(p)
&
0
\end{pmatrix}.\]

\paragraph{Spin--Taste Basis.}

To identify the proper fermionic degrees of freedom suited for the continuum limit, we perform a standard \emph{spin--taste transformation}
\cite[pp.~243--248]{Gattringer:2010zz}. Namely we define:

\[
\hat{\xi}^-_{p,\o,s,j}
\doteq
\frac14
\sum_{a\in\mathcal I}
\Gamma^{(a)}_{s,\omega}
\hat{\psi}^-_{p,a,j}, \qquad
\hat\xi^+_{p,\o,s,j}
\doteq
\frac14
\sum_{a\in\mathcal I}
\hat\psi^+_{p,a,j}
(\Gamma^{(a)})^\dagger_{\omega,s},
\]

with $\o\in\{\pm\}$ the \emph{taste index}, $s\in\{\uparrow,\downarrow\}$ the \emph{spin index} and, letting $e_A=(0,0), e_B=(1,0), e_C=(0,1), e_D=(1,1)$ be the normalized displacement within the unit cell (see Fig. \ref{figstaglat1}), 

\[
\Gamma^{(a)} \equiv \left(\begin{array}{cc}
    \Gamma^{(a)}_{\uparrow,+} & \Gamma^{(a)}_{\uparrow,-} \\
    \Gamma^{(a)}_{\downarrow,+} & \Gamma^{(a)}_{\downarrow,-}
\end{array} \right)
\doteq
\gamma_1^{e_a\cdot e_1}
\gamma_0^{e_a\cdot e_0}
=
\begin{cases}
\mathbbl{1}_2, & a=A,\\
\sigma_x, & a=B,\\
\sigma_y, & a=C,\\
i\sigma_z, & a=D.
\end{cases}
\]

Alternatively, after introducing

\[
\hat{\Xi}^-_{p,j}
=
\begin{pmatrix}
\xi^-_{p,+,\uparrow,j}
\\
\xi^-_{p,+,\downarrow,j}
\\
\xi^-_{p,-,\uparrow,j}
\\
\xi^-_{p,-,\downarrow,j}
\end{pmatrix}, \qquad \hat{\Xi}^+_{p,j}
=
\begin{pmatrix}
\xi^+_{p,+,\uparrow,j},
\xi^+_{p,+,\downarrow,j},
\xi^+_{p,-,\uparrow,j},
\xi^+_{p,-,\downarrow,j}
\end{pmatrix},
\]

the change of basis reads $\hat\Psi^-_{p,j}
=
U^\dagger \hat\Xi^-_{p,j}$ and $\hat\Psi^+_{p,j}=\hat\Xi^+_{p,j} U$, where

\[
U
=
\frac1{\sqrt2}
\begin{pmatrix}
1&0&0&i\\
0&1&-i&0\\
0&1&i&0\\
1&0&0&-i
\end{pmatrix}
\] 

is a unitary matrix. Hence:

\[
S_0^{\ell,L}
=
\frac1{4L^2}
\sum_{j=1}^{N}
\sum_{p\in (\wt{\L}_{\ell,L})^*_{--}}
\hat\Xi^+_{p,j}
\,\hat{S}(p)\,
\hat\Xi^-_{p,j},
\]

with $\hat{S}(p)= U \hat{D}(p) U^\dagger$. A direct computation gives:

\begin{equation}
\label{eq:continuum_dirac}
\hat{S}(p)
=
\mathbbl{1}_{\mathrm{taste}} \otimes \Big(
-i p_0 \gamma_0
-
i p_1 \gamma_1
\Big)
+
\mc{O}(\ell |p|^2).
\end{equation}

Equation \eqref{eq:continuum_dirac} shows that the taste sector becomes asymptotically diagonal and survives as an additional independent fermionic species. Note also that 

$$\det \hat{S}(p)= \det\hat{D}(p)= \ell^{-4}\big(\sin^2(\ell p_0)+ \sin^2(\ell p_1)\big)^2,$$

showing that $\hat{S}(p)^{-1}$ has a pole only at $p=0$. 

Observe also that, in the spin-taste basis, the interaction term \eqref{eq:sp_interaction} becomes

\[
S_{\mathrm{int}}^{\ell,L}
=
-\frac{\lambda(2\ell)^2}{8N}
\sum_{x\in\widetilde\Lambda_{\ell,L}}
\Big(
\ov\Xi_x\Xi_x
\Big)^2 \equiv -\frac{\lambda(2\ell)^2}{8N}
\sum_{x\in\widetilde\Lambda_{\ell,L}}
\Big(\sum_{j=1}^N\sum_{\o=\pm}
\ov\Xi_{x,\o,j}\Xi_{x,\o,j}
\Big)^2.
\]

\paragraph{Continuum Gross--Neveu Theory.}

The formal limit $\ell\to 0^+,L\to\infty$ is obtained by replacing Riemann sums by the corresponding integrals:

\begin{equation*}
(2\ell)^2\sum_{x\in\wt{\L}_{\ell,L}} \leadsto \iint_{\mbb{R}^2} d^2x, \qquad  L^{-2}\sum_{p\in (\wt{\L}_{\ell,L})^*_{--}} \leadsto \iint_{\mbb{R}^2} \frac{d^2p}{(2\pi)^2} 
\end{equation*}

and neglecting the remainders in the right--hand side of \eqref{eq:continuum_dirac}. We obtain:

\[
S_{\mathrm{SP}}^{\ell,L} \leadsto S_{\mathrm{SP}} = \iint d^2x \Bigg[-\frac{1}{4}i\sum_{j=1}^N\sum_{\o=\pm}\sum_{\mu=0,1} \ov{\Xi}_{x,\o,j} \gamma_{\mu}\de_{\mu} \Xi_{x,\o,j} -\frac{\lambda}{8N} 
\big(\sum_{j=1}^N\sum_{\o=\pm}\ov\Xi_{x,\o,j}\Xi_{x,\o,j}\big)^2 \Bigg].
\]

The kinetic term contains an overall factor $1/4$, which is easily removed by a trivial rescaling $\Xi\mapsto \Xi'= \frac{1}{2}\Xi$. In this way we get:

\begin{equation}
S_{\mathrm{SP}} = \iint d^2x \Bigg[- i\sum_{j=1}^N\sum_{\o=\pm}\sum_{\mu=0,1} \ov{\Xi}'_{x,\o,j} \gamma_{\mu}\de_{\mu} \Xi'_{x,\o,j} -\frac{\lambda_{\mathrm{eff}}}{2N_{\mathrm{eff}}} 
\big(\sum_{j=1}^N\sum_{\o=\pm}\ov\Xi'_{x,\o,j}\Xi'_{x,\o,j}\big)^2 \Bigg],
\end{equation}

with $N_{\mathrm{eff}}=2N$ and $\l_{\mathrm{eff}}=8\l$. We therefore conclude that the Staggered Plaquette model, after performing the formal continuum limit, reproduces the standard Gross-Neveu model with coupling $\l_{\mathrm{eff}}$ and $N_{\mathrm{eff}}$ flavors.

\paragraph{One-Loop Running Coupling.}

The Renormalization Group equation, for the running coupling constant $\l_{\mathrm{eff}}(\ell)$ at energy scale $\ell^{-1}$, reads:

\begin{equation*}
\ell \frac{d\l_{\mathrm{eff}}}{d\ell}= \b(\l_{\mathrm{eff}}),
\end{equation*}

with $\b(\l_{\mathrm{eff}})$ the Beta Function of the Gross-Neveu model, whose one-loop expression is \cite{TRACAS1990333}:

\[
\beta(\lambda_{\mathrm{eff}})
=
\frac1\pi
\left(
1-\frac1{N_{\mathrm{eff}}}
\right)
\lambda_{\mathrm{eff}}^2
+
\mc{O}(\lambda_{\mathrm{eff}}^3).
\]

The integration of the Renormalization Group equation truncated at order $\l_{\mathrm{eff}}^2$, yields

\[
\lambda_{\mathrm{eff}}(\ell)
=
\frac{\lambda_{\mathrm{eff}}(\ell_0)}{1
+
\frac1\pi
\left(
1-\frac1{N_{\mathrm{eff}}}
\right)
\lambda_{\mathrm{eff}}(\ell_0)
\log\!\big(\frac{\ell_0}{\ell}\big)}.
\]

After writing $N_{\mathrm{eff}}=2N$ and
$\lambda_{\mathrm{eff}}(\ell)=8\lambda(\ell)$, we obtain:

\[
\lambda(\ell)
=
\frac{\lambda(\ell_0)}{1
+
\frac8\pi
\left(
1-\frac1{2N}
\right)
\lambda(\ell_0)
\log\!\big(\frac{\ell_0}{\ell}\big)}.
\]

At the lowest order in $\frac{1}{N}$, this expression coincides with \eqref{eqn:bare_coupling} for $\a=\mathrm{SP}$.

\section{Upper and Lower-bound on the $2$-point function}\label{applbub}
\subsection{Lower-bound of the $2$-point function} 
Since:
\[
1 = I_x^{+K_{\alpha}}+I_x^{-K_{\alpha}}+I_x^{(-C_{\alpha},C_{\alpha})}+I_x^{\star_+}+I_x^{\star_-},
\]
the two-point function can be decomposed as:
\begin{equation}
\label{eqn:9}
\langle \phi_x\phi_y\rangle_{\mu_\alpha}
=
\sum_{i,j \in \mathcal{R}}
\big\langle
\phi_x\phi_y I_x^{(i)} I_y^{(j)}
\big\rangle_{\mu_\alpha},
\end{equation}
where the sum runs over the five regions $\mathcal{R}$ listed in
Items \ref{it:large+}--\ref{it:min-} in Section \ref{ssect:field_size}. These contributions can be naturally grouped into the three following classes. 

\begin{itemize}
\item \textbf{DT (dominant term).} The two fields lie near the \textit{same minimum}:
\[ \left \langle \phi_x \phi_y 
\right \rangle_{\mu_{\alpha}}^{\mathrm{DT}} \doteq
\big\langle
\phi_x\phi_y
\bigl(
I_x^{\star_+}I_y^{\star_+}
+
I_x^{\star_-}I_y^{\star_-}
\bigr)
\big\rangle_{\mu_\alpha}.
\]

\item \textbf{PT (Peierls' term).} The two fields lie around \textit{different minima}:
\[ \left \langle \phi_x \phi_y 
\right \rangle_{\mu_{\alpha}}^{\mathrm{PT}} \doteq
\big\langle
\phi_x\phi_y
\bigl(
I_x^{\star_+}I_y^{\star_-}
+
I_x^{\star_-}I_y^{\star_+}
\bigr)
\big\rangle_{\mu_\alpha}.
\]

\item \textbf{EUT (energetic unfavorable term).} At least one of the two fields lies in an \textit{energetically unfavorable region}:
\[\begin{split} 
\left \langle \phi_x \phi_y 
\right \rangle_{\mu_{\alpha}}^{\mathrm{EUT}}  &\doteq
\Big\langle \phi_x \phi_y 
\bigl(I_x^{+K_{\alpha}}+I_x^{-K_{\alpha}}+I_x^{(-C_{\alpha},C_{\alpha})}\bigr)
\Big\rangle_{\mu_{\alpha}}\\
&+
\Big\langle \phi_x \phi_y 
\bigl( I_x^{\star_+}+ I_x^{\star_-} \bigr)
\bigl(I_y^{+K_{\alpha}}+I_y^{-K_{\alpha}}+I_y^{(-C_{\alpha},C_{\alpha})}\bigr)
\Big\rangle_{\mu_{\alpha}}.
\end{split}\]
\end{itemize}

First, it is convenient to perform a preliminary bound of the two-point function as: 

\begin{equation}
\label{eqn:50}
\left \langle \phi_x \phi_y \right \rangle_{\mu_{\alpha}} \geq \left \langle \phi_x \phi_y \right \rangle_{\mu_{\alpha}}^{\mathrm{DT}} - \bigg|\left \langle \phi_x \phi_y \right \rangle_{\mu_{\alpha}}^{\mathrm{PT}}\bigg| - \bigg|\left \langle \phi_x \phi_y \right \rangle_{\mu_{\alpha}}^{\mathrm{EUT}}\bigg|.
\end{equation}

 The bound \eqref{eq:mainbound2} is then obtained by estimating separately the three terms in the right--hand side of \eqref{eqn:50}.

\paragraph{Dominant Term (DT).}

We shall derive the following lower bound:
    
\begin{equation}
\label{eqn:50a}
\begin{aligned}
\left \langle \phi_x \phi_y 
\right \rangle_{\mu_{\alpha}}^{\mathrm{DT}}  \geq C_{\alpha}^2 -\frac{4 C_{\alpha}^2}{K_{\alpha}^2} \left \langle \phi_x^2 I^{+K_{\alpha}}_x \right \rangle_{\mu_{\alpha}} -  \frac{8C_{\alpha}^2}{K_{\alpha}} \left \langle \phi_x^2 I^{+K_{\alpha}}_x \right \rangle_{\mu_{\alpha}}^{\frac{1}{2}} -3C_{\alpha}^2\left \langle I^{(-C_{\alpha},C_{\alpha})}_x\right \rangle_{\mu_{\alpha}} - 2C_{\alpha}^2 \left \langle I_x^{\star_+} I_y^{\star_-} \right \rangle_{\mu_{\alpha}}, 
\end{aligned}
\end{equation}

where notice that, due to the translation invariance of the model (see Footnote \ref{footnote:symmetries}), averages of the form $\langle f(\phi_x) \rangle_{\mu_{\a}}$ are actually constant with respect to $x\in \L^{\a}_L$.

\begin{proof}[Proof of \eqref{eqn:50a}.]
Since $\phi_x\phi_y\ge C_{\alpha}^2$ on the support of $I^{\star_\pm}_x I^{\star_\pm}_y$, we obtain:
\[ \left \langle \phi_x \phi_y 
\right \rangle_{\mu_{\alpha}}^{\mathrm{DT}} =
\big\langle
\phi_x\phi_y
(I_x^{\star_+}I_y^{\star_+}+I_x^{\star_-}I_y^{\star_-})
\big\rangle_{\mu_\alpha}
\ge
C_{\alpha}^2
\big\langle
I_x^{\star_+}I_y^{\star_+}+I_x^{\star_-}I_y^{\star_-}
\big\rangle_{\mu_\alpha}.
\]

Moreover, using the identity:

\[
I_x^{\star_+}I_y^{\star_+}
+
I_x^{\star_-}I_y^{\star_-}
=
(I_x^{\star_+}+I_x^{\star_-})(I_y^{\star_+}+I_y^{\star_-})
-
(I_x^{\star_+}I_y^{\star_-}+I_x^{\star_-}I_y^{\star_+}),
\]

we find that

\begin{equation}
\label{eqn:8}
\begin{aligned}
\left \langle \phi_x \phi_y 
\right \rangle_{\mu_{\alpha}}^{\mathrm{DT}}
&\ge C_{\alpha}^2\left \langle \big(I^{\star_+}_x I^{\star_+}_y + I^{\star_-}_x I^{\star_-}_y\big)\right \rangle_{\mu_{\alpha}}\\
&\ge C_{\alpha}^2\left \langle (I_x^{\star_+} + I^{\star_-}_x)(I_y^{\star_+} + I_y^{\star_-})\right \rangle_{\mu_{\alpha}}- 2C_{\alpha}^2 \left \langle  I_x^{\star_+} I_y^{\star_-} \right \rangle_{\mu_{\alpha}}\\
&= C_{\alpha}^2 \left \langle \big(1_x-I^{+K_{\alpha}}_x - I^{-K_{\alpha}}_x - I_x^{(-C_{\alpha}, C_{\alpha})}\big)\big( 1_y-I^{+K_{\alpha}}_y - I^{-K_{\alpha}}_y - I_y^{(-C_{\alpha}, C_{\alpha})}\big) \right \rangle_{\mu_{\alpha}} - 2C_{\alpha}^2 \left \langle I_x^{\star_+} I_y^{\star_-} \right \rangle_{\mu_{\alpha}}  \\
&= C_{\alpha}^2 + C_{\alpha}^2 \left \langle \big(I^{+K_{\alpha}}_x + I^{-K_{\alpha}}_x + I_x^{(-C_{\alpha}, C_{\alpha})}\big)\big( I^{+K_{\alpha}}_y + I^{-K_{\alpha}}_y + I_y^{(-C_{\alpha}, C_{\alpha})}\big) \right \rangle_{\mu_{\alpha}} \\
&\,\,\,\,\,\,\,\,\,\,\,- 2C_{\alpha}^2 \left \langle  \big(I^{+K_{\alpha}}_x + I^{-K_{\alpha}}_x + I_x^{(-C_{\alpha}, C_{\alpha})}\big)\right \rangle_{\mu_{\alpha}}- 2C_{\alpha}^2 \left \langle I_x^{\star_+} I_y^{\star_-} \right \rangle_{\mu_{\alpha}} \\
&\geq C_{\alpha}^2 - C_{\alpha}^2 \left \langle \big(I^{+K_{\alpha}}_x + I^{-K_{\alpha}}_x + I_x^{(-C_{\alpha}, C_{\alpha})}\big)\big( I^{+K_{\alpha}}_y + I^{-K_{\alpha}}_y + I_y^{(-C_{\alpha}, C_{\alpha})}\big) \right \rangle_{\mu_{\alpha}} \\
&\,\,\,\,\,\,\,\,\,\,\,- 2C_{\alpha}^2 \left \langle  \big(I^{+K_{\alpha}}_x + I^{-K_{\alpha}}_x + I_x^{(-C_{\alpha}, C_{\alpha})}\big)\right \rangle_{\mu_{\alpha}}- 2C_{\alpha}^2 \left \langle I_x^{\star_+} I_y^{\star_-} \right \rangle_{\mu_{\alpha}}.
\end{aligned}
\end{equation}

This estimate yields a strictly positive lower bound on the dominant contribution, up to additional subdominant terms of the form PT and EUT, which can be controlled as follows, using the Cauchy-Schwartz inequality and the invariance of $\mu_{\a}$ under the transformation $\phi\mapsto -\phi$.

\begin{equation*}
\begin{aligned}
&\left \langle \big(I^{+K_{\alpha}}_x + I^{-K_{\alpha}}_x + I_x^{(-C_{\alpha}, C_{\alpha})}\big)\big( I^{+K_{\alpha}}_y + I^{-K_{\alpha}}_y + I_y^{(-C_{\alpha}, C_{\alpha})}\big) \right \rangle_{\mu_{\alpha}} \\
&= \left \langle \big(I^{+K_{\alpha}}_x + I^{-K_{\alpha}}_x\big)\big( I^{+K_{\alpha}}_y + I^{-K_{\alpha}}_y \big) \right \rangle_{\mu_{\alpha}} +  \left \langle  I_x^{(-C_{\alpha}, C_{\alpha})} I_y^{(-C_{\alpha}, C_{\alpha})} \right \rangle_{\mu_{\alpha}} +\left \langle \big(I^{+K_{\alpha}}_x + I^{-K_{\alpha}}_x) I_y^{(-C_{\alpha}, C_{\alpha})} \right \rangle_{\mu_{\alpha}} \\
&+ \left \langle \ I_x^{(-C_{\alpha}, C_{\alpha})} \big(I^{+K_{\alpha}}_y + I^{-K_{\alpha}}_y\big) \right \rangle_{\mu_{\alpha}}   \\
&\leq 2 \left \langle I^{+K_{\alpha}}_x I^{+K_{\alpha}}_y\right \rangle_{\mu_{\alpha}} +  2 \left \langle I^{+K_{\alpha}}_x I^{-K_{\alpha}}_y\right \rangle_{\mu_{\alpha}} + 4   \left \langle I^{+K_{\alpha}}_x\right \rangle_{\mu_{\alpha}} + \left \langle I^{(-C_{\alpha},C_{\alpha})}_x\right \rangle_{\mu_{\alpha}}\\
&\leq \frac{2}{K_{\alpha}^2} \left \langle|\phi_x| |\phi_y| (I^{+K_{\alpha}}_x I^{+K_{\alpha}}_y+ I^{+K_{\alpha}}_x I^{-K_{\alpha}}_y)\right \rangle_{\mu_{\alpha}} + \frac{4}{K_{\alpha}}   \left \langle \phi_x I^{+K_{\alpha}}_x\right \rangle_{\mu_{\alpha}}
        + \left \langle I^{(-C_{\alpha},C_{\alpha})}_x\right \rangle_{\mu_{\alpha}}\\
        &\leq \frac{4}{K_{\alpha}^2} \left \langle \phi_x^2 I^{+K_{\alpha}}_x \right \rangle_{\mu_{\alpha}} +  \frac{4}{K_{\alpha}} \left \langle \phi_x^2 I^{+K_{\alpha}}_x \right \rangle_{\mu_{\alpha}}^{\frac{1}{2}} + \left \langle I^{(-C_{\alpha},C_{\alpha})}_x\right \rangle_{\mu_{\alpha}},\\
    \end{aligned}
\end{equation*}

and

\begin{equation*}
\begin{aligned}
        \left \langle  \big(I^{+K_{\alpha}}_x + I^{-K_{\alpha}}_x + I_x^{(-C_{\alpha}, C_{\alpha})}\big)\right \rangle_{\mu_{\alpha}} &\leq \frac{2}{K_{\alpha}} \left \langle \phi_x I_x^{+K_{\alpha}} \right \rangle_{\mu_{\alpha}} + \left \langle I^{(-C_{\alpha},C_{\alpha})}_x\right \rangle_{\mu_{\alpha}} \leq \frac{2}{K_{\alpha}} \left \langle \phi_x^2 I_x^{+K_{\alpha}} \right \rangle_{\mu_{\alpha}}^{\frac{1}{2}} + \left \langle I^{(-C_{\alpha},C_{\alpha})}_x\right \rangle_{\mu_{\alpha}}.\\
    \end{aligned}
\end{equation*}

Collecting together all the contributions above, \eqref{eqn:50a} follows.    
\end{proof}

\paragraph{Energetically Unfavorable Terms (EUT).}

We shall prove the following upper bound:
\begin{equation}
\label{eqn:50b}
\begin{aligned}
\bigg| \left \langle \phi_x \phi_y \right \rangle_{\mu_{\alpha}}^{\mathrm{EUT}} \bigg| \leq 4 \left \langle \phi_x^2 I_x^{+K_{\alpha}}  \right \rangle_{\mu_{\alpha}} + 4(C_{\alpha}+2K_{\alpha})  \left \langle \phi_x^2 I_x^{+K_{\alpha}}  \right \rangle_{\mu_{\alpha}}^{\frac{1}{2}} + (4K_{\alpha}C_{\alpha} + C_{\alpha}^2) \left \langle I_x^{(-C_{\alpha},C_{\alpha})}  \right \rangle_{\mu_{\alpha}}.   
\end{aligned}
\end{equation}

\begin{proof}[Proof of \eqref{eqn:50b}.]
 
We can rewrite

\begin{equation*}
    \left \langle \phi_x \phi_y \right \rangle_{\mu_{\alpha}}^{\mathrm{EUT}} =  \left \langle \phi_x \phi_y \right \rangle_{\mu_{\alpha}}^{\mathrm{K,K}} + \left \langle \phi_x \phi_y \right \rangle_{\mu_{\alpha}}^{\mathrm{K,C}} + \left \langle \phi_x \phi_y \right \rangle_{\mu_{\alpha}}^{\mathrm{K}, \star} + \left \langle \phi_x \phi_y \right \rangle_{\mu_{\alpha}}^{\mathrm{C}, \star} + \left \langle \phi_x \phi_y \right \rangle_{\mu_{\alpha}}^{\mathrm{C}, \mathrm{C}}, 
\end{equation*}

where:

\begin{equation}
\begin{aligned}
\left \langle \phi_x \phi_y \right \rangle_{\mu_{\alpha}}^{\mathrm{K,K}}  &\doteq \left \langle \phi_x \phi_y (I^{+K_{\alpha}}_x + I^{-K_{\alpha}}_x)(I^{+K_{\alpha}}_y+ I^{-K_{\alpha}}_y)\right \rangle_{\mu_{\alpha}}, \\
\left \langle \phi_x \phi_y \right \rangle_{\mu_{\alpha}}^{\mathrm{K,C}} &\doteq \left \langle \phi_x \phi_y  \big[ (I^{+K_{\alpha}}_x + I^{-K_{\alpha}}_x) I^{(-C_{\alpha},C_{\alpha})}_y + I^{(-C_{\alpha},C_{\alpha})}_x (I^{+K_{\alpha}}_y + I^{-K_{\alpha}}_y) \big]\right \rangle_{\mu_{\alpha}},  \\
\left \langle \phi_x \phi_y \right \rangle_{\mu_{\alpha}}^{\mathrm{K},\star}  &\doteq  \left \langle \phi_x \phi_y \big[ (I^{+K_{\alpha}}_x + I^{-K_{\alpha}}_x)( I^{\star_+}_y + I^{\star_-}_y) + ( I^{\star_+}_x + I^{\star_-}_x) (I^{+K_{\alpha}}_y + I^{-K_{\alpha}}_y) \big] \right \rangle_{\mu_{\alpha}}, \\
\left \langle \phi_x \phi_y \right \rangle_{\mu_{\alpha}}^{\mathrm{C},\star}  &\doteq \left \langle \phi_x \phi_y \big[ I^{(-C_{\alpha},C_{\alpha})}_x (I^{\star_+}_y + I^{\star_-}_y)+ (I^{\star_+}_x + I^{\star_-}_x) I^{(-C_{\alpha},C_{\alpha})}_y  \big]\right \rangle_{\mu_{\alpha}},\\
\left \langle \phi_x \phi_y \right \rangle_{\mu_{\alpha}}^{\mathrm{C,C}} &\doteq \left \langle \phi_x \phi_y I^{(-C_{\alpha},C_{\alpha})}_x I^{(-C_{\alpha},C_{\alpha})}_y\right \rangle_{\mu_{\alpha}}. 
\end{aligned}
\end{equation}

We are going to bound the several terms above separately.

\begin{enumerate}
\item \textbf{Bound on K,K term.} Using the Cauchy-Schwartz inequality and the invariance of $\mu_{\a}$ under the transformation $\phi\mapsto -\phi$, we get:
\begin{equation*}
\begin{aligned}
\sup_{x,y\in\L^{\a}_L} \bigg|\left \langle \phi_x \phi_y \right \rangle_{\mu_{\alpha}}^{\mathrm{K,K}}\bigg| &\leq 2 \sup_{x,y\in\L^{\a}_L} \left \langle \phi_x \phi_y I_x^{+K_{\alpha}} I_y^{+K_{\alpha}}  \right \rangle_{\mu_{\alpha}}  + 2 \sup_{x,y\in\L^{\a}_L} \left \langle |\phi_x \phi_y| I_x^{+K_{\alpha}} I_y^{-K_{\alpha}} \right \rangle_{\mu_{\alpha}} 
\\
&\leq 4 \sup_{x\in\L^{\a}_L} \left \langle \phi_x^2 I_x^{+K_{\alpha}}  \right \rangle_{\mu_{\alpha}}. 
\end{aligned}
\end{equation*}
\item \textbf{Bound on K,C term.} Similarly to the previous item:
\begin{equation*}
\begin{aligned}
\sup_{x,y\in\L^{\a}_L} \bigg|\left \langle \phi_x \phi_y \right \rangle_{\mu_{\alpha}}^{\mathrm{K,C}}\bigg| &\leq  4 \sup_{x,y\in\L^{\a}_L} \left \langle |\phi_x| |\phi_y| I_x^{+K_{\alpha}} I_y^{(-C_{\alpha},C_{\alpha})}  \right \rangle_{\mu_{\alpha}} 
\\
&\leq 4 C_{\alpha} \sup_{x,y\in\L^{\a}_L} \left \langle \phi_x^2 I_x^{+K_{\alpha}}  \right \rangle_{\mu_{\alpha}}^{\frac{1}{2}} \left \langle I_y^{(-C_{\alpha},C_{\alpha})} \right \rangle_{\mu_{\alpha}}^{\frac{1}{2}}\leq 4 C_{\alpha} \sup_{x\in\L^{\a}_L}\left \langle \phi_x^2 I_x^{+K_{\alpha}}  \right \rangle_{\mu_{\alpha}}^{\frac{1}{2}}. 
\end{aligned}
\end{equation*}

\item \textbf{Bound on K,$\star$ term}:
 \begin{equation*}
\begin{aligned}
\sup_{x,y\in\L^{\a}_L} \bigg|\left \langle \phi_x \phi_y \right \rangle_{\mu_{\alpha}}^{\mathrm{K,}\star}\bigg| &\leq 4 \sup_{x,y\in\L^{\a}_L} \left \langle \phi_x \phi_y I_x^{+K_{\alpha}} I_y^{\star_+}  \right \rangle_{\mu_{\alpha}} + 4 \sup_{x,y\in\L^{\a}_L} \left \langle |\phi_x \phi_y| I_x^{+K_{\alpha}} I_y^{\star_-}  \right \rangle_{\mu_{\alpha}}\\
& \leq 8K_{\alpha} \sup_{x\in\L^{\a}_L} \left \langle \phi_x^2 I_x^{+K_{\alpha}}  \right \rangle_{\mu_{\alpha}}^{\frac{1}{2}}.
\end{aligned}
\end{equation*}
\item \textbf{Bound on C,$\star$ term}:
\begin{equation*}
\begin{aligned}
\sup_{x,y\in\L^{\a}_L} \bigg|\left \langle \phi_x \phi_y \right \rangle_{\mu_{\alpha}}^{\mathrm{C,}\star}\bigg| \leq 4 \left \langle |\phi_x \phi_y| I_x^{(-C_{\alpha},C_{\alpha})} I_y^{\star_+}  \right \rangle_{\mu_{\alpha}} 
\leq 4K_{\alpha} C_{\alpha} \sup_{x\in \L^{\a}_L} \left \langle  I_x^{(-C_{\alpha},C_{\alpha})}  \right \rangle_{\mu_{\alpha}}. 
\end{aligned}
\end{equation*}
\item \textbf{Bound on C,C term}:
\begin{equation*}
\begin{aligned}
\sup_{x,y\in\L^{\a}_L} \bigg|\left \langle \phi_x \phi_y \right \rangle_{\mu_{\alpha}}^{\mathrm{C,C}}\bigg| \leq \sup_{x,y\in\L^{\a}_L} \left \langle |\phi_x \phi_y| I_x^{(-C_{\alpha},C_{\alpha})} I_y^{(-C_{\alpha},C_{\alpha})}  \right \rangle_{\mu_{\alpha}} 
\leq C_{\alpha}^2 \sup_{x\in\L^{\a}_L} \left \langle  I_x^{(-C_{\alpha},C_{\alpha})}  \right \rangle_{\mu_{\alpha}}.
\end{aligned}
\end{equation*}
\end{enumerate}

Collecting all the contributions, we find the claimed inequality \eqref{eqn:50b}.
   
\end{proof}

\paragraph{Peierls' terms (PT).} Using again the invariance of $\mu_{\a}$ under $\phi\mapsto-\phi$, we get that

\begin{equation}
\label{eqn:50c}
\bigg|\left \langle \phi_x \phi_y 
\right \rangle_{\mu_{\alpha}}^{\mathrm{PT}} \bigg| \leq 2
\big\langle
|\phi_x \phi_y| I_x^{\star_+}I_y^{\star_-}
\big\rangle_{\mu_\alpha} \leq 2 K_{\alpha}^2 \big\langle
I_x^{\star_+}I_y^{\star_-}
\big\rangle_{\mu_\alpha}.
 \end{equation}

 Plugging \eqref{eqn:50a}, \eqref{eqn:50b} and \eqref{eqn:50c} into the right--hand side of \eqref{eqn:50}, we finally get:

\begin{equation*}
\begin{split}
\inf_{x,y\in\L^{\alpha}_L} \Big\langle \phi_x \phi_y \Big\rangle_{\mu_{\alpha}} &\ge C_{\alpha}^2 - \Bigg( 2(K_{\alpha}^2+C_{\alpha}^2) \sup_{x,y\in\L^{\alpha}_L} \left\langle I^{\star_+}_x I^{\star_-}_y \right\rangle_{\mu_{\alpha}} +  4C_{\alpha} (C_{\a}+ K_{\alpha}) \sup_{x\in\L^{\alpha}_L} \left\langle I^{(-C_{\alpha},C_{\alpha})}_x \right\rangle_{\mu_{\alpha}} \\
&+ 4\big(1+ \tfrac{C_{\alpha}^2}{K_{\alpha}^2}\big)\sup_{x\in\L^{\alpha}_L} \left\langle \phi_x^2 I^{+K_{\alpha}}_x \right\rangle_{\mu_{\alpha}} + 4\big(C_{\alpha}+ 2\tfrac{C_{\alpha}^2}{K_{\alpha}}+ 2K_{\alpha}\big) \sup_{x\in\L^{\alpha}_L} \left\langle \phi_x^2  I^{+K_{\alpha}}_x\right\rangle_{\mu_{\alpha}}^{\frac{1}{2}}  \Bigg),
\end{split}
\end{equation*}

which is nothing but \eqref{eq:mainbound2}.

\subsection{Upper bound of the $2$-point function} 
In order to perform an upper bound, we can decompose the two-point function as:

\begin{equation*}
    \left \langle \phi_x \phi_y \right \rangle_{\mu_{\alpha}} = \left \langle \phi_x \phi_y (I_x^{(-K_{\alpha},K_{\alpha})} + I_x^{(-K_{\alpha},K_{\alpha})^c}) (I_y^{(-K_{\alpha},K_{\alpha})} + I_y^{(-K_{\alpha},K_{\alpha})^c})\right \rangle_{\mu_{\alpha}},
\end{equation*}

where $I_x^{(-K_{\alpha},K_{\alpha})} \doteq \mathbbl{1}_{\{-K_{\a}<\phi_x <K_{\a}\}}$ and $I_x^{(-K_{\alpha},K_{\alpha})^c} \doteq 1- I_x^{(-K_{\alpha},K_{\alpha})}$. Proceeding as in the previous subsection, we can decompose the correlator as:

\begin{equation}
\label{eq:upb}
\begin{aligned}
    \left \langle \phi_x \phi_y \right \rangle_{\mu_{\alpha}} &=  \left \langle \phi_x \phi_y I_x^{(-K_{\alpha},K_{\alpha})}  I_y^{(-K_{\alpha},K_{\alpha})}\right \rangle_{\mu_{\alpha}}  + 2 \left \langle \phi_x \phi_y I_x^{(-K_{\alpha},K_{\alpha})} I_y^{(-K_{\alpha},K_{\alpha})^c}\right \rangle_{\mu_{\alpha}}  + \\
    &\,\,\,\,\,\,\,\,\,\,\,\,\,\,\,\,\,\,\,\,\,\,\,\,\,\,+\left \langle \phi_x \phi_y I_x^{(-K_{\alpha},K_{\alpha})^c} I_y^{(-K_{\alpha},K_{\alpha})^c}\right \rangle_{\mu_{\alpha}}.
    \end{aligned}
\end{equation}

The three terms appearing in the right--hand side of \eqref{eq:upb} can be upper-bounded as follows.

\begin{equation*}
\begin{aligned}
\bigg|\left \langle \phi_x \phi_y I_x^{(-K_{\alpha},K_{\alpha})}  I_y^{(-K_{\alpha},K_{\alpha})}\right \rangle_{\mu_{\alpha}}\bigg| &\leq \left \langle |\phi_x| |\phi_y| I_x^{(-K_{\alpha},K_{\alpha})}  I_y^{(-K_{\alpha},K_{\alpha})}\right \rangle_{\mu_{\alpha}} \leq \left \langle I_x^{(-K_{\alpha},K_{\alpha})}  I_y^{(-K_{\alpha},K_{\alpha})}\right \rangle_{\mu_{\alpha}} \leq K_{\alpha}^2;\\
\bigg|\left \langle \phi_x \phi_y I_x^{(-K_{\alpha},K_{\alpha})} I_y^{(-K_{\alpha},K_{\alpha})^c}\right \rangle_{\mu_{\alpha}} \bigg| &\leq  \left \langle |\phi_x| |\phi_y| I_x^{(-K_{\alpha},K_{\alpha})} I_y^{(-K_{\alpha},K_{\alpha})^c}\right \rangle_{\mu_{\alpha}} \leq K_{\alpha} \left \langle  |\phi_y| I_y^{(-K_{\alpha},K_{\alpha})^c}\right \rangle_{\mu_{\alpha}}  \\
&= 2 K_{\alpha} \left \langle  \phi_x I_x^{+K_{\alpha}}\right \rangle_{\mu_{\alpha}}\leq 2K_{\alpha}\left \langle  \phi_x^2 I_x^{+K_{\alpha}}\right \rangle_{\mu_{\alpha}}^{\frac{1}{2}};\\
\bigg|\left \langle \phi_x \phi_y I_x^{(-K_{\alpha},K_{\alpha})^c} I_y^{(-K_{\alpha},K_{\alpha})^c}\right \rangle_{\mu_{\alpha}} \bigg| &\leq \left \langle |\phi_x| |\phi_y| I_x^{(-K_{\alpha},K_{\alpha})^c} I_y^{(-K_{\alpha},K_{\alpha})^c}\right \rangle_{\mu_{\alpha}}\\
&= 2 \left \langle \phi_x \phi_y I_x^{+K_{\alpha}} I_y^{+K_{\alpha}}\right \rangle_{\mu_{\alpha}} + 2 \left \langle |\phi_x \phi_y| I_x^{-K_{\alpha}} I_y^{+K_{\alpha}}\right \rangle_{\mu_{\alpha}}  \\
&\leq 2 \left \langle \phi_x^2 I_x^{+K_{\alpha}} \right \rangle_{\mu_{\alpha}}^{\frac{1}{2}} \left \langle \phi_y^2 I_y^{+K_{\alpha}} \right \rangle_{\mu_{\alpha}}^{\frac{1}{2}}  + 2 \left \langle \phi_x^2 I_x^{-K_{\alpha}}\right \rangle_{\mu_{\alpha}}^{\frac{1}{2}} \left \langle \phi_y^2  I_y^{+K_{\alpha}}\right \rangle_{\mu_{\alpha}}^{\frac{1}{2}}\\
&= 4 \left \langle \phi_x^2 I_x^{+K_{\alpha}} \right \rangle_{\mu_{\alpha}}^{\frac{1}{2}} \left \langle \phi_y^2 I_y^{+K_{\alpha}} \right \rangle_{\mu_{\alpha}}^{\frac{1}{2}}.
\end{aligned}
\end{equation*}

Therefore, going back to \eqref{eq:upb}:

\begin{equation*}
    \sup_{x,y \in \Lambda^{\alpha}_L} \left \langle \phi_x \phi_y \right \rangle_{\mu_{\alpha}} \leq  \sup_{x,y \in \Lambda^{\alpha}_L} \bigg|\left \langle \phi_x \phi_y \right \rangle_{\mu_{\alpha}} \bigg|\leq K_{\alpha}^2 + 2K_{\alpha}\sup_{x \in \Lambda_L^{\alpha}} \left \langle \phi_x^2 I_x^{+K_{\alpha}} \right \rangle_{\mu_{\alpha}}^{\frac{1}{2}} + 4 \sup_{x \in \Lambda_L^{\alpha}} \left \langle \phi_x^2 I_x^{+K_{\alpha}} \right \rangle_{\mu_{\alpha}},
\end{equation*}

which is nothing but \eqref{upbound}.

\printbibliography
\end{document}